\documentclass[12pt]{article}
\usepackage{eurosym}
\usepackage{inputenc}
\usepackage{geometry}
\usepackage{amsmath}
\usepackage{amsfonts}
\usepackage{amssymb}
\usepackage[nodisplayskipstretch]{setspace}
\usepackage{array}
\usepackage{rotating}
\usepackage{version}
\usepackage{graphicx}
\usepackage{url}
\usepackage{color}
\usepackage{float}
\usepackage{caption}
\usepackage{hyperref}
\usepackage[authoryear]{natbib}

\hypersetup{pdfborder={0 0 0}}
\graphicspath{{C:/TMGEfig/}}
\makeatletter
\renewcommand\section{\@startsection {section}{1}{\z@}                                   
                                      {-1.5ex \@plus -1ex \@minus -.2ex}                                   
                                      {0.8ex \@plus.2ex}                                   
                                      {\normalfont\large\bfseries}}
\renewcommand\subsection{\@startsection{subsection}{2}{\z@}                                     
                                             {-1.5ex\@plus -1ex \@minus -.2ex}                                     
                                             {0.5ex \@plus .2ex}                                     
                                             {\normalfont\normalsize\bfseries}}
\renewcommand\subsubsection{\@startsection{subsubsection}{3}{\z@}                                     
                                                   {-1ex\@plus -1ex \@minus -.2ex}                                     
                                                   {0.5ex \@plus .1ex}                                     
                                                   {\normalfont\normalsize\bfseries}}
\makeatother
\numberwithin{equation}{section}
\newtheorem{theorem}{Theorem}

\newtheorem{assumption}{Assumption}

\newtheorem{example}{Example}

\newtheorem{lemma}{Lemma}

\newtheorem{proposition}{Proposition}
\newtheorem{remark}{Remark}

\newenvironment{proof}[1][Proof]{\noindent\textbf{#1.} }{\ \rule{0.5em}{0.5em}}
\makeatletter
\g@addto@macro\normalsize{
  \setlength\abovedisplayskip{1.35ex}
  \setlength\belowdisplayskip{1.35ex}
  \setlength\abovedisplayshortskip{1.15ex}
  \setlength\belowdisplayshortskip{1.15ex}
}
\makeatother
\input{tcilatex}
\begin{document}

\title{{\large {Estimation of Average Effects in Short $T$ Heterogeneous
Panels\thanks{{\small We are grateful to Cheng Hsiao, Oliver Linton, Ron
Smith and Hayun Song for helpful comments on earlier versions of this paper.
We have also benefited from constructive comments and suggestions by Jiti
Gao, who was the discussant of our paper at the 2025 Workshop in Honour of
Professors Heather Anderson and Farshid Vahid at Monash University, two
anonymous referees and an Associate Editor. We thank Faraz Mujtaba for his
help with double-checking the R codes used to carry out the computations.
The research on this paper was started when Liying Yang was a Ph.D. student
at the Department of Economics, University of Southern California.}}}} }
\author{ {\normalsize M. Hashem Pesaran\thanks{{\small Department of
Economics, University of Southern California, and Trinity College,
University of Cambridge (pesaran@usc.edu).}}} \and {\normalsize Liying Yang%
\thanks{{\small Shenzhen Audencia Financial Technology Institute and WeBank
Institute of FinTech, Shenzhen University (lyang@szu.edu.cn). }}}}
\date{{\normalsize \today}}
\maketitle

\vspace{-9mm}
\begin{abstract}
{\small The commonly used two-way fixed effects estimator is biased under
correlated heterogeneity and can lead to misleading inference. The mean
group estimator proposed by \cite{PesaranSmith1995} is robust to correlated
heterogeneity but requires the individual estimators to have second-order
moments that could fail if the number of estimated coefficients ($k$) 
is close to the time dimension ($T$) of the panel. This paper
focuses on panels where $T-k$ is mall (including $T=k$), and proposes a trimmed mean group (TMG) estimator
that shrinks individual estimators most likely to fail the second-order
moment condition. The TMG estimator is shown to be $n^{(1-\alpha
)/2}$-consistent and asymptotically normally distributed, where $%
\alpha$ is determined by the degree to which individual estimators
might not have moments. The $\sqrt{n}$ convergence rate is achieved
only if all individual estimators have second-order moments. Extensions to
panels with time effects are provided, and a new Hausman test of correlated
heterogeneity is proposed. Small sample properties of the TMG estimator
(with and without time effects) are investigated by Monte Carlo experiments
and shown to be satisfactory. The proposed test of correlated heterogeneity
is also shown to have the correct size and satisfactory power. The utility
of the TMG approach is illustrated with an empirical application.}


\end{abstract}

\thispagestyle{empty}

\newpage \setcounter{page}{1}

\section{Introduction}

\doublespacing%

Fixed effects estimation of average effects has been predominantly utilized
for program and policy evaluation. For static panel data models where slope
heterogeneity is uncorrelated with regressors, two-way fixed effects (TWFE)
estimators that allow for unit-specific and time effects are $\sqrt{n}$%
-consistent, and if used in conjunction with robust standard errors, lead to
valid inference in panels where the time dimension, $T$, is short and the
cross section dimension, $n$, is sufficiently large. However, when the slope
heterogeneity is correlated with the regressors, the TWFE estimators could
become inconsistent even if both $n$ and $T\rightarrow \infty$.\footnote{%
The concept of the correlated random coefficient model is due to \cite%
{HeckmanVytlacil1998}. \cite{Wooldridge2005} shows that TWFE estimators
continue to be consistent if slope heterogeneity is mean-independent of all
the de-trended covariates. See also condition (\ref{FEconsistency}) given
below.} In the case of static panels with strictly exogenous regressors,
correlated (slope) heterogeneity can arise, for example, when there is a
high degree of variation in treatments across units or when there are latent
factors that influence the level and variability of treatments and their
outcomes. For example, in estimation of returns to education, the choice of
educational level is likely to be correlated with expected returns to
education. Other examples include estimation of the effects of training
programs on workers' productivity and earnings reviewed by \cite%
{CreponVandenberg2016}, evaluation of the effectiveness of micro-credit
programs discussed by \cite{BanerjeeEtal2015}, and the analysis of the
effectiveness of anti-poverty cash transfer programs considered by \cite%
{BastagliEtal2019}.

In the presence of correlated heterogeneity, \cite{PesaranSmith1995}
proposed to estimate the mean effects by simple averages of individual
estimates, which they called the mean group (MG) estimator. They showed that
the MG estimator is consistent for dynamic heterogeneous panels when $n$ and 
$T$ are both large. It was later shown that for panels with strictly
exogenous regressors, the MG estimator is in fact $\sqrt{n}$-consistent in
the presence of correlated heterogeneity even if $T$ is fixed as $%
n\rightarrow \infty $, so long as $T$ is sufficiently large such that
second-order moments of the individual estimators exist. However, such
moment conditions need not hold when $T$ is very close to the number of
estimated coefficients ($k$). In effect, we are faced with the problem of
estimating the mean of random variables with fat tails.

This problem was originally recognized by \cite{Chamberlain1992}, who showed
that one needs $T$ to be strictly larger than $k$ for regular identification
of average effects under correlated heterogeneity.\footnote{%
An unknown parameter, $\beta _{0}$, is said to be regularly identified if
there exists an estimator that converges to $\beta _{0}$ in probability at
the rate of $\sqrt{n}$. Any estimator that converges to its true value at a
rate slower than $\sqrt{n}$ is said to be irregularly identified.} In the
statistics literature, \cite{CsorgoEtal1988b,CsorgoEtal1988a} and \cite%
{GriffinPruitt1989}, among others, consider a trimmed mean estimator whereby
the observations are ordered, and those below and above a given threshold
value are excluded. However, in general such trimmed estimators may not
possess a limiting distribution if the observations are drawn from a
heavy-tailed distribution with the tail index, $\alpha _{p}<2$, and the
threshold values are chosen in an \textit{ad hoc} manner. \cite{Peng2001}
proposes a new trimmed estimator where the thresholds are endogenized and
the trimmed mean is augmented with mean estimates from the two tails. Peng's
modified trimmed estimator is shown to have the same limiting normal
distribution as the sample mean but with a slower convergence rate when $%
\alpha _{p}<2$ (see Theorem 1 and Remark 1 of \cite{Peng2001}). For settings
of inverse probability weighting where the denominator can be arbitrarily
close to zero, leading to heavy-tailed sampling distributions, \cite%
{MaWang2020} propose trimmed estimators that exclude observations with small
denominators and develop inference procedures for different trimming
threshold choices. For our purposes, trimmed estimators proposed by \cite%
{Peng2001} and \cite{MaWang2020} are subject to two limitations. They
consider only the scalar case and assume that the observations (in our
application, the individual estimates) are identically and independently
distributed. Extension of their estimators to a vector of estimators that
are not identically distributed does not seem to be straightforward.

In the context of MG estimation, we have additional information about the
precision of the individual estimates that is not used in the trimming
approaches considered in the statistics literature. One example where such
information is utilized is provided by \cite{GrahamPowell2012} (GP), who
build on the pioneering work of \cite{Chamberlain1992}. GP focus on panels
with $T=k$, where identification issues of time effects and the mean
coefficients arise especially when there are insufficient within-individual
variations for some regressors. These authors derive an irregular estimator
of the mean coefficients by excluding individual estimates from the
estimation of the average effects if the sample variance of regressors in
question is smaller than a given threshold value.

In this paper, we begin by providing conditions under which MG and fixed
effects (FE) estimators of the average effects in heterogeneous panels are $%
\sqrt{n}$-consistent, which serves as a basis for developing a diagnostic
test of the validity of (two-way) fixed effects estimators commonly used in
the literature. We then propose a trimmed mean group (TMG) estimator that
does not exclude any of the individual estimates, but uses a threshold
function similar to that of GP to shrink some of the estimates to overcome
the fat-tailed nature of the distribution of the individual estimates when $%
T $ is very close to $k$. In effect, we shrink rather than drop estimates as
done by GP. The decision on whether unit $i$ is subject to shrinkage is made
with respect to the determinant of the sample variance matrix of the
regressors, denoted by $d_{i}(>0)$. Individual estimators become
ill-conditioned when $d_{i}$ is close to zero. To characterize the trimming
process, we assume that $1/d_{i}$ follows Pareto-type distributions with the
tail index, $\alpha _{p}$, and show that individual estimators have
second-order moments when $\alpha _{p}>2$. Shrinkage of the individual
estimators is required only if $\alpha _{p}\leq 2$. The literature on
estimation of $\alpha _{p}$ is well established, for example, \cite%
{EmbrechtsEtal1997}, which can guide decisions on shrinkage/trimming before
estimating the average effects.

We also consider heterogeneous panels with time effects and show that the
TMG procedure can be applied after the time effects are eliminated. In the
case where $T=k$, we require the dependence between heterogeneous slope
coefficients and the regressors to be time-invariant. This assumption is not
required when $T>k$, and the time effects can be eliminated using the
approach first proposed by \cite{Chamberlain1992}. We refer to these
estimators as TMG-TE and derive their asymptotic distributions, bearing in
mind that the way the time effects are eliminated depends on whether $T=k$
or $T>k$.

As noted above, slope heterogeneity by itself does not render the TWFE
estimators inconsistent, which continue to have the regular convergence rate
of $\sqrt{n}$. The problem arises when slope heterogeneity is correlated
with the covariates. It is therefore important that before using the TWFE
estimators, the null of uncorrelated heterogeneity is tested. To this end,
we also propose Hausman tests of correlated heterogeneity by comparing the
FE and TWFE estimators with the associated TMG estimators and derive their
asymptotic distributions under fairly general conditions. The earlier
Hausman test of slope homogeneity developed by \cite{PesaranEtal1996} is
based on the difference between FE and MG estimators and does not apply when 
$T$ is close to $k$. The more recent dispersion-based test of slope
homogeneity proposed by \cite{PesaranYamagata2008} is shown to be quite
powerful as a test of slope homogeneity but does not distinguish between
correlated or uncorrelated heterogeneity and requires $\sqrt{n}%
/T^{2}\rightarrow 0$ as $n$ and $T\rightarrow \infty $ jointly.

We also carry out an extensive set of Monte Carlo (MC) simulations to
investigate the small sample properties of the TMG and TMG-TE estimators and
how they compare with the trimmed estimator proposed by GP. The MC evidence
on the size and empirical power of the Hausman tests of correlated
heterogeneity in panel data models with and without time effects is
provided, and the sensitivity of estimation results to the choice of the
trimming threshold parameter, $\alpha $, is also investigated. The MC and
theoretical results of the paper are all in agreement. The TMG and TMG-TE
estimators not only have the correct size but also achieve better finite
sample properties compared with the other trimmed estimator across a number
of experiments with different data generating processes, allowing for
heteroskedasticity (random and correlated), error serial correlations, and
regressors with heterogeneous dynamics and interactive effects. The
simulation results also confirm that the Hausman tests based on the
difference between FE (TWFE) and TMG (TMG-TE) estimators have the correct
size and power against the alternative of correlated heterogeneity.

Finally, we illustrate the utility of our proposed trimmed estimators by
re-examining the average effect of household expenditures on calorie demand
using a balanced panel of $1,358$ households in poor rural communities in
Nicaragua over the years 2001--2002 $(T=2)$ and 2000--2002 $(T=3)$.

The rest of the paper is organized as follows. Section \ref{HPM} sets out
the heterogeneous panel data model and discusses the asymptotic properties
of FE and MG estimators. Section \ref{IRMGE} considers ultra-short $T$
panels (including the case of $T=k$) and introduces the proposed TMG
estimator, with its asymptotic properties established in Section \ref{AsyTMG}%
. Section \ref{CRCTE} extends the TMG estimation to panels with time
effects, distinguishing between cases where $T>k$ and $T=k$. Section \ref%
{Test} sets out the Hausman test of correlated heterogeneous slope
coefficients. Section \ref{subset} discusses how to apply the TMG approach
to a subset of coefficients of interest. Section \ref{MC} provides the main
findings of the MC experiments. Section \ref{APP} presents the empirical
illustration. Section \ref{conclusion} concludes. Mathematical proofs of
propositions and theorems are given in a mathematical appendix.
Supplementary materials covering additional mathematical derivations, MC
experiments, as well as further empirical results, are provided in an online
supplement.

\textbf{Notations:} Generic positive finite constants are denoted by $C$
when large, and $c$ when small. They can take different values at different
instances. $\lambda _{\max }\left( \boldsymbol{A}\right) $ and $\lambda
_{\min }\left( \boldsymbol{A}\right) $ denote the maximum and minimum
eigenvalues of matrix $\boldsymbol{A}$. $\boldsymbol{A}\succ \boldsymbol{0}$
and $\boldsymbol{A}\succeq \boldsymbol{0}$ denote that matrix $\boldsymbol{A}
$ is positive definite and is positive semi-definite, respectively. When
matrix $\boldsymbol{A}$ is square, its adjugate (adjoint) and determinant
are denoted by $\func{adj}(\boldsymbol{A})$ and $\det (\boldsymbol{A})$,
respectively. If $\det (\boldsymbol{A})\neq 0$, then the inverse of $%
\boldsymbol{A}$ is given by $\boldsymbol{A}^{-1}=\func{adj}(\boldsymbol{A)}%
/\det (\boldsymbol{A)}$. $\left\Vert \boldsymbol{A}\right\Vert =\lambda
_{\max }^{1/2}(\boldsymbol{A}^{\prime }\boldsymbol{A)}$ and $\left\Vert 
\boldsymbol{A}\right\Vert _{1}$ denote the spectral and column norms of
matrix $\boldsymbol{A}$, respectively. $\left\Vert \boldsymbol{x}\right\Vert
_{p}=\left[ E\left( \left\Vert \boldsymbol{x}\right\Vert ^{p}\right) \right]
^{1/p}$. If $\left\{ f_{n}\right\} _{n=1}^{\infty }$ is any real sequence
and $\left\{ g_{n}\right\} _{n=1}^{\infty }$ is a sequence of positive real
numbers, then $f_{n}=O(g_{n})$ if there exists $C$ such that $\left\vert
f_{n}\right\vert /g_{n}\leq C$ for all $n$, and $f_{n}=o(g_{n})$ if $%
f_{n}/g_{n}\rightarrow 0$ as $n\rightarrow \infty $. Similarly, $%
f_{n}=O_{p}(g_{n})$ if $f_{n}/g_{n}$ is stochastically bounded, and $%
f_{n}=o_{p}(g_{n})$, if $f_{n}/g_{n}\rightarrow _{p}0$. $f_{n}=\ominus
(g_{n})$ if there exist $n_{0}\geq 1$ and positive finite constants $C_{0}$
and $C_{1}$, such that $\inf_{n\geq n_{0}}\left( \left\vert f_{n}\right\vert
/g_{n}\right) \geq C_{0}$, and $\sup_{n\geq n_{0}}\left( \left\vert
f_{n}\right\vert /g_{n}\right) \leq C_{1}$. $f_{n}=\ominus _{p}(g_{n})$ if $%
|f_{n}|/g_{n}$ is both stochastically bounded above and bounded away from
zero in probability. The operator $\rightarrow _{p}$ denotes convergence in
probability, and $\rightarrow _{d}$ denotes convergence in distribution. $%
IID $ stands for independently and identically distributed. $\boldsymbol{u}%
\perp \boldsymbol{v}$ is used to show that vectors of random variables $%
\boldsymbol{u}$ and $\boldsymbol{v}$ are independently distributed.

\section{Heterogeneous linear panel data models\label{HPM}}

Consider the following panel data model with individual fixed effects, $%
\alpha _{i}$, and heterogeneous slope coefficients, $\boldsymbol{\beta }_{i}$%
, 
\begin{equation}
y_{it}=\alpha _{i}+\boldsymbol{\beta }_{i}^{\prime }\boldsymbol{x}%
_{it}+u_{it}\text{, for }i=1,2,...,n\text{ and }t=1,2,...,T,  \label{eq2}
\end{equation}%
where $\boldsymbol{x}_{it}$ is a $k^{\prime }\times 1$ vector of regressors,
and $u_{it}$ is the error term. $\left\{ \boldsymbol{\beta }_{i}\right\}
_{i=1}^{n}$ follow the random coefficient model 
\begin{equation}
\boldsymbol{\beta }_{i}=\boldsymbol{\beta }_{0}+\boldsymbol{\eta }_{i},\text{
for }i=1,2,...,n,  \label{RCM}
\end{equation}%
where $\boldsymbol{\beta }_{0}$ is the $k^{\prime }\times 1$ vector of
average effects, and $\boldsymbol{\eta }_{i}$ are independently distributed
with mean zero and bounded positive semi-definite variances. In matrix
notations, 
\begin{equation}
\boldsymbol{y}_{i}=\alpha _{i}\boldsymbol{\tau }_{T}+\boldsymbol{X}_{i}%
\boldsymbol{\beta }_{i}+\boldsymbol{u}_{i},  \label{m1}
\end{equation}%
where $\boldsymbol{y}_{i}=(y_{i1},y_{i2},...,y_{iT})^{\prime }$, $%
\boldsymbol{\tau }_{T}$ is a $T\times 1$ vector of ones, $\boldsymbol{X}%
_{i}=(\boldsymbol{x}_{i1},\boldsymbol{x}_{i2},...,\boldsymbol{x}%
_{iT})^{\prime }$, and $\boldsymbol{u}_{i}=(u_{i1},u_{i2},...,u_{iT})^{%
\prime }$. The FE estimator of $\boldsymbol{\beta }_{0}$ is given by 
\begin{equation}
\boldsymbol{\hat{\beta}}_{FE}=\left( n^{-1}\sum_{i=1}^{n}\boldsymbol{X}%
_{i}^{\prime }\boldsymbol{M}_{T}\boldsymbol{X}_{i}\right) ^{-1}\left(
n^{-1}\sum_{i=1}^{n}\boldsymbol{X}_{i}^{\prime }\boldsymbol{M}_{T}%
\boldsymbol{y}_{i}\right) ,  \label{fee}
\end{equation}%
where $\boldsymbol{M}_{T}=\boldsymbol{I}_{T}-T^{-1}\boldsymbol{\tau }_{T}%
\boldsymbol{\tau }_{T}^{\prime }$, and $\boldsymbol{I}_{T}$ is a $T\times T$
identity matrix. It is well known that for a fixed $T\geq k=k^{\prime }+1$, $%
\boldsymbol{\hat{\beta}}_{FE}$ is a $\sqrt{n}$-consistent estimator of $%
\boldsymbol{\beta }_{0}$ and robust to possible correlations between $\alpha
_{i}$ and $\boldsymbol{x}_{it}$, even under slope heterogeneity so long as $%
\boldsymbol{\eta }_{i}$ are not correlated with $\boldsymbol{X}_{i}^{\prime }%
\boldsymbol{M}_{T}\boldsymbol{X}_{i}$. This condition is clearly satisfied
when heterogeneity is exogenous and $E\left( \boldsymbol{X}_{i}^{\prime }%
\boldsymbol{M}_{T}\boldsymbol{X}_{i}\boldsymbol{\eta }_{i}\right) =%
\boldsymbol{0}$, for the majority of the units (to be formalized below).

In the presence of correlated slope heterogeneity, the MG estimator,
initially proposed by \cite{PesaranSmith1995}, is typically considered for
consistent estimation of the average effects. When $T\geq k$, $\boldsymbol{%
\beta }_{0}$ can be estimated by the MG estimator, $\boldsymbol{\hat{\beta}}%
_{MG}$, computed as a simple average of the least squares estimates of $%
\boldsymbol{\beta }_{i}$, namely%
\begin{equation}
\boldsymbol{\hat{\beta}}_{MG}=\frac{1}{n}\sum_{i=1}^{n}\boldsymbol{\hat{\beta%
}}_{i},  \label{mge}
\end{equation}%
where 
\begin{equation}
\boldsymbol{\hat{\beta}}_{i}=\left( \boldsymbol{X}_{i}^{\prime }\boldsymbol{M%
}_{T}\boldsymbol{X}_{i}\right) ^{-1}\boldsymbol{X}_{i}^{\prime }\boldsymbol{M%
}_{T}\boldsymbol{y}_{i}.  \label{betaihat}
\end{equation}%
In contrast to the FE estimator, the MG estimator does not depend on $%
E\left( \boldsymbol{X}_{i}^{\prime }\boldsymbol{M}_{T}\boldsymbol{X}_{i}%
\boldsymbol{\eta }_{i}\right) $ and is consistent irrespective of whether
slope heterogeneity is correlated or not. As pointed out by an associate
editor, $\boldsymbol{\hat{\beta}}_{FE}$ can also be written as a weighted
average of $\boldsymbol{\hat{\beta}}_{i}$, 
\begin{equation}
\boldsymbol{\hat{\beta}}_{FE}=n^{-1}\sum_{i=1}^{n}\boldsymbol{\mathcal{W}}%
_{i}\boldsymbol{\hat{\beta}}_{i},  \label{fee2}
\end{equation}%
where $\boldsymbol{\mathcal{W}}_{i}=\left( n^{-1}\sum_{j=1}^{n}\boldsymbol{X}%
_{j}^{\prime }\boldsymbol{M}_{T}\boldsymbol{X}_{j}\right) ^{-1}\left( 
\boldsymbol{X}_{i}^{\prime }\boldsymbol{M}_{T}\boldsymbol{X}_{i}\right) $ is
the $k^{\prime }\times k^{\prime }$ weight matrix for unit $i$. The
difference between $\boldsymbol{\hat{\beta}}_{FE}$ and $\boldsymbol{\hat{%
\beta}}_{MG}$ lies in the choice of the weights. When heterogeneity is
exogenous, both estimators converge to the same limit $\boldsymbol{\beta }%
_{0}$. However, when heterogeneity is correlated, as we shall see, $%
\boldsymbol{\hat{\beta}}_{MG}$ continues to converge to $\boldsymbol{\beta }%
_{0}$, but $\boldsymbol{\hat{\beta}}_{FE}$ converges to $\boldsymbol{\beta }%
_{0}+\lim\limits_{n\rightarrow \infty }n^{-1}\sum_{i=1}^{n}E\left( 
\boldsymbol{\mathcal{W}}_{i}\boldsymbol{\eta }_{i}\right) $, where $%
\lim\limits_{n\rightarrow \infty }n^{-1}\sum_{i=1}^{n}E\left( \boldsymbol{%
\mathcal{W}}_{i}\boldsymbol{\eta }_{i}\right) \neq \boldsymbol{0}$.

To investigate the asymptotic properties of FE and MG estimators for a fixed 
$T\geq k$ as $n\rightarrow \infty $, we consider the following assumptions.

\begin{assumption}[errors and regressors]
\label{errorsx} Let $\boldsymbol{X}=(\boldsymbol{X}_{1},\boldsymbol{X}%
_{2},...,\boldsymbol{X}_{n})$ and $\boldsymbol{M}_{T}=\boldsymbol{I}%
_{T}-T^{-1}\boldsymbol{\tau }_{T}\boldsymbol{\tau }_{T}^{\prime }$. (a)
Conditional on $\boldsymbol{X}$, the errors, $u_{it}$, are cross-sectionally
independent, (b) $E(\boldsymbol{u}_{i}\left\vert \boldsymbol{X}\right. )=%
\boldsymbol{0}$, for all $i$, $\sup_{i}E\left[ \left\Vert \boldsymbol{u}%
_{i}\right\Vert ^{2+\varepsilon }\left\vert \boldsymbol{X}_{i}\right. \right]
<C$, for some $\varepsilon >0$, and (c) $E(\boldsymbol{u}_{i}\boldsymbol{u}%
_{i}^{\prime }\left\vert \boldsymbol{X}_{i}\right. )=\boldsymbol{H}_{i}(%
\boldsymbol{X}_{i})=\boldsymbol{H}_{i}$, where $\boldsymbol{H}_{i}$ is a
symmetric $T\times T$ matrix, such that $\func{inf}_{i}\lambda _{min}\left( 
\boldsymbol{H}_{i}\right) >c>0$, and $\func{sup}_{i}E\left[ \lambda
_{max}^{4+\varepsilon }\left( \boldsymbol{H}_{i}\right) \right] <C$, for
some $\varepsilon >0$. (d) $\sup_{i}E\left[ \left\Vert \func{adj}\left( 
\boldsymbol{X}_{i}^{\prime }\boldsymbol{M}_{T}\boldsymbol{X}_{i}\right)
\right\Vert ^{4+\varepsilon }\right] <C$, for some $\varepsilon >0$, where $%
\func{adj}\left( \boldsymbol{X}_{i}^{\prime }\boldsymbol{M}_{T}\boldsymbol{X}%
_{i}\right) $ is the adjugate of $\boldsymbol{X}_{i}^{\prime }\boldsymbol{M}%
_{T}\boldsymbol{X}_{i}$.
\end{assumption}

\begin{assumption}[random coefficients]
\label{rcm} (a) For $i=1,2,...,n$, the $k^{\prime }\times 1$ vector of
unit-specific slope coefficients, $\boldsymbol{\beta }_{i}$, follow the
random coefficient model 
\begin{equation}
\boldsymbol{\beta }_{i}=\boldsymbol{\beta }_{0}+\boldsymbol{\eta }_{i},
\label{RC}
\end{equation}%
where $\left\Vert \boldsymbol{\beta }_{0}\right\Vert <C$, and $\boldsymbol{%
\eta }_{i}$ are independently distributed with mean zero and $%
\sup_{i}E\left\Vert \boldsymbol{\eta }_{i}\right\Vert ^{2+\varepsilon }<C$,
for some $\varepsilon >0$. (b) $Var(\boldsymbol{\eta }_{i}\left\vert 
\boldsymbol{X}\right. )=\boldsymbol{\Omega }_{i\beta }\succeq \boldsymbol{0}$%
, $E\left( \boldsymbol{\Omega }_{i\beta }\left\vert \boldsymbol{X}%
_{i}\right. \right) =\boldsymbol{\Omega }_{\beta }$ for all $i$, and $%
\sup_{i}E\left\Vert \boldsymbol{\Omega }_{i\beta }\right\Vert <C$.
\end{assumption}

\begin{assumption}[pooling]
\label{PoolA} Let $\boldsymbol{\Psi }_{i}=\boldsymbol{X}_{i}^{\prime }%
\boldsymbol{M}_{T}\boldsymbol{X}_i$ and $\boldsymbol{\bar{\Psi}}_{n}=\frac{1%
}{n}\sum_{i=1}^{n}\boldsymbol{\Psi }_{i}\succ \boldsymbol{0}$. (a) For a
fixed $T\geq k$, $\func{inf}_{i}\lambda _{min}\left( \boldsymbol{\Psi }%
_{i}\right) >c>0$, and $\sup_{i}E\left[ \left\Vert \boldsymbol{\Psi }%
_{i}\right\Vert ^{2+\varepsilon }\right] <C$, for some $\varepsilon >0$. (b) 
$\lim_{n\rightarrow \infty }E(\boldsymbol{\bar{\Psi}}_{n})=
\lim_{n\rightarrow \infty } \frac{1}{n}\sum_{i=1}^{n}E(\boldsymbol{\Psi }%
_{i})=\boldsymbol{\bar{\Psi}}\succ \boldsymbol{0}$, $\frac{1}{n}%
\sum_{i=1}^{n}E\left\Vert \boldsymbol{\Psi }_{i}\right\Vert ^{2}<C$, and $%
\boldsymbol{\bar{\Psi}}_{n}^{-1}=\boldsymbol{\bar{\Psi}}^{-1}+o_{p}(1)$,
with $\left\Vert \boldsymbol{\bar{\Psi}}^{-1}\right\Vert ^{2+\varepsilon }<C$%
, for some $\varepsilon >0$. (c) $\frac{1}{n} \sum_{i=1}^{n}\boldsymbol{X}%
_{i}^{\prime }\boldsymbol{M}_{T}\boldsymbol{H}_{i}\boldsymbol{M}_{T}%
\boldsymbol{X}_{i}\rightarrow _{p}\boldsymbol{Q}_{FE}$, where $\boldsymbol{Q}%
_{FE}=\lim_{n\rightarrow \infty }\frac{1}{n} \sum_{i=1}^{n}E\left( 
\boldsymbol{X}_{i}^{\prime }\boldsymbol{M}_{T}\boldsymbol{H}_{i}\boldsymbol{M%
}_{T}\boldsymbol{X}_{i}\right) \succ \boldsymbol{0}$. (d) For a fixed $T\geq
k$ and some $\varepsilon >0$, $\sup_i E\left[ \lambda _{\max
}^{1+\varepsilon }\left( \boldsymbol{X}_{i}^{\prime }\boldsymbol{M}_{T}%
\boldsymbol{H}_{i}\boldsymbol{M}_{T}\boldsymbol{X}_{i}\right) \right] <C$.
\end{assumption}

\begin{remark}
The above eigenvalue condition can be derived and is included for
expositional convenience. Since $\boldsymbol{H}_{i}$ is a symmetric positive
definite matrix, 
\begin{align*}
\lambda _{\max }\left( \boldsymbol{X}_{i}^{\prime }\boldsymbol{M}_{T}%
\boldsymbol{H}_{i}\boldsymbol{M}_{T}\boldsymbol{X}_{i}\right) &\leq \lambda
_{\max }\left( \boldsymbol{X}_{i}^{\prime }\boldsymbol{M}_{T}\boldsymbol{X}%
_{i}\right) \lambda _{\max }\left( \boldsymbol{H}_{i}\right) =\lambda _{\max
}\left( \boldsymbol{\Psi }_{i}\right) \lambda _{\max }\left( \boldsymbol{H}%
_{i}\right).
\end{align*}
\end{remark}

\begin{assumption}[correlated heterogeneity]
\label{CrCorr} The $k^{\prime }\times 1$ vectors $\boldsymbol{X}_{i}^{\prime
}\boldsymbol{M}_{T}\boldsymbol{X}_{i}\boldsymbol{\eta }_{i}$, for $%
i=1,2,...,n$, are weakly cross-correlated such that for a fixed $T\geq k$, 
\begin{equation}
\sup_{j}\sum_{i=1}^{n}\left\Vert Cov\left( \boldsymbol{X}_{i}^{\prime }%
\boldsymbol{M}_{T}\boldsymbol{X}_{i}\boldsymbol{\eta }_{i},\boldsymbol{X}%
_{j}^{\prime }\boldsymbol{M}_{T}\boldsymbol{X}_{j}\boldsymbol{\eta }%
_{j}\right) \right\Vert <C.  \label{WCR}
\end{equation}
\end{assumption}

\begin{remark}
Assumption \ref{errorsx} requires the regressors, $\boldsymbol{x}_{it}$, to
be strictly exogenous, but it allows the conditional variance of $%
\boldsymbol{u}_{i}$ to depend on $\boldsymbol{X}_{i}$ and the errors, $%
u_{it} $, to be serially correlated. Under part (a) of Assumption \ref{rcm}, 
\begin{equation}
\boldsymbol{\bar{\beta}}_{n}=n^{-1}\sum_{i=1}^{n}\boldsymbol{\beta }_{i}=%
\boldsymbol{\beta }_{0}+n^{-1}\sum_{i=1}^{n}\boldsymbol{\eta }_{i}=%
\boldsymbol{\beta }_{0}+O_{p}(n^{-1/2}),  \label{Truebeta}
\end{equation}%
and hence $\boldsymbol{\beta }_{0}=\func{plim}_{n\rightarrow \infty }\left(
n^{-1}\sum_{i=1}^{n}\boldsymbol{\beta }_{i}\right) $. The condition $E\left( 
\boldsymbol{\Omega }_{i\beta }\left\vert \boldsymbol{X}_{i}\right. \right) =%
\boldsymbol{\Omega }_{\beta }$ in part (b) of Assumption \ref{rcm}, is
required when FE and MG estimators are compared in sub-section \ref{FEvsMG}
and for the Hausman test of correlated heterogeneity in Section \ref{Test}.
It is not required for establishing the asymptotic properties of the MG or
TMG estimators, which allow for dependence between $\boldsymbol{\Omega }%
_{i\beta }$ and $\boldsymbol{X}_{j}$, for all $i$ and $j$. Part (a) of
Assumption \ref{PoolA} is required for estimation of individual coefficients 
$\boldsymbol{\beta }_{i}$ and can be relaxed if particular linear
combinations of $\boldsymbol{\beta }_{i}$ are of interest. Part (b) of
Assumption \ref{PoolA} is standard in the literature on FE estimation.
Assumption \ref{CrCorr} is required for the convergence of the FE estimator
under correlated heterogeneity. Condition (\ref{WCR}) trivially holds when $%
\boldsymbol{\eta }_{i}$ is distributed independently of $\boldsymbol{X}%
_{i}^{\prime }\boldsymbol{M}_{T}\boldsymbol{X}_{i}$ and $\boldsymbol{\eta }%
_{i}$ are cross-sectionally independently distributed as under Assumption %
\ref{rcm}.
\end{remark}

\begin{remark}
The above assumptions do not require $(y_{it},\boldsymbol{x}_{it}^{\prime
},u_{it})$ to be identically and independently distributed as assumed in
some micro econometric panel data analysis with $T$ fixed. See Example \ref%
{ex1} below and the MC designs used for evaluation of the small sample
performance of the proposed TMG estimators.
\end{remark}

\subsection{Why does correlated slope heterogeneity matter?}

\label{CREmatter}

Correlated slope heterogeneity can arise in various contexts, with important
implications for estimation and inference on the average effects. For
example, return to education is likely to be positively correlated with
latent factors such as ability, talent, and degree of self-belief.
Similarly, propensity to save across households is often inversely related
to their income volatility. FE estimation allows for possible correlation
between $\boldsymbol{x}_{it}$ and $\alpha _{i}$, but not between $%
\boldsymbol{x}_{it}$ and $\boldsymbol{\beta }_{i}$. As a simple example,
consider the panel data model: 
\begin{equation*}
y_{it}=\alpha _{i}+\beta _{i}x_{it}+u_{it},
\end{equation*}%
where $x_{it}$ and $y_{it}$ could, respectively, be years of schooling and
return to schooling, or could be income and the saving rate of individual $i$
at time $t$. Suppose now that $x_{it}$ follows the following model with an
interactive effect and idiosyncratic error heteroskedasticity: 
\begin{equation*}
x_{it}=\alpha _{ix}+\gamma _{ix}f_{t}+\sigma _{ix}u_{x,it},
\end{equation*}%
where $f_{t}$ is a latent factor (could be talent in the context of return
to education) and $\sigma _{ix}u_{x,it}$ is the idiosyncratic innovation to
the $x_{it}$ process ($\sigma _{ix}$ being income volatility in the saving
rate example). A simple model of correlated slope heterogeneity is given by 
\begin{equation*}
\beta _{i}-\beta _{0}=\kappa _{\gamma }\left[ \gamma _{ix}-E\left( \gamma
_{ix}\right) \right] +\kappa _{\sigma }\left[ \sigma _{ix}^{2}-E\left(
\sigma _{ix}^{2}\right) \right] +\epsilon _{i},
\end{equation*}%
where $\kappa _{\gamma }$ and $\kappa _{\sigma }$ measure the extent to
which heterogeneity is correlated with $x_{it}$, and $\epsilon _{i}$ is the
exogenous component of heterogeneity. FE estimators are robust to exogenous
heterogeneity but can become badly biased if $\kappa _{\gamma }\neq 0$
and/or $\kappa _{\sigma }\neq 0$.

\subsection{Bias of the fixed effects estimator under correlated
heterogeneity}

It is known that fixed effects (with or without time effects) estimators are
biased under correlated heterogeneity. For example, see \cite{Wooldridge2005}%
. Here we provide a minimal set of conditions required for the FE estimator
to be $\sqrt{n}$-consistent in the presence of correlated heterogeneity.
Under the heterogeneous specification (\ref{eq2}) and noting that $%
\boldsymbol{M}_{T}\boldsymbol{\tau }_{T}=\boldsymbol{0}$, we have 
\begin{equation}
\boldsymbol{\hat{\beta}}_{FE}-\boldsymbol{\beta }_{0}=\boldsymbol{\bar{\Psi}}%
_{n}^{-1}\left[ n^{-1}\sum_{i=1}^{n}\boldsymbol{X}_{i}^{\prime }\boldsymbol{M%
}_{T}\boldsymbol{X}_{i}(\boldsymbol{\beta }_{i}-\boldsymbol{\beta }_{0})%
\right] +\boldsymbol{\bar{\Psi}}_{n}^{-1}\left( n^{-1}\sum_{i=1}^{n}%
\boldsymbol{X}_{i}^{\prime }\boldsymbol{M}_{T}\boldsymbol{u}_{i}\right) .
\label{FE1}
\end{equation}%
By Assumptions \ref{errorsx}--\ref{PoolA}, $E\left( \boldsymbol{X}%
_{i}^{\prime }\boldsymbol{M}_{T}\boldsymbol{u}_{i}\right) =\boldsymbol{0}$, $%
Var\left(\frac{1}{n}\sum_{i=1}^{n}\boldsymbol{X}_{i}^{\prime }\boldsymbol{M}%
_{T}\boldsymbol{u}_{i}\right)=O(n^{-1})$, and we have 
\begin{equation*}
\boldsymbol{\hat{\beta}}_{FE}-\boldsymbol{\beta }_{0}\rightarrow _{p}%
\boldsymbol{\bar{\Psi}}^{-1}\lim_{n\rightarrow \infty
}n^{-1}\sum_{i=1}^{n}E\left( \boldsymbol{X}_{i}^{\prime }\boldsymbol{M}_{T}%
\boldsymbol{X}_{i}\boldsymbol{\eta }_{i}\right) .
\end{equation*}%
Hence, $\boldsymbol{\hat{\beta}}_{FE}$ is a consistent estimator of the
average effect, $\boldsymbol{\beta }_{0}$, if and only if 
\begin{equation}
\lim_{n\rightarrow \infty }n^{-1}\sum_{i=1}^{n}E\left( \boldsymbol{X}%
_{i}^{\prime }\boldsymbol{M}_{T}\boldsymbol{X}_{i}\boldsymbol{\eta }%
_{i}\right) =\boldsymbol{0}.  \label{AveFE}
\end{equation}%
This condition is clearly met if 
\begin{equation}
E\left( \boldsymbol{X}_{i}^{\prime }\boldsymbol{M}_{T}\boldsymbol{X}_{i}%
\boldsymbol{\eta }_{i}\right) =\boldsymbol{0},\text{ for all }i\text{,}
\label{FEconsistency}
\end{equation}%
and has already been derived by \cite{Wooldridge2005}. But it is too
restrictive, since it is possible for the average condition in (\ref{AveFE})
to hold even though condition (\ref{FEconsistency}) is violated for some
units as $n\rightarrow \infty $. Suppose the number of units that \textit{do
not} satisfy (\ref{FEconsistency}) is given by $m_{n}=\ominus \left(
n^{a_{\eta }}\right) $, namely $m_{n}$ rises in line with $n^{a_{\eta }}$.%
\footnote{%
Note that $m_{n}=\ominus \left( n^{a_{\eta }}\right) $ differs from the
familiar big O order, $m_{n}=O(n^{a_{\eta }})$. The latter requires $%
n^{-a_{\eta }}m_{n}$ is bounded in $n$ and its limiting value could be zero.
The former requires $n^{-a_{\eta }}m_{n}$ to tend to a positive non-zero
number.} Then $n^{-1}\sum_{i=1}^{n}E\left( \boldsymbol{X}_{i}^{\prime }%
\boldsymbol{M}_{T}\boldsymbol{X}_{i}\boldsymbol{\eta }_{i}\right) =\ominus
\left( n^{a_{\eta }-1}\right) $, and condition (\ref{AveFE}) is met if $%
a_{\eta }<1$.

But for $\boldsymbol{\hat{\beta}}_{FE}$ to be a $\sqrt{n}$-consistent
estimator of $\boldsymbol{\beta }_{0}$, a much more restrictive condition on 
$a_{\eta }$ is required. Using (\ref{FE1}), 
\begin{equation}
\sqrt{n}\left( \boldsymbol{\hat{\beta}}_{FE}-\boldsymbol{\beta }_{0}\right) =%
\boldsymbol{\bar{\Psi}}_{n}^{-1}\left( n^{-1/2}\sum_{i=1}^{n}\boldsymbol{X}%
_{i}^{\prime }\boldsymbol{M}_{T}\boldsymbol{X}_{i}\boldsymbol{\eta }%
_{i}+n^{-1/2}\sum_{i=1}^{n}\boldsymbol{X}_{i}^{\prime }\boldsymbol{M}_{T}%
\boldsymbol{u}_{i}\right) .  \label{AsyBiasFE}
\end{equation}%
For a fixed $T\geq k$, under Assumptions \ref{errorsx} and \ref{PoolA}, $%
\boldsymbol{\bar{\Psi}}_{n}\rightarrow _{p}\boldsymbol{\bar{\Psi}}\succ 
\boldsymbol{0}$, and $n^{-1/2}\sum_{i=1}^{n}\boldsymbol{X}_{i}^{\prime }%
\boldsymbol{M}_{T}\boldsymbol{u}_{i}\ \rightarrow _{d}N\left( \boldsymbol{0},%
\boldsymbol{Q}_{FE}\right) $, as $n\rightarrow \infty $, where $\boldsymbol{Q%
}_{FE}=\lim_{n\rightarrow \infty }n^{-1}\sum_{i=1}^{n}E\left( \boldsymbol{X}%
_{i}^{\prime }\boldsymbol{M}_{T}\boldsymbol{H}_{i}\boldsymbol{M}_{T}%
\boldsymbol{X}_{i}\right) \succ \boldsymbol{0}$. The first term in the
brackets of (\ref{AsyBiasFE}) can be decomposed as 
\begin{equation*}
n^{-1/2}\sum_{i=1}^{n}\boldsymbol{X}_{i}^{\prime }\boldsymbol{M}_{T}%
\boldsymbol{X}_{i}\boldsymbol{\eta }_{i}=\boldsymbol{s}_{nT}+n^{-1/2}%
\sum_{i=1}^{n}E\left( \boldsymbol{X}_{i}^{\prime }\boldsymbol{M}_{T}%
\boldsymbol{X}_{i}\boldsymbol{\eta }_{i}\right) ,
\end{equation*}%
where $\boldsymbol{s}_{nT}=n^{-1/2}\sum_{i=1}^{n}\left[ \boldsymbol{X}%
_{i}^{\prime }\boldsymbol{M}_{T}\boldsymbol{X}_{i}\boldsymbol{\eta }%
_{i}-E\left( \boldsymbol{X}_{i}^{\prime }\boldsymbol{M}_{T}\boldsymbol{X}_{i}%
\boldsymbol{\eta }_{i}\right) \right] $, and 
\begin{align*}
\left\Vert Var\left( \boldsymbol{s}_{nT}\right) \right\Vert &
=n^{-1}\left\Vert \sum_{i=1}^{n}\sum_{j=1}^{n}Cov\left( \boldsymbol{X}%
_{i}^{\prime }\boldsymbol{M}_{T}\boldsymbol{X}_{i}\boldsymbol{\eta }_{i},%
\boldsymbol{X}_{j}^{\prime }\boldsymbol{M}_{T}\boldsymbol{X}_{j}\boldsymbol{%
\eta }_{j}\right) \right\Vert \\
& \leq \sup_{j}\sum_{i=1}^{n}\left\Vert Cov\left( \boldsymbol{X}_{i}^{\prime
}\boldsymbol{M}_{T}\boldsymbol{X}_{i}\boldsymbol{\eta }_{i},\boldsymbol{X}%
_{j}^{\prime }\boldsymbol{M}_{T}\boldsymbol{X}_{j}\boldsymbol{\eta }%
_{j}\right) \right\Vert ,
\end{align*}%
which is bounded by Assumption \ref{CrCorr}. Hence, $\boldsymbol{s}%
_{nT}=O_{p}(1)$. Therefore, $\sqrt{n}\left( \boldsymbol{\hat{\beta}}_{FE}-%
\boldsymbol{\beta }_{0}\right) $ will converge to a distribution with a zero
mean only if $n^{-1/2}\sum_{i=1}^{n}E\left( \boldsymbol{X}_{i}^{\prime }%
\boldsymbol{M}_{T}\boldsymbol{X}_{i}\boldsymbol{\eta }_{i}\right)
\rightarrow \boldsymbol{0}$. For this condition to hold, it is required that 
$m_{n}n^{-1/2}=\ominus \left( n^{a_{\eta }-1/2}\right) \rightarrow 0$, which
occurs only if $a_{\eta }<1/2$. A formal statement of this result is
summarized in the following proposition:

\begin{proposition}[Condition for $\protect\sqrt{n}-$consistency of the FE
estimator]
\label{prop_confee} Suppose that for $i=1,2,...,n$ and $t=1,2,...,T$, $%
y_{it} $ are generated by the heterogeneous panel data model (\ref{m1}) and
Assumptions \ref{errorsx}--\ref{CrCorr} hold. Then the FE estimator given by
(\ref{fee}) is $\sqrt{n}$-consistent if and only if%
\begin{equation}
\lim_{n\rightarrow \infty }n^{-1/2}\sum_{i=1}^{n}E\left( \boldsymbol{X}%
_{i}^{\prime }\boldsymbol{M}_{T}\boldsymbol{X}_{i}\boldsymbol{\eta }%
_{i}\right) =\boldsymbol{0},  \label{rootnFE}
\end{equation}%
where $\boldsymbol{\eta }_{i}=\boldsymbol{\beta }_{i}-\boldsymbol{\beta }%
_{0} $.
\end{proposition}

\begin{remark}
The asymptotic distribution of $\boldsymbol{\hat{\beta}}_{FE}$ is not
centered around $\boldsymbol{\beta }_{0}$ even if both $n$ and $T\rightarrow
\infty $, jointly. In this case, using (\ref{FE1}) we have 
\begin{equation*}
\sqrt{nT}\left( \boldsymbol{\hat{\beta}}_{FE}-\boldsymbol{\beta }_{0}\right)
=\boldsymbol{\bar{\Psi}}_{nT}^{-1}\left[ \sqrt{nT}\frac{1}{n}%
\sum_{i=1}^{n}\left( \frac{1}{T}\boldsymbol{X}_{i}^{\prime }\boldsymbol{M}%
_{T}\boldsymbol{X}_{i}\boldsymbol{\eta }_{i}\right) \right] +\boldsymbol{%
\bar{\Psi}}_{nT}^{-1}\left[ (nT)^{-1/2}\sum_{i=1}^{n}\boldsymbol{X}%
_{i}^{\prime }\boldsymbol{M}_{T}\boldsymbol{u}_{i}\right] ,
\end{equation*}%
where $\boldsymbol{\bar{\Psi}}_{nT}=\frac{1}{nT}\sum_{i=1}^{n}\boldsymbol{X}%
_{i}^{\prime }\boldsymbol{M}_{T}\boldsymbol{X}_{i}$, and converges to a
positive definite matrix under Assumption \ref{PoolA}, as both $%
n,T\rightarrow \infty $. For $\boldsymbol{\hat{\beta}}_{FE}$ to be $\sqrt{nT}
$-consistent, it is required that%
\begin{equation*}
\sqrt{nT}n^{-1}\sum_{i=1}^{n}E\left( T^{-1}\boldsymbol{X}_{i}^{\prime }%
\boldsymbol{M}_{T}\boldsymbol{X}_{i}\boldsymbol{\eta }_{i}\right) \rightarrow%
\boldsymbol{0}\text{, as }n\rightarrow \infty \text{ and }T\rightarrow
\infty.
\end{equation*}%
As before, suppose that $E\left( T^{-1}\boldsymbol{X}_{i}^{\prime }%
\boldsymbol{M}_{T}\boldsymbol{X}_{i}\boldsymbol{\eta }_{i}\right) \neq 
\boldsymbol{0}$ for $n^{a_{\eta }}$ of the $n$ cross-section units, and let $%
T=\ominus \left( n^{d}\right) $, $d\geq 0$. Then the above condition is met
if $\sqrt{nT}n^{a_{\eta }-1}=\ominus \left( n^{a_{\eta }-1/2+d/2}\right)
\rightarrow 0$, namely if $a_{\eta }+d/2<1/2$. Thus, increasing the time
dimension of the panel does not help with bias reduction of the FE
estimators when heterogeneity is correlated.
\end{remark}

\subsection{Mean group estimator}

In contrast to the FE estimator, the MG estimator given by (\ref{mge})
continues to be $\sqrt{n}$-consistent, so long as certain moment conditions,
to be discussed below, are met. Substituting (\ref{m1}) in (\ref{betaihat}), 
\begin{equation}
\boldsymbol{\hat{\beta}}_{i}=\boldsymbol{\beta }_{i}+\boldsymbol{\xi }_{iT},
\label{betaihat2}
\end{equation}%
where $\boldsymbol{\xi }_{iT}=\boldsymbol{R}_{i}^{\prime }\boldsymbol{u}_{i} 
$ and $\boldsymbol{R}_{i}=\boldsymbol{M}_{T}\boldsymbol{X}_{i}\left( 
\boldsymbol{X}_{i}^{\prime }\boldsymbol{M}_{T}\boldsymbol{X}_{i}\right)
^{-1} $. Averaging both sides of (\ref{betaihat2}) over $i$ yields 
\begin{equation}
\boldsymbol{\hat{\beta}}_{MG}=\boldsymbol{\bar{\beta}}_{n}+\boldsymbol{\bar{%
\xi}}_{nT},  \label{mge2}
\end{equation}%
where $\boldsymbol{\bar{\beta}}_{n}=n^{-1}\sum_{i=1}^{n}\boldsymbol{\beta }%
_{i}$ and $\boldsymbol{\bar{\xi}}_{nT}=n^{-1}\sum_{i=1}^{n}\boldsymbol{\xi }%
_{iT}$. Under Assumption \ref{errorsx}(b), $E\left( \boldsymbol{\bar{\xi}}%
_{nT}\right) =n^{-1}\sum_{i=1}^{n}E\left( \boldsymbol{\xi }_{iT}\right)
=n^{-1}\sum_{i=1}^{n}E\left[ \boldsymbol{R}_{i}^{\prime }E\left( \boldsymbol{%
u}_{i}\left\vert \boldsymbol{X}_{i}\right. \right) \right] =\boldsymbol{0}$.
Then using (\ref{mge2}), $E(\boldsymbol{\hat{\beta}}_{MG})=E(\boldsymbol{%
\bar{\beta}}_{n})+E\left( \boldsymbol{\bar{\xi}}_{nT}\right) =\boldsymbol{%
\beta }_{0}$, i.e., $\boldsymbol{\hat{\beta}}_{MG}$ is an \textit{unbiased}
estimator of $\boldsymbol{\beta }_{0}$ irrespective of the possible
dependence of $\boldsymbol{\beta }_{i}$ on $\boldsymbol{X}_{i}$. However,
the MG estimator is likely to have a large variance when $T$ is too small.
This arises, for example, when the variance of $\boldsymbol{\bar{\xi}}_{nT}$
does not exist or is very large. The conditions under which $\boldsymbol{%
\hat{\beta}}_{MG}$ converges to $\boldsymbol{\beta }_{0}$ at the regular $%
\sqrt{n}$ rate are given in the following proposition:

\begin{proposition}[Sufficient conditions for $\protect\sqrt{n}$-consistency
of $\boldsymbol{\hat{\protect\beta}}_{MG}$]
\label{prop_mexist} Suppose that $y_{it}$ for $i=1,2,...,n$ and $t=1,2,...,T$
are generated by model (\ref{m1}) and Assumption \ref{errorsx} and part (a)
of Assumption \ref{rcm} hold. Then for a fixed $T$, as $n\rightarrow \infty $%
, the MG estimator given by (\ref{mge}) is $\sqrt{n}$-consistent if 
\begin{equation}
\sup_{i}E\left( d_{i}^{-2}\right) <C\text{, and }\func{sup}_{i}E\left[
\left\Vert \func{adj}(\boldsymbol{X}_{i}^{\prime }\boldsymbol{M}_{T}%
\boldsymbol{X}_{i})\right\Vert _{1}^{2}\right] <C,  \label{sufficient}
\end{equation}%
where $d_{i}=\func{det}(\boldsymbol{X}_{i}^{\prime }\boldsymbol{M}_{T}%
\boldsymbol{X}_{i})$, and $\func{adj}(\boldsymbol{X}_{i}^{\prime }%
\boldsymbol{M}_{T}\boldsymbol{X}_{i})$ is the adjugate (or adjoint) of $%
\boldsymbol{X}_{i}^{\prime }\boldsymbol{M}_{T}\boldsymbol{X}_{i}$.
\end{proposition}

For a proof, see sub-section \ref{pf_prop_mexist} of the mathematical
appendix.

\begin{example}
\label{ex1} Consider a simple example $y_{it}=\alpha _{i}+\beta
_{i}x_{it}+u_{it}$, where $x_{it}=\alpha _{ix}+\sigma _{ix}\varepsilon
_{x,it}$, $\func{sup}_{i}\Vert \alpha _{ix}\Vert <C$, $\inf_{i}\sigma
_{ix}^{2}>c>0$, and $\varepsilon _{x,it}\thicksim IIDN(0,1)$. Suppose that $%
E(\boldsymbol{u}_{i}\boldsymbol{u}_{i}^{\prime }|\boldsymbol{x}_{i})=\sigma
_{i}^{2}\boldsymbol{I}_{T}$, for $i=1,2,...,n$. Then the individual OLS
estimator of the slope coefficient, $\hat{\beta}_{i}=(\boldsymbol{x}%
_{i}^{\prime }\boldsymbol{M}_{T}\boldsymbol{x}_{i})^{-1}\boldsymbol{x}%
_{i}^{\prime }\boldsymbol{M}_{T}\boldsymbol{y}_{i}$, has first- and
second-order moments if $E\left( u_{it}^{2}\right) <C$ and $E\left(
d_{i}^{-2}\right) <C$, where $d_{i}=\func{det}(\boldsymbol{x}_{i}^{\prime }%
\boldsymbol{M}_{T}\boldsymbol{x}_{i})$, $\boldsymbol{x}_{i}=\alpha _{ix}%
\boldsymbol{\tau }_{T}+\sigma _{ix}\boldsymbol{\varepsilon }_{ix}$, and $%
\boldsymbol{\varepsilon }_{ix}=(\varepsilon _{x,i1},\varepsilon
_{x,i2},...,\varepsilon _{x,iT})^{\prime }$. In this case, $%
1/d_{i}=(1/\sigma _{ix}^{2})\left( 1/\boldsymbol{\varepsilon }_{ix}^{\prime }%
\boldsymbol{M}_{T}\boldsymbol{\varepsilon }_{ix}\right) $, where $%
\boldsymbol{\varepsilon }_{ix}^{\prime }\boldsymbol{M}_{T}\boldsymbol{%
\varepsilon }_{ix}\thicksim \chi _{v}^{2}$, with $v=T-1$ degrees of freedom.
Suppose further that $\sigma _{ix}^{2}$ and $\boldsymbol{\varepsilon }%
_{ix}^{\prime }\boldsymbol{M}_{T}\boldsymbol{\varepsilon }_{ix}$ are
independently distributed, then 
\begin{equation*}
E\left( \frac{1}{d_{i}}\right) =E\left( \frac{1}{\sigma _{ix}^{2}}\right)
E\left( \frac{1}{\chi _{v}^{2}}\right) =\frac{1}{v-2}E\left( \frac{1}{\sigma
_{ix}^{2}}\right) \text{, for }v>2\text{.}
\end{equation*}%
\begin{equation*}
E\left( \frac{1}{d_{i}^{2}}\right) =E\left( \frac{1}{\sigma _{ix}^{4}}%
\right) E\left( \frac{1}{\chi _{v}^{4}}\right) =\frac{1}{(v-2)(v-4)}E\left( 
\frac{1}{\sigma _{ix}^{4}}\right) \text{, for }v>4.
\end{equation*}%
Since by assumption $\sigma _{ix}^{-2}<1/c<\infty $, it follows that $%
E\left( d_{i}^{-2}\right) $ exists if $T>5$. This example also illustrates
that $d_{i}$ could be random draws from a common distribution without
requiring $x_{it}$ to be IID. Note that no restrictions are placed on the
distribution of $\alpha _{ix}$ over $i$.
\end{example}

\subsection{Comparison of FE and MG estimators}

\label{FEvsMG}

Suppose now that conditions (\ref{rootnFE}) and (\ref{sufficient}) hold and
both FE and MG estimators are $\sqrt{n}$-consistent. The choice between the
two estimators will then depend on their relative efficiency, which we
measure in terms of their asymptotic covariances. We first consider their
covariances conditional on $\boldsymbol{X}$, and note that 
\begin{equation*}
Var\left( \sqrt{n}\boldsymbol{\hat{\beta}}_{MG}\left\vert \boldsymbol{X}%
\right. \right) =\boldsymbol{\bar{\Omega}}_{n,\beta }+n^{-1}\sum_{i=1}^{n}%
\boldsymbol{\Psi }_{i}^{-1}\boldsymbol{X}_{i}^{\prime }\boldsymbol{M}_{T}%
\boldsymbol{H}_{i}\boldsymbol{M}_{T}\boldsymbol{X}_{i}\boldsymbol{\Psi }%
_{i}^{-1},
\end{equation*}%
where as before $\boldsymbol{\Psi }_{i}=\boldsymbol{X}_{i}^{\prime }%
\boldsymbol{M}_{T}\boldsymbol{X}_{i}$, and 
\begin{equation}
\boldsymbol{\bar{\Omega}}_{n,\beta }=n^{-1}\sum_{i=1}^{n}\boldsymbol{\Omega }%
_{i\beta }.  \label{AveOmega}
\end{equation}%
Also under Assumptions \ref{errorsx} and \ref{rcm},%
\begin{equation*}
Var\left( \sqrt{n}\boldsymbol{\hat{\beta}}_{FE}\left\vert \boldsymbol{X}%
\right. \right) =\boldsymbol{\bar{\Psi}}_{n}^{-1}\left( n^{-1}\sum_{i=1}^{n}%
\boldsymbol{\Psi }_{i}\boldsymbol{\Omega }_{i\beta }\boldsymbol{\Psi }%
_{i}\right) \boldsymbol{\bar{\Psi}}_{n}^{-1}+\boldsymbol{\bar{\Psi}}%
_{n}^{-1}\left( n^{-1}\sum_{i=1}^{n}\boldsymbol{X}_{i}^{\prime }\boldsymbol{M%
}_{T}\boldsymbol{H}_{i}\boldsymbol{M}_{T}\boldsymbol{X}_{i}\right) 
\boldsymbol{\bar{\Psi}}_{n}^{-1},
\end{equation*}%
where $\boldsymbol{\bar{\Psi}}_{n}=n^{-1}\sum_{i=1}^{n}\boldsymbol{\Psi }%
_{i} $. Hence, 
\begin{equation}
Var\left( \sqrt{n}\boldsymbol{\hat{\beta}}_{MG}\left\vert \boldsymbol{X}%
\right. \right) -Var\left( \sqrt{n}\boldsymbol{\hat{\beta}}_{FE}\left\vert 
\boldsymbol{X}\right. \right) =\boldsymbol{A}_{n}+\boldsymbol{B}_{n},
\label{vardif}
\end{equation}%
where%
\begin{equation}
\boldsymbol{A}_{n}=\boldsymbol{\bar{\Omega}}_{n,\beta }-\boldsymbol{\bar{\Psi%
}}_{n}^{-1}\left( n^{-1}\sum_{i=1}^{n}\boldsymbol{\Psi }_{i}\boldsymbol{%
\Omega }_{i\beta }\boldsymbol{\Psi }_{i}\right) \boldsymbol{\bar{\Psi}}%
_{n}^{-1},  \label{An}
\end{equation}%
and%
\begin{equation}
\boldsymbol{B}_{n}=\left( \frac{1}{n}\sum_{i=1}^{n}\boldsymbol{\Psi }%
_{i}^{-1}\boldsymbol{X}_{i}^{\prime }\boldsymbol{M}_{T}\boldsymbol{H}_{i}%
\boldsymbol{M}_{T}\boldsymbol{X}_{i}\boldsymbol{\Psi }_{i}^{-1}\right) -%
\boldsymbol{\bar{\Psi}}_{n}^{-1}\left( \frac{1}{n}\sum_{i=1}^{n}\boldsymbol{X%
}_{i}^{\prime }\boldsymbol{M}_{T}\boldsymbol{H}_{i}\boldsymbol{M}_{T}%
\boldsymbol{X}_{i}\right) \boldsymbol{\bar{\Psi}}_{n}^{-1}.  \label{Bn}
\end{equation}%
$\boldsymbol{A}_{n}$ and $\boldsymbol{B}_{n}$ capture the effects of two
different types of heterogeneity, namely slope heterogeneity and
regressors/errors heterogeneity. The superiority of the FE estimator over
the MG estimator is readily established when the slope coefficients and
error variances are homogeneous across $i$ and the errors are serially
uncorrelated, namely if $\boldsymbol{\Omega }_{i\beta }=\boldsymbol{0}$ and $%
\boldsymbol{H}_{i}=\sigma ^{2}\boldsymbol{I}_{T}$ for all $i$. In this case, 
$\boldsymbol{A}_{n}=\boldsymbol{0}$, and we have (for any fixed $n$ and as $%
n\rightarrow \infty $) 
\begin{equation*}
\frac{Var\left( \sqrt{n}\boldsymbol{\hat{\beta}}_{MG}\left\vert \boldsymbol{X%
}\right. \right) -Var\left( \sqrt{n}\boldsymbol{\hat{\beta}}_{FE}\left\vert 
\boldsymbol{X}\right. \right) }{\sigma ^{2}}=n^{-1}\sum_{i=1}^{n}\boldsymbol{%
\Psi }_{i}^{-1}-\boldsymbol{\bar{\Psi}}_{n}^{-1},
\end{equation*}%
which is the difference between the inverse of the harmonic mean of $%
\boldsymbol{\Psi }_{i}$, and the inverse of its arithmetic mean, which is a
positive semi-definite matrix.\footnote{%
For a proof, see the Appendix to \cite{PesaranEtal1996}.} However, this
result may be reversed when we allow for heterogeneity, $\boldsymbol{\Omega }%
_{i\beta }\succeq \boldsymbol{0}$, and/or if $\boldsymbol{H}_{i}\neq \sigma
^{2}\boldsymbol{I}_{T}$. The following proposition summarizes the results of
the comparison between FE and MG estimators:

\begin{proposition}[Relative efficiency of FE and MG estimators]
\label{prop_mgvsfe} Suppose that $y_{it}$ for $i=1,2,...,n$ and $t=1,2,...,T$
are generated by the heterogeneous panel data model (\ref{m1}), Assumptions %
\ref{errorsx}--\ref{CrCorr} hold, and the uncorrelated heterogeneity
condition (\ref{rootnFE}) and second-order moment conditions (\ref%
{sufficient}) are met. Then $Var\left( \sqrt{n}\boldsymbol{\hat{\beta}}%
_{MG}\left\vert \boldsymbol{X}\right. \right) -Var\left( \sqrt{n}\boldsymbol{%
\hat{\beta}}_{FE}\left\vert \boldsymbol{X}\right. \right) =\boldsymbol{A}%
_{n}+\boldsymbol{B}_{n}$, where $\boldsymbol{A}_{n}$ and $\boldsymbol{B}_{n}$
are given by (\ref{An}) and (\ref{Bn}), respectively. The probability limit
of $\boldsymbol{A}_{n}$ as $n\rightarrow \infty$ is a negative semi-definite
matrix, but the sign of $\boldsymbol{B}_{n}$ and the sign of its probability
limit are indeterminate. Under uncorrelated heterogeneity, the FE estimator, 
$\boldsymbol{\hat{\beta}}_{FE}$, is asymptotically more efficient than the
MG estimator if the benefit from pooling (i.e., when $\boldsymbol{B}%
_{n}\succ \boldsymbol{0}$) outweighs the loss in efficiency due to slope
heterogeneity (since $\boldsymbol{A}_{n}\preceq \boldsymbol{0}$).
\end{proposition}

For a proof, see sub-section \ref{Proofmgvsfe} in the mathematical appendix.

\begin{example}
\label{ExampleMG-FE} Consider a simple case where $k^{\prime }=1$, $%
\boldsymbol{\Psi }_{i}=\psi _{i}$ and $\boldsymbol{\Omega }_{i\beta }=\omega
_{i\beta }^{2}$ are scalars and suppose that $\boldsymbol{H}_{i}(\boldsymbol{%
X}_{i})=E\left( \boldsymbol{u}_{i}\boldsymbol{u}_{i}^{\prime }\left\vert 
\boldsymbol{X}_{i}\right. \right) =\sigma ^{2}\psi _{i}\boldsymbol{I}_{T}$,
and $\omega _{i\beta }^{2}$ and $\psi _{i}$ are distributed independently,
then 
\begin{equation*}
\lim_{n\rightarrow \infty }\left[ Var\left( \sqrt{n}\boldsymbol{\hat{\beta}}%
_{MG}\right) -Var\left( \sqrt{n}\boldsymbol{\hat{\beta}}_{FE}\right) \right]
=-\frac{\left( \omega _{\beta }^{2}+\sigma ^{2}\right) }{\bar{\psi}^{2}}%
\left[ \lim_{n\rightarrow \infty }n^{-1}\sum_{i=1}^{n}E\left( \psi _{i}-\bar{%
\psi}\right) ^{2}\right] ,
\end{equation*}%
where $\bar{\psi}=\lim_{n\rightarrow \infty }n^{-1}\sum_{i=1}^{n}E\left(
\psi _{i}\right) $, and $E\left( \omega _{i\beta }^{2}\right) =$ $\omega
_{\beta }^{2}$. In this simple case, the MG estimator is more efficient than
the FE estimator even if $\omega _{\beta }^{2}=0$.
\end{example}

In general, under uncorrelated heterogeneity, the relative efficiency of the
MG and FE estimators depends on the relative magnitude of the two components
in (\ref{vardif}). Since $\boldsymbol{A}_{n}\preceq \boldsymbol{0}$, the
outcome depends on the sign and the magnitude of $\boldsymbol{B}_{n}$, which
in turn depends on the heterogeneity of error variances, $\boldsymbol{H}_{i}(%
\boldsymbol{X}_{i}$), and $\boldsymbol{\Psi }_{i}$ over $i$.

\section{Irregular mean group estimators\label{IRMGE}}

So far, we have argued that the MG estimator is robust to correlated
heterogeneity and its performance is comparable to the FE estimator even
under uncorrelated heterogeneity. However, since the MG estimator is based
on the individual estimates, $\boldsymbol{\hat{\beta}}_{i}$, $i=1,2,...,n$,
its optimality and robustness critically depend on how well individual
coefficients can be estimated. This is particularly important when $T$ is
ultra-short, which is the primary concern of this paper. In cases where $T$
is small and/or the observations on $\boldsymbol{x}_{it}$ are highly
correlated or are slowly moving, $d_{i}=\func{det}\left( \boldsymbol{X}%
_{i}^{\prime }\boldsymbol{M}_{T}\boldsymbol{X}_{i}\right) $ is likely to be
close to zero for a large number of units $i=1,2,...,n$. As a result, $%
\boldsymbol{\hat{\beta}}_{i}$ is likely to be a poor estimate of $%
\boldsymbol{\beta }_{i}$ for some $i$, and including such estimates when
computing $\boldsymbol{\hat{\beta}}_{MG}$ could be problematic, rendering
the MG estimator inefficient and unreliable.

However, as discussed above, $\boldsymbol{\hat{\beta}}_{MG}$ continues to be
an unbiased estimator of $\boldsymbol{\beta }_{0}$, even if $\boldsymbol{%
\beta }_{i}$ are correlated with $\boldsymbol{X}_{i}$, so long as the
stochastic component of $\boldsymbol{x}_{it}$ is strictly exogenous with
respect to $u_{it}$. By averaging over $\boldsymbol{\hat{\beta}}_{i}$ for $%
i=1,2,...,n$, as $n\rightarrow \infty $, the MG estimator converges to $%
\boldsymbol{\beta }_{0}$ if $T$ is sufficiently large such that $\boldsymbol{%
\hat{\beta}}_{i}$ have at least second-order moments for all $i$. The
existence of first-order moments of $\boldsymbol{\hat{\beta}}_{i}$ is
required for the MG estimator to be unbiased, and we need $\boldsymbol{\hat{%
\beta}}_{i}$ to have second-order moments for $\sqrt{n}$-consistent
estimation and valid inference about the average effects, $\boldsymbol{\beta 
}_{0}$.

When individual estimators do not have second-order moments, trimming is
required. The question is how to trim the individual estimates, $\boldsymbol{%
\hat{\beta}}_{i}$. \cite{Peng2001} proposes to categorize the individual
estimates and then use a weighted average of the estimates depending on
their left and right tail shape parameters. He assumes the estimates are
identically and independently distributed and considers only the case of a
scalar parameter. \cite{GrahamPowell2012} propose to trim by exclusion,
namely dropping individual estimates with $d_{i}$ below a threshold value.
We establish that trimming and/or shrinkage is required only if the tail
index, $\alpha _{p}$, of the distribution of $1/d_{i}$ is below or equal to $%
2$, and propose to shrink only those estimates whose $d_{i}$ are below a
threshold value, $a_{n}$, when $\alpha _{p}\leq 2$.

\subsection{A trimmed mean group estimator \label{TMGE}}

To formalize our proposed TMG estimator, we start with the least squares
estimator of $\boldsymbol{\beta }_{i}$, namely $\boldsymbol{\hat{\beta}}%
_{i}=(\boldsymbol{X}_{i}^{\prime }\boldsymbol{M}_{T}\boldsymbol{X}_{i})^{-1}%
\boldsymbol{X}_{i}^{\prime }\boldsymbol{M}_{T}\boldsymbol{y}_{i}$, and
shrink it using the threshold function $\boldsymbol{1}\{d_{i}>a_{n}\}$,
which takes the value of unity if $d_{i}>a_{n}$ and zero otherwise, where $%
d_{i}=\func{det}(\boldsymbol{X}_{i}^{\prime }\boldsymbol{M}_{T}\boldsymbol{X}%
_{i})$. The threshold value is set as 
\begin{equation}
a_{n}=C_{n}n^{-\alpha },  \label{an}
\end{equation}%
where $\alpha >0$ and $C_{n}$ is a positive constant bounded in $n$. The
choice of $\alpha $ and $C_{n}$ will be discussed below. The resultant
shrinkage estimator is given by 
\begin{equation*}
\boldsymbol{\tilde{\beta}}_{i}=\left\{ 
\begin{array}{ll}
\boldsymbol{\hat{\beta}}_{i}=(\boldsymbol{X}_{i}^{\prime }\boldsymbol{M}_{T}%
\boldsymbol{X}_{i})^{-1}\boldsymbol{X}_{i}^{\prime }\boldsymbol{M}_{T}%
\boldsymbol{y}_{i}, & \text{if }d_{i}>a_{n}, \\ 
\boldsymbol{\hat{\beta}}_{i}^{\ast }=a_{n}^{-1}\func{adj}(\boldsymbol{X}%
_{i}^{\prime }\boldsymbol{M}_{T}\boldsymbol{X}_{i})\boldsymbol{X}%
_{i}^{\prime }\boldsymbol{M}_{T}\boldsymbol{y}_{i}, & \text{if }d_{i}\leq
a_{n},%
\end{array}%
\right.
\end{equation*}%
where $\boldsymbol{\hat{\beta}}_{i}^{\ast }$ is the shrinkage version of $%
\boldsymbol{\hat{\beta}}_{i}$ since $a_{n}^{-1}\leq d_{i}^{-1}$. Written
more compactly, we have 
\begin{equation}
\boldsymbol{\tilde{\beta}}_{i}=\boldsymbol{1}\{d_{i}>a_{n}\}\boldsymbol{\hat{%
\beta}}_{i}+\boldsymbol{1}\{d_{i}\leq a_{n}\}\boldsymbol{\hat{\beta}}%
_{i}^{\ast }=(1+\delta _{i})\boldsymbol{\hat{\beta}}_{i},  \label{betai2}
\end{equation}%
where 
\begin{equation}
\delta _{i}=\left( \frac{d_{i}-a_{n}}{a_{n}}\right) \boldsymbol{1}%
\{d_{i}\leq a_{n}\}\leq 0.  \label{deltai}
\end{equation}%
We considered two versions of TMG estimators, depending on how individual
trimmed estimators, $\boldsymbol{\tilde{\beta}}_{i}$, are combined. An
obvious choice is to use a simple average of $\boldsymbol{\tilde{\beta}}_{i}$%
, namely $\overline{\boldsymbol{\tilde{\beta}}}_{n}=n^{-1}\sum_{i=1}^{n}%
\boldsymbol{\tilde{\beta}}_{i}=n^{-1}\sum_{i=1}^{n}(1+\delta _{i})%
\boldsymbol{\hat{\beta}}_{i}$, which can also be viewed as a weighted
average estimator with the weights $w_{i}=(1+\delta _{i})/n \leq 1/n$. But
it is easily seen that these weights do not add up to unity, and it might be
desirable to use the scaled weights $w_{i}/(1+\bar{\delta}%
_{n})=n^{-1}(1+\delta _{i})/(1+\bar{\delta}_{n})$, where 
\begin{equation}
\bar{\delta}_{n}=n^{-1}\sum_{i=1}^{n}\delta _{i}.  \label{deltan}
\end{equation}%
Using these modified weights, we propose the following TMG estimator 
\begin{equation}
\boldsymbol{\hat{\beta}}_{TMG}=n^{-1}\sum_{i=1}^{n}\left( \frac{1+\delta _{i}%
}{1+\bar{\delta}_{n}}\right) \boldsymbol{\hat{\beta}}_{i},  \label{TMGb}
\end{equation}%
which can be written more compactly as 
\begin{equation}
\boldsymbol{\hat{\beta}}_{TMG}=n^{-1}\sum_{i=1}^{n}\left( 1+\bar{\delta}%
_{n}\right) ^{-1}\boldsymbol{Q}_{i}^{\prime }\boldsymbol{y}_{i},
\label{TMGc}
\end{equation}%
where 
\begin{equation}
\boldsymbol{Q}_{i}^{\prime }=\left( 1+\delta _{i}\right) \left( \boldsymbol{X%
}_{i}^{\prime }\boldsymbol{M}_{T}\boldsymbol{X}_{i}\right) ^{-1}\boldsymbol{X%
}_{i}^{\prime }\boldsymbol{M}_{T}.  \label{Qi}
\end{equation}

The TMG estimator differs in two important respects from the trimmed
estimator proposed by \cite{GrahamPowell2012}, which in the context of our
setup (and abstracting from time effects for now) can be written as 
\begin{equation}
\boldsymbol{\hat{\beta}}_{GP}=\frac{\sum_{i=1}^{n}\boldsymbol{1}%
\{d_{i,GP}>h_{n}^{2}\}\boldsymbol{\hat{\beta}}_{i}}{\sum_{i=1}^{n}%
\boldsymbol{1}\{d_{i,GP}>h_{n}^{2}\}},  \label{gpe}
\end{equation}%
where $d_{i,GP}=\det (\boldsymbol{W}_{i}^{\prime }\boldsymbol{W}_{i})$, and $%
\boldsymbol{W}_{i}=\left( \boldsymbol{\tau }_{T},\boldsymbol{X}_{i}\right) $%
. In the special case where $T=k$, $d_{i,GP}=\left\vert \func{det}\left( 
\boldsymbol{W}_{i}\right) \right\vert ^{2}$ and the threshold function
considered by GP reduces to $\boldsymbol{1}\{\left\vert \func{det}\left( 
\boldsymbol{W}_{i}\right) \right\vert >h_{n}\}$.\footnote{%
Note that $d_{i,GP}=\det (\boldsymbol{W}_{i}^{\prime }\boldsymbol{W}%
_{i})=T\cdot d_{i}$, where $d_{i}=\det (\boldsymbol{X}_{i}^{\prime }%
\boldsymbol{M}_{T}\boldsymbol{X}_{i})$, and the two criteria differ only by
a multiplicative constant. Moreover, our trimming threshold $a_{n}$ in (\ref%
{an}) includes a scaling constant $C_{n}$, whose value is calibrated to the
empirical distribution of $d_{i}$ to account for its scale. Similarly, GP
also allow for a constant term in their threshold function to normalize the
scale of $d_{i,GP}$.}

The GP approach can be viewed as trimming by exclusion, which overlooks the
information that might be contained in $\func{adj}(\boldsymbol{W}%
_{i}^{\prime }\boldsymbol{W}_{i})$ when $\left\vert \func{det}\left( 
\boldsymbol{W}_{i}\right) \right\vert \leq h_{n}$. More specifically, to
relate $\boldsymbol{\hat{\beta}}_{TMG}$ to the GP estimator given by (\ref%
{gpe}), using (\ref{betai2}) in (\ref{TMGb}), we note that%
\begin{equation}
\boldsymbol{\hat{\beta}}_{TMG}=\frac{1-\pi _{n}}{1+\bar{\delta}_{n}}\left( 
\frac{\sum_{i=1}^{n}\boldsymbol{1}\{d_{i}>a_{n}\}\boldsymbol{\hat{\beta}}_{i}%
}{\sum_{i=1}^{n}\boldsymbol{1}\{d_{i}>a_{n}\}}\right) +\frac{\pi _{n}}{1+%
\bar{\delta}_{n}}\left( \frac{\sum_{i=1}^{n}\boldsymbol{1}\{d_{i}\leq a_{n}\}%
\boldsymbol{\hat{\beta}}_{i}^{\ast }}{\sum_{i=1}^{n}\boldsymbol{1}%
\{d_{i}\leq a_{n}\}}\right),  \label{betaTn}
\end{equation}%
where $\pi _{n}$ is the fraction of the estimates being trimmed given by 
\begin{equation}
\pi _{n}=\frac{\sum_{i=1}^{n}\boldsymbol{1}\{d_{i}\leq a_{n}\}}{n}.
\label{pin}
\end{equation}%
Compared to $\boldsymbol{\hat{\beta}}_{TMG}$, the GP estimator places zero
weights on the estimates with $d_{i}\leq a_{n}$, and leaves out the scaling
factor, $\left( 1+\bar{\delta}_{n}\right) ^{-1}$, which, as already noted,
can play an important role in the \textit{small sample} performance of the
TMG estimator. Our proposed method also differs from GP in the way we
motivate and calibrate the threshold function.

\section{Asymptotic properties of the TMG estimator\label{AsyTMG}}

To investigate the asymptotic properties of the TMG estimator, $\boldsymbol{%
\hat{\beta}}_{TMG}$, we introduce the following additional assumption about
the distribution of $d_{i}$:

\begin{assumption}[distribution of $d_{i}$]
\label{distributiondi} Let $d_{i}=\func{det}\left( \boldsymbol{X}%
_{i}^{\prime }\boldsymbol{M}_{T}\boldsymbol{X}_{i}\right) $, where $%
\boldsymbol{X}_{i}=(\boldsymbol{x}_{i1},\boldsymbol{x}_{i2},...,\boldsymbol{x%
}_{iT})^{\prime }$ is the $T\times k^{\prime }$ matrix of observations on $%
\boldsymbol{x}_{it}$ in the heterogeneous panel data model given by (\ref{m1}%
) and $\boldsymbol{M}_{T}=\boldsymbol{I}_{T}-\boldsymbol{\tau }_{T}%
\boldsymbol{\tau }_{T}^{\prime }/T$. For $i=1,2,...,n$, $d_{i}$ are random
draws from the probability distribution function, $F_{d}(u)$, with a
continuously differentiable density function, $f_{d}(u)$, over $u\in
(0,\infty )$, such that $F_{d}(0)=0$, and $F_{d}(a_{n})\thicksim
C_{f}a_{n}^{\alpha _{p}}$, as $n\rightarrow \infty $, where $a_{n}$ is given
by (\ref{an}), $C_{f}$ is bounded in $n$, and $\alpha _{p}>0$.
\end{assumption}

\begin{remark}
\label{Pareto}Parameter $\alpha _{p}$ in Assumption \ref{distributiondi} is
related to the tail index of Pareto-type distributions of $z=1/d$ given by%
\footnote{%
Pareto-type distributions cover a wide class of distributions including
Pareto, Cauchy, Burr and Stable distributions with exponent $\alpha _{p}<2$.
See Example 3.3.10 (p. 133) in \cite{EmbrechtsEtal1997}.} 
\begin{equation*}
\Pr \left( z>b_{n}\right) \thicksim C_{p}b_{n}^{-\alpha _{p}}\text{, }%
b_{n}\rightarrow \infty .
\end{equation*}%
For $d>0$, 
\begin{equation*}
\Pr \left( z>b_{n}\right) =\Pr \left( 1/d>b_{n}\right) =\Pr \left(
d<1/b_{n}\right) =F_{d}\left( b_{n}^{-1}\right) \thicksim
C_{f}b_{n}^{-\alpha _{p}}.
\end{equation*}%
Connecting $\alpha _{p}$ to the shape parameter of the Pareto-type
distributions of $1/d$ allows us to utilize the rich and extensive
literature that already exists on estimation of the tail index. See, for
example, Section 6.4.2 of \cite{EmbrechtsEtal1997}.
\end{remark}

\begin{remark}
Note that when $\alpha _{p}>2$, $E\left( d_{i}^{-2}\right) <C$, condition (%
\ref{sufficient}) will be met, and as a result, the MG estimator becomes $%
\sqrt{n}$-consistent. Trimming is only required if $\alpha _{p}\leq 2$.
\end{remark}

\begin{remark}
\label{Fu} By the mean value theorem, $F_{d}(a_{n})=%
\int_{0}^{a_{n}}f_{d}(u)du=f_{d}(\bar{a}_{n})a_{n}$, where $\bar{a}_{n}$
lies on the segment $(0,a_{n})$. Hence, under Assumption \ref{distributiondi}%
, it follows that $f_{d}(\bar{a}_{n})=O\left( a_{n}^{\alpha _{p}-1}\right) $.
\end{remark}

\begin{remark}
Assumption \ref{distributiondi}, by imposing $F_{d}(0)=0$, rules out the
case where $d_{i}=0$ for some units. In practice, units with $d_{i}=0$ can
be excluded when computing the TMG estimates. This case arises, for example,
when $k^{\prime }=1$ and $T=2$, and $d_{i}=\func{det}(\boldsymbol{x}%
_{i}^{\prime }\boldsymbol{M}_{T}\boldsymbol{x}_{i})=\frac{1}{2}\left(
x_{i2}-x_{i1}\right) ^{2}=0$, for some units. Units with $x_{i2}=x_{i1}$ are
known as stayers, and those with very small values of $\left\vert
x_{i2}-x_{i1}\right\vert $ (below a given threshold value) are known as slow
movers and discussed by GP, and more recently by \cite{SasakiUra2026} (SU).
Our proposed TMG estimator does exploit the information on slow movers,
namely those with $a_{n}\geq d_{i}>0$, and excludes units with $d_{i}=0$
only. Due to the error cross-sectional independence, we conjecture that the
TMG estimator will perform well even if there are stayers that are excluded
from our analysis, so long as the fraction of stayers in the sample is not
too large. In cases where stayers form a large fraction of the units, the
TMG estimator can be augmented with an estimator for the stayers computed
using the extrapolation approach proposed by SU. Following such a mixed
strategy, the extrapolation error must be balanced against the loss of
efficiency associated with using TMG that excludes the stayers.
\end{remark}

\begin{assumption}[dependence of $\boldsymbol{\protect\beta }_{i}$ on $d_{i}$%
]
\label{CRE} Dependence of $\boldsymbol{\beta }_{i}=(\beta _{i1},\beta
_{i2},...,\beta _{ik^{\prime }})^{\prime }$ on $d_{i}$ is characterized by%
\begin{equation}
\boldsymbol{\beta }_{i}=\boldsymbol{\beta }_{0}+\boldsymbol{\eta }_{i},
\label{etai}
\end{equation}%
where $\boldsymbol{\eta }_{i}$ are independently distributed over $i$,%
\begin{equation}
\boldsymbol{\eta }_{i}=\boldsymbol{B}_{i}\left\{ \boldsymbol{g}(d_{i})-E%
\left[ \boldsymbol{g}(d_{i})\right] \right\} +\boldsymbol{\epsilon }_{i},
\label{etai2}
\end{equation}%
$E(\boldsymbol{\epsilon }_{i}|d_{i})=\boldsymbol{0}$, $\func{sup}%
_{i}E\left\Vert \boldsymbol{\epsilon }_{i}\right\Vert ^{4}<C$, $\boldsymbol{g%
}(u)=(g_{1}(u),g_{2}(u),...,g_{k^{\prime }}(u))^{\prime }$, and $g_{j}(u)$
for $j=1,2,...,k^{\prime }$ are bounded and continuously differentiable
functions of $u$ on $(0,\infty )$. $\boldsymbol{B}_{i}$ are bounded $%
k^{\prime }\times k^{\prime }$ matrices of fixed constants with $\func{sup}%
_{i}\left\Vert \boldsymbol{B}_{i}\right\Vert <C$.
\end{assumption}

\begin{remark}
Under Assumption \ref{distributiondi}, $d_{i}$ are distributed independently
over $i$, which implies that $\delta _{i}$, defined by (\ref{deltai}), are
also distributed independently over $i$. The cross-sectional independence
assumption is made to simplify the mathematical exposition. It can be
relaxed by requiring $\boldsymbol{\eta }_{i}$ and $\delta _{i}$ to be weakly
cross-sectionally correlated.
\end{remark}

\begin{remark}
\label{nonstationarity} Our framework does not require the regressors or
error processes to be stationary over time. Under Assumptions \ref{errorsx}, %
\ref{rcm}(a), \ref{distributiondi}, and \ref{CRE}, the TMG estimation allows
for some forms of nonstationarity, including when regressors are generated
by factor-augmented autoregressions and serially correlated errors with
time-varying volatility. In general, TMG estimators are robust to
heterogeneity in the second-order moments of $u_{it}$, $\boldsymbol{x}_{it}$%
, and $\boldsymbol{\eta }_{i}$.
\end{remark}

Using (\ref{betaihat2}) and (\ref{etai}) in (\ref{betai2}), we have 
\begin{equation*}
\boldsymbol{\tilde{\beta}}_{i}=(1+\delta _{i})\boldsymbol{\beta }%
_{0}+(1+\delta _{i})\left( \boldsymbol{\eta }_{i}+\boldsymbol{\xi }%
_{iT}\right) ,
\end{equation*}%
and $\boldsymbol{\hat{\beta}}_{TMG}$ defined by (\ref{TMGb}) can be written
as 
\begin{equation}
\boldsymbol{\hat{\beta}}_{TMG}-\boldsymbol{\beta }_{0}=\left( \frac{1}{1+%
\bar{\delta}_{n}}\right) n^{-1}\sum_{i=1}^{n}(1+\delta _{i})\left( 
\boldsymbol{\eta }_{i}+\boldsymbol{\xi }_{iT}\right) .  \label{TMG_egzi}
\end{equation}%
(\ref{TMG_egzi}) can be written equivalently as 
\begin{equation}
\boldsymbol{\hat{\beta}}_{TMG}-\boldsymbol{\beta }_{0}=\left( \frac{%
1+E\left( \bar{\delta}_{n}\right) }{1+\bar{\delta}_{n}}\right) \left( 
\boldsymbol{b}_{n}+n^{-1}\sum_{i=1}^{n}\left[ \boldsymbol{p}_{i}-E\left( 
\boldsymbol{p}_{i}\right) \right] +n^{-1}\sum_{i=1}^{n}\boldsymbol{q}%
_{iT}\right) ,  \label{TMGgap1}
\end{equation}%
where 
\begin{equation}
\boldsymbol{b}_{n}=n^{-1}\sum_{i=1}^{n}E\left( \boldsymbol{p}_{i}\right)
=n^{-1}\sum_{i=1}^{n}\frac{E(\delta _{i}\boldsymbol{\eta }_{i})}{1+E\left( 
\bar{\delta}_{n}\right) },\text{ }  \label{bn}
\end{equation}
\begin{equation}
\boldsymbol{p}_{i}=\left[ 1+E\left( \bar{\delta}_{n}\right) \right]
^{-1}\left( 1+\delta _{i}\right) \boldsymbol{\eta }_{i} \text{, and }%
\boldsymbol{q}_{iT}=\left[ 1+E\left( \bar{\delta}_{n}\right) \right]
^{-1}\left( 1+\delta _{i}\right) \boldsymbol{\xi }_{iT}.  \label{piqiT}
\end{equation}%
Also, by Lemma \ref{deltaeta}, $E(\delta _{i})=O(a_{n}^{\alpha _{p}})$, $%
E\left( \bar{\delta}_{n}\right) =O(a_{n}^{\alpha _{p}})$, $E\left( \delta
_{i}\boldsymbol{\eta }_{i}\right) =O(a_{n}^{\alpha _{p}})$, and hence 
\begin{equation*}
\boldsymbol{b}_{n}=\frac{1}{1+E\left( \bar{\delta}_{n}\right) }\left[
n^{-1}\sum_{i=1}^{n}E(\delta _{i}\boldsymbol{\eta }_{i})\right] =\frac{%
O(a_{n}^{\alpha _{p}})}{1+O(a_{n}^{\alpha _{p}})}=O(a_{n}^{\alpha
_{p}})=O(n^{-\alpha \alpha _{p}}).
\end{equation*}%
Since under Assumption \ref{distributiondi}, $\delta _{i}-E\left( \delta
_{i}\right) $ is distributed independently over $i$ with a zero mean and
bounded variance, then 
\begin{equation}
\frac{1+E\left( \bar{\delta}_{n}\right) }{1+\bar{\delta}_{n}}=1-\frac{\bar{%
\delta}_{n}-E\left( \bar{\delta}_{n}\right) }{1+E\left( \bar{\delta}%
_{n}\right) +\left( \bar{\delta}_{n}-E\left( \bar{\delta}_{n}\right) \right) 
}=1+O_{p}\left( n^{-1/2-\alpha \alpha _{p}}/2\right) .  \label{deltaOrder}
\end{equation}%
Similarly, under Assumptions \ref{distributiondi} and \ref{CRE}, $%
\boldsymbol{p}_{i}-E\left( \boldsymbol{p}_{i}\right) $ is distributed
independently over $i$ with zero means and bounded variances, and we have%
\begin{equation*}
\frac{1}{n}\sum_{i=1}^{n}\left[ \boldsymbol{p}_{i}-E\left( \boldsymbol{p}%
_{i}\right) \right] =\left[ 1+E\left( \bar{\delta}_{n}\right) \right]
^{-1}\left\{ \frac{1}{n}\sum_{i=1}^{n}\boldsymbol{\eta }_{i}+\frac{1}{n}%
\sum_{i=1}^{n}\left[ \delta _{i}\boldsymbol{\eta }_{i}-E\left( \delta _{i}%
\boldsymbol{\eta }_{i}\right) \right] \right\} =O_{p}(n^{-1/2}).
\end{equation*}%
Consider now $\boldsymbol{\bar{q}}_{nT}=n^{-1}\sum_{i=1}^{n}\boldsymbol{q}%
_{iT}$, where $\boldsymbol{q}_{iT}$ is defined by (\ref{piqiT}), and note
that 
\begin{equation}
\boldsymbol{\bar{q}}_{nT}=\left( \frac{1}{1+E\left( \bar{\delta}_{n}\right) }%
\right)\boldsymbol{\bar{\varkappa}}_{nT},  \label{qbarnt}
\end{equation}%
where $\boldsymbol{\bar{\varkappa}}_{nT}=n^{-1}\sum_{i=1}^{n}\left( 1+\delta
_{i}\right) \boldsymbol{\xi }_{iT}$. Using results of Lemma \ref{EVegzi}, we
have $E\left( \boldsymbol{\bar{q}}_{nT}\right) =\boldsymbol{0}$, and for any 
$\alpha _{p}>0$, 
\begin{equation*}
Var\left( \boldsymbol{\bar{q}}_{nT}\right) =\left( \frac{1}{1+E\left( \bar{%
\delta}_{n}\right) }\right) ^{2}Var\left( \boldsymbol{\bar{\varkappa}}%
_{nT}\right) =O(n^{-1}a_{n}^{-1})+O(n^{-1}a_{n}^{-1+\alpha
_{p}/2})=O(n^{-1}a_{n}^{-1}).
\end{equation*}%
Hence, $\boldsymbol{\bar{q}}_{nT}=O_{p}\left( n^{-1/2+\alpha /2}\right) $.
Using the above results in (\ref{TMGgap1}), we have%
\begin{equation}
\boldsymbol{\hat{\beta}}_{TMG}-\boldsymbol{\beta }_{0}=O(n^{-\alpha \alpha
_{p}})+O_{p}\left( n^{-\frac{(1-\alpha )}{2}}\right) .  \label{TMGgap2}
\end{equation}%
Thus, $\boldsymbol{\hat{\beta}}_{TMG}$ converges to $\boldsymbol{\beta }_{0}$
asymptotically so long as $0<\alpha <1$ as $n\rightarrow \infty $. The
convergence rate of $\boldsymbol{\hat{\beta}}_{TMG}$ to $\boldsymbol{\beta }%
_{0}$ will depend on the trade-off between the bias and variance of $%
\boldsymbol{\hat{\beta}}_{TMG}$. Though it is possible to reduce the bias of 
$\boldsymbol{\hat{\beta}}_{TMG}$ by choosing a value of $\alpha $ close to
unity, it will be at the expense of a larger variance. In what follows, we
shed light on the choice of $\alpha $ by considering the conditions under
which the asymptotic distribution of $\boldsymbol{\hat{\beta}}_{TMG}$ is
centered around $\boldsymbol{\beta }_{0}$ and at the same time, its
asymptotic variance tends to zero at a reasonably fast rate.

\subsection{The choice of the trimming threshold value \label{Threshold}}

We begin by assuming that the rate at which $\boldsymbol{\hat{\beta}}_{TMG}$
converges to $\boldsymbol{\beta }_{0}$ is given by $n^{\gamma }$, where $%
\gamma $ is set in relation to $\alpha $. Since it is not guaranteed that
the individual estimators have second order moments when $T$ is ultra-short,
we expect the rate, $n^{\gamma }$, to be below the regular rate of $n^{1/2}$%
. Using (\ref{TMGgap1}) and (\ref{deltaOrder}) and noting that $\gamma \leq
1/2$ (with equality holding only under regular convergence), we have 
\begin{equation}
n^{\gamma }\left( \boldsymbol{\hat{\beta}}_{TMG}-\boldsymbol{\beta }%
_{0}\right) =n^{\gamma }\boldsymbol{b}_{n}+n^{\gamma -\frac{1-\alpha }{2}}%
\left[ n^{-\frac{1+\alpha }{2}}\sum_{i=1}^{n}\left[ \boldsymbol{p}%
_{i}-E\left( \boldsymbol{p}_{i}\right) \right] +n^{-\frac{1+\alpha }{2}%
}\sum_{i=1}^{n}\boldsymbol{q}_{iT}\right] +o_{p}(1).  \label{Asy3}
\end{equation}%
To ensure that the asymptotic distribution of $\boldsymbol{\hat{\beta}}%
_{TMG} $ is correctly centered, we must have $n^{\gamma }\boldsymbol{b}%
_{n}\rightarrow \boldsymbol{0,}$ as $n\rightarrow \infty $. Since $n^{\gamma
}\boldsymbol{b}_{n}=O(n^{\gamma }a_{n}^{\alpha _{p}})=O(n^{\gamma -\alpha
\alpha _{p}})$, this condition is ensured if $\gamma <\alpha \alpha _{p}$.
Turning to the second term of the above, we note that to obtain a
non-degenerate distribution, we also need to set $\gamma =\left( 1-\alpha
\right) /2$. Combining these two requirements yields $\left( 1-\alpha
\right) /2<\alpha \alpha _{p}$, or 
\begin{equation}
\alpha >\frac{1}{1+2\alpha _{p}}.  \label{calpha}
\end{equation}%
In view of (\ref{TMGgap2}), the convergence rate of $\boldsymbol{\hat{\beta}}%
_{TMG}-\boldsymbol{\beta }_{0}$, namely $n^{-\frac{(1-\alpha )}{2}}$,
depends on $\alpha _{p}$, which governs the tail property of the
distribution of $1/d_{i}$. In the simple example \ref{ex1}, assuming $\sigma
_{ix}^{2}=\sigma _{x}^{2}$, then $1/d_{i}=(1/\sigma _{x}^{2})(1/\chi
_{v}^{2})$, and the shape parameter of $1/d_{i}$ will equal the shape
parameter of $1/\chi _{v}^{2}$, which is given by $\alpha _{p}=v/2$. In the
simple case where $k^{\prime }=1$, we have $\alpha _{p}=(T-1)/2$. In the
worst case scenario with $T=2=k$, $\alpha _{p}=1/2$, and the convergence
rate of $\boldsymbol{\hat{\beta}}_{TMG}$ is at most $n^{-1/4}$, well below
the regular convergence rate of $n^{-1/2}$. The optimum choice of $\alpha $
depends on $\alpha _{p}$, which in turn is determined by the rate at which $%
F_{d}(a_{n})$ tends to zero as $a_{n}\rightarrow 0$.

To ensure that the threshold function, $\boldsymbol{1}\{d_{i}>C_{n}n^{-%
\alpha }\}$, does not depend on the scale of $\boldsymbol{x}_{it}$, we
suggest setting $C_{n}=\bar{d}_{n}=n^{-1}\sum_{i=1}^{n}d_{i}>0$. With this
choice, the threshold function can be written as $\boldsymbol{1}%
\{d_{i}>C_{n}n^{-\alpha }\}=\boldsymbol{1}\{d_{i}/\bar{d}_{n}>n^{-\alpha }\}$%
, noting that $d_{i}/\bar{d}_{n}$ is scale free. The selection of $\alpha $
is more complicated. In practice, a two-step procedure can be implemented,
whereby in the first step the value of $\alpha _{p}$ is estimated using the
observations $\left\{ 1/d_{i}\right\} _{i=1}^{n}$.\footnote{%
Asymptotically, the estimator of $\alpha _{p}$ does not depend on the scale
of $d_{i}$.} Then TMG estimation can be carried out using $\hat{\alpha}%
=1/(1+2\hat{\alpha}_{p})+\epsilon $, where $\hat{\alpha}_{p}$ is a
consistent estimator of $\alpha _{p}$, and $\epsilon $ is a small positive
constant. See also Remark \ref{Pareto}.

Given the uncertainty associated with the estimation of $\alpha _{p}$, we
also considered setting $\alpha _{p}=1$, and using the threshold function $%
\boldsymbol{1}\{d_{i}/\bar{d}_{n}>n^{-1/3}\}$. This approach has the
advantage of being simple to implement. The small sample performances of
these two approaches to setting $\alpha$ are compared using MC experiments.
See sub-section \ref{MChatalpha} of the online supplement. The simulation
results show that the simple threshold rule $\boldsymbol{1}\{d_{i}/\bar{d}%
_{n}>n^{-1/3}\}$ performs better when $T$ is ultra-short ($T=k$), and has a
similar performance when the estimated value of $\hat{\alpha}_{p}$ is used
if $T>k$.

GP show that to correctly center the limiting distribution of their proposed
estimator, $\boldsymbol{\hat{\beta}}_{GP}$ given by (\ref{gpe}), $h_{n}$ in
their threshold function $\boldsymbol{1}\{|\det (\boldsymbol{W}%
_{i})|>h_{n}\} $ must be set as $h_{n}=C_{GP}n^{-\alpha _{GP}}$, such that $%
(nh_{n})^{1/2}h_{n}\rightarrow 0$, as $n\rightarrow \infty $, which implies
that $\alpha _{GP}>1/3$ (see p. 2125 and p. 2138 of GP). SU adopt the same
choice and set $h_{n}=C_{GP}n^{-\alpha _{GP}}n^{-L/(1+2L)}$, with $L\in
\{1,2\}$ being the order of local polynomials used to infer the average
effects of stayers (see sub-section 5.2 on p. 16 of SU).

\begin{theorem}[Asymptotic distribution of TMG estimator]
\label{thm_asytmg} Suppose for $i=1,2,...,n$ and $t=1,2,...,T$, $y_{it}$ are
generated by the heterogeneous panel data model (\ref{m1}), Assumptions \ref%
{errorsx}, \ref{rcm}(a), \ref{distributiondi}, and \ref{CRE} hold, and $%
\alpha _{p}\in (0,2]$ is the shape parameter of the tail probability
distribution of $1/d_{i}$ over $i$, such that $\func{plim}_{n\rightarrow
\infty }\bar{d}_{n}>0$, where $\bar{d}_{n}=n^{-1}\sum_{i=1}^{n}d_{i}>0$.
Then for $\alpha >1/(1+2\alpha _{p})$, 
\begin{equation*}
n^{(1-\alpha )/2}\left( \boldsymbol{\hat{\beta}}_{TMG}-\boldsymbol{\beta }%
_{0}\right) \rightarrow _{d}N\left( \boldsymbol{0}_{k^{\prime }},\boldsymbol{%
V}_{\beta }\right) ,\text{ as }n\rightarrow \infty ,
\end{equation*}%
where $\boldsymbol{\hat{\beta}}_{TMG}$ is given by (\ref{TMGb}), 
\begin{equation}
\boldsymbol{V}_{\beta }=Avar\left( n^{(1-\alpha )/2}\boldsymbol{\hat{\beta}}%
_{TMG}\right) =C^{-1}\lim_{n\rightarrow \infty }n^{-1}\sum_{i=1}^{n}E\left[
a_{n}\boldsymbol{1}\{d_{i}>a_{n}\}\boldsymbol{R}_{i}^{\prime }\boldsymbol{H}%
_{i}\boldsymbol{R}_{i}\right] ,  \label{Varbeta}
\end{equation}%
$\boldsymbol{H}_{i}=E\left( \boldsymbol{u}_{i}\boldsymbol{u}_{i}^{\prime
}\left\vert \boldsymbol{X}_{i}\right. \right) $, $\boldsymbol{R}_{i}^{\prime
}=\left( \boldsymbol{X}_{i}^{\prime }\boldsymbol{M}_{T}\boldsymbol{X}%
_{i}\right) ^{-1}\boldsymbol{X}_{i}^{\prime }\boldsymbol{M}_{T}$, $d_{i}=%
\func{det}\left( \boldsymbol{X}_{i}^{\prime }\boldsymbol{M}_{T}\boldsymbol{X}%
_{i}\right) >0$, $a_{n}=C_{n}n^{-\alpha }$, and $C=\func{plim}_{n\rightarrow
\infty }C_{n}=\func{plim}_{n\rightarrow \infty }\bar{d}_{n}$.
\end{theorem}

For a proof, see sub-section \ref{Proofthm1} of the mathematical appendix.

\subsection{Robust estimation of the covariance matrix of the TMG estimator}

As with standard MG estimation, consistent estimation of $\boldsymbol{V}%
_{\beta }$ using (\ref{Varbeta}) requires knowledge of $\boldsymbol{H}_{i}$
which cannot be estimated consistently when $T$ is short. We follow the
literature and propose a robust covariance estimator of $\boldsymbol{V}%
_{\beta }$, which is asymptotically unbiased for a wide class of error
variances, $E\left( \boldsymbol{u}_{i}\boldsymbol{u}_{i}^{\prime }\left\vert 
\boldsymbol{X}_{i}\right. \right) =\boldsymbol{H}_{i}(\boldsymbol{X}_{i})$,
thus allowing for serially correlated and conditionally heteroskedastic
errors. The following theorem summarizes the main result.

\begin{theorem}[Robust covariance matrix of TMG estimator]
\label{VarCon} Suppose Assumptions \ref{errorsx}, \ref{rcm}(a), \ref%
{distributiondi}, and \ref{CRE} hold and $\boldsymbol{\beta }_{0}$ is
estimated by $\boldsymbol{\hat{\beta}}_{TMG}$ given by (\ref{TMGb}). Then
for $\alpha >1/(1+2\alpha _{p})$, where $\alpha _{p}\in (0,2]$ is the shape
parameter of the tail probability distribution of $1/d_{i}$ over $i$, as $%
n\rightarrow \infty $, 
\begin{equation}
\boldsymbol{V}_{\beta }=\func{plim}_{n\rightarrow \infty }\left[
n^{-(1+\alpha )}\sum_{i=1}^{n}\left( \boldsymbol{\tilde{\beta}}_{i}-%
\boldsymbol{\hat{\beta}}_{TMG}\right) \left( \boldsymbol{\tilde{\beta}}_{i}-%
\boldsymbol{\hat{\beta}}_{TMG}\right) ^{\prime }\right] ,  \label{Vbeta}
\end{equation}%
where $\boldsymbol{V}_{\beta }=Avar\left( n^{(1-\alpha )/2}\boldsymbol{\hat{%
\beta}}_{TMG}\right) $ is defined by (\ref{Varbeta}). Accordingly, $Var(%
\boldsymbol{\hat{\beta}}_{TMG})$ can be consistently estimated by%
\begin{equation}
\widehat{Var(\boldsymbol{\hat{\beta}}_{TMG})}=\frac{1}{n(n-1)(1+\bar{\delta}%
_{n})^{2}}\sum_{i=1}^{n}\left( \boldsymbol{\tilde{\beta}}_{i}-\boldsymbol{%
\hat{\beta}}_{TMG}\right) \left( \boldsymbol{\tilde{\beta}}_{i}-\boldsymbol{%
\hat{\beta}}_{TMG}\right) ^{\prime }.  \label{varC}
\end{equation}%
where $\boldsymbol{\tilde{\beta}}_{i}$ and $\bar{\delta}_{n}$ are given by (%
\ref{betai2}) and (\ref{deltan}), respectively.
\end{theorem}

For a proof, see sub-section \ref{ProofVarCon} of the mathematical appendix.

\section{Panels with time effects\label{CRCTE}}

The panel data model with time effects is given by 
\begin{equation}
y_{it}=\alpha _{i}+\phi _{t}+\boldsymbol{x}_{it}^{\prime }\boldsymbol{\beta }%
_{i}+u_{it},  \label{panTE1}
\end{equation}%
where $\phi _{t}$ for $t=1,2,...,T$ are the time effects. Without loss of
generality, we adopt the normalization $\boldsymbol{\tau }_{T}^{\prime }%
\boldsymbol{\phi }=0$, where $\boldsymbol{\phi }=(\phi _{1},\phi
_{2},...,\phi _{T})^{\prime }$. To estimate $\boldsymbol{\beta }_{0}$,
initially we suppose $\boldsymbol{\phi }$ is known, and the trimmed
estimator of $\boldsymbol{\beta }_{i}$ is now given by $\boldsymbol{\tilde{%
\beta}}_{i}(\boldsymbol{\phi })=\boldsymbol{Q}_{i}^{\prime }(\boldsymbol{y}%
_{i}-\boldsymbol{\phi })=\boldsymbol{\tilde{\beta}}_{i}-\boldsymbol{Q}%
_{i}^{\prime }\boldsymbol{\phi }$, where $\boldsymbol{Q}_{i}$ is defined by (%
\ref{Qi}). The associated TMG-TE estimator of $\boldsymbol{\beta }_{0}$ then
follows as 
\begin{equation*}
\boldsymbol{\hat{\beta}}_{TMG-TE}(\boldsymbol{\phi })=n^{-1}\sum_{i=1}^{n}%
\left( 1+\bar{\delta}_{n}\right) ^{-1}\boldsymbol{\tilde{\beta}}_{i}(%
\boldsymbol{\phi })=\boldsymbol{\hat{\beta}}_{TMG}-\boldsymbol{\bar{Q}}%
_{n}^{\prime }\boldsymbol{\phi },
\end{equation*}%
where $\boldsymbol{\hat{\beta}}_{TMG}$ is given by (\ref{TMGc}), and 
\begin{equation}
\boldsymbol{\bar{Q}}_{n}=\frac{1}{1+\bar{\delta}_{n}}\left(
n^{-1}\sum_{i=1}^{n}\boldsymbol{Q}_{i}\right) .  \label{Qn}
\end{equation}

From our earlier analysis, it is clear that for a known $\boldsymbol{\phi }$%
, $\boldsymbol{\hat{\beta}}_{TMG-TE}(\boldsymbol{\phi })$ has the same
asymptotic distribution as $\boldsymbol{\hat{\beta}}_{TMG}$ with $%
\boldsymbol{y}_{i}$ replaced by $\boldsymbol{y}_{i}-\boldsymbol{\phi }$.
When $T>k$, we can follow \cite{Chamberlain1992} and eliminate the time
effects by the de-meaning transformation $\boldsymbol{M}_{i}=\boldsymbol{I}%
_{T}-\boldsymbol{M}_{T}\boldsymbol{X}_{i}(\boldsymbol{X}_{i}^{\prime }%
\boldsymbol{M}_{T}\boldsymbol{X}_{i})^{-1}\boldsymbol{X}_{i}^{\prime }%
\boldsymbol{M}_{T}$. Under the normalization $\boldsymbol{\tau }_{T}^{\prime
}\boldsymbol{\phi }=0$, we have $\boldsymbol{M}_{T}\boldsymbol{\phi }=%
\boldsymbol{\phi }$, and$\ \boldsymbol{M}_{T}\boldsymbol{y}_{i}=\boldsymbol{M%
}_{T}\boldsymbol{X}_{i}\boldsymbol{\beta }_{i}+\boldsymbol{\phi }+%
\boldsymbol{M}_{T}\boldsymbol{u}_{i}$. Then $\boldsymbol{M}_{i}\boldsymbol{M}%
_{T}\boldsymbol{y}_{i}=\boldsymbol{M}_{i}\boldsymbol{\phi }+\boldsymbol{M}%
_{i}\boldsymbol{M}_{T}\boldsymbol{u}_{i}$, and averaging over $i$ we obtain 
\begin{equation}
n^{-1}\sum_{i=1}^{n}\boldsymbol{M}_{i}\boldsymbol{M}_{T}\boldsymbol{y}%
_{i}=\left( n^{-1}\sum_{i=1}^{n}\boldsymbol{M}_{i}\right) \boldsymbol{\phi }%
+n^{-1}\sum_{i=1}^{n}\boldsymbol{M}_{i}\boldsymbol{M}_{T}\boldsymbol{u}_{i}.
\label{TE-C}
\end{equation}%
Hence, $\boldsymbol{\phi }$ can be estimated without knowing $\boldsymbol{%
\beta }_{0}$, if $\boldsymbol{\bar{M}}_{n}=n^{-1}\sum_{i=1}^{n}\boldsymbol{M}%
_{i}$ is a positive definite matrix. This requires $T>k$, since $\boldsymbol{%
\bar{M}}_{n}$ is singular if $T=k$. Therefore, to implement Chamberlain's
estimation approach, we require the following assumption:

\begin{assumption}[identification of time effects]
\label{invertm} For $T>k$, $\boldsymbol{\bar{M}}_{n}=n^{-1}\sum_{i=1}^{n}%
\boldsymbol{M}_{i}\rightarrow _{p}\boldsymbol{M}\succ \boldsymbol{0}$, where 
$\boldsymbol{M}_{i}=\boldsymbol{I}_{T}-\boldsymbol{M}_{T}\boldsymbol{X}_{i}(%
\boldsymbol{X}_{i}^{\prime }\boldsymbol{M}_{T}\boldsymbol{X}_{i})^{-1}%
\boldsymbol{X}_{i}^{\prime }\boldsymbol{M}_{T}$.
\end{assumption}

Under this Assumption, $\boldsymbol{\phi }$ can be estimated by%
\begin{equation}
\boldsymbol{\hat{\phi}}_{C}=\left( n^{-1}\sum_{i=1}^{n}\boldsymbol{M}%
_{i}\right) ^{-1}\left( n^{-1}\sum_{i=1}^{n}\boldsymbol{M}_{i}\boldsymbol{M}%
_{T}\boldsymbol{y}_{i}\right) ,  \label{TE2}
\end{equation}%
and its asymptotic distribution follows straightforwardly. Specifically,
using (\ref{TE-C}) we have 
\begin{equation}
\sqrt{n}\left( \boldsymbol{\hat{\phi}}_{C}-\boldsymbol{\phi }\right) =%
\boldsymbol{\bar{M}}_{n}^{-1}\left( n^{-1/2}\sum_{i=1}^{n}\boldsymbol{M}_{i}%
\boldsymbol{M}_{T}\boldsymbol{u}_{i}\right) ,  \label{TE-Ca}
\end{equation}%
and $\sqrt{n}\left( \boldsymbol{\hat{\phi}}_{C}-\boldsymbol{\phi }%
_{0}\right) \rightarrow _{d}N(\boldsymbol{0},\boldsymbol{V}_{\phi ,C})$,
where $\boldsymbol{V}_{\phi ,C}=Avar\left( \sqrt{n}\boldsymbol{\hat{\phi}}%
_{C}\right) $ is given by 
\begin{equation*}
\boldsymbol{V}_{\phi ,C}=\boldsymbol{M}^{-1}\lim_{n\rightarrow \infty
}E\left( n^{-1}\sum_{i=1}^{n}\boldsymbol{M}_{i}\boldsymbol{M}_{T}\boldsymbol{%
u}_{i}\boldsymbol{u}_{i}^{\prime }\boldsymbol{M}_{T}\boldsymbol{M}%
_{i}\right) \boldsymbol{M}^{-1}.
\end{equation*}%
Since $\boldsymbol{M}_{i}\boldsymbol{M}_{T}\boldsymbol{u}_{i}=\boldsymbol{M}%
_{i}\boldsymbol{M}_{T}(\boldsymbol{y}_{i}-\boldsymbol{\phi })$, the
asymptotic variance of $\boldsymbol{\hat{\phi}}_{C}$ can be consistently
estimated by 
\begin{equation}
\widehat{\boldsymbol{V}}_{\phi ,C}=\boldsymbol{\bar{M}}_{n}^{-1}\left[
n^{-1}\sum_{i=1}^{n}\boldsymbol{M}_{i}\boldsymbol{M}_{T}(\boldsymbol{y}_{i}-%
\boldsymbol{\hat{\phi}}_{C})(\boldsymbol{y}_{i}-\boldsymbol{\hat{\phi}}%
_{C})^{\prime }\boldsymbol{M}_{T}\boldsymbol{M}_{i}\right] \boldsymbol{\bar{M%
}}_{n}^{-1}.  \label{Varphi2}
\end{equation}%
Using $\boldsymbol{\hat{\phi}}_{C}$, the TMG-TE estimator of $\boldsymbol{%
\beta }_{0}$ is now given by 
\begin{equation}
\boldsymbol{\hat{\beta}}_{C,TMG-TE}=\frac{1}{1+\bar{\delta}_{n}}\left[
n^{-1}\sum_{i=1}^{n}\boldsymbol{Q}_{i}^{\prime }(\boldsymbol{y}_{i}-%
\boldsymbol{\hat{\phi}}_{C})\right] \text{, for }T>k,  \label{BetaTE2}
\end{equation}%
and 
\begin{equation}
n^{\frac{1-\alpha }{2}}\left( \boldsymbol{\hat{\beta}}_{C,TMG-TE}-%
\boldsymbol{\beta }_{0}\right) =n^{\frac{1-\alpha }{2}}\left( \boldsymbol{%
\hat{\beta}}_{C,TMG-TE}(\boldsymbol{\phi })-\boldsymbol{\beta }_{0}\right)
-n^{-\alpha /2}\boldsymbol{\bar{Q}}_{n}^{\prime }\sqrt{n}\left( \boldsymbol{%
\hat{\phi}}_{C}-\boldsymbol{\phi }\right) ,  \label{AsyDTE-C1}
\end{equation}%
where $\boldsymbol{\bar{Q}}_{n}^{\prime }\sqrt{n}\left( \boldsymbol{\hat{\phi%
}}_{C}-\boldsymbol{\phi }\right) =O_{p}(1)$. Hence%
\begin{equation}
n^{(1-\alpha )/2}\left( \boldsymbol{\hat{\beta}}_{C,TMG-TE}-\boldsymbol{%
\beta }_{0}\right) =n^{(1-\alpha )/2}\left( \boldsymbol{\hat{\beta}}%
_{C,TMG-TE}(\boldsymbol{\phi })-\boldsymbol{\beta }_{0}\right)
+O_{p}(n^{-\alpha /2}).  \label{AsyDTE-C}
\end{equation}%
Also since $\alpha >0$, then $n^{(1-\alpha )/2}\left( \boldsymbol{\hat{\beta}%
}_{C,TMG-TE}-\boldsymbol{\beta }_{0}\right) $ has the same asymptotic
distribution as $n^{(1-\alpha )/2}\left( \boldsymbol{\hat{\beta}}_{C,TMG-TE}(%
\boldsymbol{\phi })-\boldsymbol{\beta }_{0}\right) $, with $\boldsymbol{\phi 
}$ treated as known. This result follows since $\boldsymbol{\hat{\phi}}_{C}-%
\boldsymbol{\phi }\rightarrow _{p}\boldsymbol{0}$ at the faster rate of $%
\sqrt{n}$, compared with the rate $n^{\frac{1-\alpha }{2}}$ of $\boldsymbol{%
\hat{\beta}}_{C,TMG-TE}-\boldsymbol{\beta }_{0}\rightarrow _{p}\boldsymbol{0}
$. In short, $Avar\left( n^{(1-\alpha )/2}\boldsymbol{\hat{\beta}}%
_{C,TMG-TE}\right) $ is not affected by the estimation uncertainty of the
time effects when $\alpha >0$. A consistent estimator is given by 
\begin{equation}
\widehat{Var\left( \boldsymbol{\hat{\beta}}_{C,TMG-TE}\right) }=\frac{1}{%
n(n-1)(1+\bar{\delta}_{n})^{2}}\sum_{i=1}^{n}\left( \boldsymbol{\tilde{\beta}%
}_{i,C}-\boldsymbol{\hat{\beta}}_{C,TMG-TE}\right) \left( \boldsymbol{\tilde{%
\beta}}_{i,C}-\boldsymbol{\hat{\beta}}_{C,TMG-TE}\right) ^{\prime },
\label{VarbetaC}
\end{equation}%
where $\boldsymbol{\tilde{\beta}}_{i,C}=\boldsymbol{Q}_{i}^{\prime }(%
\boldsymbol{y}_{i}-\boldsymbol{\hat{\phi}}_{C})$.

When $T=k$, Assumption \ref{invertm} does not hold and standard de-meaning
techniques can not be used to eliminate $\boldsymbol{\phi }$. This in turn
requires the correlation of slope coefficients with the regressors to be
time-invariant, a condition that holds trivially under uncorrelated
heterogeneity but could be restrictive when heterogeneity is correlated. The
TMG-TE estimator for the case of $T= k$ and its asymptotic properties are
derived in Section \ref{Teqk} of the mathematical appendix.

For the Monte Carlo experiments and the empirical application, $\boldsymbol{%
\hat{\phi}}_{C}$, given by (\ref{TE2}), is used as an estimator of $%
\boldsymbol{\phi }$ for panels with time effects when $T>k$. When $T=k$, we
use the estimator set out in Section \ref{Teqk} of the mathematical appendix.

\section{A test of correlated heterogeneity\label{Test}}

As summarized in Proposition \ref{prop_confee}, the FE estimator is$\sqrt{n}$%
-consistent if the slope coefficients, $\boldsymbol{\beta }_{i}$, and the
regressors, $\boldsymbol{X}_{i}=\left( \boldsymbol{x}_{i1},\boldsymbol{x}%
_{i2},...,\boldsymbol{x}_{iT}\right) ^{\prime }$, are independently
distributed. As a diagnostic check, we consider testing the null hypothesis 
\begin{equation}
H_{0}:\boldsymbol{\eta }_{i}\perp \boldsymbol{X}_{i}\text{ and }E\left( 
\boldsymbol{\eta }_{i}\right) =\boldsymbol{0}\text{, for all }i,
\label{null}
\end{equation}%
where $\boldsymbol{\eta }_{i}=\boldsymbol{\beta }_{i}-\boldsymbol{\beta }%
_{0} $. To test $H_{0}$, we follow \cite{Hausman1978} and compare FE and TMG
estimators given by (\ref{fee}) and (\ref{TMGb}), respectively. Under $H_{0}$%
, both estimators are consistent. For the alternative hypothesis, we consider%
\begin{equation}
H_{1}:E\left( n^{-1}\sum_{i=1}^{n}\boldsymbol{\bar{\Psi}}_{n}^{-1}%
\boldsymbol{X}_{i}^{\prime }\boldsymbol{M}_{T}\boldsymbol{X}_{i}\boldsymbol{%
\eta }_{i}\right) =\ominus \left( n^{a_{\eta }-1}\right) \text{, } a_{\eta
}>1/2.  \label{H1}
\end{equation}%
Also by Proposition \ref{prop_confee}, the FE estimator is $\sqrt{n}$%
-consistent only if $a_{\eta }<1/2$. See also equation (\ref{rootnFE}).
Therefore, as required by Hausman's testing procedure, the FE estimator is
inconsistent under $H_{1}$ where $a_{\eta }>1/2$, but the TMG estimator is
consistent irrespective of the value of $a_{\eta }$.

The Hausman test of $H_{0}$ is based on the difference $\sqrt{n}\boldsymbol{%
\hat{\Delta}}_{\beta }=\sqrt{n}\left( \boldsymbol{\hat{\beta}}_{FE}-%
\boldsymbol{\hat{\beta}}_{TMG}\right) $. Such a test has been considered by 
\cite{PesaranEtal1996} and \cite{PesaranYamagata2008}, assuming the MG
estimator has at least second-order moments.\footnote{%
See pp. 160--162 of \cite{PesaranEtal1996} and p. 53 of \cite%
{PesaranYamagata2008}.} Here we extend this test to cover cases when $T$ is
ultra-short and the moment condition (\ref{sufficient}) is not met. Also,
the earlier tests were derived under the null of homogeneity (namely $%
\boldsymbol{\eta }_{i}=\boldsymbol{0}$ for all $i$), whilst the null that we
consider is more general and covers the null of homogeneity as a special
case.

In the development of his test, Hausman originally assumed that one of the
estimators under consideration is asymptotically efficient and showed that
in this case, the asymptotic variance of the difference between the two
estimators is equal to the difference between the two variances. But in the
present application, as shown in Proposition \ref{prop_mgvsfe} and
illustrated by Example \ref{ExampleMG-FE}, neither of the two estimators, $%
\boldsymbol{\hat{\beta}}_{FE}$ and $\boldsymbol{\hat{\beta}}_{TMG}$, are
efficient, and the asymptotic variance of $\sqrt{n}\boldsymbol{\hat{\Delta}}%
_{\beta }$ will not be equal to the difference between the asymptotic
variances of FE and TMG estimators. See also Section 26.9.1 of \cite%
{Pesaran2015}.

The following theorem summarizes our main result for the Hausman test
applied to $\sqrt{n}\left( \boldsymbol{\hat{\beta}}_{FE}-\boldsymbol{\hat{%
\beta}}_{TMG}\right) $.

\begin{theorem}
\label{AsyDTest}Suppose for $i=1,2,...,n$ and $t=1,2,...,T$, $y_{it}$ are
generated by the heterogeneous panel data model (\ref{m1}) and Assumptions %
\ref{errorsx}, \ref{rcm}, \ref{PoolA}(b)--(d), and \ref{CrCorr}--\ref{CRE}
hold. Then under $H_{0}$, defined by (\ref{null}), and for a fixed $T\geq k$%
, $\sqrt{n}\boldsymbol{\hat{\Delta}}_{\beta }=\sqrt{n}\left( \boldsymbol{%
\hat{\beta}}_{FE}-\boldsymbol{\hat{\beta}}_{TMG}\right) \rightarrow _{d}N(%
\boldsymbol{0},\boldsymbol{V}_{\Delta })$, as $n\rightarrow \infty $, where 
\begin{equation*}
\boldsymbol{V}_{\Delta }=\lim_{n\rightarrow \infty }\frac{1}{n}%
\sum_{i=1}^{n}E\left( \boldsymbol{\Gamma }_{i}\boldsymbol{X}_{i}^{\prime }%
\boldsymbol{M}_{T}\boldsymbol{H}_{i}\boldsymbol{M}_{T}\boldsymbol{X}_{i}%
\boldsymbol{\Gamma }_{i}\right) +\lim_{n\rightarrow \infty }\frac{1}{n}%
\sum_{i=1}^{n}E\left( \boldsymbol{\Gamma }_{i}\boldsymbol{\Psi }_{i}%
\boldsymbol{\Omega }_{\beta }\boldsymbol{\Psi }_{i}\boldsymbol{\Gamma }%
_{i}\right) ,
\end{equation*}%
$\boldsymbol{H}_{i}=E(\boldsymbol{u}_{i}\boldsymbol{u}_{i}^{\prime
}\left\vert \boldsymbol{X}_{i}\right. )$, $\boldsymbol{\Omega }_{\beta
}=E\left( \boldsymbol{\Omega }_{i\beta }\right) \succeq \boldsymbol{0}$, $%
\boldsymbol{\Gamma }_{i}=\boldsymbol{\bar{\Psi}}_{n}^{-1}-\left( \frac{%
1+\delta _{i}}{1+\bar{\delta}_{n}}\right) \boldsymbol{\Psi }_{i}^{-1}$, $%
\boldsymbol{\bar{\Psi}}_{n}=\frac{1}{n}\sum_{i=1}^{n}\boldsymbol{\Psi }_{i}$%
, $\boldsymbol{\Psi }_{i}=\boldsymbol{X}_{i}^{\prime }\boldsymbol{M}_{T}%
\boldsymbol{X}_{i}$, $\bar{\delta}_{n}=\frac{1}{n}\sum_{i=1}^{n}\delta _{i}$%
, and $\delta _{i}$ is defined by (\ref{deltai}). The asymptotic variance
matrix, $\boldsymbol{V}_{\Delta }$, is positive definite if either $%
\lim\limits_{n\rightarrow \infty }\frac{1}{n}\sum_{i=1}^{n}E\left( 
\boldsymbol{\Gamma }_{i}\boldsymbol{X}_{i}^{\prime }\boldsymbol{M}_{T}%
\boldsymbol{H}_{i}\boldsymbol{M}_{T}\boldsymbol{X}_{i}\boldsymbol{\Gamma }%
_{i}\right) \succ \boldsymbol{0}$, and/or $\lim\limits_{n\rightarrow \infty }%
\frac{1}{n}\sum_{i=1}^{n}E\left( \boldsymbol{\Gamma }_{i}\boldsymbol{\Psi }%
_{i}^{2}\boldsymbol{\Gamma }_{i}\right) \succ \boldsymbol{0}$, and $\lambda
_{\min }\left( \boldsymbol{\Omega }_{\beta }\right) >c>0$. When one of these
conditions is met, under $H_{0}$ and for a fixed $T\geq k$, $H_{\beta }=n%
\boldsymbol{\hat{\Delta}}_{\beta }^{\prime }\boldsymbol{V}_{\Delta }^{-1}%
\boldsymbol{\hat{\Delta}}_{\beta }\rightarrow _{d}\chi _{k^{\prime }}^{2}$
as $n\rightarrow \infty $. Under the alternative hypothesis, $H_{1}$ defined
by (\ref{H1}), provided that, in addition, $\sup_{i}E\left\Vert \boldsymbol{%
\Psi }_{i}\boldsymbol{\eta }_{i}\right\Vert ^{2+\varepsilon }<C$ for some $%
\varepsilon >0$, $H_{\beta }=\ominus _{p}(n^{2a_{\eta }-1})+O_{p}(n^{a_{\eta
}-\alpha \alpha _{p}})+O_{p}\left( n^{1-2\alpha \alpha _{p}}\right) ,$ and $%
H_{\beta }\rightarrow _{p}\infty $, as $n\rightarrow \infty $, if $\alpha
_{\eta }>1/2$.
\end{theorem}

For a proof, see sub-section \ref{ProofAsyDTest} in the mathematical
appendix.

\begin{remark}
The proposed test of correlated heterogeneity is consistent, and its power
rises not only with $a_{\eta }$, but is also enhanced from the correlation
between $\boldsymbol{\eta }_{i}$ and $\delta _{i}$, induced from the
shrinkage of some of the individual estimators used in construction of the
TMG estimator (see equation (\ref{meuB}) in the mathematical appendix). But
under $H_{0}$, $\boldsymbol{\eta }_{i}$ and $\boldsymbol{X}_{i}$ are
independently distributed, which implies that $\boldsymbol{\eta }_{i}$ and $%
\delta _{i}$ are also independently distributed, and shrinkage of individual
estimators does not lead to size distortions.
\end{remark}

To implement the $H_{\beta }$ test, the asymptotic variance, $\boldsymbol{V}%
_{\Delta }$, can be rewritten as $\boldsymbol{V}_{\Delta }=\lim_{n \to
\infty} n^{-1}\sum_{i=1}^{n}\sum_{t=1}^{T}\sum_{t^{\prime }=1}^{T}E\left( 
\boldsymbol{g}_{it}\boldsymbol{g}_{it^{\prime }}^{\prime }\tilde{\nu}_{it}%
\tilde{\nu}_{it^{\prime }}\right) $, where $\tilde{\nu}_{it}=\nu _{it}-\bar{%
\nu}_{i\circ }$, $\bar{\nu}_{i\circ } = T^{-1}\sum_{t=1}^{T}\nu_{it}$, and $%
\boldsymbol{g}_{it}$ is the $t^{th}$ column of $\boldsymbol{G}_{i}^{\prime}$%
, where (see also (\ref{GiApp}) in the mathematical appendix)%
\begin{equation}
\boldsymbol{G}_{i}=\boldsymbol{X}_{i}\left[ \left( n^{-1}\sum_{i=1}^{n}%
\boldsymbol{X}_{i}^{\prime }\boldsymbol{M}_{T}\boldsymbol{X}_{i}\right)
^{-1}-\left( \frac{1+\delta _{i}}{1+\bar{\delta}_{n}}\right) \left( 
\boldsymbol{X}_{i}^{\prime }\boldsymbol{M}_{T}\boldsymbol{X}_{i}\right) ^{-1}%
\right].  \label{Gi}
\end{equation}%
For $T$ fixed, a consistent estimator of $\boldsymbol{V}_{\Delta }$, which
is robust to the choices of $\boldsymbol{H}_{i}$ and $\boldsymbol{\Omega }%
_{i \beta }$, is 
\begin{equation}
\boldsymbol{\widehat{V}}_{\Delta }=\frac{1}{n}\sum_{i=1}^{n}\sum_{t=1}^{T}%
\sum_{t^{\prime }=1}^{T}\boldsymbol{g}_{it}\boldsymbol{g}_{it^{\prime
}}^{\prime }\hat{\tilde{\nu}}_{it}\hat{\tilde{\nu}}_{it^{\prime }},
\label{VarHbeta}
\end{equation}%
where $\hat{\tilde{\nu}}_{it}=(y_{it}-\bar{y}_{i\circ })-\boldsymbol{\hat{%
\beta}}_{FE}^{\prime }(\boldsymbol{x}_{it}-\boldsymbol{\bar{x}}_{i\circ })$, 
$\bar{y}_{i\circ }=\frac{1}{T}\sum_{t=1}^{T}y_{it}$, and $\boldsymbol{\bar{x}%
}_{i\circ }=\frac{1}{T}\sum_{t=1}^{T}\boldsymbol{x}_{it}$. Using $%
\boldsymbol{\widehat{V}}_{\Delta }$, the Hausman test statistic for
uncorrelated slope heterogeneity is given by 
\begin{equation}
\hat{H}_{\beta }=n\left( \boldsymbol{\hat{\beta}}_{FE}-\boldsymbol{\hat{\beta%
}}_{TMG}\right) ^{\prime }\boldsymbol{\widehat{V}}_{\Delta }^{-1}\left( 
\boldsymbol{\hat{\beta}}_{FE}-\boldsymbol{\hat{\beta}}_{TMG}\right) .
\label{htest}
\end{equation}%
Under $H_{0}$, the consistency of $\boldsymbol{\widehat{V}}_{\Delta }$ as an
estimator of $\boldsymbol{V}_{\Delta }$ follows from the $\sqrt{n}$%
-consistency of the FE estimator.

An extension of the $\hat{H}_{\beta }$ test statistic to panel data models
with time effects is provided in Section \ref{TestTE} of the online
supplement.

\section{Trimmed mean group estimation for a subset of coefficients}

\label{subset}

The TMG approach can be applied to a subset of coefficients of interest. Let 
$\boldsymbol{X}_{i}=(\boldsymbol{X}_{i1},\boldsymbol{X}_{i2})$, where $%
\boldsymbol{X}_{i1}$ is the $T\times p$ matrix of observations on the focal
(treatment) variables and $\boldsymbol{X}_{i2}$ is the $T\times (k^{\prime
}-p)$ matrix of observations on the auxiliary (control) variables, and
partition the coefficients accordingly as $\boldsymbol{\beta }_{i}=(%
\boldsymbol{\beta }_{i1}^{\prime },\boldsymbol{\beta }_{i2}^{\prime
})^{\prime }$, where $\boldsymbol{\beta }_{i1}$ is the $p\times 1$ vector of
the coefficients of interest. Using results from the partitioned
regressions, we have 
\begin{equation*}
\boldsymbol{\hat{\beta}}_{i1}=(\boldsymbol{X}_{i1}^{\prime }\boldsymbol{M}%
_{i2}\boldsymbol{X}_{i1})^{-1}\boldsymbol{X}_{i1}^{\prime }\boldsymbol{M}%
_{i2}\boldsymbol{y}_{i},
\end{equation*}%
where $\boldsymbol{M}_{i2}=\boldsymbol{M}_{T}-\boldsymbol{M}_{T}\boldsymbol{X%
}_{i2}\left( \boldsymbol{X}_{i2}^{\prime }\boldsymbol{M}_{T}\boldsymbol{X}%
_{i2}\right) ^{-}\boldsymbol{X}_{i2}^{\prime }\boldsymbol{M}_{T}$, and $%
\left( \boldsymbol{X}_{i2}^{\prime }\boldsymbol{M}_{T}\boldsymbol{X}%
_{i2}\right) ^{-}$ denotes a generalized inverse of $\left( \boldsymbol{X}%
_{i2}^{\prime }\boldsymbol{M}_{T}\boldsymbol{X}_{i2}\right) $.\footnote{%
Note that $\boldsymbol{\hat{\beta}}_{i1}$ is invariant to the choice of the
generalized inverse, $\left( \boldsymbol{X}_{i2}^{\prime }\boldsymbol{X}%
_{i2}\right) ^{-}$.}

The TMG estimator of $\boldsymbol{\beta }_{01}=E\left( \boldsymbol{\beta }%
_{i1}\right) $ is given by 
\begin{equation*}
\boldsymbol{\hat{\beta}}_{1,TMG}=n^{-1}\sum_{i=1}^{n}\left( \frac{1+\delta
_{i1}}{1+\bar{\delta}_{n1}}\right) \boldsymbol{\hat{\beta}}_{i1},
\end{equation*}%
where $\delta _{i1}=\left( \frac{d_{i1}-a_{n1}}{a_{n1}}\right) \boldsymbol{1}%
\{d_{i1}\leq a_{n1}\}$, $\bar{\delta}_{n1}=\frac{1}{n}\sum_{i=1}^{n}\delta
_{i1}$, and $d_{i1}=\det \left( \boldsymbol{X}_{i1}^{\prime }\boldsymbol{M}%
_{i2}\boldsymbol{X}_{i1}\right) $. Similarly, for the choice of threshold
value, we consider $a_{n1}=\bar{d}_{n1}n^{-\alpha _{1}}$, where $\alpha _{1}>%
\frac{1}{1+2\alpha _{1,p}}$, $\alpha _{1,p}$ is the shape parameter of the
tail probability distribution of $1/d_{i1}$ over $i$, and $\bar{d}%
_{n1}=n^{-1}\sum_{i=1}^{n}d_{i1}$.

The above partitioned formula can also be adapted for testing general linear
restrictions $H_{0}:\boldsymbol{S\beta }_{0}=\boldsymbol{r}$ against $H_{1}:%
\boldsymbol{S\beta }_{0}\neq \boldsymbol{r}$, where $\boldsymbol{r}$ is a $%
p\times 1$ vector and $\boldsymbol{S}$ is a $p\times k^{\prime }$ ($%
p<k^{\prime }$) full-rank matrix of fixed constants. Let $\boldsymbol{\gamma 
}_{i}=$ $\boldsymbol{S\beta }_{i}-\boldsymbol{r}$ and partition $\boldsymbol{%
S=(S}_{1},\boldsymbol{S}_{2})$ such that $\boldsymbol{S}_{1}$ is a $p\times
p $ non-singular matrix.\footnote{%
This can be achieved by a suitable reordering of the elements of $%
\boldsymbol{\beta }_{i}$.} Then (\ref{m1}) can be written as%
\begin{equation*}
\boldsymbol{\tilde{y}}_{i}=\alpha _{i}\boldsymbol{\tau }_{T}+\boldsymbol{%
\tilde{X}}_{i1}\boldsymbol{\gamma }_{i}+\boldsymbol{\tilde{X}}_{i2}%
\boldsymbol{\beta }_{i2}+\boldsymbol{u}_{i},
\end{equation*}%
where $\boldsymbol{\tilde{y}}_{i}=\boldsymbol{y}_{i}-\boldsymbol{X}_{i1}%
\boldsymbol{S}_{1}^{-1}\boldsymbol{r}$, $\boldsymbol{\tilde{X}}_{i1}=%
\boldsymbol{X}_{i1}\boldsymbol{S}_{1}^{-1}$, and $\boldsymbol{\tilde{X}}%
_{i2}=\boldsymbol{X}_{i2}-\boldsymbol{X}_{i1}\boldsymbol{S}_{1}^{-1}%
\boldsymbol{S}_{2}$. The TMG estimator of $\boldsymbol{\gamma }_{0}=E\left( 
\boldsymbol{\gamma }_{i}\right) $ can be computed using the unit-specific
estimates of $\boldsymbol{\gamma }_{i}$ as above.

The Hausman test proposed in the previous section can also be applied to a
subset of the coefficients in a straightforward manner.

\section{Monte Carlo experiments\label{MC}}

\subsection{Data generating processes (DGP) \label{DGP}}

The outcome variable, $y_{it}$, is generated as 
\begin{equation}
y_{it}=\alpha _{i}+\sum_{j=1}^{k^{\prime }}\beta _{ij}x_{j,it}+\kappa \sigma
_{it}e_{it}\text{, for }i=1,2,...,n\text{, and }t=1,2,...,T,  \label{ydgp}
\end{equation}%
where the errors, $u_{it}$, are allowed to be serially correlated and
heteroskedastic. Specifically, we set $u_{it}=\kappa \sigma _{it}e_{it}$ and
generate $e_{it}$ as AR(1) processes 
\begin{equation}
e_{it}=\rho _{ie}e_{i,t-1}+\left( 1-\rho _{ie}^{2}\right) ^{1/2}\varsigma
_{it},\text{ for }i=1,2,...,n.  \label{yerror}
\end{equation}%
In the baseline DGP, we set $\sigma _{it}=\sigma _{iu}$ for all $t$ and
generate $\sigma _{iu}^{2}\sim IID\frac{1}{2}\left( 1+z_{iu}^{2}\right)$,
with $z_{iu}\sim IIDN(0,1)$.\footnote{%
More general heteroskedastic specifications for $\sigma _{it}$ are
considered in sub-section \ref{dgprobust} of the online supplement.} Both
Gaussian and non-Gaussian errors are considered: $\varsigma _{it}\sim
IIDN(0,1)$ and $\varsigma _{it}\sim IID\frac{1}{2}\left( \chi
_{2}^{2}-2\right) $.

The regressors, $x_{j,it}$, for $j=1,2,....k^{\prime }$, $i=1,2,...,n$, and $%
t=1,2,...,T$, are generated as factor-augmented AR(1) processes 
\begin{equation}
x_{j,it}=\alpha _{j,ix}(1-\rho _{j,ix})+\gamma _{j,ix}f_{j,t}+\rho
_{j,ix}x_{j,i,t-1}+\left( 1-\rho _{j,ix}^{2}\right) ^{1/2}u_{xj,it},
\label{xdgp}
\end{equation}%
where $\alpha _{j,ix}\sim IIDN(1,1)$, $u_{xj,it}=\sigma _{j,ix}e_{xj,it}$, $%
e_{xj,it}\sim IID(0,1)$, $\sigma _{j,ix}^{2}=\frac{1}{2}\left(
1+z_{j,ix}^{2}\right) $, and $z_{j,ix}\sim IIDN(0,1)$. We consider panels
with $k^{\prime }=1,2$ and $3$ regressors and investigate the extent to
which trimming is required for different combinations of $T$ and $k^{\prime
} $. The common factors are generated as $%
f_{j,t}=0.9f_{j,t-1}+(1-0.9^{2})^{1/2}v_{j,t}$, for $%
t=-49,-48,...,-1,0,1,...,T$, where $v_{j,t}\sim IIDN(0,1)$, and $f_{j,-50}=0 
$. The factor loadings are generated as $\gamma _{j, ix}\sim IIDU(0,2)$.
When time effects are included in the model, we set $\phi _{t}=t$, for $%
t=1,2,...,T-1$, and $\phi _{T}=-T(T-1)/2$, so that $\boldsymbol{\tau }%
_{T}^{\prime }\boldsymbol{\phi }=0$.

For the slope coefficients, we experiment with both correlated and
uncorrelated effects in $\beta _{i1}$, while considering uncorrelated
heterogeneous effects for the other slope coefficients when the respective
regressors are included. Specifically, $\alpha _{i}$ and $\beta _{i1}$ are
generated as 
\begin{equation}
(\alpha _{i},\beta _{i1})^{\prime }=\left( \alpha _{0},\beta _{01}\right)
^{\prime }+\left( \eta _{i\alpha },\eta _{i\beta _{1}}\right) ^{\prime },
\label{coefdgp}
\end{equation}%
where 
\begin{equation}
\boldsymbol{\tilde{\eta}}_{i}=\left( \eta _{i\alpha },\eta _{i\beta
_{1}}\right) ^{\prime }=\left( \frac{\sigma _{1,ix}^{2}-E\left( \sigma
_{1,ix}^{2}\right) }{\sqrt{Var\left( \sigma _{1,ix}^{2}\right) }}\right) 
\boldsymbol{\psi }+\boldsymbol{\epsilon }_{i}=\sqrt{2}(\sigma _{1,ix}^{2}-1)%
\boldsymbol{\psi }+\boldsymbol{\epsilon }_{i},  \label{eta_i}
\end{equation}%
$\boldsymbol{\psi =(}\psi _{\alpha },\psi _{\beta _{1}})^{\prime }$, $%
\boldsymbol{\epsilon }_{i}=(\epsilon _{i\alpha },\epsilon _{i\beta
_{1}})^{\prime }\sim IIDN\left( \boldsymbol{0},\boldsymbol{V}_{\epsilon
}\right) $, and $\boldsymbol{V}_{\epsilon }=Diag(\boldsymbol{\sigma }%
_{\epsilon }^{2})$ with $\boldsymbol{\sigma }_{\epsilon }^{2}=\left( \sigma
_{\epsilon \alpha }^{2},\sigma _{\epsilon \beta _{1}}^{2}\right) ^{\prime }$%
. It follows that 
\begin{equation*}
E(\boldsymbol{\tilde{\eta}}_{i})=\boldsymbol{0}\text{, and }\boldsymbol{V}_{%
\boldsymbol{\tilde{\eta}}}=E(\boldsymbol{\tilde{\eta}}_{i}\boldsymbol{\tilde{%
\eta}}_{i}^{\prime })=\left( 
\begin{array}{cc}
\sigma _{\alpha }^{2} & \sigma _{\alpha \beta _{1}} \\ 
\sigma _{\alpha \beta _{1}} & \sigma _{\beta _{1}}^{2}%
\end{array}%
\right) =\boldsymbol{\psi \psi }^{\prime }+\boldsymbol{V}_{\epsilon }.
\end{equation*}%
The degree of correlated heterogeneity is determined by $\boldsymbol{\psi
\psi }^{\prime }$, with $Cov(\alpha _{i},\beta _{i1})=\sigma _{\alpha \beta
_{1}}\neq 0$ when $\psi _{\alpha }$ and $\psi _{\beta _{1}}$ are both
non-zero. Specifically, $\sigma _{\alpha }^{2}=\psi _{\alpha }^{2}+\sigma
_{\epsilon \alpha }^{2}$, $\sigma _{\alpha \beta _{1}}=\psi _{\alpha }\psi
_{\beta _{1}}$, and $\sigma _{\beta _{1}}^{2}=\psi _{\beta _{1}}^{2}+\sigma
_{\epsilon \beta _{1}}^{2}$. The coefficients of $x_{j,it}$, for $%
j=2,3,...,k^{\prime }$ are generated as $\beta _{ij}=\beta _{0j}+\epsilon
_{i\beta _{j}}$ with $\epsilon _{i\beta _{j}}\sim IIDN(0,\sigma _{\epsilon
\beta _{j}}^{2})$, which are not correlated with any of the regressors. The
true values of the parameters of interest are set as follows: $E(\alpha
_{i})=\alpha _{0}=1$ and $E(\beta _{ij})=\beta _{0j}=1$ for $%
j=1,2,...,k^{\prime }$. We also set $\sigma _{\alpha }^{2}=0.5$, $\sigma
_{\beta _{1}}^{2}=0.75$, and $\sigma _{\epsilon \beta _{j}}^{2}=0.5$ for $%
j=2,...,k^{\prime }$.

When $k=k^{\prime }+1=2$, for a fixed $T\geq k$, the asymptotic bias of the
FE estimator of $\beta _{01}=E(\beta _{i1})$ is given by (see (\ref%
{AsyBiasFE})) 
\begin{equation*}
\func{plim}_{n\rightarrow \infty }\left( \hat{\beta}_{1,FE}-\beta
_{01}\right) =\frac{\sum_{t=1}^{T}\lim_{n\rightarrow \infty
}n^{-1}\sum_{i=1}^{n}E\left[ (x_{1,it}-\bar{x}_{1,iT})^{2}\eta _{i\beta _{1}}%
\right] }{\sum_{t=1}^{T}\lim_{n\rightarrow \infty }n^{-1}\sum_{i=1}^{n}E%
\left[ (x_{1,it}-\bar{x}_{1,iT})^{2}\right] }.
\end{equation*}%
The size of the bias will depend on $\psi _{\beta _{1}}$ and the parameters
of $x_{1,it}$ process. The exact expression for this bias simplifies
considerably if $\rho _{1, ix}=0$ (no dynamics in the $x_{1, it}$ equation).
In this case, 
\begin{equation*}
\func{plim}_{n\rightarrow \infty }\left( \hat{\beta}_{1,FE}-\beta
_{01}\right) =\frac{\sqrt{2}\left[ E\left( \sigma _{1,ix}^{4}\right)
-E\left( \sigma _{1,ix}^{2}\right) \right] \psi _{\beta _{1}}\left( \frac{T-1%
}{T}\right) }{\left[ T^{-1}\sum_{t=1}^{T}\left( f_{1,t}-\bar{f}_{1,T}\right)
^{2}\right] E\left( \gamma _{1,ix}^{2}\right) +\left( \frac{T-1}{T}\right)
E\left( \sigma _{1,ix}^{2}\right) },
\end{equation*}%
which, noting that $E\left( \sigma _{1,ix}^{2}\right) =1$, $E\left( \sigma
_{1,ix}^{4}\right) =3/2$ and $E\left( \gamma _{1,ix}^{2}\right) =4/3$,
simplifies to 
\begin{equation*}
\func{plim}_{n\rightarrow \infty }\left( \hat{\beta}_{1,FE}-\beta
_{01}\right) =\frac{\left( \frac{T-1}{T}\right) \sqrt{2}/2\psi _{\beta _{1}}%
}{\left[ T^{-1}\sum_{t=1}^{T}\left( f_{1,t}-\bar{f}_{1,T}\right) ^{2}\right]
\left( 4/3\right) +\left( \frac{T-1}{T}\right) }.
\end{equation*}%
It is also worth noting that when $n$ and $T\rightarrow \infty $, jointly,
then $\func{plim}_{n,T\rightarrow \infty }(\hat{\beta}_{1,FE}-\beta
_{01})=0.707\psi _{\beta _{1}}$ without interactive effects in $x_{1, it}$
process, and the bias of the FE estimator does not vanish even if $%
T\rightarrow \infty $.

For the baseline DGP, we generate the errors in the outcome equation as
chi-squared without serial correlation ($\rho _{ie}=0$ in (\ref{yerror})). $%
x_{j,it}$ are generated without dynamics or interactive effects ($\rho
_{j,ix}=0$ and $\gamma _{j,ix}=0$ in (\ref{xdgp})). We set $\psi _{\alpha
}=0.5$ for individual fixed effects. For the slope coefficient, we consider
three possibilities: (a) uncorrelated heterogeneity, $\psi _{\beta _{1}}=0$,
(b) a medium level of correlated heterogeneity, $\psi _{\beta _{1}}=0.5$,
generating a bias around $35\%$ when $k^{\prime }=1$, and (c) a high level, $%
\psi _{\beta _{1}}=0.8$, leading to a bias of around $57\%$ for the FE
estimator when $k^{\prime }=1$. For each choice of $\psi _{\beta _{1}}$, the
scalar parameter $\kappa $ in (\ref{ydgp}) is set such that the pooled $%
R^{2} $ ($PR^{2}$) of panel regressions is around $0.2$. This is achieved by
stochastic simulation for each $T$, as described in sub-section \ref{Simuk}
of the online supplement. We also experiment with a medium level of fit by
setting $PR^{2}=0.4$ when $k^{\prime }=1$.\footnote{%
The scaling of the errors in the outcome equation, $\kappa $, is set when
the DGP contains one regressor. As a result, the $PR^{2}$ will be slightly
higher than the target values of $0.2$ and $0.4$, when $k^{\prime }=2$ or $3$%
.}

We consider different values of $T$ depending on $k^{\prime }$ up to $T=15$.
We also carried out a number of robustness checks, detailed in Section \ref%
{dgprobust} of the online supplement.

\subsection{Monte Carlo findings}

\subsubsection{Estimates of the tail index of the distribution of $1/d_{i}$}

\label{MCalphap}

We first provide estimates of the shape parameter of Pareto distributions
fitted to $z_{i}=1/\func{det}(\boldsymbol{X}_{i}^{\prime }\boldsymbol{M}_{T}%
\boldsymbol{X}_{i})$, $i=1,2,...,n$, for different choices of $T$ and $%
k^{\prime }$. We use the estimator of $\alpha _{p}$ proposed by \cite%
{Hill1975}, which is given by 
\begin{equation}
\hat{\alpha}_{p,Hill}=\frac{(m+1)}{\sum_{j=1}^{m}\ln \left( z^{(j)}\right)
-m\ln \left( z^{(m+1)}\right) },  \label{aphill}
\end{equation}%
where $z^{(1)}\geq z^{(2)}....\geq z^{(m)}$ are the first $m$ largest values
of $z_{i}$, and $m$ is the cut-off point set such that $m$ and $%
n/m\rightarrow \infty $, as $n\rightarrow \infty $.\footnote{%
See also \cite{Pickands1975} and Section 6 of \cite{PesaranYang2020} for
more recent literature on estimation of $\alpha _{p}$.}

Table \ref{tab:ap_mc} summarizes the estimates of $\alpha _{p}$ for $n=5,000$
and two cut-off values: $m=n^{1/2}$ and $n^{1/3}$, in the case of panels
with $k^{\prime }=1,2$ and $3$ regressors. As expected, the estimates of $%
\alpha _{p}$ are larger for the lower cut-off value, although the
differences between the two estimates are quite small when $T=k^{\prime }+1$%
. It is also interesting that when $T=k^{\prime }+1$, the estimates of $%
\alpha _{p}$ lie in the narrow interval [$0.51,0.57$], irrespective of the
value of $k^{\prime }=1,2$ and $3$; thus indicating lack of moments for the
individual estimators and the need for trimming. This finding is not
affected when we consider more general $\{\boldsymbol{x}_{it}\}$ processes.
See Table \ref{tab:ap_mc_3} in the online supplement for $k^{\prime }=1$.
Also, as expected, the estimates of $\alpha _{p}$ rise with $T-k^{\prime}$
and exceed the threshold value of $2$ for both choices of cut-off values
when $T>2k^{\prime }+3=2k+1$. $T=2k+1$ seems to be a borderline case where $%
\alpha _{p}$ is estimated to be very close to $2$, but there is still a wide
margin of uncertainty. For example, for $k^{\prime }=1$ and using the
cut-off value of $n^{1/2}$, estimates of $\alpha _{p}$ for $T=5$ and $6$ are
given by $1.97 $ $(0.23)$ and $2.38$ $(0.28)$, respectively. The standard
errors are in parentheses.\footnote{%
Additional estimation results for $\alpha _{p}$ are provided in Tables \ref%
{tab:ap_mc_2} and \ref{tab:ap_mc_3} in Section \ref{MCap} of the online
supplement.}

\begin{table}[h]
\caption{The estimates of $\protect\alpha _{p}$, the tail index of the
distribution of $1/d_{i}$, where $d_{i}=\func{det}(\boldsymbol{X}%
_{i}^{\prime }\boldsymbol{M}_{T}\boldsymbol{X}_{i})$, by Hill's method for
two choices of cut-off values}
\label{tab:ap_mc}\vspace{-6mm}
\par
\begin{center}
\scalebox{0.75}{
\begin{tabular}{cccccccccccccccccccc}
\hline \hline
\multicolumn{6}{c}{One regressor $(k^{\prime}=1)$} &  & \multicolumn{6}{c}{Two regressors $(k^{\prime}=2)$} &  & \multicolumn{6}{c}{Three regressors $(k^{\prime}=3)$} \\ \cline{1-6} \cline{8-13} \cline{15-20}
Cut-off & \multicolumn{2}{c}{$n^{1/2}$} && \multicolumn{2}{c}{$n^{1/3}$} &  &  & \multicolumn{2}{c}{$n^{1/2}$} && \multicolumn{2}{c}{$n^{1/3}$} &  &  & \multicolumn{2}{c}{$n^{1/2}$} && \multicolumn{2}{c}{$n^{1/3}$} \\ \cline{1-3} \cline{5-6} \cline{9-10} \cline{12-13} \cline{16-17} \cline{19-20}
$T$ & $\hat{\alpha}_{p}$ & & & $\hat{\alpha}_{p}$ & &   & $T$ & $\hat{\alpha}_{p}$ & & & $\hat{\alpha}_{p}$ & &   & $T$ & $\hat{\alpha}_{p}$ & & & $\hat{\alpha}_{p}$  & \\ \hline 
& \multicolumn{19}{c}{$n=5,000$} \\ \hline
2  & 0.51 & (0.06) &  & 0.56 & (0.13) &  & 3  & 0.51 & (0.06) &  & 0.57 & (0.13) &  & 4  & 0.51 & (0.06) &  & 0.57 & (0.13) \\
3  & 1.03 & (0.12) &  & 1.12 & (0.26) &  & 4  & 0.98 & (0.12) &  & 1.09 & (0.26) &  & 5  & 0.95 & (0.11) &  & 1.08 & (0.25) \\
4  & 1.50 & (0.18) &  & 1.66 & (0.39) &  & 5  & 1.38 & (0.16) &  & 1.58 & (0.37) &  & 6  & 1.30 & (0.15) &  & 1.52 & (0.36) \\
5  & 1.97 & (0.23) &  & 2.21 & (0.52) &  & 6  & 1.72 & (0.20) &  & 2.00 & (0.47) &  & 7  & 1.59 & (0.19) &  & 1.88 & (0.44) \\
6  & 2.38 & (0.28) &  & 2.70 & (0.64) &  & 7  & 2.03 & (0.24) &  & 2.35 & (0.55) &  & 8  & 1.83 & (0.22) &  & 2.20 & (0.52) \\
8  & 3.12 & (0.37) &  & 3.56 & (0.84) &  & 8  & 2.30 & (0.27) &  & 2.72 & (0.64) &  & 9  & 2.03 & (0.24) &  & 2.47 & (0.58) \\
10 & 3.76 & (0.44) &  & 4.37 & (1.03) &  & 10 & 2.75 & (0.32) &  & 3.27 & (0.77) &  & 10 & 2.22 & (0.26) &  & 2.68 & (0.63) \\
15 & 5.10 & (0.60) &  & 6.12 & (1.44) &  & 15 & 3.70 & (0.44) &  & 4.57 & (1.08) &  & 15 & 2.93 & (0.35) &  & 3.66 & (0.86)
\\\hline\hline
\end{tabular}
}
\end{center}
\par
\vspace{-1mm} 
\begin{spacing}{1}
{\footnotesize 
Notes: The estimates of $\alpha_{p}$ and their standard errors are computed using Hill's estimation procedure, with the cut-off values $n^{1/2}$ and $n^{1/3}$. $\boldsymbol{X}_{i} = (\boldsymbol{x}_{i1}, \boldsymbol{x}_{i2}, ..., \boldsymbol{x}_{iT})^{\prime}$, where the $k^{\prime} \times 1$ regressors, $\boldsymbol{x}_{it}$, are generated as specified in the baseline DGP. $\boldsymbol{M}_{T} = \boldsymbol{I}_{T} - \boldsymbol{\tau}_{T}\boldsymbol{\tau}_{T}^{\prime}/T$. The numbers in brackets are standard errors. }
\end{spacing}
\end{table}

\subsubsection{Comparison of TMG, FE, and MG estimators \label{MCfe}}

We begin by comparing the performance of the TMG estimator with those of FE
and MG estimators under both uncorrelated and correlated heterogeneity. We
consider the sample size combinations $n=1,000,2,000,5,000$, $10,000$, $%
T=2,3,...,15$ subject to $T\geq k^{\prime }+1$, for $k^{\prime }=1,2$ and $3$%
. Recall that $k^{\prime }$ is the number of regressors and $k=k^{\prime }+1$%
. The TMG estimator depends on the indicator, $\boldsymbol{1}\{d_{i}>a_{n}\}$%
, where $a_{n}=C_{n}n^{-\alpha }$. In view of the discussion in Section \ref%
{Threshold} on the choice of $\alpha $, we consider the values of $\alpha
=1/3,0.35$ and $1/2$, and as discussed earlier we set $C_{n}=\bar{d}%
_{n}=n^{-1}\sum_{i=1}^{n}d_{i}>0$, where $d_{i}=\func{det}(\boldsymbol{X}%
_{i}^{\prime }\boldsymbol{M}_{T}\boldsymbol{X}_{i})$. In what follows, we
report the results for the TMG estimator with $\alpha =1/3$, but discuss the
sensitivity of the TMG estimator to the choice of $\alpha $ in Section \ref%
{MChatalpha} of the online supplement. All MC experiments are based on $%
R=2,000 $ replications.

\begin{sidewaystable}
\caption{Bias, RMSE and size of FE, MG and TMG estimators of $\beta_{01}$ $(E(\beta_{i1}) = \beta_{01}=1)$ in the baseline DGP with one regressor, without time effects} 
\label{tab:T_d1_c12_chi2_tex0}
\vspace{-7mm}
\begin{center}
\scalebox{0.65}{
\begin{tabular}{rrcrrrrrrrrrrrrrrrrrrrrrrrrr}
 \hline\hline &  \multicolumn{13}{c}{ Uncorrelated heterogeneity: $\psi_{\beta_{1}} = 0 $, $PR^{2}=0.2$ } &  & \multicolumn{13}{c}{ Correlated heterogeneity: $\psi_{\beta_{1}}=0.5$, $PR^{2}=0.2$ }  \\ \cline{2-14} \cline{16-28}   
& $\hat{\pi}$ $(\times 100)$  & & \multicolumn{3}{c}{Bias}  & & \multicolumn{3}{c}{RMSE} && \multicolumn{3}{c}{Size $(\times 100)$} && $\hat{\pi}$ $(\times 100)$ & & \multicolumn{3}{c}{Bias}  & & \multicolumn{3}{c}{RMSE} && \multicolumn{3}{c}{Size $(\times 100)$} \\ \cline{2-2} \cline{4-6} \cline{8-10} \cline{12-14} \cline{16-16} \cline{18-20} \cline{22-24} \cline{26-28}  
 $T$ & TMG & & FE & MG & TMG && FE & MG & TMG && FE & MG & TMG &&  TMG & & FE & MG & TMG && FE & MG & TMG && FE & MG & TMG    \\ \hline 
&   \multicolumn{27}{c}{$n=1,000$} \\ \hline
2 & 27.50 &  & 0.002 & 1.616 & 0.003 &  & 0.141 & 72.785 & 0.263 &  & 4.6 & 1.8 & 5.1 &  & 27.50 &  & 0.357 & 1.831 & 0.026 &  & 0.397 & 82.491 & 0.299 &  & 54.5 & 1.8 & 5.1 \\
3 & 12.20 &  & -0.002 & -0.001 & -0.004 &  & 0.104 & 0.374 & 0.162 &  & 5.2 & 4.3 & 4.5 &  & 12.20 &  & 0.352 & -0.001 & 0.007 &  & 0.375 & 0.425 & 0.183 &  & 79.1 & 4.5 & 4.7 \\
4 & 6.10 &  & 0.000 & 0.000 & 0.000 &  & 0.086 & 0.151 & 0.122 &  & 5.4 & 5.6 & 5.4 &  & 6.10 &  & 0.353 & 0.001 & 0.006 &  & 0.370 & 0.171 & 0.138 &  & 90.0 & 5.4 & 5.0 \\
5 & 3.20 &  & -0.001 & 0.001 & 0.000 &  & 0.076 & 0.109 & 0.100 &  & 4.9 & 4.6 & 4.6 &  & 3.20 &  & 0.352 & 0.002 & 0.004 &  & 0.365 & 0.123 & 0.113 &  & 95.8 & 4.4 & 4.6 \\
6 & 1.80 &  & 0.001 & 0.001 & 0.001 &  & 0.070 & 0.092 & 0.088 &  & 5.1 & 5.2 & 4.9 &  & 1.80 &  & 0.353 & 0.001 & 0.002 &  & 0.365 & 0.103 & 0.099 &  & 97.9 & 5.3 & 5.0 \\
8 & 0.60 &  & -0.001 & 0.002 & 0.002 &  & 0.061 & 0.071 & 0.070 &  & 4.2 & 5.1 & 5.1 &  & 0.60 &  & 0.350 & 0.002 & 0.002 &  & 0.360 & 0.080 & 0.079 &  & 99.1 & 5.4 & 5.7 \\
10 & 0.20 &  & 0.000 & 0.000 & 0.000 &  & 0.057 & 0.063 & 0.063 &  & 4.7 & 5.1 & 5.0 &  & 0.20 &  & 0.352 & 0.000 & 0.000 &  & 0.361 & 0.069 & 0.069 &  & 99.6 & 5.4 & 5.3 \\
15 & 0.00 &  & 0.000 & 0.000 & 0.000 &  & 0.052 & 0.052 & 0.052 &  & 5.4 & 5.3 & 5.3 &  & 0.00 &  & 0.351 & 0.000 & 0.000 &  & 0.358 & 0.057 & 0.057 &  & 100.0 & 6.0 & 6.0 \\
 & \multicolumn{27}{c}{$n=2,000$} \\ \hline
2 & 24.60 &  & -0.003 & -0.909 & -0.004 &  & 0.102 & 101.969 & 0.197 &  & 5.7 & 1.5 & 5.4 &  & 24.60 &  & 0.349 & -1.031 & 0.015 &  & 0.370 & 115.567 & 0.224 &  & 81.6 & 1.5 & 5.2 \\
3 & 9.80 &  & 0.001 & -0.024 & -0.003 &  & 0.072 & 0.930 & 0.120 &  & 4.3 & 4.1 & 5.0 &  & 9.80 &  & 0.354 & -0.027 & 0.005 &  & 0.366 & 1.055 & 0.136 &  & 97.4 & 4.2 & 4.8 \\
4 & 4.40 &  & -0.001 & -0.002 & -0.002 &  & 0.061 & 0.108 & 0.087 &  & 5.0 & 4.9 & 5.2 &  & 4.40 &  & 0.353 & -0.002 & 0.002 &  & 0.361 & 0.122 & 0.098 &  & 99.7 & 4.8 & 5.1 \\
5 & 2.10 &  & 0.000 & -0.003 & -0.003 &  & 0.055 & 0.080 & 0.075 &  & 5.4 & 5.9 & 5.7 &  & 2.10 &  & 0.352 & -0.003 & -0.001 &  & 0.359 & 0.090 & 0.084 &  & 99.9 & 6.0 & 6.0 \\
6 & 1.10 &  & -0.001 & 0.001 & 0.001 &  & 0.049 & 0.062 & 0.061 &  & 5.1 & 4.7 & 4.4 &  & 1.10 &  & 0.354 & 0.001 & 0.002 &  & 0.360 & 0.070 & 0.068 &  & 100.0 & 4.6 & 4.4 \\
8 & 0.30 &  & -0.001 & -0.001 & -0.001 &  & 0.043 & 0.050 & 0.049 &  & 4.6 & 4.4 & 4.4 &  & 0.30 &  & 0.352 & -0.001 & -0.001 &  & 0.357 & 0.055 & 0.055 &  & 100.0 & 4.4 & 4.3 \\
10 & 0.10 &  & 0.000 & -0.001 & -0.001 &  & 0.040 & 0.043 & 0.043 &  & 4.3 & 4.9 & 4.9 &  & 0.10 &  & 0.354 & 0.000 & 0.000 &  & 0.358 & 0.048 & 0.048 &  & 100.0 & 4.3 & 4.4 \\
15 & 0.00 &  & -0.001 & 0.000 & 0.000 &  & 0.035 & 0.037 & 0.037 &  & 4.3 & 5.3 & 5.3 &  & 0.00 &  & 0.353 & 0.000 & 0.000 &  & 0.357 & 0.040 & 0.040 &  & 100.0 & 5.4 & 5.4 \\
 & \multicolumn{27}{c}{$n=5,000$} \\ \hline
2 & 21.20 &  & 0.000 & 1.415 & 0.000 &  & 0.065 & 98.557 & 0.132 &  & 5.5 & 2.0 & 4.8 &  & 21.20 &  & 0.353 & 1.603 & 0.016 &  & 0.362 & 111.699 & 0.150 &  & 99.5 & 2.0 & 4.6 \\
3 & 7.30 &  & -0.001 & 0.005 & 0.000 &  & 0.047 & 0.385 & 0.077 &  & 5.3 & 4.6 & 5.8 &  & 7.30 &  & 0.353 & 0.006 & 0.006 &  & 0.358 & 0.436 & 0.088 &  & 100.0 & 4.4 & 5.8 \\
4 & 2.90 &  & -0.001 & 0.000 & 0.000 &  & 0.038 & 0.068 & 0.057 &  & 4.4 & 4.5 & 5.2 &  & 2.90 &  & 0.352 & 0.000 & 0.002 &  & 0.355 & 0.077 & 0.065 &  & 100.0 & 4.9 & 5.1 \\
5 & 1.20 &  & 0.000 & 0.000 & 0.000 &  & 0.034 & 0.048 & 0.046 &  & 4.3 & 4.0 & 3.9 &  & 1.20 &  & 0.354 & 0.000 & 0.002 &  & 0.357 & 0.054 & 0.051 &  & 100.0 & 4.2 & 4.2 \\
6 & 0.50 &  & 0.000 & 0.000 & 0.000 &  & 0.031 & 0.040 & 0.040 &  & 4.0 & 4.5 & 4.6 &  & 0.50 &  & 0.354 & 0.000 & 0.000 &  & 0.357 & 0.045 & 0.044 &  & 100.0 & 4.7 & 4.6 \\
8 & 0.10 &  & 0.000 & 0.000 & 0.000 &  & 0.028 & 0.032 & 0.032 &  & 5.4 & 5.0 & 5.0 &  & 0.10 &  & 0.353 & 0.000 & 0.000 &  & 0.355 & 0.036 & 0.036 &  & 100.0 & 5.4 & 5.6 \\
10 & 0.00 &  & 0.000 & 0.000 & 0.000 &  & 0.026 & 0.028 & 0.028 &  & 5.3 & 5.0 & 5.0 &  & 0.00 &  & 0.353 & 0.000 & 0.000 &  & 0.354 & 0.031 & 0.031 &  & 100.0 & 4.6 & 4.6 \\
15 & 0.00 &  & 0.000 & 0.000 & 0.000 &  & 0.022 & 0.022 & 0.022 &  & 4.7 & 4.6 & 4.6 &  & 0.00 &  & 0.354 & 0.000 & 0.000 &  & 0.356 & 0.024 & 0.024 &  & 100.0 & 3.7 & 3.7 \\
 & \multicolumn{27}{c}{$n=10,000$} \\ \hline
2 & 19.00 &  & 0.002 & 0.769 & 0.001 &  & 0.045 & 114.192 & 0.102 &  & 5.1 & 2.4 & 6.5 &  & 19.00 &  & 0.355 & 0.871 & 0.016 &  & 0.360 & 129.420 & 0.116 &  & 100.0 & 2.4 & 6.2 \\
3 & 5.90 &  & 0.000 & 0.000 & 0.001 &  & 0.033 & 0.130 & 0.057 &  & 5.2 & 4.9 & 4.8 &  & 5.90 &  & 0.353 & 0.000 & 0.006 &  & 0.356 & 0.148 & 0.065 &  & 100.0 & 4.8 & 4.3 \\
4 & 2.10 &  & -0.001 & 0.000 & -0.001 &  & 0.028 & 0.047 & 0.042 &  & 5.4 & 5.2 & 5.1 &  & 2.10 &  & 0.352 & 0.000 & 0.001 &  & 0.354 & 0.053 & 0.047 &  & 100.0 & 5.2 & 5.3 \\
5 & 0.80 &  & -0.001 & 0.000 & 0.000 &  & 0.025 & 0.035 & 0.034 &  & 5.1 & 6.2 & 5.8 &  & 0.80 &  & 0.353 & 0.000 & 0.001 &  & 0.355 & 0.040 & 0.038 &  & 100.0 & 5.9 & 5.4 \\
6 & 0.30 &  & 0.000 & -0.001 & -0.001 &  & 0.022 & 0.028 & 0.028 &  & 5.0 & 5.2 & 4.7 &  & 0.30 &  & 0.353 & -0.001 & 0.000 &  & 0.354 & 0.032 & 0.031 &  & 100.0 & 4.8 & 4.8 \\
8 & 0.10 &  & -0.001 & 0.000 & 0.000 &  & 0.020 & 0.023 & 0.023 &  & 5.1 & 6.2 & 6.2 &  & 0.10 &  & 0.352 & 0.000 & 0.000 &  & 0.354 & 0.026 & 0.026 &  & 100.0 & 6.1 & 6.2 \\
10 & 0.00 &  & 0.000 & -0.001 & -0.001 &  & 0.018 & 0.020 & 0.020 &  & 5.5 & 5.4 & 5.4 &  & 0.00 &  & 0.353 & -0.001 & -0.001 &  & 0.354 & 0.023 & 0.022 &  & 100.0 & 5.4 & 5.4 \\ 
 15 & 0.00 &  & 0.000 & 0.000 & 0.000 &  & 0.016 & 0.016 & 0.016 &  & 5.3 & 5.0 & 5.0 &  & 0.00 &  & 0.353 & 0.000 & 0.000 &  & 0.354 & 0.018 & 0.018 &  & 100.0 & 5.2 & 5.2 \\
\hline
\hline
\end{tabular}
}
\end{center}
\vspace{-3mm}
{\footnotesize
Notes: 
(i) The baseline DGP is generated as $y_{it}=\alpha_{i} + \beta_{i1} x_{1, it} + u_{it}$, where the error processes for $y_{it}$ and $x_{1, it}$ equations are chi-squared and Gaussian, respectively, $x_{1, it}$ are generated without autoregressions, $\rho_{1, ix}=0$, or interactive effects, $\gamma_{1, ix}=0$, and $\psi_{\beta_{1}}$ (the degree of correlated heterogeneity) is defined by (\ref{eta_i}).  
For further details, see Section \ref{DGP}. 
(ii) FE and MG estimators are given by (\ref{fee}) and (\ref{mge}), respectively. 
The TMG estimator and its variance estimator are given by (\ref{TMGb}) and (\ref{varC}). 
(iii) The trimming threshold value for the TMG estimator is given by $a_{n}=\bar{d}_{n} n^{-\alpha}$, where $\bar{d}_{n} =\frac{1}{n} \sum_{i}^{n} d_{i}$ and $d_{i} =\func{det}(\boldsymbol{X}_{i}^{\prime} \boldsymbol{M}_{T} \boldsymbol{X}_{i})$ with $\boldsymbol{X}_{i}=(\boldsymbol{x}_{i1},\boldsymbol{x}_{i2},...,\boldsymbol{x}_{iT})^{\prime}$, $\boldsymbol{M}_{T} = \boldsymbol{I}_{T} - \boldsymbol{\tau}_{T}\boldsymbol{\tau}_{T}^{\prime}/T$, $\boldsymbol{I}_{T}$ being a $T \times T$ identity matrix and $\boldsymbol{\tau}_{T}$ being a $T \times 1$ vector of ones. 
$\alpha$ is set to $1/3$. 
$\hat{\pi}$ is the simulated fraction of individual estimates being trimmed, defined by (\ref{pin}). 
}
\end{sidewaystable}

Table \ref{tab:T_d1_c12_chi2_tex0} reports bias, root mean squared errors
(RMSE) and size for estimation of $E(\beta _{i1})=\beta _{01}$ in the case
of DGPs with one regressor, $k^{\prime }=1$. The left panel of the table
provides results when heterogeneity is uncorrelated (i.e. $\psi _{\beta
_{1}}=0$), whilst the right panel of the table gives the results for the
case of correlated heterogeneity with $\psi _{\beta _{1}}=0.5$. The fraction
of the trimmed estimates, $\pi _{n}$, defined by (\ref{pin}), tends to be
quite large for the case where $T=k$, but falls quite rapidly as $T-k$ is
increased. For example, for $T=2$ and $n=1,000$, as many as $27.5$ per cent
of the individual estimates are trimmed when computing the TMG estimates,
but this fraction falls to $1.8$ per cent when the number of time periods is
increased to $T=6$. However, recall that the TMG estimator continues to make
use of the trimmed estimates, as can be seen from (\ref{betaTn}), and the
TMG estimator shows little bias compared to the (untrimmed) MG estimator.
The TMG and MG estimators converge as $T$ is increased, and they are almost
identical for the panels when $T\geq 8$. This is in line with the two
estimates of $\alpha _{p}$ reported in Table \ref{tab:ap_mc} for $k^{\prime
}=1$ and $T=8$. Both estimates ($3.12$ and $3.56$) are well in excess of $2$%
, such that all individual estimators have second-order moments and
therefore no trimming is required.

Comparing TMG and FE estimators, we first note that in line with the theory,
the FE estimator performs very well under uncorrelated heterogeneity but is
badly biased when heterogeneity is correlated. Further, this bias does not
diminish if $n$ and $T$ are increased. The simulated bias of the FE
estimator in the case where $T=k=2$ and $n=1,000$ amounts to $0.357$, which
is close to the analytical result presented in Section \ref{DGP}. When
heterogeneity is correlated, the FE estimator also exhibits substantial size
distortions, which tend to get accentuated as $n$ is increased for a given $%
T $. In contrast, the TMG estimator is robust to the choice of $\psi _{\beta
_{1}}$ and delivers size very close to the assumed five per cent level.%
\footnote{%
Increasing $PR^{2}$ from $0.2$ to $0.4$ does not affect the bias of the FE
estimator and slightly reduces its RMSE, but results in a higher degree of
size distortion under correlated heterogeneity. See the results summarized
in the right panel of Table \ref{tab:T_d1_c12_chi2_tex0} and Table \ref%
{tab:T_d1_c34_chi2_tex0} in the online supplement.}

The empirical power functions for TMG and FE estimators in the case of a
single regressor and for the sample sizes $n=10,000$ and $T=2$, $3$, and $4$
, are displayed in Figure \ref{fig:fe_tmg_k2_base_1}. As can be seen, under
uncorrelated heterogeneity (the left panel with $\psi _{\beta _{1}}=0$),
both estimators are centered correctly around $\beta _{01}=1$, with the FE
estimator having better power properties. But the differences between the
power of FE and TMG estimators shrink rapidly and become negligible as $T$
is increased from $T=2$ to $T=4$.\footnote{%
But as shown in Example \ref{ExampleMG-FE}, it does not necessarily follow
that the FE estimator will dominate the TMG estimator in terms of efficiency
under uncorrelated heterogeneity. See the left panel of Table \ref%
{tab:T_d1_c2_hk23_chi2_tex0} and Figure \ref{fig:fe_tmg_k2_hetrosk} in the
online supplement.} The right panel provides the power plots under
correlated heterogeneity with $\psi _{\beta _{1}}=0.5$. In this case, the
empirical power functions of the FE estimator now shift markedly to the
right, away from the true value, an outcome that becomes more pronounced as $%
T $ is increased. In contrast, the empirical power functions for the TMG
estimator are always centered correctly and are robust to the choice of $%
\psi _{\beta _{1}}$.

\begin{figure}[h!]
\caption{Empirical power functions for FE, MG, and TMG estimators of $%
\protect\beta _{01}$ $(E(\protect\beta_{i1}) = \protect\beta_{01}=1)$ in the
baseline DGP with one regressor, without time effects, for $n=10,000$ and $%
T=2,3,4$}
\label{fig:fe_tmg_k2_base_1}\vspace{-7mm}
\par
\begin{center}
\includegraphics[scale=0.25]{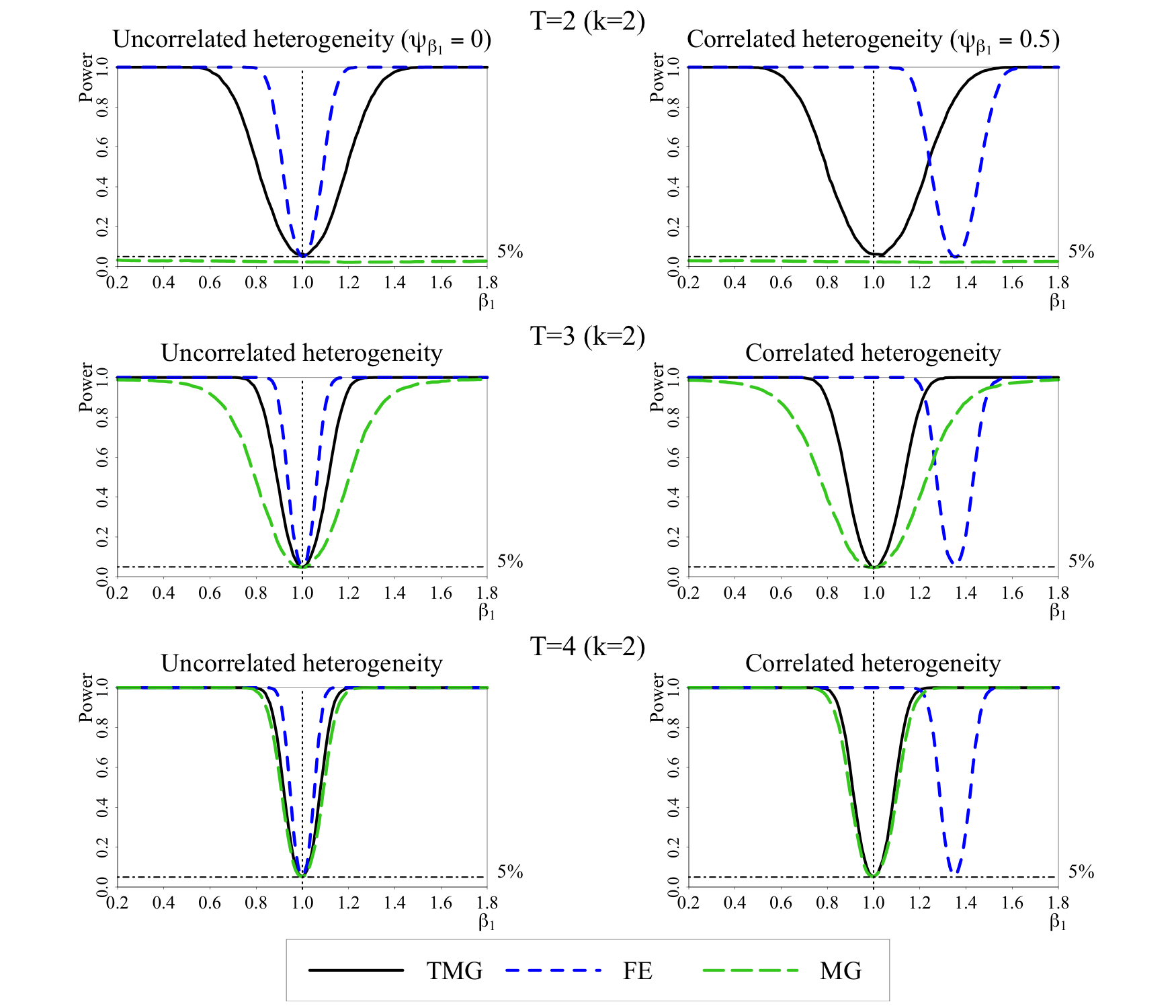}
\end{center}
\par
\vspace{-3mm} 
\begin{spacing}{1}
{\footnotesize
Notes: See footnotes to Table \ref{tab:T_d1_c12_chi2_tex0}.}
\end{spacing}
\end{figure}

Turning to DGPs with more than one regressor, to save space, the results are
summarized in Tables \ref{tab:T_d1_c12_chi2_tex0_k3} and \ref%
{tab:T_d1_c12_chi2_tex0_k4} in Section \ref{MCk34} of the online supplement.
These tables present estimation results for the baseline DGP with two and
three regressors, respectively. The FE estimator continues to display
substantial bias under correlated heterogeneity. For the TMG estimator, as
the number of regressors, $k^{\prime }$, increases, the trimmed fraction
rises for a given $T$. With $n=1,000$, it grows from $27.5$ per cent to $%
41.8 $ percent for $T=k=3$, and $50.8$ per cent for $T=k=4$. Moreover, as $%
k^{\prime }$ increases, the trimmed fraction decreases with $T$ at a slower
rate, in line with the estimates of $\alpha _{p}$ in Table \ref{tab:ap_mc}.

To summarize, under uncorrelated heterogeneity, the FE estimator performs
well despite the heterogeneity, and is more efficient than the TMG estimator
in the case of the baseline DGP used in our MCs. But, in general, the
relative efficiency of TMG and FE estimators depends on the underlying DGP.
The situation is markedly different when heterogeneity is correlated, and
the FE estimator can be badly biased, leading to incorrect inference, whilst
the TMG estimator provides valid inference with size around the nominal five
per cent level and reasonable power, irrespective of whether $\beta _{i1}$
is correlated with $x_{1,it}$ or not.

\subsubsection{Comparison of TMG and GP estimators \label{MCcorr}}

Focusing on the case of correlated heterogeneity, we now compare the
relative performance of TMG and GP estimators. To implement the GP
estimator, defined by (\ref{gpe}), for $T=k$, we follow GP and set $%
h_{n}=C_{GP}n^{-\alpha _{GP}}$, with $\alpha _{GP}=1/3$ and $C_{GP}=\frac{1}{%
2}\min \left( \hat{\sigma}_{D},\hat{r}_{D}/1.34\right) $, where $\hat{\sigma}%
_{D}$ and $\hat{r}_{D}$ are the respective sample standard deviation and
interquartile range of $|\func{det}\left( \boldsymbol{W}_{i}\right) |$ with $%
\boldsymbol{W}_{i}=\left( \boldsymbol{\tau }_{T},\boldsymbol{X}_{i}\right) $%
. For further details of GP's choice of $\alpha _{GP}$ when $T=k$, see p.
2138 of \cite{GrahamPowell2012}. The asymptotic variance of their estimator
in the case of models with and without time effects is provided in equation
(30) on p. 2126 of their paper. There is no clear guidance by GP as to the
choice of $h_{n}$ when $T>k$.\footnote{%
For $T=3$ with $k=2$, GP do not use the bandwidth parameter, $h_{n}$, but
directly select the \textquotedblleft percent trimmed\textquotedblright , $%
\pi _{n}$. In their empirical application, they report estimates with 4 per
cent being trimmed for $T=3$ with $k=2$. See the last column of Table 3 on
p. 2136 of GP.} For consistency, when $T>k$, for computing GP estimates we
continue to use their bandwidth, $h_{n}=C_{GP}n^{-\alpha _{GP}}$ with $%
\alpha _{GP}=1/3$, but set $C_{GP}=\left(
n^{-1}\sum_{i=1}^{n}d_{i,GP}\right) ^{1/2}$ and trim if $d_{i,GP} = \func{det%
}\left( \boldsymbol{W}_{i}^{\prime} \boldsymbol{W}_{i}\right) < h_{n}^{2}$.

The bias, RMSE, and size for the two estimators are summarized in Table \ref%
{tab:beta_d1_c2_chi2_tex0_b} for $T=2,3,4,5,6$, and $n=1000,2000,5000,10000$%
. The associated empirical power functions are displayed in Figure \ref%
{fig:tmg_gp_k2_base}.

\begin{table}[h]
\caption{Bias, RMSE and size of TMG and GP estimators of $\protect\beta %
_{01} $ $(E(\protect\beta _{i1})=\protect\beta _{01}=1)$ in the baseline DGP
with one regressor, without time effects, but with correlated heterogeneity, 
$\protect\psi _{\protect\beta_{1}}=0.5$}
\label{tab:beta_d1_c2_chi2_tex0_b}\vspace{-6mm}
\par
\begin{center}
\scalebox{0.78}{
\begin{tabular}{rrrrrrrrrrrr}
\hline\hline
 & \multicolumn{2}{c}{$\hat{\pi}$ $(\times 100)$} &  & \multicolumn{2}{c}{Bias} &  & \multicolumn{2}{c}{RMSE} &  & \multicolumn{2}{c}{Size $(\times 100)$} \\ \cline{2-3} \cline{5-6} \cline{8-9} \cline{11-12}
$T$ & TMG   & GP   &  & TMG    & GP     &  & TMG   & GP    &  & TMG & GP  \\ \hline
  & \multicolumn{11}{c}{$n=1,000$}                                        \\ \hline
2 & 27.5 & 4.1 &  & 0.026 & -0.003 &  & 0.299 & 0.653 &  & 5.1 & 4.6 \\
3 & 12.2 & 1.3 &  & 0.007 & -0.001 &  & 0.183 & 0.234 &  & 4.7 & 4.7 \\
4 & 6.1 & 0.2 &  & 0.006 & 0.001 &  & 0.138 & 0.158 &  & 5.0 & 5.1 \\
5 & 3.2 & 0.0 &  & 0.004 & 0.002 &  & 0.113 & 0.121 &  & 4.6 & 4.4 \\
6 & 1.8 & 0.0 &  & 0.002 & 0.001 &  & 0.099 & 0.102 &  & 5.0 & 5.4 \\
 & \multicolumn{11}{c}{$n = 2,000$} \\ \hline
2 & 24.6 & 3.2 &  & 0.015 & 0.004 &  & 0.224 & 0.523 &  & 5.2 & 4.7 \\
3 & 9.8 & 0.8 &  & 0.005 & -0.003 &  & 0.136 & 0.176 &  & 4.8 & 4.9 \\
4 & 4.4 & 0.1 &  & 0.002 & -0.004 &  & 0.098 & 0.113 &  & 5.1 & 4.8 \\
5 & 2.1 & 0.0 &  & -0.001 & -0.003 &  & 0.084 & 0.089 &  & 6.0 & 6.2 \\
6 & 1.1 & 0.0 &  & 0.002 & 0.001 &  & 0.068 & 0.070 &  & 4.4 & 4.6 \\
 & \multicolumn{11}{c}{$n = 5,000$} \\ \hline
2 & 21.2 & 2.4 &  & 0.016 & 0.008 &  & 0.150 & 0.398 &  & 4.6 & 6.1 \\
3 & 7.3 & 0.4 &  & 0.006 & 0.001 &  & 0.088 & 0.117 &  & 5.8 & 5.3 \\
4 & 2.9 & 0.0 &  & 0.002 & -0.001 &  & 0.065 & 0.074 &  & 5.1 & 5.1 \\
5 & 1.2 & 0.0 &  & 0.002 & 0.000 &  & 0.051 & 0.054 &  & 4.2 & 4.2 \\
6 & 0.5 & 0.0 &  & 0.000 & 0.000 &  & 0.044 & 0.045 &  & 4.6 & 4.6 \\
 & \multicolumn{11}{c}{$n = 10,000$} \\ \hline
2 & 19.0 & 1.9 &  & 0.016 & 0.008 &  & 0.116 & 0.302 &  & 6.2 & 4.4 \\
3 & 5.9 & 0.3 &  & 0.006 & 0.000 &  & 0.065 & 0.086 &  & 4.3 & 4.7 \\
4 & 2.1 & 0.0 &  & 0.001 & -0.001 &  & 0.047 & 0.052 &  & 5.3 & 5.6 \\
5 & 0.8 & 0.0 &  & 0.001 & 0.000 &  & 0.038 & 0.040 &  & 5.4 & 5.8 \\
6 & 0.3 & 0.0 &  & 0.000 & -0.001 &  & 0.031 & 0.032 &  & 4.8 & 4.8 \\
\hline\hline
\end{tabular}
}
\end{center}
\par
\vspace{-1mm} 
\begin{spacing}{1}
{\footnotesize 
Notes: 
(i) The GP estimator proposed by \cite{GrahamPowell2012} is given by (\ref{gpe}). For $T=k$, GP compare $d_{i,GP}^{1/2}$ with the bandwidth $h_{n} = C_{GP}n^{-\alpha_{GP}}$, where $d_{i,GP} = \func{det}(\boldsymbol{W}_{i}^{\prime}\boldsymbol{W}_{i})$ and $\boldsymbol{W}_{i} = (\boldsymbol{\tau}_{T}, \boldsymbol{X}_{i})$. $\alpha_{GP}$ is set to 1/3. $C_{GP}=\frac{1}{2}\min \left( \hat{\sigma}_{D},\hat{r}_{D}/1.34\right) $, where $\hat{\sigma}_{D}$ and $\hat{r}_{D}$ are the respective sample standard deviation and interquartile range of $d_{i,GP}^{1/2}$. For $T>k$, $C_{GP}=\left(n^{-1}\sum_{i=1}^{n}d_{i,GP}\right)^{1/2}$. See sub-section \ref{MCcorr} for details. 
(ii) For details of the baseline DGP without time effects and the TMG estimator, see footnotes to Table \ref{tab:T_d1_c12_chi2_tex0}. 
$\hat{\pi}$ is the simulated fraction of individual estimates being trimmed, defined by (\ref{pin}). } 
\end{spacing}
\end{table}

The fractions of the trimmed estimates, $\hat{\pi}$, differ markedly across
the estimators. For example, when $T=2$ and $n=1,000$, the fraction of
trimmed estimates for the TMG estimator is around $27.5$ per cent as
compared to $4.1$ per cent for the GP estimator, and falls to $19.0$ per
cent as $n$ is increased to $10,000$. Increasing $T$ from $2$ to $3$ with $%
n=1,000$ reduces this fraction to $12.2$ per cent as compared to $1.3$ per
cent for the GP estimator. The heavy trimming causes the TMG estimator to
have a larger bias than the GP estimator, particularly when $T=2$ and $n$ is
large. However, the TMG estimator continues to have better overall small
sample performance. Recall that the TMG estimator makes use of the trimmed
estimates, as set out in the second term of (\ref{betaTn}), but the trimmed
estimates are not used in the GP estimator. This difference in the way
trimmed estimates are treated is reflected in the lower RMSE of the TMG
estimator as compared to MG and GP estimators for all $T$ and $n$
combinations. For example, when $T=2$ and $n=1,000$, the RMSE of the TMG is $%
0.299$ as compared to $0.653$ for the GP estimator. The relative advantage
of the TMG estimator continues when $T$ increases from $2$ to $3$. For $T=3$
and $n=1,000$, the RMSE of the TMG estimator stands at $0.183$ compared to $%
0.234$ for the GP estimator. The larger the value of $T $, the less
important the trimming becomes.

The empirical power functions for TMG, GP, and MG estimators are shown in
Figure \ref{fig:tmg_gp_k2_base}. As can be seen, the TMG estimator is
uniformly more powerful than the GP estimator across different sample sizes.
To save space, the MC results for models with time effects are summarized in
sub-section \ref{MCte} of the online supplement. Similar outcomes are
obtained when we consider DGPs with two or three regressors. See Tables \ref%
{tab:beta_d1_c2_chi2_tex0_b_k3} and \ref{tab:beta_d1_c2_chi2_tex0_b_k4}, and
the corresponding empirical power functions, Figures \ref{fig:tmg_gp_k3_base}
and \ref{fig:tmg_gp_k4_base}, in sub-section \ref{MCk34} of the online
supplement. This is particularly the case when $T=k\in \{3,4\}$, where the
trimmed fraction of the GP estimator rises only slightly with the number of
regressors, resulting in substantial declines in the empirical powers.

\begin{figure}[h!]
\caption{Empirical power functions for TMG, GP, and MG estimators of $%
\protect\beta_{01}$ $(E(\protect\beta_{i1}) = \protect\beta_{01}=1)$ in the
baseline DGP with one regressor, without time effects, but with correlated
heterogeneity, $\protect\psi_{\protect\beta_{1}}=0.5$, for $n=10,000$ and $%
T=2,3,4,5$}
\label{fig:tmg_gp_k2_base}\vspace{-8mm}
\par
\begin{center}
\includegraphics[scale=0.14]{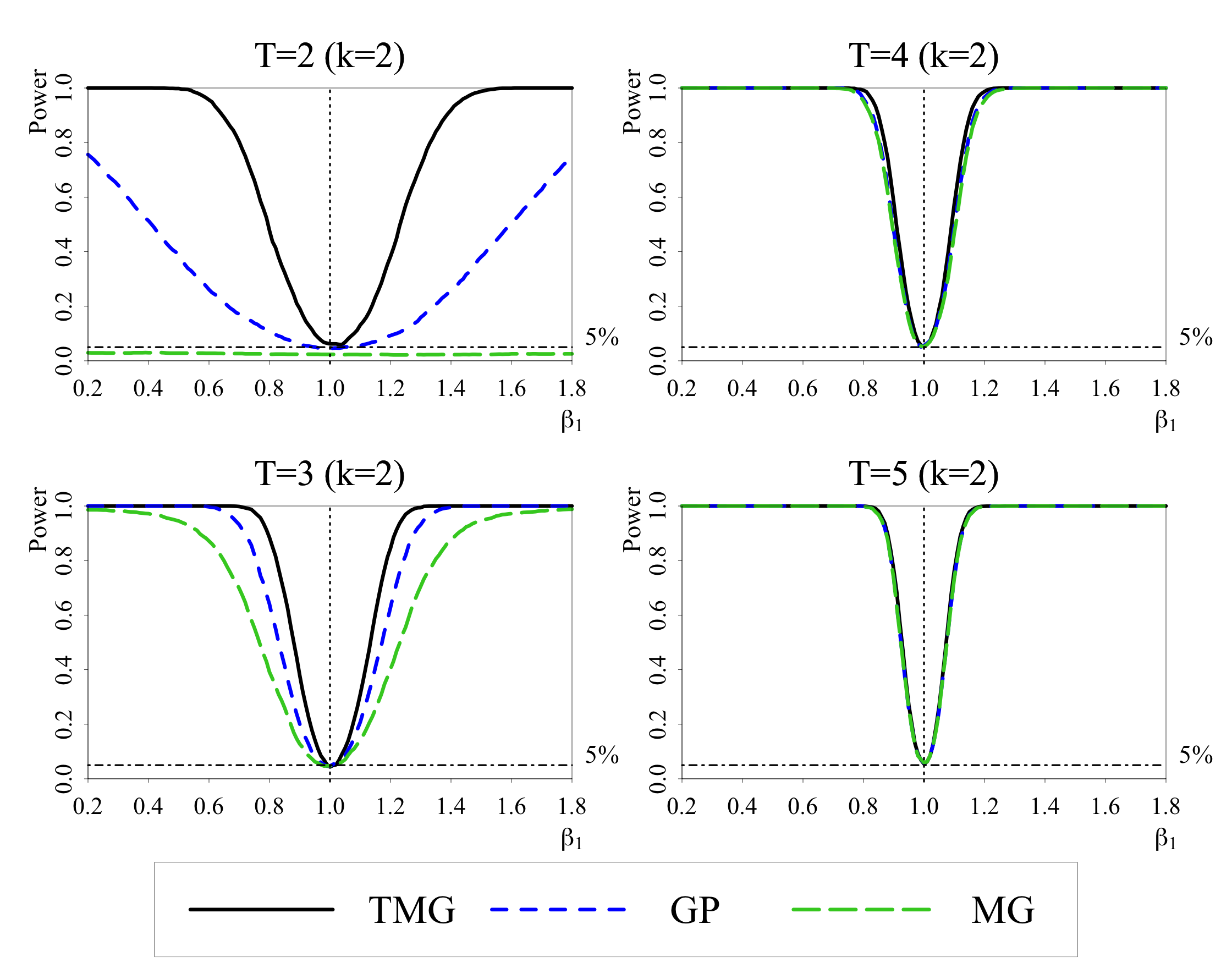}
\end{center}
\par
\vspace{-3mm} 
\begin{spacing}{1}
{\footnotesize
Notes: See footnotes to Tables \ref{tab:T_d1_c12_chi2_tex0} and \ref{tab:beta_d1_c2_chi2_tex0_b}.
}
\end{spacing}
\end{figure}

For the purpose of comparisons, in addition to the choice of $\alpha
_{GP}=1/3$ by GP, we also considered the threshold values $2\alpha
_{GP}=\alpha \in \{0.35,1/2\}$, so that the two threshold functions (ours
and the one suggested by GP) share the same exponents. The results for $T=2$
summarized in Table \ref{tab:thresh_d1_c2_chi2_tex0_b_t2} show that the TMG
estimator has smaller RMSE as compared to the GP estimator for $\alpha
=2\alpha _{GP}\in \{0.35,1/2\}$, and has better power as can be seen from
Figures \ref{fig:tmg_gp_alpha_1} and \ref{fig:tmg_gp_alpha_2} in the online
supplement. Additional results for $T=3$ are provided in sub-section \ref%
{MCthresh} of the online supplement, where the differences between TMG and
GP estimators are much smaller.

\begin{table}[!htb]
\caption{Bias, RMSE and size of TMG and GP estimators of $\protect\beta_{01}$
$(E(\protect\beta_{i1})=\protect\beta_{01}=1)$ for different threshold
exponents, $\protect\alpha $ and $\protect\alpha _{GP}$, in the baseline DGP
with one regressor, without time effects, but with correlated heterogeneity, 
$\protect\psi _{\protect\beta_{1}}=0.5$ ($T=2$)}
\label{tab:thresh_d1_c2_chi2_tex0_b_t2}\vspace{-6mm}
\par
\begin{center}
\scalebox{0.45}{
\LARGE
\begin{tabular}{lcrrcclrrcclrrcclrrcc}
\hline\hline
 &  & \multicolumn{4}{c}{$n=1,000$} &  & \multicolumn{4}{c}{$n=2,000$} &  & \multicolumn{4}{c}{$n=5,000$} &  & \multicolumn{4}{c}{$n=10,000$} \\ \cline{3-6} \cline{8-11} \cline{13-16} \cline{18-21} 
 & $\alpha$/$\alpha_{GP}$ & \multicolumn{1}{c}{$\hat{\pi}$} & \multicolumn{1}{c}{Bias} & RMSE & Size &  &  \multicolumn{1}{c}{$\hat{\pi}$} & \multicolumn{1}{c}{Bias} & RMSE & Size &  & \multicolumn{1}{c}{$\hat{\pi}$} &\multicolumn{1}{c}{Bias} & RMSE & Size &  & \multicolumn{1}{c}{$\hat{\pi}$} & \multicolumn{1}{c}{Bias} & RMSE & Size \\ \hline
TMG & 1/3 & 27.5 & 0.026 & 0.299 & 5.1 &  & 24.6 & 0.015 & 0.224 & 5.2 &  & 21.2 & 0.016 & 0.150 & 4.6 &  & 19.0 & 0.016 & 0.116 & 6.2 \\
TMG & 0.35 & 26.0 & 0.025 & 0.307 & 5.1 &  & 23.2 & 0.013 & 0.230 & 5.6 &  & 19.8 & 0.015 & 0.155 & 4.7 &  & 17.6 & 0.014 & 0.120 & 6.3 \\
TMG & 0.50 & 15.7 & 0.013 & 0.395 & 5.5 &  & 13.3 & 0.001 & 0.300 & 5.4 &  & 10.5 & 0.008 & 0.215 & 5.9 &  & 8.9 & 0.008 & 0.165 & 5.6 \\
GP & 0.35/2 & 12.0 & 0.018 & 0.392 & 5.1 &  & 10.7 & 0.002 & 0.294 & 5.7 &  & 9.1 & 0.011 & 0.201 & 5.0 &  & 8.1 & 0.010 & 0.153 & 5.6 \\
GP & 0.25 & 7.2 & 0.000 & 0.513 & 6.2 &  & 6.1 & -0.005 & 0.383 & 5.4 &  & 4.8 & 0.007 & 0.279 & 5.8 &  & 4.1 & 0.008 & 0.212 & 5.5 \\
GP & 1/3 & 4.1 & -0.003 & 0.653 & 4.6 &  & 3.2 & 0.004 & 0.523 & 4.7 &  & 2.4 & 0.008 & 0.398 & 6.1 &  & 1.9 & 0.008 & 0.302 & 4.4\\
\hline\hline
\end{tabular}}
\end{center}
\par
\vspace{-2mm} 
\begin{spacing}{1}
{\footnotesize
Notes: See footnotes to Table \ref{tab:beta_d1_c2_chi2_tex0_b}. $\hat{\pi}$ is the simulated fraction of individual estimates being trimmed, defined by (\ref{pin}); the reported values are $100 \times \hat{\pi}$. Size is reported in per cent.}
\end{spacing}
\end{table}

\subsubsection{MC evidence on the Hausman test of correlated heterogeneity}

Table \ref{tab:Test_d1_chi2_tex0} reports empirical size and power of the
Hausman test of correlated heterogeneity given by (\ref{htest}) under three
scenarios: homogeneity (left), uncorrelated heterogeneity (middle), and
correlated heterogeneity (right). We continue to focus on the simple case
where $k^{\prime }=1$. The size of the test is around the nominal level of 5
per cent. When $\boldsymbol{x}_{it}$ is strictly exogenous and $\boldsymbol{%
\beta }_{i1}$ is distributed independently of $\boldsymbol{X}_{i} $, FE, MG
and TMG estimators are all consistent under homogeneity and if heterogeneity
is present but uncorrelated. In such a case, the Hausman test does not have
power. However, in the case where slope coefficients are heterogeneous 
\textit{and} correlated with the regressors, the TMG estimator is consistent
while the FE estimator is biased for all $T$. In this case, we would expect
the proposed test to have power, and this is indeed evident in the right
panel of Table \ref{tab:Test_d1_chi2_tex0}. Also, the power of the test
rises with increases in $n$ even when $T=2$, illustrating the (ultra) small $%
T$ consistency of the proposed test. We also obtain similar test results
when we allow for time effects. See sub-section \ref{MCtest} of the online
supplement.

\begin{table}[h]
\caption{Empirical size and power of the Hausman test of correlated
heterogeneity in the baseline DGP with one regressor, without time effects}
\label{tab:Test_d1_chi2_tex0}\vspace{-6mm}
\par
\begin{center}
\scalebox{0.8}{
\begin{tabular}{lccccccccccrrrr}
\hline\hline 
     &   \multicolumn{9}{c}{Under $H_{0}$: Uncorrelated heterogeneity} &  & \multicolumn{4}{c}{Under $H_{1}$: Correlated heterogeneity}   \\ \cline{2-10} \cline{12-15}
     &   \multicolumn{4}{c}{$\sigma_{\beta_{1}}^{2}=0$}   &  & \multicolumn{4}{c}{$\sigma_{\beta_{1}}^{2}=0.75$, $\psi_{\beta_{1}} =0$} &  & \multicolumn{4}{c}{$\sigma_{\beta_{1}}^{2}=0.75$, $\psi_{\beta_{1}} =0.5$ }  \\ \cline{2-5} \cline{7-10} \cline{12-15} 
$T/n$   & 1,000   & 2,000   & 5,000   & 10,000  &  & 1,000 & 2,000 & 5,000  & 10,000  &  & $\,\,\,\,\,$1,000 & $\,\,\,\,\,$2,000 & $\,\,\,\,\,$5,000 & $\,\,\,\,$10,000 \\ \hline  
2   & 5.1 & 4.4 & 5.1 & 4.0 &  & 5.6 & 5.1 & 4.9 & 5.5 &  & 23.6 & 39.0 & 69.0 & 88.3 \\
3   & 4.1 & 4.6 & 5.4 & 5.1 &  & 4.8 & 5.1 & 5.1 & 4.6 &  & 58.7 & 83.2 & 98.9 & 100.0 \\
4   & 5.6 & 5.1 & 4.3 & 5.4 &  & 4.3 & 4.8 & 5.0 & 5.1 &  & 85.0 & 98.8 & 100.0 & 100.0 \\
5   & 4.8 & 4.9 & 5.1 & 5.0 &  & 5.6 & 5.6 & 4.0 & 5.1 &  & 95.9 & 99.9 & 100.0 & 100.0 \\
6   & 5.3 & 6.0 & 5.2 & 5.7 &  & 5.4 & 4.6 & 5.0 & 5.1 &  & 98.8 & 100.0 & 100.0 & 100.0 \\
8   & 4.5 & 5.3 & 5.6 & 4.3 &  & 4.3 & 5.1 & 5.5 & 5.7 &  & 100.0 & 100.0 & 100.0 & 100.0 \\
\hline\hline
\end{tabular}}
\end{center}
\par
\vspace{-3mm} 
\begin{spacing}{1}
\footnotesize{
Notes: 
(i) In the baseline DGP for the test, the outcome variable is generated as $y_{it}=\alpha_{i} + \beta_{i1} x_{1, it} + u_{it}$, with $\alpha_{i}$ correlated with $x_{1, it}$ under both the null and alternative hypotheses. For details of the baseline DGP without time effects, see sub-section \ref{DGP}. (ii) The null hypothesis is given by (\ref{null}), including the case of homogeneity with $\sigma_{\beta_{1}}^{2}=0$ and the case of uncorrelated heterogeneity with $\psi_{\beta_{1}}=0$ (the degree of correlated heterogeneity defined by (\ref{eta_i})) and $\sigma_{\beta_{1}}^{2}=0.75$. The alternative of correlated heterogeneity is generated with $\psi_{\beta_{1}}=0.5$ and $\sigma_{\beta_{1}}^{2}=0.75$. 
(iii) The test statistic is calculated based on the difference between FE and TMG estimators, given by (\ref{htest}). 
Size and power are in per cent.}
\end{spacing}
\end{table}

Hausman test results for DGPs with two or three regressors are summarized in
Tables \ref{tab:Test_d1_chi2_tex0_k3} through \ref%
{tab:Test_te_d1_chi2_tex0_k4} of the online supplement. Recall that we allow
for correlated heterogeneity only in the coefficients of the first
regressor. Consequently, we observe a decline in the power of the Hausman
test as we add regressors with uncorrelated heterogeneous coefficients. The
empirical power of the test decreases with $k$, particularly when $%
T=k=k^{\prime }+1$.

Finally, as with all other Monte Carlo studies our small-sample analysis and
its conclusions remain specific to the Monte Carlo designs studied.

\section{Empirical application \label{APP}}

In this section, we revisit the empirical application in \cite%
{GrahamPowell2012} who provide estimates of the average effect of household
expenditures on calorie demand, based on a sample of households from poor
rural communities in Nicaragua that participated in a conditional cash
transfer program. The data set is a balanced panel with $n=1,358$ households
observed from 2000 to 2002. We present estimates of the average effects
using the following panel data model with time effects: 
\begin{equation}
\ln (Cal_{it})=\alpha _{i}+\phi _{t}+\beta _{i}\ln (Exp_{it})+u_{it},
\label{calm1}
\end{equation}%
where $\ln (Cal_{it})$ denotes the logarithm of household calorie
availability per capita in year $t$ of household $i$, and $\ln (Exp_{it})$
denotes the logarithm of real household expenditures per capita (in
thousands of 2001 cordobas) of household $i$ in year $t$. The parameter of
interest is the average effect defined by $\beta _{0}=E(\beta_{i})$.

We first estimate $\alpha _{p}$ as a diagnostic to see if trimming is
needed. For the panel covering the period 2001--2002 $(T=2)$, Hill's
estimates of $\alpha _{p}$ are 0.71 (0.12) and 0.59 (0.17) for the two
cut-off values of $n^{1/2}$ and $n^{1/3}$, respectively, with standard
errors in parentheses. Similarly, for the dataset 2000--2002 $(T=3)$, the
estimates of $\alpha _{p}$ are 0.94 (0.15) and 0.96 (0.28), respectively.
All these estimates are well below two, and it is advisable that the trimmed
MG estimator is used for consistent estimation of the average effects, $%
\beta _{0}$, in case heterogeneity in $\beta _{i}$ is correlated.

Table \ref{tab:RPS_cal_htest} reports results of the Hausman test of
correlated heterogeneity in the effects of household expenditures on calorie
demand. The null hypothesis of uncorrelated heterogeneity is rejected for
both the panels and irrespective of whether time effects are included.
Therefore, for this application, the FE and TWFE estimates of $\beta _{0}$
could be biased.

\begin{table}[h]
\caption{Hausman statistics for testing correlated heterogeneity in the
effects of household expenditures on calorie demand in Nicaragua}
\label{tab:RPS_cal_htest}
\begin{center}
\vspace{-6mm} 
\scalebox{0.85}{
\begin{tabular}{lccccccc}
\hline\hline
 & \multicolumn{3}{c}{Without time effects} &  & \multicolumn{3}{c}{With time effects} \\ \cline{2-4} \cline{6-8}
 & 2001--2002 &  & 2000--2002 &  & 2001--2002 &  & 2000--2002 \\ \hline
Statistics & 5.918 &  & 7.626 &  & 5.959 & & 7.653 \\
$p$-value  & 0.015 &  & 0.006 &  & 0.015 & & 0.006\\
$T$ & 2 &  & 3 &  & 2 &  & 3\\
\hline\hline
\end{tabular}}
\end{center}
\par
\vspace{-2mm}{\footnotesize Notes: The test is applied to the average effect 
$\beta_{0}=E(\beta_{i})$ in the model (\ref{calm1}) based on a panel of $%
1,358$ households. The test statistic for panels without time effects is
described in the footnote (iii) to Table \ref{tab:Test_d1_chi2_tex0}. For
panels with time effects, the test statistic is based on the difference
between the TWFE and TMG-TE estimators given by (\ref{htetest}) with $T=2$
and (\ref{htetest2}) with $T>2$. For further details, see sub-section \ref%
{TestTE} in the online supplement. }
\end{table}

Table \ref{tab:RPS_cal_T2} presents the estimates of $\beta _{0}$ based on
the panel of 2001--2002 (with $T=2$) without time effects (left panel), and
with time effects (right panel). The estimates are not affected by the
inclusion of time effects but differ considerably across different methods.%
\footnote{%
When $T=2$, adding time effects does not change the estimated average effect.%
} Turning to the trimmed estimators, we find that only the TMG estimator is
heavily trimmed with 27.1 per cent of the estimates being trimmed, whilst
the rate of trimming is only around 3.8 per cent for the GP estimator.%
\footnote{%
For the 2001--2002 panel, $\hat{\pi}$ of the GP estimator is identical to
the one reported in Table 3 of \cite{GrahamPowell2012}. GP estimated a model
with time-varying coefficients, $y_{it}=\alpha _{i}+\phi _{t}+\left( \beta
_{i}+\phi _{t,\beta }\right) x_{it}+u_{it}$, where $\left( \boldsymbol{\phi }%
^{\prime},\boldsymbol{\phi }_{\beta }^{\prime}\right)^{\prime}$ are
identified by stayers but estimated by near stayers with $\boldsymbol{\phi }%
_{\beta } = (\phi _{1,\beta }, \phi _{2,\beta }, ..., \phi _{T,\beta
})^{\prime}$. While $\boldsymbol{\phi }_{\beta }$ is not included in (\ref%
{calm1}), the GP estimates we compute are close to the trimmed estimates in
Table 3 of \cite{GrahamPowell2012}.} Focusing on the estimates without time
effects, we find the FE estimate, 0.6568 (0.0287), is much larger and more
precisely estimated than either the GP or TMG estimates, given by 0.4549
(0.1003) and 0.5623 (0.0425), respectively, with standard errors in
brackets. Judging by the standard errors, it is also noticeable that the TMG
is more precisely estimated than the GP estimate and lies somewhere between
the FE and GP estimates. These estimates are in line with the MC results
reported in the previous section, where we found that in the presence of
correlated heterogeneity, FE estimates are biased with smaller standard
errors (thus leading to incorrect inference), whilst GP and TMG estimators
are correctly centered, with the TMG estimator being more efficient. Similar
results are obtained when we use the extended panel with $T=3$ (2000--2002),
presented in Section \ref{appt3} of the online supplement.

\begin{table}[h]
\caption{Alternative estimates of the average effect of household
expenditures on calorie demand in Nicaragua over the period 2001--2002 ($T=2$%
) }
\label{tab:RPS_cal_T2}\vspace{-6mm}
\par
\begin{center}
\scalebox{0.85}{
\begin{tabular}{lccccccc}
\hline\hline & \multicolumn{3}{c}{Without time effects} &  & \multicolumn{3}{c}{With time effects} \\ \cline{2-4} \cline{6-8}
 & (1) & (2) & (3)  &  & (4) & (5) & (6)  \\ 
 & FE & GP  & TMG &  & TWFE & GP-TE  & TMG-TE \\ \hline
$\hat{\beta}_{0}$ & 0.6568 & 0.4549  & 0.5623 &  & 0.6554 & 0.4629 & 0.5612 \\
 & (0.0287) & (0.1003)  & (0.0425) &  & (0.0284) & (0.1025)  & (0.0424) \\
$\hat{\phi}_{2002}$ & ...  & ...  & ...  &  & 0.0172 & -0.0181  & 0.0178 \\
 & ...  & ...  & ...   &  & (0.0063) & (0.0296)  & (0.0064) \\
$\hat{\pi}$ $(\times 100)$ & ...  & 3.8 & 27.1 &  & ... & 3.8  & 27.1\\
\hline\hline
\end{tabular}}
\end{center}
\par
\vspace{-2mm} {\footnotesize Notes: (i) The estimates of $\beta _{0}=E(\beta
_{i}) $ and $\phi _{2002}$ in the model (\ref{calm1}) are based on a panel
of $1,358$ households. (ii) FE and GP estimators are given by (\ref{fee})
and (\ref{gpe}), respectively. The TMG estimator is given by (\ref{TMGb}),
and its variance is estimated by (\ref{varC}). (iii) The TWFE estimator is
given by (\ref{TWFEhat}) in the online supplement. The GP-TE estimator of $%
\boldsymbol{\beta }_{0}$ and $\boldsymbol{\phi}_{0}$ is given by equations
(25) and (24) in \cite{GrahamPowell2012}. When $T=k=2$, the TMG-TE estimator
of $\boldsymbol{\beta }_{0}$ and $\boldsymbol{\phi }_{0}$ are given by (\ref%
{TMG-TE1}) and (\ref{phihat_b}), and their variances are estimated by (\ref%
{varbetaTE}) and (\ref{VarPhicombined}), respectively, in the mathematical
appendix. (iv) The trimming threshold value for TMG and TMG-TE estimators is
given by $a_{n}=\bar{d}_{n} n^{-\alpha}$, where $\bar{d}_{n} =\frac{1}{n}
\sum_{i}^{n} d_{i}$, $d_{i} =\func{det}(\boldsymbol{X}_{i}^{\prime} 
\boldsymbol{M}_{T} \boldsymbol{X}_{i}) $, and $\boldsymbol{X}_{i}=(%
\boldsymbol{x}_{i1},\boldsymbol{x}_{i2},...,\boldsymbol{x}_{iT})^{\prime}$
and $\boldsymbol{M}_{T} = \boldsymbol{I}_{T} - \boldsymbol{\tau}_{T}%
\boldsymbol{\tau}_{T}^{\prime}/T$. $\alpha$ is set to $1/3$. $\hat{\pi}$ is
the estimated fraction of individual estimates being trimmed given by (\ref%
{pin}). \textquotedblleft ...\textquotedblright denotes that the estimation
algorithms are not applicable. The numbers in brackets are standard errors. }
\end{table}

\section{Conclusions \label{conclusion}}

This paper studies the estimation of average effects in panel data models
with possibly correlated heterogeneous coefficients, when the number of
cross-sectional units is large, but the number of time periods can be as
small as the number of regression coefficients. We recall that the FE
estimator is inconsistent under correlated heterogeneity, and the MG
estimator could not have second-order (or even first-order) moments when
applied to ultra-short panels. The TMG estimator is therefore proposed to
deal with the fat-tailed distributions of the individual estimates (which
arise if $T$ is very close to $k$) by shrinking (not trimming) individual
estimators that are most likely to fail the second-order moment condition.
The TMG estimator is shown to be consistent and asymptotically normally
distributed, but at a slower rate than $\sqrt{n}$. The paper also proposes
TMG estimators for panels with time effects, distinguishing between cases
where $T=k$ and $T>k$. The TMG estimators play a crucial role in assessing
the robustness of FE and TWFE estimators against slope correlated
heterogeneity. The dispersion slope homogeneity tests by \cite%
{PesaranYamagata2008} require large $T$ and do not differentiate between
uncorrelated heterogeneity and correlated heterogeneity.

We highlight the bias and size distortion properties of the FE and TWFE
estimators under correlated heterogeneity. In contrast, the TMG and TMG-TE
estimators are shown to have desirable finite sample performance under a
number of different MC designs, allowing for Gaussian and non-Gaussian
heteroskedastic error processes, dynamic heterogeneity and interactive
effects in the covariates, different numbers of regressors, and different
choices of the trimming threshold parameter, $\alpha $. In particular, since
the TMG and TMG-TE estimators exploit information on all available
individual estimates, they have the smallest RMSE, and tests based on them
have the correct size and are more powerful than the other trimmed
estimators currently proposed in the literature. The Hausman tests based on
TMG and TMG-TE estimators are also shown to have very good small sample
properties, with their size controlled and their power rising strongly with $%
n$ even when $T=k=2$.

\newpage \medskip {\small \setstretch{1.02} 
\bibliographystyle{chicago}
\bibliography{TMGref}
}

\newpage

\appendix

\begin{center}
{\large Mathematical Appendix}
\end{center}

\setcounter{page}{1}\renewcommand{\thepage}{A\arabic{page}} %
\setcounter{table}{0} \renewcommand{\thetable}{A.\arabic{table}} %
\setcounter{section}{0} \renewcommand{\thesection}{A.\arabic{section}} %
\setcounter{section}{0}\renewcommand{\theHsection}{appendixsection.A.%
\arabic{section}} 
\setcounter{example}{0} \renewcommand{\theexample}{A.\arabic{example}} %
\setcounter{figure}{0} \renewcommand{\thefigure}{A.\arabic{figure}} %
\setcounter{footnote}{0} \renewcommand{\thefootnote}{A\arabic{footnote}} %
\renewcommand{\thetheorem}{A.\arabic{theorem}}\setcounter{theorem}{0} %
\renewcommand{\theproposition}{A.\arabic{proposition}}%
\setcounter{proposition}{0} \renewcommand{\theassumption}{A.%
\arabic{assumption}}\setcounter{assumption}{0} \renewcommand{\thelemma}{A.%
\arabic{lemma}}\setcounter{lemma}{0} \renewcommand{\theremark}{A.%
\arabic{remark}}\setcounter{remark}{0}

\section{Lemmas}

\begin{lemma}
\label{deltaeta}Suppose that Assumptions \ref{rcm}(a), \ref{distributiondi}
and \ref{CRE} hold. Then for each $i$, we have 
\begin{align}
E\left[ d_{i}^{s}\boldsymbol{1}\{d_{i}\leq a_{n}\}\right] &
=O(a_{n}^{s+\alpha_{p}})\text{, for }s=1,2,...,  \label{ds} \\
E(\delta _{i})& =O(a_{n}^{\alpha_{p}})\text{, }E(\delta
_{i}^{2})=O(a_{n}^{\alpha_{p}}),  \label{dbar2} \\
E\left( \delta _{i}\boldsymbol{\eta }_{i}\right) & =O(a_{n}^{\alpha_{p}})%
\text{, }E\left( \delta _{i}^{2}\boldsymbol{\eta }_{i}\right)
=O(a_{n}^{\alpha_{p}}),  \label{detax}
\end{align}
\begin{equation}
n^{-1}\sum_{i=1}^{n}\left\{ E\left[ d_{i}^{2}\boldsymbol{1}\{d_{i}\leq
a_{n}\}\right] \right\} ^{1/2} =O(a_{n}^{1+\alpha_{p}/2}),  \label{a}
\end{equation}
and 
\begin{equation}
n^{-1}\sum_{i=1}^{n}\left\{ E\left[ d_{i}^{-2}\boldsymbol{1}\{d_{i}>a_{n}\}%
\right] \right\} ^{1/2} =O(a_{n}^{-1}).  \label{b}
\end{equation}
\end{lemma}

\begin{proof}
By the mean value theorem (MVT), under Assumption \ref{distributiondi}, 
\begin{equation}
F_{d}(a_{n})=\int_{0}^{a_{n}}f_{d}(u)du=F_{d}(0)+f_{d}(\bar{a}%
_{n})a_{n}=f_{d}(\bar{a}_{n})a_{n}=O(a_{n}^{\alpha _{p}}),  \label{Fd}
\end{equation}%
where $\bar{a}_{n}$ lies on the line segment $(0,a_{n})$. Let $\psi
_{s}(a_{n})=\int_{0}^{a_{n}}u^{s}f_{d}(u)du$, and note that $\psi _{s}(0)=0$%
, and $\psi _{s}^{\prime }(a_{n})=a_{n}^{s}f_{d}(a_{n})$. Then by the MVT, 
\begin{equation*}
\psi _{s}(a_{n})=\psi _{s}(0)+\left[ \bar{a}_{n}^{s}f_{d}(\bar{a}_{n})\right]
a_{n},
\end{equation*}%
where $\bar{a}_{n}$ lies on the line segment $(0,a_{n})$. Also, from Remark %
\ref{Fu} we have $f_{d}(\bar{a}_{n})=O\left( a_{n}^{\alpha _{p}-1}\right) $,
and since $\bar{a}_{n}^{s}\leq a_{n}^{s}$ then $\psi _{s}(a_{n})=\bar{a}%
_{n}^{s}f_{d}(\bar{a}_{n})a_{n}=O\left( a_{n}^{s+\alpha _{p}}\right) $.
Using this result, it now follows that 
\begin{equation}
E\left[ d_{i}^{s}\boldsymbol{1}\{d_{i}\leq a_{n}\}\right] =%
\int_{0}^{a_{n}}u^{s}f_{d}(u)du=\psi _{s}(a_{n})=O(a_{n}^{s+\alpha _{p}})%
\text{, for }s=1,2,....  \label{Eds}
\end{equation}%
Further 
\begin{align}
E(\delta _{i})=& E\left[ \left( \frac{d_{i}-a_{n}}{a_{n}}\right) \boldsymbol{%
1}\{d_{i}\leq a_{n}\}\right] =a_{n}^{-1}E\left[ d_{i}\boldsymbol{1}%
\{d_{i}\leq a_{n}\}\right] -E\left[ \boldsymbol{1}\{d_{i}\leq a_{n}\}\right]
\notag \\
=& a_{n}^{-1}O(a_{n}^{1+\alpha _{p}})-F_{d}(a_{n})=O(a_{n}^{\alpha _{p}}).
\label{d1}
\end{align}%
Similarly 
\begin{align}
E(\delta _{i}^{2})& =E\left[ \left( \frac{d_{i}-a_{n}}{a_{n}}\right) ^{2}%
\boldsymbol{1}\{d_{i}\leq a_{n}\}\right]  \notag \\
& =a_{n}^{-2}E\left[ d_{i}^{2}\boldsymbol{1}\{d_{i}\leq a_{n}\}\right] +E%
\left[ \boldsymbol{1}\{d_{i}\leq a_{n}\}\right] -2a_{n}^{-1}E\left[ d_{i}%
\boldsymbol{1}\{d_{i}\leq a_{n}\}\right]  \notag \\
& =a_{n}^{-2}O(a_{n}^{2+\alpha
_{p}})+F_{d}(a_{n})-2a_{n}^{-1}O(a_{n}^{1+\alpha _{p}})=O(a_{n}^{\alpha
_{p}}).  \label{d2}
\end{align}%
Consider now the terms involving the products of $\delta _{i}$ and $%
\boldsymbol{\eta }_{i}$, and we have 
\begin{equation}
E\left( \delta _{i}\boldsymbol{\eta }_{i}\right) =\boldsymbol{B}_{i}E\left\{
\left( \frac{d_{i}-a_{n}}{a_{n}}\right) \boldsymbol{1}\{d_{i}\leq a_{n}\}%
\left[ \boldsymbol{g}(d_{i})-E\left[ \boldsymbol{g}(d_{i})\right] \right]
\right\} .  \label{deta1}
\end{equation}%
Since $\boldsymbol{B}_{i}$ is bounded and does not depend on $d_{i}$,
without loss of generality, we set $\boldsymbol{B}_{i}=\boldsymbol{I}%
_{k^{\prime }}$ and consider the $j^{th}$ term of (\ref{deta1}), namely 
\begin{align*}
s_{j}(a_{n})=& E\left\{ \left( \frac{d_{i}-a_{n}}{a_{n}}\right) \boldsymbol{1%
}\{d_{i}\leq a_{n}\}\left[ g_{j}(d_{i})-E\left[ g_{j}(d_{i})\right] \right]
\right\} \\
=& \frac{1}{a_{n}}\int_{0}^{a_{n}}ug_{j}(u)f_{d}(u)du-%
\int_{0}^{a_{n}}g_{j}(u)f_{d}(u)du \\
& -E\left[ g_{j}(d_{i})\right] \left[ \frac{1}{a_{n}}%
\int_{0}^{a_{n}}uf_{d}(u)du\right] +E\left[ g_{j}(d_{i})\right] \left[
\int_{0}^{a_{n}}f_{d}(u)du\right] .
\end{align*}%
By Assumption \ref{CRE}, $E\left[ g_{j}(d_{i})\right] <C$, and using (\ref%
{Fd}) and (\ref{ds}) we have%
\begin{equation*}
\int_{0}^{a_{n}}f_{d}(u)du=F_{d}(a_{n})=O(a_{n}^{\alpha _{p}})\text{, and }%
a_{n}^{-1}\int_{0}^{a_{n}}uf_{d}(u)du=a_{n}^{-1}O(a_{n}^{1+\alpha
_{p}})=O\left( a_{n}^{\alpha _{p}}\right) .
\end{equation*}%
Also by the MVT (and recalling that $F_{d}(a_{n})=f_{d}(\bar{a}%
_{n})a_{n}=O(a_{n}^{\alpha _{p}})$), 
\begin{equation*}
\int_{0}^{a_{n}}g_{j}(u)f_{d}(u)du=g_{j}(\bar{a}_{n})f_{d}(\bar{a}%
_{n})a_{n}=g_{j}(\bar{a}_{n})F_{d}\left( a_{n}\right) =O(a_{n}^{\alpha
_{p}}),
\end{equation*}%
and 
\begin{equation*}
\frac{1}{a_{n}}\int_{0}^{a_{n}}ug_{j}(u)f_{d}(u)du=\frac{1}{a_{n}}\left[ 
\bar{a}_{n}g_{j}(\bar{a}_{n})f_{d}(\bar{a}_{n})a_{n}\right] =\frac{1}{a_{n}}%
\left[ \bar{a}_{n}g_{j}(\bar{a}_{n})F_{d}\left( a_{n}\right) \right]
=O(a_{n}^{\alpha _{p}}).
\end{equation*}%
Hence, $E\left( \delta _{i}\boldsymbol{\eta }_{i}\right) =O(a_{n}^{\alpha
_{p}})$. Similarly, the $j^{th}$ term of $E\left( \delta _{i}^{2}\boldsymbol{%
\eta }_{i}\right) $ (setting $\boldsymbol{B}_{i}=\boldsymbol{I}_{k^{\prime
}} $) is given by 
\begin{eqnarray}
s_{j2}(a_{n}) &=&E\left\{ \left( \frac{d_{i}-a_{n}}{a_{n}}\right) ^{2}%
\boldsymbol{1}\{d_{i}\leq a_{n}\}\left[ g_{j}(d_{i})-E\left[ g_{j}(d_{i})%
\right] \right] \right\}  \notag \\
&=&E\left\{ \left( \frac{d_{i}^{2}}{a_{n}^{2}} + 1-2\frac{d_{i}}{a_{n}}%
\right) \boldsymbol{1}\{d_{i}\leq a_{n}\}\left[ g_{j}(d_{i})-E\left[
g_{j}(d_{i})\right] \right] \right\} .  \label{sj2A}
\end{eqnarray}%
Consider the first term 
\begin{align*}
&E\left\{ \frac{d_{i}^{2}}{a_{n}^{2}}\boldsymbol{1}\{d_{i}\leq a_{n}\}\left[
g_{j}(d_{i})-E\left[ g_{j}(d_{i})\right] \right] \right\} \\
=&\frac{1}{a_{n}^{2}}\int_{0}^{a_{n}}u^{2}g_{j}(u)f_{d}(u)du-\frac{1}{%
a_{n}^{2}}E\left[ g_{j}(d_{i})\right] E\left[ d_{i}^{2}\boldsymbol{1}%
\{d_{i}\leq a_{n}\}\right] ,
\end{align*}%
and again by the MVT, $a_{n}^{-2}%
\int_{0}^{a_{n}}u^{2}g_{j}(u)f_{d}(u)du=O(a_{n}^{\alpha _{p}})$, $E\left[
g_{j}(d_{i})\right] <C$, and using (\ref{Eds}), $E\left[ d_{i}^{2}%
\boldsymbol{1}\{d_{i}\leq a_{n}\}\right] =O(a_{n}^{2+\alpha _{p}})$. Hence,
the first term of (\ref{sj2A}) is $O(a_{n}^{\alpha _{p}})$. For its second
term, we have 
\begin{equation*}
E\left\{ \boldsymbol{1}\{d_{i}\leq a_{n}\}\left[ g_{j}(d_{i})-E\left[
g_{j}(d_{i})\right] \right] \right\} =\int_{0}^{a_{n}}g_{j}(u)f_{d}(u)du-E%
\left[ g_{j}(d_{i})\right] \int_{0}^{a_{n}}f_{d}(u)du=O(a_{n}^{\alpha _{p}}),
\end{equation*}%
and the order of the third term is already established to be $%
O(a_{n}^{\alpha _{p}})$. Hence, it follows that $E\left( \delta _{i}^{2}%
\boldsymbol{\eta }_{i}\right) =O(a_{n}^{\alpha _{p}})$. Finally, result (\ref%
{a}) follows from (\ref{ds}), and (\ref{b}) follows noting that $d_{i}^{-2}%
\boldsymbol{1}\{d_{i}>a_{n}\}\leq a_{n}^{-2}$.
\end{proof}

\begin{lemma}
\label{EVegzi} Suppose Assumptions \ref{errorsx}, \ref{rcm}(a), \ref%
{distributiondi} and \ref{CRE} hold, and let%
\begin{equation}
\boldsymbol{\bar{\varkappa}}_{nT}=n^{-1}\sum_{i=1}^{n}\left( 1+\delta
_{i}\right) \boldsymbol{\xi }_{iT},  \label{egsidel}
\end{equation}%
where $\delta _{i}=\left( \frac{d_{i}-a_{n}}{a_{n}}\right) \boldsymbol{1}%
\{d_{i}\leq a_{n}\}$, $a_{n}=C_{n}n^{-\alpha }$, $C_{n}<C$, $d_{i}=\func{det}%
(\boldsymbol{X}_{i}^{\prime }\boldsymbol{M}_{T}\boldsymbol{X}_{i})$, $%
\boldsymbol{\xi }_{iT}=\boldsymbol{R}_{i}^{\prime }\boldsymbol{u}_{i}$, and $%
\boldsymbol{R}_{i}=\boldsymbol{M}_{T}\boldsymbol{X}_{i}\left( \boldsymbol{X}%
_{i}^{\prime }\boldsymbol{M}_{T}\boldsymbol{X}_{i}\right) ^{-1}$. Then 
\begin{equation}
E\left( \boldsymbol{\bar{\varkappa}}_{nT}\right) =\boldsymbol{0}\text{, }
\label{Eegzi}
\end{equation}%
\begin{equation}
Var\left( \boldsymbol{\bar{\varkappa}}_{nT}\right) =\frac{1}{n^{2}}%
\sum_{i=1}^{n}E\left[ \boldsymbol{1}\{d_{i}>a_{n}\}\boldsymbol{R}%
_{i}^{\prime }\boldsymbol{H}_{i}\boldsymbol{R}_{i}\right] +\frac{1}{n^{2}}%
\sum_{i=1}^{n}a_{n}^{-2}E\left[ d_{i}^{2}\boldsymbol{1}\{d_{i}\leq a_{n}\}%
\boldsymbol{R}_{i}^{\prime }\boldsymbol{H}_{i}\boldsymbol{R}_{i}\right] ,
\label{Vegzi}
\end{equation}%
\begin{equation}
Var\left( \boldsymbol{\bar{\varkappa}}_{nT}\right)
=O(n^{-1}a_{n}^{-1})+O(n^{-1}a_{n}^{-1+\alpha _{p}/2}),  \label{VegziOrder}
\end{equation}%
\begin{equation}
\sup_{i}E\left( \left\Vert \boldsymbol{\varkappa }_{iT,n}\right\Vert
^{2+\varepsilon }\right) <C\text{, for some }\varepsilon >0,  \label{BMegzi}
\end{equation}%
and%
\begin{equation}
E\left[ n^{-1}\sum_{i=1}^{n}a_{n}^{-1}d_{i}^{2}\boldsymbol{1}\{d_{i}\leq
a_{n}\}\boldsymbol{R}_{i}^{\prime }\boldsymbol{H}_{i}(\boldsymbol{X}_{i})%
\boldsymbol{R}_{i}\right] =O\left( a_{n}^{\alpha _{p}/2}\right) .
\label{VegziA}
\end{equation}
\end{lemma}

\begin{proof}
Under Assumption \ref{errorsx} and conditional on $\boldsymbol{X}$ (and
hence on $d_{i}$), $E\left( \boldsymbol{\bar{\varkappa}}_{nT}\left\vert 
\boldsymbol{X}\right. \right) =n^{-1}\sum_{i=1}^{n}\left( 1+\delta
_{i}\right) \boldsymbol{R}_{i}^{\prime }E\left( \boldsymbol{u}_{i}\left\vert 
\boldsymbol{X}_{i}\right. \right) =\boldsymbol{0}$, and result (\ref{Eegzi})
follows. Also since conditional on $\boldsymbol{X}$, $\left\{ \left(
1+\delta _{i}\right) \boldsymbol{\xi }_{iT}\text{, for }i=1,2,...,n\right\} $
are distributed independently, we have 
\begin{equation*}
Var\left( \boldsymbol{\bar{\varkappa}}_{nT}\left\vert \boldsymbol{X}%
_{i}\right. \right) =n^{-2}\sum_{i=1}^{n}\left( 1+\delta _{i}\right)
^{2}E\left( \boldsymbol{\xi }_{iT}\boldsymbol{\xi }_{iT}^{\prime }\left\vert 
\boldsymbol{X}_{i}\right. \right) =n^{-2}\sum_{i=1}^{n}\left( 1+\delta
_{i}\right) ^{2}\boldsymbol{R}_{i}^{\prime }E\left( \boldsymbol{u}_{i}%
\boldsymbol{u}_{i}^{\prime }\left\vert \boldsymbol{X}_{i}\right. \right) 
\boldsymbol{R}_{i}
\end{equation*}%
and%
\begin{equation}
Var\left( \boldsymbol{\bar{\varkappa}}_{nT}\left\vert \boldsymbol{X}%
_{i}\right. \right) =n^{-2}\sum_{i=1}^{n}\left( 1+\delta _{i}\right) ^{2}%
\boldsymbol{R}_{i}^{\prime }\boldsymbol{H}_{i}\boldsymbol{R}_{i},
\label{Cvar}
\end{equation}%
where $\boldsymbol{H}_{i}=E\left( \boldsymbol{u}_{i}\boldsymbol{u}%
_{i}^{\prime }\left\vert \boldsymbol{X}_{i}\right. \right) $.\footnote{%
We have suppressed the dependence of $\boldsymbol{H}_{i}$ on $\boldsymbol{X}%
_{i}$ to simplify the exposition.} Also, $Var\left( \boldsymbol{\bar{%
\varkappa}}_{nT}\right) =E\left[ Var\left( \boldsymbol{\bar{\varkappa}}%
_{nT}\left\vert \boldsymbol{X}_{i}\right. \right) \right] +Var\left[ E\left( 
\boldsymbol{\bar{\varkappa}}_{nT}\left\vert \boldsymbol{X}\right. \right) %
\right] =E\left[ Var\left( \boldsymbol{\bar{\varkappa}}_{nT}\left\vert 
\boldsymbol{X}_{i}\right. \right) \right] $. Result (\ref{Vegzi}) now
follows if we use 
\begin{equation}
\left( 1+\delta _{i}\right) ^{2}=\boldsymbol{1}\{d_{i}>a_{n}%
\}+a_{n}^{-2}d_{i}^{2}\boldsymbol{1}\{d_{i}\leq a_{n}\},  \label{1d2}
\end{equation}%
in (\ref{Cvar}) and taking unconditional expectations. To establish (\ref%
{VegziOrder}), we first note that under Assumption \ref{errorsx}, $\func{sup}%
_{i}\lambda _{max}\left( \boldsymbol{H}_{i}\right) <C$, and 
\begin{equation}
\left\Vert \boldsymbol{R}_{i}^{\prime }\boldsymbol{H}_{i}\boldsymbol{R}%
_{i}\right\Vert \leq \lambda _{max}\left( \boldsymbol{H}_{i}\right)
\left\Vert \boldsymbol{R}_{i}^{\prime }\boldsymbol{R}_{i}\right\Vert
=\lambda _{max}\left( \boldsymbol{H}_{i}\right) \left\Vert (\boldsymbol{X}%
_{i}^{\prime }\boldsymbol{M}_{T}\boldsymbol{X}_{i})^{-1}\right\Vert
<Cd_{i}^{-1}\left\Vert \func{adj}(\boldsymbol{X}_{i}^{\prime }\boldsymbol{M}%
_{T}\boldsymbol{X}_{i})\right\Vert .  \label{NormRHR}
\end{equation}%
Then using (\ref{Vegzi}) we have%
\begin{align*}
\left\Vert Var\left( \boldsymbol{\bar{\varkappa}}_{nT}\right) \right\Vert
\leq & Cn^{-2}\sum_{i=1}^{n}E\left[ \boldsymbol{1}\{d_{i}>a_{n}\}d_{i}^{-1}%
\left\Vert \func{adj}(\boldsymbol{X}_{i}^{\prime }\boldsymbol{M}_{T}%
\boldsymbol{X}_{i})\right\Vert \right] \\
& +Cn^{-2}E\left[ \sum_{i=1}^{n}a_{n}^{-2}d_{i}\boldsymbol{1}\{d_{i}\leq
a_{n}\}\left\Vert \func{adj}(\boldsymbol{X}_{i}^{\prime }\boldsymbol{M}_{T}%
\boldsymbol{X}_{i})\right\Vert \right] .
\end{align*}%
By the Cauchy-Schwarz inequality,%
\begin{equation*}
E\left[ \boldsymbol{1}\{d_{i}>a_{n}\}d_{i}^{-1}\left\Vert \func{adj}(%
\boldsymbol{X}_{i}^{\prime }\boldsymbol{M}_{T}\boldsymbol{X}_{i})\right\Vert %
\right] \leq \left\{ E\left[ d_{i}^{-2}\boldsymbol{1}\{d_{i}>a_{n}\}\right]
\right\} ^{\frac{1}{2}}\left\{ E\left[ \left\Vert \func{adj}\left( 
\boldsymbol{X}_{i}^{\prime }\boldsymbol{M}_{T}\boldsymbol{X}_{i}\right)
\right\Vert ^{2}\right] \right\} ^{\frac{1}{2}},
\end{equation*}%
and 
\begin{align*}
& E\left[ d_{i}\boldsymbol{1}\{d_{i}\leq a_{n}\}\left\Vert \func{adj}(%
\boldsymbol{X}_{i}^{\prime }\boldsymbol{M}_{T}\boldsymbol{X}_{i})\right\Vert %
\right] \\
\leq & \left\{ E\left[ d_{i}^{2}\boldsymbol{1}\{d_{i}\leq a_{n}\}\right]
\right\} ^{1/2}\left\{ E\left[ \left\Vert \func{adj}\left( \boldsymbol{X}%
_{i}^{\prime }\boldsymbol{M}_{T}\boldsymbol{X}_{i}\right) \right\Vert ^{2}%
\right] \right\} ^{1/2}.
\end{align*}%
Also, by part (d) of Assumption \ref{errorsx}, $\func{sup}_{i}E\left[
\left\Vert \func{adj}\left( \boldsymbol{X}_{i}^{\prime }\boldsymbol{M}_{T}%
\boldsymbol{X}_{i}\right) \right\Vert ^{2}\right] <C$, then%
\begin{equation*}
\left\Vert Var\left( \boldsymbol{\bar{\varkappa}}_{nT}\right) \right\Vert
\leq C\left[ n^{-2}\sum_{i=1}^{n}\left\{ E\left[ d_{i}^{-2}\boldsymbol{1}%
\{d_{i}>a_{n}\}\right] \right\} ^{1/2}+a_{n}^{-2}n^{-2}\sum_{i=1}^{n}\left\{
E\left[ d_{i}^{2}\boldsymbol{1}\{d_{i}\leq a_{n}\}\right] \right\} ^{1/2}%
\right].
\end{equation*}%
Using results (\ref{b}) and (\ref{a}) of Lemma \ref{deltaeta}, we now have%
\begin{equation*}
n^{-2}\sum_{i=1}^{n}\left\{ E\left[ d_{i}^{-2}\boldsymbol{1}\{d_{i}>a_{n}\}%
\right] \right\} ^{1/2}=O(n^{-1}a_{n}^{-1}),
\end{equation*}%
and 
\begin{equation*}
a_{n}^{-2}n^{-2}\sum_{i=1}^{n}\left\{ E\left[ d_{i}^{2}\boldsymbol{1}%
\{d_{i}\leq a_{n}\}\right] \right\} ^{1/2}=a_{n}^{-2}O(n^{-1}a_{n}^{1+\alpha
_{p}/2})=O(n^{-1}a_{n}^{-1+\alpha _{p}/2}),
\end{equation*}%
and result (\ref{VegziOrder}) follows. To establish (\ref{BMegzi}), we first
note that under Assumption \ref{errorsx}, 
\begin{equation*}
E\left( \left\Vert \boldsymbol{\varkappa }_{iT,n}\right\Vert ^{2+\varepsilon
}\left\vert \boldsymbol{X}_{i}\right. \right) \leq C n^{-\alpha
(1+\varepsilon /2)}(1+\delta _{i})^{2+\varepsilon }\lambda
_{max}^{1+\varepsilon /2}\left( \boldsymbol{R}_{i}^{\prime }\boldsymbol{H}%
_{i}\boldsymbol{R}_{i}\right) .
\end{equation*}%
Since $\boldsymbol{H}_{i}$ is positive definite and $\boldsymbol{R}%
_{i}^{\prime }\boldsymbol{R}_{i}=d_{i}^{-1}\func{adj}(\boldsymbol{X}%
_{i}^{\prime }\boldsymbol{M}_{T}\boldsymbol{X}_{i})$, then 
\begin{equation*}
\lambda _{max}^{1+\varepsilon /2}\left( \boldsymbol{R}_{i}^{\prime }%
\boldsymbol{H}_{i}\boldsymbol{R}_{i}\right) \leq \lambda
_{max}^{1+\varepsilon /2}\left( \boldsymbol{H}_{i}\right) \lambda
_{max}^{1+\varepsilon /2}\left( \boldsymbol{R}_{i}^{\prime }\boldsymbol{R}%
_{i}\right) \leq \frac{1}{d_{i}^{1+\varepsilon /2}}\lambda
_{max}^{1+\varepsilon /2}\left( \boldsymbol{H}_{i}\right) \lambda
_{max}^{1+\varepsilon /2}\left[ \func{adj}(\boldsymbol{X}_{i}^{\prime }%
\boldsymbol{M}_{T}\boldsymbol{X}_{i})\right] ,
\end{equation*}%
and 
\begin{equation*}
E\left( \left\Vert \boldsymbol{\varkappa }_{iT,n}\right\Vert ^{2+\varepsilon
}\left\vert \boldsymbol{X}_{i}\right. \right) \leq C n^{-\alpha
(1+\varepsilon /2)}\left[ \frac{(1+\delta _{i})^{2+\varepsilon }}{%
d_{i}^{1+\varepsilon /2}}\right] \lambda _{max}^{1+\varepsilon /2}\left( 
\boldsymbol{H}_{i}\right) \lambda _{max}^{1+\varepsilon /2}\left[ \func{adj}(%
\boldsymbol{X}_{i}^{\prime }\boldsymbol{M}_{T}\boldsymbol{X}_{i})\right] .
\end{equation*}%
Taking unconditional expectations of the above and applying the
Cauchy-Schwarz inequality (recalling that $d_{i}$ and hence $\delta _{i}$
are IID under Assumption \ref{distributiondi}), 
\begin{align}
\sup_{i}E\left\Vert \boldsymbol{\varkappa }_{iT,n}\right\Vert
^{2+\varepsilon } & \leq C n^{-\alpha (1+\frac{\varepsilon}{2}%
)}\sup_{i}E\left\{ \left[ \frac{(1+\delta _{i})^{2+\varepsilon }}{%
d_{i}^{1+\varepsilon /2}}\right] \lambda _{max}^{1+\varepsilon /2}\left( 
\boldsymbol{H}_{i}\right) \lambda _{max}^{1+\varepsilon /2}\left[ \func{adj}(%
\boldsymbol{X}_{i}^{\prime }\boldsymbol{M}_{T}\boldsymbol{X}_{i})\right]
\right\}  \notag \\
& \leq C E\left\{ n^{-\alpha (2+\varepsilon )}\left[ \frac{(1+\delta
_{i})^{4+2\varepsilon }}{d_{i}^{2+\varepsilon }}\right] \right\} ^{\frac{1}{2%
}}\sup_{i}E\left\{ \lambda _{max}^{2+\varepsilon }\left( \boldsymbol{H}%
_{i}\right) \lambda _{max}^{2+\varepsilon }\left[ \func{adj}(\boldsymbol{X}%
_{i}^{\prime }\boldsymbol{M}_{T}\boldsymbol{X}_{i})\right] \right\} ^{\frac{1%
}{2}}.  \label{Echi2}
\end{align}%
Applying the Cauchy-Schwarz inequality and by Assumption \ref{errorsx}, 
\begin{align*}
& \sup_{i}E\left\{ \lambda _{max}^{2+\varepsilon }\left( \boldsymbol{H}%
_{i}\right) \lambda _{max}^{2+\varepsilon }\left[ \func{adj}(\boldsymbol{X}%
_{i}^{\prime }\boldsymbol{M}_{T}\boldsymbol{X}_{i})\right] \right\} \\
\leq & \sup_{i}E\left[ \lambda _{max}^{4+2\varepsilon }\left( \boldsymbol{H}%
_{i}\right) \right] ^{1/2}\sup_{i}E\left\{ \lambda _{max}^{4+2\varepsilon }%
\left[ \func{adj}(\boldsymbol{X}_{i}^{\prime }\boldsymbol{M}_{T}\boldsymbol{X%
}_{i})\right] \right\} ^{1/2}<C.
\end{align*}%
Using these results in (\ref{Echi2}), 
\begin{equation}
\sup_{i}E\left\Vert \boldsymbol{\varkappa }_{iT,n}\right\Vert
^{2+\varepsilon }\leq C\left\{ E\left[ n^{-\alpha (2+\varepsilon )}\left(
1+\delta _{i}\right) ^{4+2\varepsilon }d_{i}^{-(2+\varepsilon )}\right]
\right\} ^{1/2}.  \label{Echi3}
\end{equation}%
Using (\ref{1d2}), it is now easily seen that%
\begin{equation*}
\left( 1+\delta _{i}\right) ^{4+2\varepsilon }=\boldsymbol{1}%
\{d_{i}>a_{n}\}+a_{n}^{-(4+2\varepsilon )}d_{i}^{4+2\varepsilon }\boldsymbol{%
1}\{d_{i}\leq a_{n}\},
\end{equation*}%
and 
\begin{align*}
& E\left[ n^{-(2+\varepsilon )\alpha }\left( 1+\delta _{i}\right)
^{4+2\varepsilon }d_{i}^{-(2+\varepsilon )}\right] \\
=& n^{-(2+\varepsilon )\alpha }E\left[ d_{i}^{-(2+\varepsilon )}\boldsymbol{1%
}\{d_{i}>a_{n}\}\right] +n^{-(2+\varepsilon )\alpha }a_{n}^{-(4+2\varepsilon
)}E\left[ d_{i}^{2+\varepsilon }\boldsymbol{1}\{d_{i}\leq a_{n}\}\right] .
\end{align*}%
Using (\ref{ds}), $E\left[ d_{i}^{2+\varepsilon }\boldsymbol{1}\{d_{i}\leq
a_{n}\}\right] =O(a_{n}^{2+\varepsilon +\alpha _{p}})$, and recalling that $%
a_{n}=C_{n}n^{-\alpha }$, we have 
\begin{equation*}
n^{-(2+\varepsilon )\alpha }a_{n}^{-(4+2\varepsilon )}E\left[
d_{i}^{2+\varepsilon }\boldsymbol{1}\{d_{i}\leq a_{n}\}\right] =O\left(
n^{-(2+\varepsilon )\alpha +(4+2\varepsilon )\alpha -(2+\varepsilon )\alpha
-\alpha \alpha _{p}}\right) =O\left( n^{-\alpha \alpha _{p}}\right) .
\end{equation*}%
Also, since $d_{i}^{-(2+\varepsilon )}\boldsymbol{1}\{d_{i}>a_{n}\}\leq
a_{n}^{-(2+\varepsilon )}\boldsymbol{1}\{d_{i}>a_{n}\}$, then 
\begin{align*}
n^{-(2+\varepsilon )\alpha }E\left[ d_{i}^{-(2+\varepsilon )}\boldsymbol{1}%
\{d_{i}>a_{n}\}\right] & \leq n^{-(2+\varepsilon )\alpha
}a_{n}^{-(2+\varepsilon )}E\left[ \boldsymbol{1}\{d_{i}>a_{n}\}\right]
=C_{n}^{-(2+\varepsilon )}\Pr (d_{i}>a_{n}) \\
& =C_{n}^{-(2+\varepsilon )}\left[ 1-F_{d}(a_{n})\right] .
\end{align*}%
Overall, 
\begin{equation*}
E\left[ n^{-(2+\varepsilon )\alpha }\left( 1+\delta _{i}\right)
^{4+2\varepsilon }d_{i}^{-(2+\varepsilon )}\right] =O\left( n^{-\alpha
\alpha _{p}}\right) +C_{n}^{-(2+\varepsilon )}\left[ 1-F_{d}(a_{n})\right] ,
\end{equation*}%
which is bounded since $\alpha \alpha _{p}>0$, and $C_{n}$ is bounded. Using
this result in (\ref{Echi3}), it follows that $\sup_{i}E\left\Vert \varkappa
_{iT,n}\right\Vert ^{2+\varepsilon }$ is bounded as required. To establish (%
\ref{VegziA}), using (\ref{NormRHR}) we have 
\begin{equation}
\left\Vert n^{-1}\sum_{i=1}^{n}a_{n}^{-1}d_{i}^{2}\boldsymbol{1}\{d_{i}\leq
a_{n}\}\boldsymbol{R}_{i}^{\prime }\boldsymbol{H}_{i}\boldsymbol{R}%
_{i}\right\Vert \leq Cn^{-1}\sum_{i=1}^{n}a_{n}^{-1}d_{i}\boldsymbol{1}%
\{d_{i}\leq a_{n}\}\left\Vert \func{adj}(\boldsymbol{X}_{i}^{\prime }%
\boldsymbol{M}_{T}\boldsymbol{X}_{i})\right\Vert ,  \label{Vi2}
\end{equation}%
and by the Cauchy-Schwarz inequality,%
\begin{equation*}
E\left[ d_{i}\boldsymbol{1}\{d_{i}\leq a_{n}\}\left\Vert \func{adj}(%
\boldsymbol{X}_{i}^{\prime }\boldsymbol{M}_{T}\boldsymbol{X}_{i})\right\Vert %
\right] \leq \left\{ E\left[ d_{i}^{2}\boldsymbol{1}\{d_{i}\leq a_{n}\}%
\right] \right\} ^{1/2}\left[ E\left\Vert \func{adj}\left( \boldsymbol{X}%
_{i}^{\prime }\boldsymbol{M}_{T}\boldsymbol{X}_{i}\right) \right\Vert ^{2}%
\right] ^{1/2}.
\end{equation*}%
Hence 
\begin{align*}
& E\left\Vert n^{-1}\sum_{i=1}^{n}a_{n}^{-1}d_{i}^{2}\boldsymbol{1}%
\{d_{i}\leq a_{n}\}\boldsymbol{R}_{i}^{\prime }\boldsymbol{H}_{i}\boldsymbol{%
R}_{i}\right\Vert \\
\leq & Cn^{-1}\sum_{i=1}^{n}a_{n}^{-1}\left\{ E\left[ d_{i}^{2}\boldsymbol{1}%
\{d_{i}\leq a_{n}\}\right] \right\} ^{1/2}\left[ E\left\Vert \func{adj}%
\left( \boldsymbol{X}_{i}^{\prime }\boldsymbol{M}_{T}\boldsymbol{X}%
_{i}\right) \right\Vert ^{2}\right] ^{1/2}.
\end{align*}%
Also, by part (d) of Assumption \ref{errorsx}, $\func{sup}_{i}E\left\Vert 
\func{adj}\left( \boldsymbol{X}_{i}^{\prime }\boldsymbol{M}_{T}\boldsymbol{X}%
_{i}\right) \right\Vert ^{2}<C$, and we have 
\begin{equation*}
E\left\Vert n^{-1}\sum_{i=1}^{n}a_{n}^{-1}d_{i}^{2}\boldsymbol{1}\{d_{i}\leq
a_{n}\}\boldsymbol{R}_{i}^{\prime }\boldsymbol{H}_{i}(\boldsymbol{X}_{i})%
\boldsymbol{R}_{i}\right\Vert \leq C\left[ n^{-1}\sum_{i=1}^{n}\left\{ E%
\left[ a_{n}^{-2}d_{i}^{2}\boldsymbol{1}\{d_{i}\leq a_{n}\}\right] \right\}
^{1/2}\right] .
\end{equation*}%
Result (\ref{VegziA}) follows, noting that by (\ref{ds}), $E\left[
a_{n}^{-2}d_{i}^{2}\boldsymbol{1}\{d_{i}\leq a_{n}\}\right]
=a_{n}^{-2}O(a_{n}^{2+\alpha _{p}})=O\left( a_{n}^{\alpha _{p}}\right) $.
\end{proof}

\begin{lemma}
\label{asysit} Let 
\begin{align*}
v_{it}-\bar{v}_{i\circ } &=u_{it}-\bar{u}_{i\circ }+\left( \boldsymbol{x}%
_{it}-\boldsymbol{\bar{x}}_{i\circ }\right) ^{\prime }\boldsymbol{\eta }_{i
},\text{ for }i=1,2,...,n;t=1,2,...,T, \\
\text{and }\bar{v}_{\circ t}-\bar{v}_{\circ \circ }
&=n^{-1}\sum_{i=1}^{n}\left( v_{it}-\bar{v}_{i\circ }\right) =\left( \bar{u}%
_{\circ t}-\bar{u}_{\circ \circ }\right) +n^{-1}\sum_{i=1}^{n}\left( 
\boldsymbol{x}_{it}-\boldsymbol{\bar{x}}_{i\circ }\right) ^{\prime }%
\boldsymbol{\eta }_{i }, \text{ for }t=1,2,...,T,
\end{align*}%
where $\boldsymbol{\eta }_{i }=\boldsymbol{\beta }_{i}-\boldsymbol{\beta }%
_{0}$, and suppose that Assumptions \ref{errorsx}, \ref{CRE} and \ref%
{stationaryCRE} hold. Then%
\begin{equation}
E\left( v_{it}-\bar{v}_{i\circ }\right) =0\text{, for }%
i=1,2,...,n;t=1,2,...,T,  \label{Lem3-1}
\end{equation}%
\begin{equation}
\bar{v}_{\circ t}-\bar{v}_{\circ \circ }=O_{p}(n^{-1/2}),\text{ for }%
t=1,2,...,T,  \label{Lem3-2}
\end{equation}%
and (noting that $T$ is fixed as $n\rightarrow \infty $)%
\begin{equation}
\sqrt{n}\left( \boldsymbol{\bar{v}}_{T}-\bar{v}_{\circ \circ }\boldsymbol{%
\tau }_{T}\right) \rightarrow _{d}N(\boldsymbol{0},\boldsymbol{\Omega }_{\nu
}),  \label{Lem3-3}
\end{equation}%
where $\boldsymbol{\bar{v}}_{T}=\left( \bar{v}_{\circ 1},\bar{v}_{\circ
2},...,\bar{v}_{\circ T}\right) ^{\prime }=n^{-1}\sum_{i=1}^{n}\boldsymbol{%
\nu }_{i\circ }$, $\boldsymbol{\nu }_{i\circ }=(\nu _{i1},\nu _{i2},...,\nu
_{iT})^{\prime }$,%
\begin{equation}
\boldsymbol{\Omega }_{\nu }=\boldsymbol{M}_{T}\left[ \lim_{n\rightarrow
\infty }n^{-1}\sum_{i=1}^{n}E\left( \boldsymbol{\nu }_{i\circ }\boldsymbol{%
\nu }_{i\circ }^{\prime }\right) \right] \boldsymbol{M}_{T},  \label{Omeganu}
\end{equation}%
and $\boldsymbol{M}_{T}=\boldsymbol{I}_{T}-T^{-1}\boldsymbol{\tau }_{T}%
\boldsymbol{\tau }_{T}^{\prime }$.
\end{lemma}

\begin{proof}
Under Assumptions \ref{errorsx} and \ref{stationaryCRE}, $E(u_{it})=0$ and $%
E\left( \boldsymbol{x}_{it}^{\prime }\boldsymbol{\eta }_{i }\right) =E\left( 
\boldsymbol{x}_{is}^{\prime }\boldsymbol{\eta }_{i }\right) $ for all $t$
and $s$. Hence, 
\begin{equation*}
E\left( u_{it}-u_{i\circ }\right) =0 \text{, and } E\left[ \left( 
\boldsymbol{x}_{it}-\boldsymbol{\bar{x}}_{i\circ }\right) ^{\prime }%
\boldsymbol{\eta }_{i }\right] =E\left( \boldsymbol{x}_{it}^{\prime }%
\boldsymbol{\eta }_{i }\right) -T^{-1} \sum_{t^{\prime }=1}^{T}E\left( 
\boldsymbol{x}_{it^{\prime }}^{\prime }\boldsymbol{\eta }_{i }\right) =0,
\end{equation*}%
then result (\ref{Lem3-1}) follows. Result (\ref{Lem3-2}) also follows
noting that under Assumptions \ref{errorsx} and \ref{CRE}, $\left\{ v_{it}-%
\bar{v}_{i\circ }\text{, for }i=1,2,...,n\right\} $ are cross-sectionally
independent with mean zero and finite variances. To establish (\ref{Lem3-3}%
), we first note that $\bar{v}_{\circ \circ }=T^{-1}\left( \boldsymbol{\tau }%
_{T}^{\prime }\boldsymbol{\bar{v}}_{T}\right) $, and hence 
\begin{equation*}
\sqrt{n}\left( \boldsymbol{\bar{v}}_{T}-\bar{v}_{\circ \circ }\boldsymbol{%
\tau }_{T}\right) =\boldsymbol{M}_{T }\sqrt{n}\boldsymbol{\bar{v}}%
_{T}=n^{-1/2}\sum_{i=1}^{n}\boldsymbol{M}_{T}\boldsymbol{\nu }_{i\circ },
\end{equation*}%
where $\boldsymbol{M}_{T }\boldsymbol{\nu }_{i\circ }$ is a $T\times 1$
vector ($T$ is fixed) with zero means and finite variances, and by
Assumption \ref{CRE} is cross-sectionally independent. Therefore, the result
(\ref{Lem3-3}) follows by standard central limit theorems for independent
but not identically distributed random variables.
\end{proof}

\section{Proof Propositions and Theorems}

\subsection{Proof Proposition \protect\ref{prop_mexist}}

\label{pf_prop_mexist}

\begin{proof}
Using (\ref{mge2}) and (\ref{Truebeta}) we note that 
\begin{equation}
\boldsymbol{\hat{\beta}}_{MG}-\boldsymbol{\beta }_{0}=\boldsymbol{\bar{\xi}}%
_{nT}+O_{p}\left( n^{-1/2}\right),  \label{MGgap}
\end{equation}%
where $\boldsymbol{\bar{\xi}}_{nT} = n^{-1} \sum_{i=1}^{n} \boldsymbol{\xi }%
_{iT}$ and $\boldsymbol{\xi }_{iT} = (\boldsymbol{X}_i^{\prime }\boldsymbol{M%
}_T \boldsymbol{X}_i)^{-1} \boldsymbol{X}_i^{\prime }\boldsymbol{M}_T 
\boldsymbol{u}_i$. By the Markov inequality, for any fixed $\epsilon >0$,%
\begin{equation}
Pr\left( \left\Vert \boldsymbol{\bar{\xi}}_{nT}\right\Vert \geq \epsilon
\right) \leq \frac{E\left\Vert \boldsymbol{\bar{\xi}}_{nT}\right\Vert ^{2}}{%
\epsilon ^{2}}.  \label{MI}
\end{equation}%
Further, 
\begin{equation*}
\left\Vert \boldsymbol{\bar{\xi}}_{nT}\right\Vert ^{2}=n^{-2}\left\Vert
\sum_{i=1}^{n}\boldsymbol{\xi }_{iT}\right\Vert ^{2}=n^{-2}\left(
\sum_{i=1}^{n}\boldsymbol{\xi }_{iT}\right) ^{\prime }\left( \sum_{i=1}^{n}%
\boldsymbol{\xi }_{iT}\right) =n^{-2}\sum_{i=1}^{n}\sum_{j=1}^{n}\boldsymbol{%
\xi }_{iT}^{\prime }\boldsymbol{\xi }_{jT}.
\end{equation*}%
Hence, $E\left\Vert \boldsymbol{\bar{\xi}}_{nT}\right\Vert
^{2}=n^{-2}\sum_{i=1}^{n}\sum_{j=1}^{n}E\left( \boldsymbol{\xi }%
_{iT}^{\prime }\boldsymbol{\xi }_{jT}\right) $. Since under Assumption \ref%
{errorsx}(a), $u_{it}$'s are cross-sectionally independent, then 
\begin{equation}
E\left\Vert \boldsymbol{\bar{\xi}}_{nT}\right\Vert
^{2}=n^{-2}\sum_{i=1}^{n}E\left( \boldsymbol{\xi }_{iT}^{\prime }\boldsymbol{%
\xi }_{iT}\right) .  \label{egzibar2}
\end{equation}%
Using $\boldsymbol{\xi }_{iT}=\left( \boldsymbol{X}_{i}^{\prime }\boldsymbol{%
M}_{T}\boldsymbol{X}_{i}\right) ^{-1}\boldsymbol{X}_{i}^{\prime }\boldsymbol{%
M}_{T}\boldsymbol{u}_{i} =\boldsymbol{R}_{i}^{\prime }\boldsymbol{u}_{i}$,
where $\boldsymbol{R}_{i} = \boldsymbol{M}_{T}\boldsymbol{X}_{i}\left( 
\boldsymbol{X}_{i}^{\prime }\boldsymbol{M}_{T}\boldsymbol{X}_{i}\right)
^{-1} $, we have 
\begin{equation*}
E\left( \boldsymbol{\xi }_{iT}^{\prime }\boldsymbol{\xi }_{iT}|\boldsymbol{X}%
_{i}\right) =E\left( \boldsymbol{u}_{i}^{\prime }\boldsymbol{R}_{i}%
\boldsymbol{R}_{i}^{\prime }\boldsymbol{u}_{i}|\boldsymbol{X}_{i}\right) =E%
\left[ Tr\left( \boldsymbol{R}_{i}^{\prime }\boldsymbol{u}_{i}\boldsymbol{u}%
_{i}^{\prime }\boldsymbol{R}_{i}\right) |\boldsymbol{X}_{i}\right] =Tr\left( 
\boldsymbol{R}_{i}^{\prime }\boldsymbol{H}_{i}(\boldsymbol{X}_{i})%
\boldsymbol{R}_{i}\right) ,
\end{equation*}%
where by Assumption \ref{errorsx}(c), $\boldsymbol{H}_{i}(\boldsymbol{X}%
_{i})=E(\boldsymbol{u}_{i}\boldsymbol{u}_{i}^{\prime }|\boldsymbol{X}_{i})$,
and 
\begin{align*}
Tr\left( \boldsymbol{R}_{i}^{\prime }\boldsymbol{H}_{i}(\boldsymbol{X}_{i})%
\boldsymbol{R}_{i}\right) & \leq T\lambda _{max}\left[ \boldsymbol{H}_{i}(%
\boldsymbol{X}_{i})\right] Tr\left( \boldsymbol{R}_{i}^{\prime }\boldsymbol{R%
}_{i}\right) =T\lambda _{max}\left[ \boldsymbol{H}_{i}(\boldsymbol{X}_{i})%
\right] Tr\left[ \left( \boldsymbol{X}_{i}^{\prime }\boldsymbol{M}_{T}%
\boldsymbol{X}_{i}\right) ^{-1}\right] \\
& \leq T\lambda _{max}\left[ \boldsymbol{H}_{i}(\boldsymbol{X}_{i})\right]
\left\{ k^{\prime }\lambda _{max}\left[ \left( \boldsymbol{X}_{i}^{\prime }%
\boldsymbol{M}_{T}\boldsymbol{X}_{i}\right) ^{-1}\right] \right\} .
\end{align*}%
Since $T$ and $k^{\prime }$ are finite, and under Assumption \ref{errorsx}%
(c), $\func{sup}_{i}\lambda _{max} \boldsymbol{H}_{i}(\boldsymbol{X}_{i})<C$%
, it follows that $Tr\left( \boldsymbol{R}_{i}^{\prime }\boldsymbol{H}_{i}(%
\boldsymbol{X}_{i})\boldsymbol{R}_{i}\right) \leq C\lambda _{max}[\left( 
\boldsymbol{X}_{i}^{\prime }\boldsymbol{M}_{T}\boldsymbol{X}_{i}\right)
^{-1} ]$. Then given (\ref{egzibar2}), we have 
\begin{align}
E\left\Vert \boldsymbol{\bar{\xi}}_{nT}\right\Vert ^{2}
&=n^{-2}\sum_{i=1}^{n}E\left[ Tr\left( \boldsymbol{R}_{i}^{\prime }%
\boldsymbol{H}_{i}(\boldsymbol{X}_{i})\boldsymbol{R}_{i}\right) \right] \leq
Cn^{-2}\sum_{i=1}^{n} E\{ \lambda _{max}[( \boldsymbol{X}_{i}^{\prime }%
\boldsymbol{M}_{T}\boldsymbol{X}_{i}) ^{-1}] \}  \notag \\
&\leq Cn^{-1}\func{sup}_{i}E\{ \lambda _{max}[ ( \boldsymbol{X}_{i}^{\prime }%
\boldsymbol{M}_{T}\boldsymbol{X}_{i}) ^{-1}] \} .  \label{cb1}
\end{align}%
Also $(\boldsymbol{X}_{i}^{\prime }\boldsymbol{M}_{T}\boldsymbol{X}%
_{i})^{-1}=d_{i}^{-1}\func{adj}(\boldsymbol{X}_{i}^{\prime }\boldsymbol{M}%
_{T}\boldsymbol{X}_{i})$, where $d_{i}=\func{det}(\boldsymbol{X}_{i}^{\prime
}\boldsymbol{M}_{T}\boldsymbol{X}_{i})$ and $\func{adj}(\boldsymbol{X}%
_{i}^{\prime }\boldsymbol{M}_{T}\boldsymbol{X}_{i})$ is the adjugate of $%
\boldsymbol{X}_{i}^{\prime }\boldsymbol{M}_{T}\boldsymbol{X}_{i}$. Since $%
\boldsymbol{X}_{i}^{\prime }\boldsymbol{M}_{T}\boldsymbol{X}_{i}$ is a
symmetric matrix, $\lambda _{\max }\left[ (\boldsymbol{X}_{i}^{\prime }%
\boldsymbol{M}_{T}\boldsymbol{X}_{i})^{-1}\right] \leq \left\Vert (%
\boldsymbol{X}_{i}^{\prime }\boldsymbol{M}_{T}\boldsymbol{X}%
_{i})^{-1}\right\Vert _{1}$ , where $\left\Vert \boldsymbol{A}\right\Vert
_{1}$ denotes the column norm of $\boldsymbol{A}$, and $\lambda _{\max }%
\left[ (\boldsymbol{X}_{i}^{\prime }\boldsymbol{M}_{T}\boldsymbol{X}%
_{i})^{-1}\right] \leq d_{i}^{-1}\left\Vert \func{adj}(\boldsymbol{X}%
_{i}^{\prime }\boldsymbol{M}_{T}\boldsymbol{X}_{i})\right\Vert _{1}$.
Further, by the Cauchy-Schwarz inequality, 
\begin{equation*}
E\left\{ \lambda _{\max }\left[ (\boldsymbol{X}_{i}^{\prime }\boldsymbol{M}%
_{T}\boldsymbol{X}_{i})^{-1}\right] \right\} \leq \left[ E\left(
d_{i}^{-2}\right) \right] ^{1/2}\left\{ E\left[ \left\Vert \func{adj}(%
\boldsymbol{X}_{i}^{\prime }\boldsymbol{M}_{T}\boldsymbol{X}_{i})\right\Vert
_{1}^{2}\right] \right\} ^{1/2},
\end{equation*}%
and by conditions (\ref{sufficient}), it follows that $\sup_i E\left\{
\lambda _{\max }\left[ (\boldsymbol{X}_{i}^{\prime }\boldsymbol{M}_{T}%
\boldsymbol{X}_{i})^{-1}\right] \right\} <C$. Using this result in (\ref{cb1}%
) now yields 
\begin{equation}
E\left\Vert \boldsymbol{\bar{\xi}}_{nT}\right\Vert ^{2}=O(n^{-1}).
\label{cb}
\end{equation}%
Hence, from (\ref{MI}), it follows that $\boldsymbol{\bar{\xi}}%
_{nT}=O_{p}(n^{-1/2})$. In conjunction with (\ref{MGgap}), this establishes
that $\boldsymbol{\hat{\beta}}_{MG}-\boldsymbol{\beta }_{0}=O_{p}\left(
n^{-1/2}\right) $, namely $\boldsymbol{\hat{\beta}}_{MG}$ converges in
probability to $\boldsymbol{\beta }_{0}$ at the regular rate of $n^{-1/2}$,
irrespective of whether $\boldsymbol{\beta }_{i}$ are correlated with the
regressors or not, and this result is robust to error serial correlation and
conditional heteroskedasticity.
\end{proof}

\subsection{Proof Proposition \protect\ref{prop_mgvsfe}\label{Proofmgvsfe}}

\begin{proof}
Let $\boldsymbol{V}_{n}^{P}=\boldsymbol{\bar{\Psi}}_{n}\boldsymbol{A}_{n}%
\boldsymbol{\bar{\Psi}}_{n}$, where $\boldsymbol{\bar{\Psi}}_{n}=\frac{1}{n}%
\sum_{i=1}^{n}\boldsymbol{\Psi }_{i}$, and $\boldsymbol{\Psi }_{i}=%
\boldsymbol{X}_{i}^{\prime }\boldsymbol{M}_{T}\boldsymbol{X}_{i}$. Then 
\begin{equation*}
\func{plim}_{n\rightarrow \infty }\boldsymbol{V}_{n}^{P}=\func{plim}%
_{n\rightarrow \infty }\boldsymbol{\bar{\Psi}}_{n}\boldsymbol{A}_{n}%
\boldsymbol{\bar{\Psi}}_{n}=\func{plim}_{n\rightarrow \infty }\boldsymbol{%
\bar{\Psi}}_{n}\boldsymbol{\bar{\Omega}}_{n,\beta }\boldsymbol{\bar{\Psi}}%
_{n}-\func{plim}_{n\rightarrow \infty }\left( n^{-1}\sum_{i=1}^{n}%
\boldsymbol{\Psi }_{i}\boldsymbol{\Omega }_{i\beta }\boldsymbol{\Psi }%
_{i}\right) .
\end{equation*}%
For the first term, by Assumption \ref{rcm}, $\func{plim}_{n\rightarrow
\infty }\boldsymbol{\bar{\Omega}}_{n,\beta }=\lim_{n\rightarrow \infty }%
\frac{1}{n}\sum_{i=1}^{n}E\left( \boldsymbol{\Omega }_{i\beta }\right) =%
\boldsymbol{\Omega }_{\beta }$, and by Assumption \ref{PoolA}, $\func{plim}%
_{n\rightarrow \infty }\boldsymbol{\bar{\Psi}}_{n}=\lim_{n\rightarrow \infty
}\frac{1}{n}\sum_{i=1}^{n}E\left( \boldsymbol{\Psi }_{i}\right) $. Thus, 
\begin{equation*}
\func{plim}_{n\rightarrow \infty }\boldsymbol{\bar{\Psi}}_{n}\boldsymbol{%
\bar{\Omega}}_{n,\beta }\boldsymbol{\bar{\Psi}}_{n}=\left[
\lim_{n\rightarrow \infty }n^{-1}\sum_{i=1}^{n}E\left( \boldsymbol{\Psi }%
_{i}\right) \right] \boldsymbol{\Omega }_{\beta }\left[ \lim_{n\rightarrow
\infty }n^{-1}\sum_{i=1}^{n}E\left( \boldsymbol{\Psi }_{i}\right) \right] .
\end{equation*}%
For the second term, under part (b) of Assumption \ref{rcm}, $\boldsymbol{%
\Omega }_{i\beta }$ is mean independent of $\boldsymbol{\Psi }_{i}$, and
together with Assumption \ref{PoolA}, we have 
\begin{align*}
E\left\Vert n^{-1}\sum_{i=1}^{n}\boldsymbol{\Psi }_{i}\boldsymbol{\Omega }%
_{i\beta }\boldsymbol{\Psi }_{i}\right\Vert & \leq
n^{-1}\sum_{i=1}^{n}E\left( \left\Vert \boldsymbol{\Psi }_{i}\right\Vert
^{2}\left\Vert \boldsymbol{\Omega }_{i\beta }\right\Vert \right)
=n^{-1}\sum_{i=1}^{n}E\left( \left\Vert \boldsymbol{\Psi }_{i}\right\Vert
^{2}\right) E\left\Vert \boldsymbol{\Omega }_{i\beta }\right\Vert \\
& \leq \left( \sup_{i}E\left\Vert \boldsymbol{\Omega }_{i\beta }\right\Vert
\right) \left[ n^{-1}\sum_{i=1}^{n}E\left( \left\Vert \boldsymbol{\Psi }%
_{i}\right\Vert ^{2}\right) \right] <C.
\end{align*}%
Also by Assumption \ref{rcm}, $E\left( \boldsymbol{\Psi }_{i}\boldsymbol{%
\Omega }_{i\beta }\boldsymbol{\Psi }_{i}\right) =E\left[ E\left( \boldsymbol{%
\Psi }_{i}\boldsymbol{\Omega }_{i\beta }\boldsymbol{\Psi }_{i}\left\vert 
\boldsymbol{\Psi }_{i}\right. \right) \right] =E\left( \boldsymbol{\Psi }_{i}%
\boldsymbol{\Omega }_{\beta }\boldsymbol{\Psi }_{i}\right) $, then 
\begin{equation*}
\func{plim}_{n\rightarrow \infty }\left( n^{-1}\sum_{i=1}^{n}\boldsymbol{%
\Psi }_{i}\boldsymbol{\Omega }_{i\beta }\boldsymbol{\Psi }_{i}\right)
=\lim_{n\rightarrow \infty }n^{-1}\sum_{i=1}^{n}E\left( \boldsymbol{\Psi }%
_{i}\boldsymbol{\Omega }_{i\beta }\boldsymbol{\Psi }_{i}\right)
=\lim_{n\rightarrow \infty }n^{-1}\sum_{i=1}^{n}E\left( \boldsymbol{\Psi }%
_{i}\boldsymbol{\Omega }_{\beta }\boldsymbol{\Psi }_{i}\right) .
\end{equation*}%
Hence, 
\begin{equation*}
\func{plim}_{n\rightarrow \infty }\boldsymbol{V}_{n}^{P}=\left[
\lim_{n\rightarrow \infty }\frac{1}{n}\sum_{i=1}^{n}E\left( \boldsymbol{\Psi 
}_{i}\right) \right] \boldsymbol{\Omega }_{\beta }\left[ \lim_{n\rightarrow
\infty }\frac{1}{n}\sum_{i=1}^{n}E\left( \boldsymbol{\Psi }_{i}\right) %
\right] -\lim_{n\rightarrow \infty }\frac{1}{n}\sum_{i=1}^{n}E\left( 
\boldsymbol{\Psi }_{i}\boldsymbol{\Omega }_{\beta }\boldsymbol{\Psi }%
_{i}\right) .
\end{equation*}%
Let $\boldsymbol{P}_{i}=\boldsymbol{\Psi }_{i}\boldsymbol{\Omega }_{\beta
}^{1/2}$ and $\boldsymbol{\bar{P}}_{n}=\frac{1}{n}\sum_{i=1}^{n}\boldsymbol{P%
}_{i}\rightarrow _{p}\boldsymbol{\bar{P}}=\lim_{n\rightarrow \infty }\frac{1%
}{n}\sum_{i=1}^{n}E\left( \boldsymbol{P}_{i}\right) $ as $n\rightarrow
\infty $, then 
\begin{align*}
\func{plim}_{n\rightarrow \infty }\boldsymbol{V}_{n}^{P}=& \left[
\lim_{n\rightarrow \infty }\frac{1}{n}\sum_{i=1}^{n}E\left( \boldsymbol{\Psi 
}_{i}\boldsymbol{\boldsymbol{\Omega }}_{\beta }^{1/2}\right) \right] \left[
\lim_{n\rightarrow \infty }\frac{1}{n}\sum_{i=1}^{n}E\left( \boldsymbol{%
\Omega }_{\beta }^{1/2}\boldsymbol{\Psi }_{i}\right) \right]
-\lim_{n\rightarrow \infty }\frac{1}{n}\sum_{i=1}^{n}E\left( \boldsymbol{%
\Psi }_{i}\boldsymbol{\Omega }_{\beta }\boldsymbol{\Psi }_{i}\right) \\
=& \boldsymbol{\bar{P}\bar{P}}^{\prime }-\lim_{n\rightarrow \infty
}n^{-1}\sum_{i=1}^{n}E\left( \boldsymbol{P}_{i}\boldsymbol{P}_{i}^{\prime
}\right) =-\lim_{n\rightarrow \infty }\frac{1}{n}\sum_{i=1}^{n}E\left[
\left( \boldsymbol{P}_{i}-\boldsymbol{\bar{P}}\right) \left( \boldsymbol{P}%
_{i}-\boldsymbol{\bar{P}}\right) ^{\prime }\right] .
\end{align*}%
Hence, 
\begin{equation*}
\func{plim}_{n\rightarrow \infty }\boldsymbol{A}_{n}=\func{plim}%
_{n\rightarrow \infty }\boldsymbol{\bar{\Psi}}_{n}^{-1}\boldsymbol{V}_{n}^{P}%
\boldsymbol{\bar{\Psi}}_{n}^{-1}=-\boldsymbol{\bar{\Psi}}^{-1}\lim_{n%
\rightarrow \infty }\frac{1}{n}\sum_{i=1}^{n}E\left[ \left( \boldsymbol{P}%
_{i}-\boldsymbol{\bar{P}}\right) \left( \boldsymbol{P}_{i}-\boldsymbol{\bar{P%
}}\right) ^{\prime }\right] \boldsymbol{\bar{\Psi}}^{-1},
\end{equation*}%
Since $\lim_{n\rightarrow \infty }\frac{1}{n}\sum_{i=1}^{n}E\left[ \left( 
\boldsymbol{P}_{i}-\boldsymbol{\bar{P}}\right) \left( \boldsymbol{P}_{i}-%
\boldsymbol{\bar{P}}\right) ^{\prime }\right] \succeq \boldsymbol{0}$ and $%
\boldsymbol{\bar{\Psi}}\succ \boldsymbol{0}$, it then follows that the
probability limit of $\boldsymbol{A}_{n}$ is negative semi-definite. For $%
\boldsymbol{B}_{n}$ given by (\ref{Bn}), we have 
\begin{align*}
\func{plim}_{n\rightarrow \infty }\left( \boldsymbol{\bar{\Psi}}_{n}%
\boldsymbol{B}_{n}\boldsymbol{\bar{\Psi}}_{n}\right) =& \boldsymbol{\bar{\Psi%
}}\left[ \lim_{n\rightarrow \infty }n^{-1}\sum_{i=1}^{n}E\left( \boldsymbol{%
\Psi }_{i}^{-1}\boldsymbol{X}_{i}^{\prime }\boldsymbol{M}_{T}\boldsymbol{H}%
_{i}\boldsymbol{M}_{T}\boldsymbol{X}_{i}\boldsymbol{\Psi }_{i}^{-1}\right) %
\right] \boldsymbol{\bar{\Psi}} \\
& -\left[ \lim_{n\rightarrow \infty }n^{-1}\sum_{i=1}^{n}\boldsymbol{X}%
_{i}^{\prime }\boldsymbol{M}_{T}\boldsymbol{H}_{i}\boldsymbol{M}_{T}%
\boldsymbol{X}_{i}\right] ,
\end{align*}%
and in general it is not possible to sign $\boldsymbol{\bar{\Psi}}_{n}%
\boldsymbol{B}_{n}\boldsymbol{\bar{\Psi}}_{n}$ or its probability limit. The
outcome depends on the heterogeneity of error variances and their
interactions with the heterogeneity of regressors. We have already seen that 
$\boldsymbol{B}_{n}\succeq \boldsymbol{0}$ when $\boldsymbol{H}_{i}=\sigma
^{2}\boldsymbol{I}_{T}$, but this result need not hold in a more general
setting where $\boldsymbol{H}_{i}$ varies across $i$.
\end{proof}

\subsection{Proof Theorem \protect\ref{thm_asytmg}}

\label{Proofthm1}

\begin{proof}
To establish Theorem \ref{thm_asytmg}, we first write (\ref{Asy3}) as 
\begin{equation*}
n^{(1-\alpha )/2}\left( \boldsymbol{\hat{\beta}}_{TMG}-\boldsymbol{\beta }%
_{0}\right) =n^{(1-\alpha )/2}\boldsymbol{b}_{n}+\boldsymbol{z}_{p,n}+%
\boldsymbol{z}_{q,nT}+o_{p}(1),
\end{equation*}%
where $\boldsymbol{b}_{n}$ is given by (\ref{bn}), 
\begin{equation*}
\boldsymbol{z}_{p,n}=n^{-(1+\alpha )/2}\sum_{i=1}^{n}\left[ \boldsymbol{p}%
_{i}-E\left( \boldsymbol{p}_{i}\right) \right] \text{, }\boldsymbol{z}%
_{q,nT}=n^{-(1+\alpha )/2}\sum_{i=1}^{n}\boldsymbol{q}_{iT},
\end{equation*}%
$\boldsymbol{p}_{i}=\frac{1+\delta _{i}}{1+E\left( \bar{\delta}_{n}\right) }%
\boldsymbol{\eta }_{i}$, $\boldsymbol{q}_{iT}=\frac{1+\delta _{i}}{1+E\left( 
\bar{\delta}_{n}\right) }\boldsymbol{R}_{i}^{\prime }\boldsymbol{u}_{i}$, $%
\bar{\delta}_{n}=n^{-1}\sum_{i=1}^{n}\delta _{i}$, $\delta _{i}$ is given by
(\ref{deltai}), and $\boldsymbol{R}_{i}^{\prime }=(\boldsymbol{X}%
_{i}^{\prime }\boldsymbol{M}_{T}\boldsymbol{X}_{i})^{-1}\boldsymbol{X}%
_{i}^{\prime }\boldsymbol{M}_{T}$. Recall that $n^{(1-\alpha )/2}\boldsymbol{%
b}_{n}=O\left( n^{\frac{1-\alpha (1+2\alpha _{p})}{2}}\right) $, which is
negligible since $\alpha >\frac{1}{1+2\alpha _{p}}$. Under Assumptions \ref%
{distributiondi} and \ref{CRE}, $\boldsymbol{p}_{i}$ are cross-sectionally
independent, and we have $Var\left( \boldsymbol{z}_{p,n}\right) =n^{-\alpha }%
\left[ n^{-1}\sum_{i=1}^{n}Var\left( \boldsymbol{p}_{i}\right) \right]
=O(n^{-\alpha })$. Since $E\left( \boldsymbol{z}_{p,n}\right) =\boldsymbol{0}
$, it follows that $\boldsymbol{z}_{p,n}=O_{p}(n^{-\alpha /2})$. Hence, 
\begin{equation*}
n^{(1-\alpha )/2}\left( \boldsymbol{\hat{\beta}}_{TMG}-\boldsymbol{\beta }%
_{0}\right) =\boldsymbol{z}_{q,nT}+O_{p}(n^{-\alpha /2})+o_{p}(1).
\end{equation*}%
Using (\ref{qbarnt}) and recalling that $E\left( \bar{\delta}_{n}\right)
=O\left( a_{n}^{\alpha _{p}}\right) =o(1)$, the first term of the above
equation can be written as 
\begin{equation*}
\boldsymbol{z}_{q,nT}=n^{-(1+\alpha )/2}\sum_{i=1}^{n}\boldsymbol{q}%
_{iT}=n^{-1/2}\sum_{i=1}^{n}\boldsymbol{\varkappa }_{iT,n}+o_{p}(1),
\end{equation*}%
where%
\begin{equation}
\boldsymbol{\varkappa }_{iT,n}=n^{-\alpha /2}\left( 1+\delta _{i}\right) 
\boldsymbol{R}_{i}^{\prime }\boldsymbol{u}_{i}.  \label{Echi1}
\end{equation}%
By Assumptions \ref{errorsx} and \ref{distributiondi}, conditional on $%
\boldsymbol{X}$, $\boldsymbol{\varkappa }_{iT,n}$ are independently
distributed over $i$ with zero means.\footnote{%
Since by Assumption \ref{errorsx}, $E\left( \boldsymbol{u}_{i}\left\vert 
\boldsymbol{X}\right. \right) =\boldsymbol{0}$, it follows that $E\left( 
\boldsymbol{\varkappa }_{iT,n}\right) =E\left[ E\left( \boldsymbol{\varkappa 
}_{iT,n}\left\vert \boldsymbol{X}\right. \right) \right] =n^{-\alpha /2}E%
\left[ \left( 1+\delta _{i}\right) \boldsymbol{R}_{i}^{\prime }E\left( 
\boldsymbol{u}_{i}\left\vert \boldsymbol{X}\right. \right) \right] =%
\boldsymbol{0}$.} Using (\ref{BMegzi}) of Lemma \ref{EVegzi}, $%
\sup_{i}E\left\Vert \varkappa _{iT,n}\right\Vert ^{2+\varepsilon }<C$, and
the Lyapunov condition is met, and by the Lindeberg-Feller Central Limit
Theorem, $n^{-1/2}\sum_{i=1}^{n}\varkappa _{iT,n}$ converges to a normal
distribution with a zero mean and a finite variance. Also, using (\ref{Vegzi}%
) of Lemma \ref{EVegzi}, we note that%
\begin{align*}
Var\left( \boldsymbol{z}_{q,nT}\right) =& n^{-1-\alpha }\sum_{i=1}^{n}E\left[
\boldsymbol{1}\{d_{i}>a_{n}\}\boldsymbol{R}_{i}^{\prime }\boldsymbol{H}_{i}%
\boldsymbol{R}_{i}\right] \\
& +n^{-1-\alpha }\sum_{i=1}^{n}a_{n}^{-2}E\left[ d_{i}^{2}\boldsymbol{1}%
\{d_{i}\leq a_{n}\}\boldsymbol{R}_{i}^{\prime }\boldsymbol{H}_{i}\boldsymbol{%
R}_{i}\right] +o(1),
\end{align*}%
which can be written equivalently as (noting that $a_{n}=C_{n}n^{-\alpha }$) 
\begin{align*}
Var\left( \boldsymbol{z}_{q,nT}\right) =& C_{n}^{-1}\left[
n^{-1}\sum_{i=1}^{n}E\left[ a_{n}\boldsymbol{1}\{d_{i}>a_{n}\}\boldsymbol{R}%
_{i}^{\prime }\boldsymbol{H}_{i}\boldsymbol{R}_{i}\right] \right] \\
& +C_{n}^{-1}\left[ n^{-1}\sum_{i=1}^{n}a_{n}^{-1}E\left[ d_{i}^{2}%
\boldsymbol{1}\{d_{i}\leq a_{n}\}\boldsymbol{R}_{i}^{\prime }\boldsymbol{H}%
_{i}\boldsymbol{R}_{i}\right] \right] +o(1).
\end{align*}%
By (\ref{VegziA}) of Lemma \ref{EVegzi}, $E\left[ n^{-1}%
\sum_{i=1}^{n}a_{n}^{-1}d_{i}^{2}\boldsymbol{1}\{d_{i}\leq a_{n}\}%
\boldsymbol{R}_{i}^{\prime }\boldsymbol{H}_{i}(\boldsymbol{X}_{i})%
\boldsymbol{R}_{i}\right] =O\left( a_{n}^{\alpha _{p}/2}\right) =O\left(
n^{-\alpha \alpha _{p}/2}\right) =o(1)$, and it follows that 
\begin{equation}
\lim_{n\rightarrow \infty }Var\left( \boldsymbol{z}_{q,nT}\right)
=C^{-1}\lim_{n\rightarrow \infty }n^{-1}\sum_{i=1}^{n}E\left[ a_{n}%
\boldsymbol{1}\{d_{i}>a_{n}\}\boldsymbol{R}_{i}^{\prime }\boldsymbol{H}_{i}%
\boldsymbol{R}_{i}\right] ,  \label{ConTMG}
\end{equation}%
where $C=\func{plim}_{n\rightarrow \infty }C_{n}>0$.
\end{proof}

\subsection{Proof Theorem \protect\ref{VarCon}\label{ProofVarCon}}

\begin{proof}
Using (\ref{betaihat}), (\ref{betai2}) and (\ref{etai}), we have 
\begin{equation}
\boldsymbol{\tilde{\beta}}_{i}-\boldsymbol{\beta }_{0}=\delta _{i}%
\boldsymbol{\beta }_{0}+\boldsymbol{\zeta }_{iT},  \label{V1}
\end{equation}%
where $\boldsymbol{\zeta }_{iT}=(1+\delta _{i})\boldsymbol{\eta }%
_{i}+(1+\delta _{i})\boldsymbol{\xi }_{iT}$, $\boldsymbol{\xi }_{iT}=\left( 
\boldsymbol{X}_{i}^{\prime }\boldsymbol{M}_{T}\boldsymbol{X}_{i}\right) ^{-1}%
\boldsymbol{X}_{i}^{\prime }\boldsymbol{M}_{T}\boldsymbol{u}_{i}$, and using
(\ref{TMG_egzi}) 
\begin{equation}
\boldsymbol{\hat{\beta}}_{TMG}-\boldsymbol{\beta }_{0}=\left( \frac{1}{1+%
\bar{\delta}_{n}}\right) \boldsymbol{\bar{\zeta}}_{nT},  \label{V2}
\end{equation}%
where $\boldsymbol{\bar{\zeta}}_{nT}=\frac{1}{n}\sum_{i=1}^{n}\boldsymbol{%
\eta }_{i}+\frac{1}{n}\sum_{i=1}^{n}\delta _{i}\boldsymbol{\eta }_{i}+\frac{1%
}{n}\sum_{i=1}^{n}(1+\delta _{i})\boldsymbol{\xi }_{iT}$. Subtracting (\ref%
{V2}) from (\ref{V1}) now yields 
\begin{equation*}
\boldsymbol{\tilde{\beta}}_{i}-\boldsymbol{\hat{\beta}}_{TMG}=\boldsymbol{%
\zeta }_{iT}+\delta _{i}\boldsymbol{\beta }_{0}-\left( \frac{1}{1+\bar{\delta%
}_{n}}\right) \boldsymbol{\bar{\zeta}}_{nT},
\end{equation*}%
and we have%
\begin{align}
& n^{-1}\sum_{i=1}^{n}\left( \boldsymbol{\tilde{\beta}}_{i}-\boldsymbol{\hat{%
\beta}}_{TMG}\right) \left( \boldsymbol{\tilde{\beta}}_{i}-\boldsymbol{\hat{%
\beta}}_{TMG}\right) ^{\prime } & &  \notag \\
=\,\,& n^{-1}\sum_{i=1}^{n}\boldsymbol{\zeta }_{iT}\boldsymbol{\zeta }%
_{iT}^{\prime }+\left( n^{-1}\sum_{i=1}^{n}\delta _{i}^{2}\right) 
\boldsymbol{\beta }_{0}\boldsymbol{\beta }_{0}^{^{\prime }}+\left(
n^{-1}\sum_{i=1}^{n}\delta _{i}\boldsymbol{\zeta }_{iT}\right) \boldsymbol{%
\beta }_{0}^{\prime }+\boldsymbol{\beta }_{0}\left(
n^{-1}\sum_{i=1}^{n}\delta _{i}\boldsymbol{\zeta }_{iT}^{\prime }\right) & &
\notag \\
\,\,& +\left[ \left( \frac{1}{1+\bar{\delta}_{n}}\right) ^{2}-2\left( \frac{1%
}{1+\bar{\delta}_{n}}\right) \right] \boldsymbol{\bar{\zeta}}_{nT}%
\boldsymbol{\bar{\zeta}}_{nT}^{\prime }-\left( \frac{\bar{\delta}_{n}}{1+%
\bar{\delta}_{n}}\right) \boldsymbol{\beta }_{0}\boldsymbol{\bar{\zeta}}%
_{nT}^{\prime }-\left( \frac{\bar{\delta}_{n}}{1+\bar{\delta}_{n}}\right) 
\boldsymbol{\bar{\zeta}}_{nT}\boldsymbol{\beta }_{0}^{\prime }. & &
\label{VTMG0}
\end{align}%
By the results in Lemma \ref{deltaeta}, $E\left( \bar{\delta}_{n}\right)
=O(a_{n}^{\alpha _{p}})$, and $E\left( \boldsymbol{\bar{\zeta}}_{nT}\right)
=E\left( \delta _{i}\boldsymbol{\eta }_{i}\right) =O(a_{n}^{\alpha _{p}})$, $%
\bar{\delta}_{n}=O(a_{n}^{\alpha _{p}})+o_{p}(1)$ and $n^{-1}\sum_{i=1}^{n}%
\delta _{i}^{2}=O(a_{n}^{\alpha _{p}})+o_{p}(1)$. Also, using (\ref%
{deltaOrder}) and (\ref{TMGgap2}), we have $\boldsymbol{\bar{\zeta}}%
_{nT}=O(n^{-\alpha \alpha _{p}})+O_{p}\left( n^{-\frac{(1-\alpha )}{2}%
}\right) $, and since $\alpha >1/(1+2\alpha _{p})$, 
\begin{align}
\left[ \left( \frac{1}{1+\bar{\delta}_{n}}\right) ^{2}-2\left( \frac{1}{1+%
\bar{\delta}_{n}}\right) \right] \boldsymbol{\bar{\zeta}}_{nT}\boldsymbol{%
\bar{\zeta}}_{nT}^{\prime }& =O\left( n^{-2\alpha \alpha _{p}}\right)
+O_{p}\left( n^{-\frac{1-\alpha }{2}-\alpha \alpha _{p}}\right) +O_{p}\left(
n^{-(1-\alpha )}\right)  \notag \\
& =O_{p}\left( n^{-(1-\alpha )}\right) \text{, }  \label{VTMG01}
\end{align}%
and 
\begin{equation}
\left( \frac{\bar{\delta}_{n}}{1+\bar{\delta}_{n}}\right) \boldsymbol{\bar{%
\zeta}}_{nT}\boldsymbol{\beta }_{0}^{\prime }=O_{p}(a_{n}^{\alpha
_{p}}n^{-\alpha \alpha _{p}})+O_{p}\left( a_{n}^{\alpha _{p}}n^{-\frac{%
(1-\alpha )}{2}}\right) =O_{p}\left( n^{-\frac{(1-\alpha )}{2}-\alpha \alpha
_{p}}\right) .  \label{VTMG02}
\end{equation}%
Consider now%
\begin{equation}
\boldsymbol{\bar{\zeta}}_{\delta ,nT}=n^{-1}\sum_{i=1}^{n}\delta _{i}%
\boldsymbol{\zeta }_{iT}=n^{-1}\sum_{i=1}^{n}\delta _{i}\left( 1+\delta
_{i}\right) \boldsymbol{\eta }_{i}+n^{-1}\sum_{i=1}^{n}\delta _{i}\left(
1+\delta _{i}\right) \boldsymbol{\xi }_{iT}.  \label{egzibard}
\end{equation}%
By (\ref{detax}) of Lemma \ref{deltaeta}, $E\left[ \delta _{i}\left(
1+\delta _{i}\right) \boldsymbol{\eta }_{i}\right] =O(a_{n}^{\alpha _{p}})$,
and since $\delta _{i}\left( 1+\delta _{i}\right) \boldsymbol{\eta }_{i}$
are distributed independently over $i$, we have 
\begin{equation}
n^{-1}\sum_{i=1}^{n}\delta _{i}\left( 1+\delta _{i}\right) \boldsymbol{\eta }%
_{i}=O_{p}(a_{n}^{\alpha _{p}})\text{.}  \label{egzibard1}
\end{equation}%
Since conditional on $\boldsymbol{X}_{i}$, $\delta _{i}\left( 1+\delta
_{i}\right) \boldsymbol{\xi }_{iT}$ are distributed independently over $i$
with zero means, following the same line of argument as in the proof Lemma %
\ref{EVegzi}, $E\left[ \delta _{i}\left( 1+\delta _{i}\right) \boldsymbol{%
\xi }_{iT}\right] =\boldsymbol{0}$, and 
\begin{align*}
Var\left[ n^{-1}\sum_{i=1}^{n}\delta _{i}\left( 1+\delta _{i}\right) 
\boldsymbol{\xi }_{iT}\right] &=n^{-2}\sum_{i=1}^{n}E\left[ \delta
_{i}^{2}\left( 1+\delta _{i}\right) ^{2}\boldsymbol{R}_{i}^{\prime }%
\boldsymbol{H}_{i}\boldsymbol{R}_{i}\right] \\
&\leq Cn^{-2}\sum_{i=1}^{n}E\left[ \delta _{i}^{2}\left( 1+\delta
_{i}\right) ^{2}d_{i}^{-1}\left\Vert \func{adj}(\boldsymbol{X}_{i}^{\prime }%
\boldsymbol{M}_{T}\boldsymbol{X}_{i})\right\Vert \right] .
\end{align*}%
Further using (\ref{1d2}), we have 
\begin{align*}
\left( 1+\delta _{i}\right) ^{2}\delta _{i}^{2} &=\left[ \boldsymbol{1}%
\{d_{i}>a_{n}\}+a_{n}^{-2}d_{i}^{2}\boldsymbol{1}\{d_{i}\leq a_{n}\}\right]
\left( \frac{d_{i}-a_{n}}{a_{n}}\right) ^{2}\boldsymbol{1}\{d_{i}\leq a_{n}\}
\\
&=a_{n}^{-2}d_{i}^{2}\left( \frac{d_{i}-a_{n}}{a_{n}}\right) ^{2}\boldsymbol{%
1}\{d_{i}\leq a_{n}\},
\end{align*}%
and 
\begin{equation*}
Var\left[ n^{-1}\sum_{i=1}^{n}\delta _{i}\left( 1+\delta _{i}\right) 
\boldsymbol{\xi }_{iT}\right] \leq Cn^{-2}\sum_{i=1}^{n}E\left[
a_{n}^{-2}d_{i}\left( \frac{d_{i}-a_{n}}{a_{n}}\right) ^{2}\boldsymbol{1}%
\{d_{i}\leq a_{n}\}\left\Vert \func{adj}(\boldsymbol{X}_{i}^{\prime }%
\boldsymbol{M}_{T}\boldsymbol{X}_{i})\right\Vert \right] .
\end{equation*}%
By the Cauchy-Schwarz inequality, 
\begin{align*}
&E\left[ a_{n}^{-2}d_{i}\left( \frac{d_{i}-a_{n}}{a_{n}}\right) ^{2}%
\boldsymbol{1}\{d_{i}\leq a_{n}\}\left\Vert \func{adj}(\boldsymbol{X}%
_{i}^{\prime }\boldsymbol{M}_{T}\boldsymbol{X}_{i})\right\Vert \right] \\
\leq &a_{n}^{-2}\left\{ E\left[ d_{i}^{2}\left( \frac{d_{i}-a_{n}}{a_{n}}%
\right) ^{4}\boldsymbol{1}\{d_{i}\leq a_{n}\}\right] \right\} ^{1/2}\left[
E\left\Vert \func{adj}\left( \boldsymbol{X}_{i}^{\prime }\boldsymbol{M}_{T}%
\boldsymbol{X}_{i}\right) \right\Vert ^{2}\right] ^{1/2},
\end{align*}%
and since by part (d) of Assumption \ref{errorsx}, $\func{sup}%
_{i}E\left\Vert \func{adj}\left( \boldsymbol{X}_{i}^{\prime }\boldsymbol{M}%
_{T}\boldsymbol{X}_{i}\right) \right\Vert ^{2}<C$, we have%
\begin{equation*}
Var\left[ n^{-1}\sum_{i=1}^{n}\delta _{i}\left( 1+\delta _{i}\right) 
\boldsymbol{\xi }_{iT}\right] \leq Cn^{-2}a_{n}^{-2}\sum_{i=1}^{n}\left\{ E%
\left[ d_{i}^{2}\left( \frac{d_{i}-a_{n}}{a_{n}}\right) ^{4}\boldsymbol{1}%
\{d_{i}\leq a_{n}\}\right] \right\} ^{1/2}.
\end{equation*}%
Also using (\ref{ds}) of Lemma \ref{deltaeta} and noting that $0 \leq \left[%
(d_{i} - a_{n})/ a_{n}\right]^4 \leq 1$ for $d_{i} \leq a_{n}$, 
\begin{equation*}
E\left[ d_{i}^{2}\left( \frac{d_{i}-a_{n}}{a_{n}}\right) ^{4}\boldsymbol{1}%
\{d_{i}\leq a_{n}\}\right] =O\left( a_{n}^{2+\alpha _{p}}\right) ,
\end{equation*}%
which yields%
\begin{equation*}
Var\left[ n^{-1}\sum_{i=1}^{n}\delta _{i}\left( 1+\delta _{i}\right) 
\boldsymbol{\xi }_{iT}\right] =O\left( n^{-1}a_{n}^{-2}a_{n}^{(2+\alpha
_{p})/2}\right) =O\left( n^{-1}a_{n}^{-(2-\alpha _{p})/2}\right) ,
\end{equation*}%
and by the Markov inequality, 
\begin{equation}
n^{-1}\sum_{i=1}^{n}\delta _{i}\left( 1+\delta _{i}\right) \boldsymbol{\xi }%
_{iT}=O_{p}\left( n^{-1/2}a_{n}^{-(2-\alpha _{p})/4}\right) =O_{p}\left(
n^{-1/2+\alpha (2-\alpha _{p})/4}\right) .  \label{egzibard2}
\end{equation}%
Using (\ref{egzibard1}) and (\ref{egzibard2}) in (\ref{egzibard}), we have $%
\boldsymbol{\bar{\zeta}}_{\delta ,nT}=O_{p}(n^{-\alpha \alpha
_{p}})+O_{p}\left( n^{-1/2+\alpha (2-\alpha _{p})/4}\right) $, which if used
with (\ref{VTMG01}) and (\ref{VTMG02}) in (\ref{VTMG0}) now yields 
\begin{align*}
& \frac{1}{n}\sum_{i=1}^{n}(\boldsymbol{\tilde{\beta}}_{i}-\boldsymbol{\hat{%
\beta}}_{TMG})(\boldsymbol{\tilde{\beta}}_{i}-\boldsymbol{\hat{\beta}}%
_{TMG})^{\prime } \\
=& \frac{1}{n}\sum_{i=1}^{n}\boldsymbol{\zeta }_{iT}\boldsymbol{\zeta }%
_{iT}^{\prime }+O_{p}\left( n^{-(1-\alpha )}\right) +O_{p}\left( n^{-\frac{%
(1-\alpha )}{2}-\alpha \alpha _{p}}\right) +O_{p}(n^{-\alpha \alpha
_{p}})+O_{p}\left( n^{-1/2+\alpha (2-\alpha _{p})/4}\right) .
\end{align*}%
Since $\boldsymbol{\zeta }_{iT}$ are independently distributed over $i$, 
\begin{equation*}
\func{plim}_{n\rightarrow \infty }n^{-1}\sum_{i=1}^{n}\left( \boldsymbol{%
\tilde{\beta}}_{i}-\boldsymbol{\hat{\beta}}_{TMG}\right) \left( \boldsymbol{%
\tilde{\beta}}_{i}-\boldsymbol{\hat{\beta}}_{TMG}\right) ^{\prime
}=\lim_{n\rightarrow \infty }\left[ n^{-1}\sum_{i=1}^{n}E\left( \boldsymbol{%
\zeta }_{iT}\boldsymbol{\zeta }_{iT}^{\prime }\right) \right] .
\end{equation*}%
Using (\ref{V2}) and recalling that $\bar{\delta}_{n}=O(a_{n}^{\alpha _{p}})$%
, we have%
\begin{equation*}
Avar\left( n^{(1-\alpha )/2}\boldsymbol{\hat{\beta}}_{TMG}\right)
=\lim_{n\rightarrow \infty }Var\left( n^{(1-\alpha )/2}\boldsymbol{\bar{\zeta%
}}_{nT}\right) =Avar\left( n^{(1-\alpha )/2}\boldsymbol{\bar{\zeta}}%
_{nT}\right) .
\end{equation*}%
Also (recall that $E\left( \boldsymbol{\bar{\zeta}}_{nT}\right)
=O(a_{n}^{\alpha _{p}}))$, 
\begin{align*}
\lim_{n\rightarrow \infty }Var\left( n^{(1-\alpha )/2}\boldsymbol{\bar{\zeta}%
}_{nT}\right) & =\lim_{n\rightarrow \infty }E\left\{ n^{-(1+\alpha
)}\sum_{i=1}^{n}\left[ \boldsymbol{\zeta }_{iT}-E\left( \boldsymbol{\zeta }%
_{iT}\right) \right] \left[ \boldsymbol{\zeta }_{iT}-E\left( \boldsymbol{%
\zeta }_{iT}\right) \right] ^{\prime }\right\} \\
& =\lim_{n\rightarrow \infty }n^{-(1+\alpha )}\sum_{i=1}^{n}E\left( 
\boldsymbol{\zeta }_{iT}\boldsymbol{\zeta }_{iT}^{\prime }\right)
+O(n^{-\alpha -2\alpha \alpha _{p}}).
\end{align*}%
Hence, 
\begin{align*}
Avar\left( n^{(1-\alpha )/2}\boldsymbol{\hat{\beta}}_{TMG}\right) &
=\lim_{n\rightarrow \infty }n^{-(1+\alpha )}\sum_{i=1}^{n}E\left( 
\boldsymbol{\zeta }_{iT}\boldsymbol{\zeta }_{iT}^{\prime }\right) \\
& =\func{plim}_{n\rightarrow \infty }n^{-(1+\alpha )}\sum_{i=1}^{n}\left( 
\boldsymbol{\tilde{\beta}}_{i}-\boldsymbol{\hat{\beta}}_{TMG}\right) \left( 
\boldsymbol{\tilde{\beta}}_{i}-\boldsymbol{\hat{\beta}}_{TMG}\right)
^{\prime }.
\end{align*}
\end{proof}

\subsection{Proof Theorem \protect\ref{AsyDTest}}

\label{ProofAsyDTest}

\begin{proof}
To derive the asymptotic distribution of $\sqrt{n}\boldsymbol{\hat{\Delta}}%
_{\beta }=\sqrt{n}$ $\left( \boldsymbol{\hat{\beta}}_{FE}-\boldsymbol{\hat{%
\beta}}_{TMG}\right) $, we first note from (\ref{FE1}) and (\ref{TMGb}) that 
\begin{equation*}
\sqrt{n}\left( \boldsymbol{\hat{\beta}}_{FE}-\boldsymbol{\beta }_{0}\right) =%
\boldsymbol{\bar{\Psi}}_{n}^{-1}\left( n^{-1/2}\sum_{i=1}^{n}\boldsymbol{X}%
_{i}^{\prime }\boldsymbol{M}_{T}\boldsymbol{\nu }_{i}\right) ,
\end{equation*}%
and%
\begin{equation*}
\sqrt{n}\left( \boldsymbol{\hat{\beta}}_{TMG}-\boldsymbol{\beta }_{0}\right)
=n^{-1/2}\sum_{i=1}^{n}\left( \frac{1+\delta _{i}}{1+\bar{\delta}_{n}}%
\right) \boldsymbol{\Psi }_{i}^{-1}\boldsymbol{X}_{i}^{\prime }\boldsymbol{M}%
_{T}\boldsymbol{\nu }_{i},
\end{equation*}%
where $\boldsymbol{\bar{\Psi}}_{n}=n^{-1}\sum_{i=1}^{n}\boldsymbol{\Psi }%
_{i} $, $\boldsymbol{\Psi }_{i}=\boldsymbol{X}_{i}^{\prime }\boldsymbol{M}%
_{T}\boldsymbol{X}_{i}$, $\bar{\delta}_{n}=n^{-1}\sum_{i=1}^{n}\delta _{i}$, 
$\delta _{i}$ is given by (\ref{deltai}), and $\boldsymbol{\nu }_{i}=%
\boldsymbol{u}_{i}+\boldsymbol{X}_{i}\boldsymbol{\eta }_{i}$. Using the
above expressions, 
\begin{equation}
\sqrt{n}\boldsymbol{\hat{\Delta}}_{\beta }=\frac{1}{\sqrt{n}}\sum_{i=1}^{n}%
\boldsymbol{\Gamma }_{i}\boldsymbol{X}_{i}^{\prime }\boldsymbol{M}_{T}%
\boldsymbol{\nu }_{i},  \label{DeltaBeta}
\end{equation}%
where 
\begin{equation}
\boldsymbol{\Gamma }_{i}=\boldsymbol{\bar{\Psi}}_{n}^{-1}-\left( \frac{%
1+\delta _{i}}{1+\bar{\delta}_{n}}\right) \boldsymbol{\Psi }_{i}^{-1}.
\label{GiApp}
\end{equation}%
Under $H_{0}$ defined by (\ref{null}) and Assumptions \ref{errorsx} and \ref%
{rcm}, $E(\boldsymbol{\nu }_{i}|\boldsymbol{X})=\boldsymbol{0}$, and
conditional on $\boldsymbol{X}$, the errors $\boldsymbol{u}_{i}$ and $%
\boldsymbol{\eta }_{i}$ are cross-sectionally independent. Thus, conditional
on $\boldsymbol{X}$ (or conditional on $\boldsymbol{\Gamma }_{i}\boldsymbol{X%
}_{i}^{\prime })$, $\boldsymbol{\Gamma }_{i}\boldsymbol{X}_{i}^{\prime }%
\boldsymbol{M}_{T}\boldsymbol{\nu }_{i}$ are cross-sectionally independent
over $i=1,2,...,n$, with $E\left( \boldsymbol{\Gamma }_{i}\boldsymbol{X}%
_{i}^{\prime }\boldsymbol{M}_{T}\boldsymbol{\nu }_{i}\right) =\boldsymbol{0}$%
, and by the central limit theorem, $\sqrt{n}\boldsymbol{\hat{\Delta}}%
_{\beta }$ tends to a normal distribution as $n\rightarrow \infty $, if for
some $\varepsilon >0$, $\sup_{i}E\left\Vert \boldsymbol{\Gamma }_{i}%
\boldsymbol{X}_{i}^{\prime }\boldsymbol{M}_{T}\boldsymbol{\nu }%
_{i}\right\Vert ^{2+\varepsilon }<C$ is met and $Avar\left( \sqrt{n}%
\boldsymbol{\hat{\Delta}}_{\beta }\right) \succ \boldsymbol{0}$. Since $%
\boldsymbol{\Gamma }_{i}\boldsymbol{X}_{i}^{\prime }\boldsymbol{M}_{T}%
\boldsymbol{\nu }_{i}=\boldsymbol{\Gamma }_{i}\boldsymbol{\Psi }_{i}%
\boldsymbol{\eta }_{i}+\boldsymbol{\Gamma }_{i}\boldsymbol{X}_{i}^{\prime }%
\boldsymbol{M}_{T}\boldsymbol{u}_{i}$, by the Minkowski inequality we have%
\begin{equation*}
\sup_{i}\left[ E\left\Vert \boldsymbol{\Gamma }_{i}\boldsymbol{X}%
_{i}^{\prime }\boldsymbol{M}_{T}\boldsymbol{\nu }_{i}\right\Vert
^{2+\varepsilon }\right] ^{\frac{1}{2+\varepsilon }}\leq \sup_{i}\left[
E\left\Vert \boldsymbol{\Gamma }_{i}\boldsymbol{\Psi }_{i}\boldsymbol{\eta }%
_{i}\right\Vert ^{2+\varepsilon }\right] ^{\frac{1}{2+\varepsilon }}+\sup_{i}%
\left[ E\left\Vert \boldsymbol{\Gamma }_{i}\boldsymbol{X}_{i}^{\prime }%
\boldsymbol{M}_{T}\boldsymbol{u}_{i}\right\Vert ^{2+\varepsilon }\right] ^{%
\frac{1}{2+\varepsilon }}.
\end{equation*}%
Under $H_{0}$ 
\begin{equation*}
\sup_{i}E\left\Vert \boldsymbol{\Gamma }_{i}\boldsymbol{\Psi }_{i}%
\boldsymbol{\eta }_{i}\right\Vert ^{2+\varepsilon }\leq \sup_{i}E\left\Vert 
\boldsymbol{\Gamma }_{i}\boldsymbol{\Psi }_{i}\right\Vert ^{2+\varepsilon
}\sup_{i}E\left\Vert \boldsymbol{\eta }_{i}\right\Vert ^{2+\varepsilon },
\end{equation*}%
and under Assumption \ref{rcm}, $\sup_{i}E\left\Vert \boldsymbol{\eta }%
_{i}\right\Vert ^{2+\varepsilon }<C$. Also, $\boldsymbol{\Gamma }_{i}%
\boldsymbol{\Psi }_{i}=\boldsymbol{\bar{\Psi}}_{n}^{-1}\boldsymbol{\Psi }%
_{i}-\frac{1+\delta _{i}}{1+\bar{\delta}_{n}}\boldsymbol{I}_{k^{\prime }}$, $%
\bar{\delta}_{n}=o_{p}(1)$, $|1+\delta _{i}| \leq 1$, and $%
\sup_{i}E\left\Vert \boldsymbol{\Gamma }_{i}\boldsymbol{\Psi }%
_{i}\right\Vert ^{2+\varepsilon }$ $<C$, if $\sup_{i}E\left\Vert \boldsymbol{%
\bar{\Psi}}_{n}^{-1}\boldsymbol{\Psi }_{i}\right\Vert ^{2+\varepsilon }$ $<C$%
. Under Assumption \ref{PoolA}, $\boldsymbol{\bar{\Psi}}_{n}^{-1}=%
\boldsymbol{\bar{\Psi}}^{-1}+o_{p}(1)$, $\left\Vert \boldsymbol{\bar{\Psi}}%
^{-1}\right\Vert ^{2+\varepsilon }<C$, and $\sup_{i}E\left\Vert \boldsymbol{%
\Psi }_{i}\right\Vert ^{2+\varepsilon }<C$, then 
\begin{equation*}
\sup_{i}E\left\Vert \boldsymbol{\bar{\Psi}}_{n}^{-1}\boldsymbol{\Psi }%
_{i}\right\Vert ^{2+\varepsilon }\leq \left\Vert \boldsymbol{\bar{\Psi}}%
^{-1}\right\Vert ^{2+\varepsilon }\sup_{i}E\left\Vert \boldsymbol{\Psi }%
_{i}\right\Vert ^{2+\varepsilon },\text{ as }n\rightarrow \infty ,
\end{equation*}%
and $\sup_{i}E\left\Vert \boldsymbol{\Gamma }_{i}\boldsymbol{\Psi }_{i}%
\boldsymbol{\eta }_{i}\right\Vert ^{2+\varepsilon }<C$, as $n\rightarrow
\infty $.\footnote{%
For the condition $\sup_{i}E\left\Vert \boldsymbol{\bar{\Psi}}_{n}^{-1}%
\boldsymbol{\Psi }_{i}\right\Vert ^{2+\varepsilon }<C$ to hold for any $n$,
the stronger moment conditions $E\left\Vert \boldsymbol{\bar{\Psi}}%
_{n}^{-1}\right\Vert ^{4+2\varepsilon }<C$ and $\sup_{i}E\left\Vert 
\boldsymbol{\Psi }_{i}\right\Vert ^{4+2\varepsilon }<C$ are required.}
Similarly,%
\begin{align*}
& \sup_{i}\left[ E\left\Vert \boldsymbol{\Gamma }_{i}\boldsymbol{X}%
_{i}^{\prime }\boldsymbol{M}_{T}\boldsymbol{u}_{i}\right\Vert
^{2+\varepsilon }\right] ^{\frac{1}{2+\varepsilon }} \\
\leq & \sup_{i}\left[ E\left\Vert \boldsymbol{\bar{\Psi}}_{n}^{-1}%
\boldsymbol{X}_{i}^{\prime }\boldsymbol{M}_{T}\boldsymbol{u}_{i}\right\Vert
^{2+\varepsilon }\right] ^{\frac{1}{2+\varepsilon }}+\sup_{i}\left[
E\left\Vert \left( \frac{1+\delta _{i}}{1+\bar{\delta}_{n}}\right) 
\boldsymbol{\Psi }_{i}^{-1}\boldsymbol{X}_{i}^{\prime }\boldsymbol{M}_{T}%
\boldsymbol{u}_{i}\right\Vert ^{2+\varepsilon }\right] ^{\frac{1}{%
2+\varepsilon }}.
\end{align*}%
Note that as $n\rightarrow \infty $, $\left\vert 1+\delta _{i}\right\vert
\leq 1$, and $E\left\Vert \boldsymbol{\bar{\Psi}}_{n}^{-1}\boldsymbol{X}%
_{i}^{\prime }\boldsymbol{M}_{T}\boldsymbol{u}_{i}\right\Vert
^{2+\varepsilon }\leq \left\Vert \boldsymbol{\bar{\Psi}}^{-1}\right\Vert
^{2+\varepsilon }E\left\Vert \boldsymbol{X}_{i}^{\prime }\boldsymbol{M}_{T}%
\boldsymbol{u}_{i}\right\Vert ^{2+\varepsilon }$. Further, under part(d) of
Assumption \ref{PoolA}, 
\begin{equation*}
\sup_{i}E\left\Vert \boldsymbol{X}_{i}^{\prime }\boldsymbol{M}_{T}%
\boldsymbol{u}_{i}\right\Vert ^{2+\varepsilon }\leq \sup_{i}E\left[ \lambda
_{\max }^{1+\varepsilon /2}\left( \boldsymbol{X}_{i}^{\prime }\boldsymbol{M}%
_{T}\boldsymbol{H}_{i}\boldsymbol{M}_{T}\boldsymbol{X}_{i}\right) \right] <C.
\end{equation*}%
Also, using (\ref{BMegzi}) of Lemma \ref{EVegzi}, $\sup_{i}E\left\Vert
\varkappa _{iT,n}\right\Vert ^{2+\varepsilon }<C$, and recalling that $\bar{%
\delta}_{n}=o_{p}(1)$, we have 
\begin{equation}
\sup_{i}E\left\Vert \boldsymbol{\Gamma }_{i}\boldsymbol{X}_{i}^{\prime }%
\boldsymbol{M}_{T}\boldsymbol{u}_{i}\right\Vert ^{2+\varepsilon }<C.
\label{Bgxu}
\end{equation}%
Overall, condition $\sup_{i}E\left\Vert \boldsymbol{\Gamma }_{i}\boldsymbol{X%
}_{i}^{\prime }\boldsymbol{M}_{T}\boldsymbol{\nu }_{i}\right\Vert
^{2+\varepsilon }<C$ is met, and under $H_{0}$, $\sqrt{n}\boldsymbol{\hat{%
\Delta}}_{\beta }\rightarrow _{d}N(\boldsymbol{0},\boldsymbol{V}_{\Delta })$%
, where%
\begin{equation*}
\boldsymbol{V}_{\Delta }=\lim_{n\rightarrow \infty }\frac{1}{n}%
\sum_{i=1}^{n}E\left( \boldsymbol{\Gamma }_{i}\boldsymbol{X}_{i}^{\prime }%
\boldsymbol{M}_{T}\boldsymbol{\nu }_{i}\boldsymbol{\nu }_{i}^{\prime }%
\boldsymbol{M}_{T}\boldsymbol{X}_{i}\boldsymbol{\Gamma }_{i}\right) .
\end{equation*}%
Since under $H_{0}$, $\boldsymbol{\eta }_{i}$ and $\boldsymbol{X}_{i}$ are
independently distributed, then $\boldsymbol{V}_{\Delta }=\boldsymbol{V}%
_{1\Delta }+\boldsymbol{V}_{2\Delta }$, where 
\begin{align*}
\boldsymbol{V}_{1\Delta }& =\lim_{n\rightarrow \infty }\frac{1}{n}%
\sum_{i=1}^{n}E\left( \boldsymbol{\Gamma }_{i}\boldsymbol{X}_{i}^{\prime }%
\boldsymbol{M}_{T}\boldsymbol{H}_{i}\boldsymbol{M}_{T}\boldsymbol{X}_{i}%
\boldsymbol{\Gamma }_{i}\right) , \\
\boldsymbol{V}_{2\Delta }& =\lim_{n\rightarrow \infty }\frac{1}{n}%
\sum_{i=1}^{n}E\left[ E\left( \boldsymbol{\Gamma }_{i}\boldsymbol{\Psi }_{i}%
\boldsymbol{\Omega }_{i\beta }\boldsymbol{\Psi }_{i}\boldsymbol{\Gamma }_{i}|%
\boldsymbol{X}\right) \right] =\lim_{n\rightarrow \infty }\frac{1}{n}%
\sum_{i=1}^{n}E\left( \boldsymbol{\Gamma }_{i}\boldsymbol{\Psi }_{i}%
\boldsymbol{\Omega }_{\beta }\boldsymbol{\Psi }_{i}\boldsymbol{\Gamma }%
_{i}\right) ,
\end{align*}%
and $\boldsymbol{\Gamma }_{i}=\boldsymbol{\bar{\Psi}}_{n}^{-1}-\left( \frac{%
1+\delta _{i}}{1+\bar{\delta}_{n}}\right) \boldsymbol{\Psi }_{i}^{-1}$. We
first note that $\boldsymbol{V}_{1\Delta }\succeq \boldsymbol{0}$. Since $%
\boldsymbol{H}_{i}\succ \boldsymbol{0}$, for any non-zero $k^{\prime }\times
1$ vector $\boldsymbol{b}$, 
\begin{equation*}
\frac{1}{n}\sum_{i=1}^{n}\boldsymbol{b}^{\prime }\boldsymbol{\Gamma }_{i}%
\boldsymbol{X}_{i}^{\prime }\boldsymbol{M}_{T}\boldsymbol{H}_{i}\boldsymbol{M%
}_{T}\boldsymbol{X}_{i}\boldsymbol{\Gamma }_{i}\boldsymbol{b}=\frac{1}{n}%
\sum_{i=1}^{n}\left( \boldsymbol{M}_{T}\boldsymbol{X}_{i}\boldsymbol{\Gamma }%
_{i}\boldsymbol{b}\right) ^{\prime }\boldsymbol{H}_{i}\left( \boldsymbol{M}%
_{T}\boldsymbol{X}_{i}\boldsymbol{\Gamma }_{i}\boldsymbol{b}\right) \geq 0.
\end{equation*}%
Also, since $\boldsymbol{\Omega }_{\beta }\succeq \boldsymbol{0}$, 
\begin{equation*}
\frac{1}{n}\sum_{i=1}^{n}\boldsymbol{b}^{\prime }\boldsymbol{\Gamma }_{i}%
\boldsymbol{\Psi }_{i}\boldsymbol{\Omega }_{\beta }\boldsymbol{\Psi }_{i}%
\boldsymbol{\Gamma }_{i}\boldsymbol{b}\geq \lambda _{\min }\left( 
\boldsymbol{\Omega }_{\beta }\right) \boldsymbol{b}^{\prime }\left( \frac{1}{%
n}\sum_{i=1}^{n}\boldsymbol{\Gamma }_{i}\boldsymbol{\Psi }_{i}^{2}%
\boldsymbol{\Gamma }_{i}\right) \boldsymbol{b}\geq 0.
\end{equation*}%
Hence, taking expectations and the limit as $n\rightarrow \infty $ yields $%
\boldsymbol{V}_{1\Delta }\succeq \boldsymbol{0}$, $\boldsymbol{V}_{2\Delta
}\succeq \boldsymbol{0}$, and overall, $\boldsymbol{V}_{\Delta }=\boldsymbol{%
V}_{1\Delta }+\boldsymbol{V}_{2\Delta }\succeq \boldsymbol{0}$. $\boldsymbol{%
V}_{\Delta }$ becomes positive definite if either $\boldsymbol{V}_{1\Delta
}\succ \boldsymbol{0}$, or $\lim\limits_{n\rightarrow \infty }\frac{1}{n}%
\sum_{i=1}^{n}E\left( \boldsymbol{\Gamma }_{i}\boldsymbol{\Psi }_{i}^{2}%
\boldsymbol{\Gamma }_{i}\right) \succ \boldsymbol{0}$ and $\lambda _{\min
}\left( \boldsymbol{\Omega }_{\beta }\right) >c>0$ such that $\boldsymbol{V}%
_{2\Delta }\succ \boldsymbol{0}$. Hence, under $H_{0}$, $H_{\beta }=n%
\boldsymbol{\hat{\Delta}}_{\beta }^{\prime }\boldsymbol{V}_{\Delta }^{-1}%
\boldsymbol{\hat{\Delta}}_{\beta }\rightarrow _{d}\chi _{k^{\prime }}^{2}$
as $n\rightarrow \infty $, where $\chi _{k^{\prime }}^{2}$ is a chi-squared
distribution with $k^{\prime }=\func{dim}(\boldsymbol{\beta })$ degrees of
freedom. To derive the limiting value of $H_{\beta }$ under the alternative
hypothesis defined by (\ref{H1}), we first write (\ref{DeltaBeta}) as 
\begin{equation*}
\sqrt{n}\left( \boldsymbol{\hat{\Delta}}_{\beta }-\boldsymbol{\mu }%
_{n}\right) =\frac{1}{\sqrt{n}}\sum_{i=1}^{n}\left[ \boldsymbol{z}%
_{i}-E\left( \boldsymbol{z}_{i}\right) \right] ,
\end{equation*}%
where $\boldsymbol{z}_{i}=\boldsymbol{\Gamma }_{i}\boldsymbol{X}_{i}^{\prime
}\boldsymbol{M}_{T}\boldsymbol{\nu }_{i}-E\left( \boldsymbol{\Gamma }_{i}%
\boldsymbol{X}_{i}^{\prime }\boldsymbol{M}_{T}\boldsymbol{\nu }_{i}\right) $%
, and $\boldsymbol{\mu }_{n}=n^{-1}\sum_{i=1}^{n}E\left( \boldsymbol{\Gamma }%
_{i}\boldsymbol{X}_{i}^{\prime }\boldsymbol{M}_{T}\boldsymbol{\nu }%
_{i}\right) $. By construction $E\left( \boldsymbol{z}_{i}\right) =%
\boldsymbol{0}$. Note that condition (\ref{Bgxu}) holds both under the null
and alternative hypotheses. Also, by the Minkowski inequality, 
\begin{equation*}
\sup_{i}\left[ E\left\Vert \boldsymbol{\Gamma }_{i}\boldsymbol{\Psi }_{i}%
\boldsymbol{\eta }_{i}\right\Vert ^{2+\varepsilon }\right] ^{\frac{1}{%
2+\varepsilon }} \leq \sup_{i}\left[ E\left\Vert \boldsymbol{\bar{\Psi}}%
_{n}^{-1}\boldsymbol{\Psi }_{i}\boldsymbol{\eta }_{i}\right\Vert
^{2+\varepsilon }\right] ^{\frac{1}{2+\varepsilon }}+\sup_{i}\left[
E\left\Vert \frac{1+\delta _{i}}{1+\bar{\delta}_{n}}\boldsymbol{\eta }%
_{i}\right\Vert ^{2+\varepsilon }\right] ^{\frac{1}{2+\varepsilon }}.
\end{equation*}%
Since $|1+\delta _{i}| \leq 1$ and under Assumption \ref{rcm}, $%
\sup_{i}E\left\Vert \boldsymbol{\eta }_{i}\right\Vert ^{2+\varepsilon }<C$,
then $\sup_{i}\left[ E\left\Vert \frac{1+\delta _{i}}{1+\bar{\delta}_{n}}%
\boldsymbol{\eta }_{i}\right\Vert ^{2+\varepsilon }\right] <C$. Further,
under Assumption \ref{PoolA}, $\left\Vert \boldsymbol{\bar{\Psi}}%
^{-1}\right\Vert ^{2+\varepsilon }<C$ for some $\varepsilon >0.$Then $%
\sup_{i}E\left\Vert \boldsymbol{\Gamma }_{i}\boldsymbol{\Psi }_{i}%
\boldsymbol{\eta }_{i}\right\Vert ^{2+\varepsilon }<C$ if 
\begin{equation*}
\sup_{i}E\left\Vert \boldsymbol{\bar{\Psi}}_{n}^{-1}\boldsymbol{\Psi }_{i}%
\boldsymbol{\eta }_{i}\right\Vert ^{2+\varepsilon }\leq \left\Vert 
\boldsymbol{\bar{\Psi}}^{-1}\right\Vert ^{2+\varepsilon }\sup_{i}E\left\Vert 
\boldsymbol{\Psi }_{i}\boldsymbol{\eta }_{i}\right\Vert ^{2+\varepsilon }<C%
\text{, as }n\rightarrow \infty .
\end{equation*}%
Thus, under $H_{1}$, and given the additional condition $\sup_{i}E\left\Vert 
\boldsymbol{\Psi }_{i}\boldsymbol{\eta }_{i}\right\Vert ^{2+\varepsilon }<C$%
, it also follows that $\sup_{i}\left\Vert \boldsymbol{\Gamma }_{i}%
\boldsymbol{X}_{i}^{\prime }\boldsymbol{M}_{T}\boldsymbol{\nu }%
_{i}\right\Vert ^{2+\varepsilon }<C$, and the Lyapunov condition $%
E\left\Vert \boldsymbol{z}_{i}\right\Vert ^{2+\varepsilon }<C$ holds for
some $\varepsilon >0$. Hence, by the Lindeberg-Feller Central Limit Theorem,
we have $\sqrt{n}\left( \boldsymbol{\hat{\Delta}}_{\beta }-\boldsymbol{\mu }%
_{n}\right) \rightarrow _{d}N\left( \boldsymbol{0},\boldsymbol{V}_{\Delta
}^{\ast }\right) $, where $\boldsymbol{V}_{\Delta }^{\ast }=Avar\left( \sqrt{%
n}\boldsymbol{\hat{\Delta}}_{\beta }\left\vert H_{1}\right. \right) $, and 
\begin{equation}
\sqrt{n}\left( \boldsymbol{\hat{\Delta}}_{\beta }-\boldsymbol{\mu }%
_{n}\right) =O_{p}(1).  \label{DeltaH1}
\end{equation}%
To derive the non-centrality parameter of $\sqrt{n}\boldsymbol{\hat{\Delta}}%
_{\beta }$ under $H_{1}$, we note that%
\begin{align*}
\boldsymbol{\mu }_{n}=& n^{-1}\sum_{i=1}^{n}E\left[ \left( \boldsymbol{\bar{%
\Psi}}_{n}^{-1}\boldsymbol{X}_{i}^{\prime }\boldsymbol{M}_{T}\boldsymbol{u}%
_{i}-\left( \frac{1+\delta _{i}}{1+\bar{\delta}_{n}}\right) \boldsymbol{\Psi 
}_{i}^{-1}\boldsymbol{X}_{i}^{\prime }\boldsymbol{M}_{T}\boldsymbol{u}%
_{i}\right) \right] \\
& +n^{-1}\sum_{i=1}^{n}E\left[ \left( \boldsymbol{\bar{\Psi}}_{n}^{-1}%
\boldsymbol{X}_{i}^{\prime }\boldsymbol{M}_{T}\boldsymbol{X}_{i}\boldsymbol{%
\eta }_{i}-\left( \frac{1+\delta _{i}}{1+\bar{\delta}_{n}}\right) 
\boldsymbol{\Psi }_{i}^{-1}\boldsymbol{X}_{i}^{\prime }\boldsymbol{M}_{T}%
\boldsymbol{X}_{i}\boldsymbol{\eta }_{i}\right) \right].
\end{align*}%
Recalling that under Assumption \ref{errorsx}, $E\left( \boldsymbol{u}%
_{i}\left\vert \boldsymbol{X}\right. \right) =\boldsymbol{0}$, then $E\left( 
\boldsymbol{\bar{\Psi}}_{n}^{-1}\boldsymbol{X}_{i}^{\prime }\boldsymbol{M}%
_{T}\boldsymbol{u}_{i}\right) =\boldsymbol{0}$, and $E\left[ \left( \frac{%
1+\delta _{i}}{1+\bar{\delta}_{n}}\right) \boldsymbol{\Psi }_{i}^{-1}%
\boldsymbol{X}_{i}^{\prime }\boldsymbol{M}_{T}\boldsymbol{u}_{i}\right] =%
\boldsymbol{0}$. Further, $E\left[ \left( \frac{1+\delta _{i}}{1+\bar{\delta}%
_{n}}\right) \boldsymbol{\Psi }_{i}^{-1}\boldsymbol{X}_{i}^{\prime }%
\boldsymbol{M}_{T}\boldsymbol{X}_{i}\boldsymbol{\eta }_{i}\right] =E\left[
\left( \frac{\delta _{i}-\bar{\delta}_{n}}{1+\bar{\delta}_{n}}\right) 
\boldsymbol{\eta }_{i}\right] $, and using results in Lemma \ref{deltaeta},
we have 
\begin{equation}
E\left[ \left( \frac{\delta _{i}-\bar{\delta}_{n}}{1+\bar{\delta}_{n}}%
\right) \boldsymbol{\eta }_{i}\right] =O\left( a_{n}^{\alpha _{p}}\right)
=O(n^{-\alpha \alpha _{p}}).  \label{meuB}
\end{equation}%
Hence, 
\begin{equation}
\boldsymbol{\mu }_{n}=n^{-1}\sum_{i=1}^{n}E\left( \boldsymbol{\bar{\Psi}}%
_{n}^{-1}\boldsymbol{X}_{i}^{\prime }\boldsymbol{M}_{T}\boldsymbol{X}_{i}%
\boldsymbol{\eta }_{i}\right) +O(n^{-\alpha \alpha _{p}}),  \label{meu}
\end{equation}%
where $\ \alpha >(1+2\alpha _{p})^{-1}$. Now write $H_{\beta }=n\boldsymbol{%
\hat{\Delta}}_{\beta }^{\prime }\boldsymbol{V}_{\Delta }^{-1}\boldsymbol{%
\hat{\Delta}}_{\beta }$ as 
\begin{align}
H_{\beta }=& n\boldsymbol{\hat{\Delta}}_{\beta }^{\prime }\boldsymbol{V}%
_{\Delta }^{-1}\boldsymbol{\hat{\Delta}}_{\beta }=\left[ \sqrt{n}\left( 
\boldsymbol{\hat{\Delta}}_{\beta }-\boldsymbol{\mu }_{n}\right) ^{\prime }+%
\sqrt{n}\boldsymbol{\mu }_{n}^{\prime }\right] \boldsymbol{V}_{\Delta }^{-1}%
\left[ \sqrt{n}\left( \boldsymbol{\hat{\Delta}}_{\beta }-\boldsymbol{\mu }%
_{n}\right) +\sqrt{n}\boldsymbol{\mu }_{n}\right]  \notag \\
=& \sqrt{n}\left( \boldsymbol{\hat{\Delta}}_{\beta }-\boldsymbol{\mu }%
_{n}\right) ^{\prime }\boldsymbol{V}_{\Delta }^{-1}\sqrt{n}\left( 
\boldsymbol{\hat{\Delta}}_{\beta }-\boldsymbol{\mu }_{n}\right) +\sqrt{n}%
\boldsymbol{\mu }_{n}^{\prime }\boldsymbol{V}_{\Delta }^{-1}\sqrt{n}\left( 
\boldsymbol{\hat{\Delta}}_{\beta }-\boldsymbol{\mu }_{n}\right)  \notag \\
& +\sqrt{n}\left( \boldsymbol{\hat{\Delta}}_{\beta }-\boldsymbol{\mu }%
_{n}\right) ^{\prime }\boldsymbol{V}_{\Delta }^{-1}\sqrt{n}\boldsymbol{\mu }%
_{n}+n\left( \boldsymbol{\mu }_{n}^{\prime }\boldsymbol{V}_{\Delta }^{-1}%
\boldsymbol{\mu }_{n}\right) .  \label{Hbeta2}
\end{align}%
By (\ref{DeltaH1}), the first three terms of the above equation are
asymptotically dominated by $n\left( \boldsymbol{\mu }_{n}^{\prime }%
\boldsymbol{V}_{\Delta }^{-1}\boldsymbol{\mu }_{n}\right) $. Provided that
the eigenvalues of $\boldsymbol{V}_{\Delta }^{-1}$ are bounded above and
away from zero, $H_{\beta }=\ominus _{p}\left[ n\left( \boldsymbol{\mu }%
_{n}^{\prime }\boldsymbol{\mu }_{n}\right) \right] $. Furthermore, using (%
\ref{meu}) and noting that under $H_{1}$, defined by (\ref{H1}), $E\left( 
\frac{1}{n}\sum_{i=1}^{n}\boldsymbol{\bar{\Psi}}_{n}^{-1}\boldsymbol{X}%
_{i}^{\prime }\boldsymbol{M}_{T}\boldsymbol{X}_{i}\boldsymbol{\eta }%
_{i}\right) =\ominus \left( n^{a_{\eta }-1}\right) $, then $H_{\beta
}=\ominus _{p}(n^{2a_{\eta }-1})+O_{p}(n^{a_{\eta }-\alpha \alpha
_{p}})+O_{p}\left( n^{1-2\alpha \alpha _{p}}\right) $, and $H_{\beta
}\rightarrow _{p}\infty $, as $n\rightarrow \infty $, for $\alpha _{\eta
}>1/2$.
\end{proof}

\begin{remark}
Under $H_{0}$, consider the case where $\boldsymbol{H}_{i}=\sigma ^{2}%
\boldsymbol{I}_{T}$ with $\sigma ^{2}>0$. Then 
\begin{align*}
\boldsymbol{V}_{1\Delta }& =\lim_{n\rightarrow \infty }\frac{1}{n}%
\sum_{i=1}^{n}E\left( \boldsymbol{\Gamma }_{i}\boldsymbol{X}_{i}^{\prime }%
\boldsymbol{M}_{T}\boldsymbol{H}_{i}\boldsymbol{M}_{T}\boldsymbol{X}_{i}%
\boldsymbol{\Gamma }_{i}\right)  \\
& =\sigma ^{2}\lim_{n\rightarrow \infty }E\left[ \left( \frac{1}{1+\bar{%
\delta}_{n}}\right) ^{2}\frac{1}{n}\sum_{i=1}^{n}\left( 1+\delta _{i}\right)
^{2}\boldsymbol{\Psi }_{i}^{-1}-\boldsymbol{\bar{\Psi}}_{n}^{-1}\right] .
\end{align*}%
Hence, if $\lim_{n\rightarrow \infty }E\left[ \left( \frac{1}{1+\bar{\delta}%
_{n}}\right) ^{2}\frac{1}{n}\sum_{i=1}^{n}\left( 1+\delta _{i}\right) ^{2}%
\boldsymbol{\Psi }_{i}^{-1}-\boldsymbol{\bar{\Psi}}_{n}^{-1}\right] \succ 
\boldsymbol{0}$, then $\boldsymbol{V}_{1\Delta }\succ \boldsymbol{0}$, and
consequently, $\boldsymbol{V}_{\Delta }\succ \boldsymbol{0}$. In the absence
of trimming (with $\delta _{i}=0$ for all $i$), this condition reduces to $%
\lim_{n\rightarrow \infty }\left[ \frac{1}{n}\sum_{i=1}^{n}E\left( 
\boldsymbol{\Psi }_{i}^{-1}\right) -E\left( \boldsymbol{\bar{\Psi}}%
_{n}^{-1}\right) \right] \succ \boldsymbol{0}$, which holds so long as $%
\boldsymbol{\Psi }_{i}$ are not identically distributed. See the Appendix to 
\cite{PesaranEtal1996}.
\end{remark}

\section{TMG-TE estimator with $T=k$ \label{Teqk}}

For panels with time effects when $T=k$, we make the following additional
assumption:

\begin{assumption}
\label{stationaryCRE} For $i=1,2,...,n$, 
\begin{equation}
E\left( \boldsymbol{x}_{it}^{\prime }\boldsymbol{\eta }_{i }\right) =E\left( 
\boldsymbol{x}_{is}^{\prime }\boldsymbol{\eta }_{i }\right) \text{, for all }%
t,s=1,2,...,T,  \label{IDtime}
\end{equation}%
where $\boldsymbol{\eta }_{i }=\boldsymbol{\beta }_{i}-\boldsymbol{\beta }%
_{0}$, $\boldsymbol{\beta}_0 = \func{plim}_{n \to \infty} E(\boldsymbol{\beta%
}_i)$, and $\Vert \boldsymbol{\beta }_{0}\Vert <C$.
\end{assumption}

\begin{remark}
Assumption \ref{stationaryCRE} allows for dependence between $\boldsymbol{x}%
_{it} $ and $\boldsymbol{\eta }_{i}$, but requires this dependence to be
time-invariant.
\end{remark}

\begin{remark}
The irregular identification of $\boldsymbol{\phi }$ when $T=k$ in \cite%
{GrahamPowell2012} is based on moments conditional on the sub-population of
\textquotedblleft stayers\textquotedblright. Under Assumption \ref%
{distributiondi} $d_{i}>0$ for all $i$, i.e., there are no \textquotedblleft
stayers\textquotedblright in the population, this identification strategy
cannot be used. Moreover, GP assume that the joint distribution of $(u_{it},
\alpha_{i}, \boldsymbol{\beta}_{i}^{\prime })^{\prime }$ given $\boldsymbol{X%
}_{i}$ does not depend on $t$, which is similar to Assumption \ref%
{stationaryCRE}. See interpretations of Assumption 1.1 part (ii) on p. 2111
in \cite{GrahamPowell2012}.
\end{remark}

Averaging (\ref{panTE1}) over $i$, 
\begin{equation}
\bar{y}_{\circ t}=\bar{\alpha}_{n}+\phi _{t}+\boldsymbol{\bar{x}}_{\circ
t}^{\prime }\boldsymbol{\beta }_{0}+\bar{\nu}_{\circ t},  \label{ybar1}
\end{equation}%
where $\bar{\nu}_{\circ t}=n^{-1}\sum_{i=1}^{n}\nu _{it}$, $\nu _{it}=%
\boldsymbol{x}_{it}^{\prime }\boldsymbol{\eta }_{i }+u_{it}$, $\bar{y}%
_{\circ t}=n^{-1}\sum_{i=1}^{n}y_{it}$, $\boldsymbol{\bar{x}}_{\circ
t}=n^{-1}\sum_{i=1}^{n}\boldsymbol{x}_{it}$, $\bar{u}_{\circ
t}=n^{-1}\sum_{i=1}^{n}u_{it}$, and $\bar{\alpha}_{n}=n^{-1}\sum_{i=1}^{n}%
\alpha _{i}$. Averaging over $t$, under the normalization $%
\sum_{t=1}^{T}\phi _{t}=0$, 
\begin{equation}
\bar{y}_{\circ \circ }=\bar{\alpha}_{n}+\boldsymbol{\bar{x}}_{\circ \circ
}^{\prime }\boldsymbol{\beta }_{0}+n^{-1}\sum_{i=1}^{n}\boldsymbol{\bar{x}}%
_{i\circ }^{\prime }\boldsymbol{\eta }_{i }+\bar{u}_{\circ \circ },
\label{ybar0}
\end{equation}%
where $\bar{y}_{\circ \circ }=T^{-1}\sum_{t=1}^{T}\bar{y}_{\circ t}$, $%
\boldsymbol{\bar{x}}_{\circ \circ }=T^{-1}\sum_{t=1}^{T}\boldsymbol{\bar{x}}%
_{\circ t}$ and $\bar{u}_{\circ \circ }=T^{-1}\sum_{t=1}^{T}\bar{u}_{\circ
t} $. Subtracting (\ref{ybar0}) from (\ref{ybar1}), yields 
\begin{equation}
\phi _{t}=\left( \bar{y}_{\circ t}-\bar{y}_{\circ \circ }\right) -\left( 
\boldsymbol{\bar{x}}_{\circ t}-\boldsymbol{\bar{x}}_{\circ \circ }\right)
^{\prime }\boldsymbol{\beta }_{0}-\left( \bar{\nu}_{\circ t}-\bar{\nu}%
_{\circ \circ }\right) \text{, for }t=1,2,...,T,  \label{phit0}
\end{equation}%
where $\bar{\nu}_{\circ t}-\bar{\nu}_{\circ \circ }=\left( \bar{u}_{\circ t}-%
\bar{u}_{\circ \circ }\right) +n^{-1}\sum_{i=1}^{n}\left( \boldsymbol{x}%
_{it}-\boldsymbol{\bar{x}}_{i\circ }\right) ^{\prime }\boldsymbol{\eta }_{i
} $. Under Assumptions \ref{errorsx}, \ref{CRE} and \ref{stationaryCRE}, $%
\bar{\nu}_{\circ t}-\bar{\nu}_{\circ \circ }=O_{p}(n^{-1/2})$,\footnote{%
For a proof, see Lemma \ref{asysit}.} which suggests the following estimator
of $\phi _{t}$ 
\begin{equation}
\hat{\phi}_{t}=\left( \bar{y}_{\circ t}-\bar{y}_{\circ \circ }\right)
-\left( \boldsymbol{\bar{x}}_{\circ t}-\boldsymbol{\bar{x}}_{\circ \circ
}\right) ^{\prime }\boldsymbol{\hat{\beta}}_{TMG-TE},\text{ for }t=1,2,...,T,
\label{phihat}
\end{equation}%
where 
\begin{equation}
\boldsymbol{\hat{\beta}}_{TMG-TE}=\boldsymbol{\hat{\beta}}_{TMG}-\boldsymbol{%
\bar{Q}}_{n}^{\prime }\boldsymbol{\boldsymbol{\hat{\phi}}},  \label{TMG-TE1}
\end{equation}%
and $\boldsymbol{\bar{Q}}_{n}$ is given by (\ref{Qn}). Stacking the
equations in (\ref{phihat}) over $t=1,2,...,T$ we have 
\begin{equation}
\boldsymbol{\boldsymbol{\hat{\phi}}=M}_{T}\left(\boldsymbol{\bar{y}-%
\boldsymbol{\bar{X}}\hat{\beta}}_{TMG-TE}\right),  \label{phihat_a}
\end{equation}%
where $\boldsymbol{M}_{T}=\boldsymbol{I}_{T}-T^{-1}\boldsymbol{\tau }_{T}%
\boldsymbol{\tau }_{T}^{\prime }$, $\boldsymbol{\bar{y}}=n^{-1}\sum_{i=1}^{n}%
\boldsymbol{y}_{i}$, and$\text{ }\boldsymbol{\bar{X}}=n^{-1}\sum_{i=1}^{n}%
\boldsymbol{X}_{i}$. The above system of equations can now be solved in
terms of $\boldsymbol{\hat{\beta}}_{TMG}$ if $( \boldsymbol{I}_{T}-%
\boldsymbol{M}_{T}\boldsymbol{\bar{X}}\boldsymbol{\bar{Q}}_{n}^{\prime}) $
is non-singular. Under this condition, we have 
\begin{equation}
\boldsymbol{\boldsymbol{\hat{\phi}}}=\left( \boldsymbol{I}_{T}-\boldsymbol{M}%
_{T}\boldsymbol{\bar{X}\bar{Q}}_{n}^{\prime }\right) ^{-1}\boldsymbol{M}%
_{T}\left( \boldsymbol{\bar{y}}-\boldsymbol{\bar{X}\hat{\beta}}_{TMG}\right),
\label{phihat_b}
\end{equation}%
and substituting $\boldsymbol{\boldsymbol{\hat{\phi}}}$ from (\ref{phihat_a}%
) in (\ref{TMG-TE1}) we have 
\begin{equation}
\boldsymbol{\hat{\beta}}_{TMG-TE}=\left( \boldsymbol{I}_{k^{\prime}}-%
\boldsymbol{\bar{Q}}_{n}^{\prime }\boldsymbol{M}_{T}\boldsymbol{\bar{X}}%
\right) ^{-1}\left( \boldsymbol{\hat{\beta}}_{TMG}-\boldsymbol{\bar{Q}}%
_{n}^{\prime }\boldsymbol{M}_{T}\boldsymbol{\bar{y}}\right) .
\label{TMG_TE_2}
\end{equation}

\begin{remark}
Note that $\left( \boldsymbol{M}_{T}\boldsymbol{\bar{X}}\right) \boldsymbol{%
\bar{Q}}_{n}^{\prime }$ and $\boldsymbol{\bar{Q}}_{n}^{\prime }\left( 
\boldsymbol{M}_{T}\boldsymbol{\bar{X}}\right) $ have the same $k^{\prime}$ ($%
k^{\prime}\leq T-1$) non-zero eigenvalues, $\func{det}\left( \boldsymbol{I}%
_{T}-\boldsymbol{M}_{T}\boldsymbol{\bar{X}\bar{Q}}_{n}^{\prime }\right) =%
\func{det}\left( \boldsymbol{I}_{k^{\prime}}-\boldsymbol{\bar{Q}}%
_{n}^{\prime }\boldsymbol{M}_{T}\boldsymbol{\bar{X}}\right) $, and if $%
\left( \boldsymbol{I}_{T}-\boldsymbol{M}_{T}\boldsymbol{\bar{X}}\boldsymbol{%
\bar{Q}}_{n}^{\prime}\right) $ is invertible so will $\left( \boldsymbol{I}%
_{k^{\prime}}-\boldsymbol{\bar{Q}}_{n}^{\prime }\boldsymbol{M}_{T}%
\boldsymbol{\bar{X}}\right) $.
\end{remark}

The following theorem provides a summary of the results for estimation of $%
\boldsymbol{\phi }_{0}$ and $\boldsymbol{\beta }_{0}$, and their asymptotic
distributions when $T\geq k$. Using results similar to the ones employed to
establish Theorem \ref{VarCon}, robust covariance matrices for $\boldsymbol{%
\hat{\beta}}_{TMG-TE}$ and $\boldsymbol{\hat{\phi}}$ are given by (\ref%
{varbetaTE}) and (\ref{VarPhicombined}), respectively. In particular, the
asymptotic covariance of $\boldsymbol{\hat{\phi}}$ is applicable to both
cases (a) and (b) of Theorem \ref{thm_asytmgte}, and does not require
knowing if $\func{plim}_{n\rightarrow \infty }\boldsymbol{M}_{T}\boldsymbol{%
\bar{X}}=\boldsymbol{0}$, or not.

\begin{theorem}[Asymptotic distributions of $\boldsymbol{\hat{\protect\beta}}%
_{TMG-TE}$ and $\boldsymbol{\hat{\protect\phi}}$ when $T\geq k$]
\label{thm_asytmgte}

Suppose

\noindent that for $i=1,2,...,n$ and $t=1,2,...,T$, $y_{it}$ are generated
by (\ref{panTE1}), $T\geq k$, Assumptions \ref{errorsx}, \ref{rcm}(a), \ref%
{distributiondi}, \ref{CRE} and \ref{stationaryCRE} hold, and $\boldsymbol{I}%
_{k^{\prime }}-\boldsymbol{\bar{Q}}_{n}^{\prime }\boldsymbol{M}_{T}%
\boldsymbol{\bar{X}}$ and its probability limit as $n \to \infty$ are
invertible, where $\boldsymbol{\bar{Q}}_{n}$ is defined by (\ref{Qn}) and $%
\boldsymbol{\bar{X}}=n^{-1}\sum_{i=1}^{n}\boldsymbol{X}_{i}$. Then as $%
n\rightarrow \infty $, for $\alpha >\frac{1}{1+2\alpha _{p}}$, where $\alpha
_{p}$ is the tail index of the distribution of $\left\{
1/d_{i},i=1,2,...,n\right\} $ and $d_{i}=\det (\boldsymbol{X}_{i}^{\prime }%
\boldsymbol{M}_{T}\boldsymbol{X}_{i}\boldsymbol{)}$, 
\begin{equation}
n^{(1-\alpha )/2}\left( \boldsymbol{\hat{\beta}}_{TMG-TE}-\boldsymbol{\beta }%
_{0}\right) \rightarrow _{d}N\left( \boldsymbol{0}_{k^{\prime }},\boldsymbol{%
V}_{\beta ,TMG-TE}\right) ,  \label{Asytmgte}
\end{equation}%
where $\boldsymbol{\hat{\beta}}_{TMG-TE}$ is given by (\ref{TMG_TE_2}), 
\begin{equation*}
\boldsymbol{V}_{\beta ,TMG-TE}=\left( \boldsymbol{I}_{k^{\prime }}-%
\boldsymbol{G}_{x}\right) ^{-1}\boldsymbol{V}_{\beta }(\boldsymbol{\phi }%
_{0})\left( \boldsymbol{I}_{k^{\prime }}-\boldsymbol{G}_{x}^{\prime }\right)
^{-1},
\end{equation*}%
$\boldsymbol{G}_{x}=\lim_{n\rightarrow \infty }\left( \boldsymbol{\bar{Q}}%
_{n}^{\prime }\boldsymbol{M}_{T}\boldsymbol{\bar{X}}\right) $, and $%
\boldsymbol{V}_{\beta }(\boldsymbol{\phi }_{0})=\lim_{n\rightarrow \infty
}Var\left[ n^{(1-\alpha )/2}\boldsymbol{\hat{\beta}}_{TMG-TE}(\boldsymbol{%
\phi }_{0})\right] $. Also, 
\begin{equation*}
\boldsymbol{\hat{\phi}}=\boldsymbol{M}_{T}\left( \boldsymbol{\bar{y}-\bar{X}%
\hat{\beta}}_{TMG-TE}\right) .
\end{equation*}

\begin{itemize}
\item[(a)] If $\func{plim}_{n\rightarrow \infty }\boldsymbol{M}_{T}%
\boldsymbol{\bar{X}}=\boldsymbol{0}$, we have 
\begin{equation}
\sqrt{n}\left( \boldsymbol{\hat{\phi}}-\boldsymbol{\phi }_{0}\right)
\rightarrow _{d}N(\boldsymbol{0}_{T},\boldsymbol{\Omega }_{\nu }),
\label{Asyphi1}
\end{equation}%
where $\boldsymbol{\Omega }_{\nu }=\boldsymbol{M}_{T}\left[%
\lim_{n\rightarrow \infty }n^{-1}\sum_{i=1}^{n}E\left( \boldsymbol{\nu }_{i}%
\boldsymbol{\nu }_{i}^{\prime }\right) \right]\boldsymbol{M}_{T}$, $%
\boldsymbol{\nu }_{i}=\boldsymbol{(}\nu_{i1},\nu_{i2},...,\nu_{iT})^{\prime
} $, and $\nu _{it}=u_{it}+\boldsymbol{x}_{it}^{\prime }\boldsymbol{\eta}%
_{i} $.

\item[(b)] If $\func{plim}_{n\rightarrow \infty }\boldsymbol{M}_{T}%
\boldsymbol{\bar{X}}\neq \boldsymbol{0}$, for $\alpha >\frac{1}{1+2\alpha
_{p}}$, we have 
\begin{equation}
n^{(1-\alpha )/2}\left( \boldsymbol{\hat{\phi}}-\boldsymbol{\phi }%
_{0}\right) \rightarrow _{d}N\left( \boldsymbol{0}_{T},\boldsymbol{V}_{\phi
}\right) ,  \label{Asyphi2}
\end{equation}%
where $\boldsymbol{V}_{\phi }=\func{plim}_{n\rightarrow \infty }\boldsymbol{M%
}_{T}\boldsymbol{\bar{X}}Var\left( n^{(1-\alpha )/2}\boldsymbol{\hat{\beta}}%
_{TMG-TE}\right) \boldsymbol{\bar{X}}^{\prime }\boldsymbol{M}_{T}$.
\end{itemize}
\end{theorem}

\begin{proof}
To derive the asymptotic distribution of $\boldsymbol{\hat{\beta}}_{TMG-TE}$%
, we first note that $\boldsymbol{\hat{\beta}}_{TMG-TE}(\boldsymbol{\phi }%
_{0})=\boldsymbol{\hat{\beta}}_{TMG}-\boldsymbol{\bar{Q}}_{n}^{\prime }%
\boldsymbol{\phi }_{0}$, and $\boldsymbol{\hat{\beta}}_{TMG-TE}=\boldsymbol{%
\hat{\beta}}_{TMG}-\boldsymbol{\bar{Q}}_{n}^{\prime }\boldsymbol{\hat{\phi}}$%
, where $\boldsymbol{\bar{Q}}_{n}$ is given by (\ref{Qn}). Hence%
\begin{equation}
\left( \boldsymbol{\hat{\beta}}_{TMG-TE}-\boldsymbol{\beta }_{0}\right)
-\left( \boldsymbol{\hat{\beta}}_{TMG-TE}(\boldsymbol{\phi }_{0})-%
\boldsymbol{\beta }_{0}\right) =-\boldsymbol{\bar{Q}}_{n}^{\prime }\left( 
\boldsymbol{\hat{\phi}-\phi }_{0}\right) .  \label{thetagap}
\end{equation}%
Also stacking (\ref{phit0}) over $t$ and subtracting the results from (\ref%
{phihat_a}) yields 
\begin{equation}
\boldsymbol{\hat{\phi}}-\boldsymbol{\phi }_{0}=-\boldsymbol{M}_{T}%
\boldsymbol{\bar{X}}\left( \boldsymbol{\hat{\beta}}_{TMG-TE}-\boldsymbol{%
\beta }_{0}\right) +\boldsymbol{M}_{T}\boldsymbol{\bar{\nu}},  \label{phigap}
\end{equation}%
where $\boldsymbol{\bar{\nu}=}n^{-1}\sum_{i=1}^{n}\boldsymbol{\nu }_{i}$, $%
\boldsymbol{\nu }_{i}=(\nu _{i1},\nu _{i2},...,\nu _{iT})^{\prime }$, and $%
\nu _{it}=u_{it}+\boldsymbol{x}_{it}^{\prime }\boldsymbol{\eta }_{i}$. Using
this result in (\ref{thetagap}), we have 
\begin{equation*}
\left( \boldsymbol{I}_{k^{\prime }}-\boldsymbol{\bar{Q}}_{n}^{\prime }%
\boldsymbol{M}_{T}\boldsymbol{\bar{X}}\right) \left( \boldsymbol{\hat{\beta}}%
_{TMG-TE}-\boldsymbol{\beta }_{0}\right) =\left( \boldsymbol{\hat{\beta}}%
_{TMG-TE}(\boldsymbol{\phi }_{0})-\boldsymbol{\beta }_{0}\right) -%
\boldsymbol{\bar{Q}}_{n}^{\prime }\boldsymbol{M}_{T}\boldsymbol{\bar{\nu}}.
\end{equation*}%
For a known value of $\boldsymbol{\phi }_{0}$, the asymptotic distribution
of $\left( \boldsymbol{\hat{\beta}}_{TMG-TE}(\boldsymbol{\phi }_{0})-%
\boldsymbol{\beta }_{0}\right) $ is the same as $\boldsymbol{\hat{\beta}}%
_{TMG}$ with $\boldsymbol{y}_{i}$ replaced by $\boldsymbol{y}_{i}-%
\boldsymbol{\phi }_{0}$. Under the assumption that $\boldsymbol{I}%
_{k^{\prime }}-\boldsymbol{\bar{Q}}_{n}^{\prime }\boldsymbol{M}_{T}%
\boldsymbol{\bar{X}}$ is invertible, we have 
\begin{equation*}
\boldsymbol{\hat{\beta}}_{TMG-TE}-\boldsymbol{\beta }_{0}=\left( \boldsymbol{%
I}_{k^{\prime }}-\boldsymbol{\bar{Q}}_{n}^{\prime }\boldsymbol{M}_{T}%
\boldsymbol{\bar{X}}\right) ^{-1}\left( \boldsymbol{\hat{\beta}}_{TMG-TE}(%
\boldsymbol{\phi }_{0})-\boldsymbol{\beta }_{0}\right) -\left( \boldsymbol{I}%
_{k^{\prime }}-\boldsymbol{\bar{Q}}_{n}^{\prime }\boldsymbol{M}_{T}%
\boldsymbol{\bar{X}}\right) ^{-1}\boldsymbol{\bar{Q}}_{n}^{\prime }%
\boldsymbol{M}_{T}\boldsymbol{\bar{\nu}}.
\end{equation*}%
Hence using Lemma \ref{asysit}, $\boldsymbol{M}_T\boldsymbol{\bar{\nu}}%
=O_{p}\left( n^{-1/2}\right) $, and we have 
\begin{equation*}
n^{\frac{1-\alpha}{2}}\left( \boldsymbol{\hat{\beta}}_{TMG-TE}-\boldsymbol{%
\beta }_{0}\right) =\left( \boldsymbol{I}_{k^{\prime }}-\boldsymbol{\bar{Q}}%
_{n}^{\prime }\boldsymbol{M}_{T}\boldsymbol{\bar{X}}\right) ^{-1}\left[ n^{%
\frac{1-\alpha}{2}}\left( \boldsymbol{\hat{\beta}}_{TMG-TE}(\boldsymbol{\phi 
}_{0})-\boldsymbol{\beta }_{0}\right) \right] +O_{p}(n^{-\alpha /2}),
\end{equation*}%
where for a known $\boldsymbol{\phi }_{0}$ we have already established in
Theorem \ref{thm_asytmg} that 
\begin{equation*}
n^{(1-\alpha )/2}\left( \boldsymbol{\hat{\beta}}_{TMG-TE}(\boldsymbol{\phi }%
_{0})-\boldsymbol{\beta }_{0}\right) \rightarrow _{d}N\left( \boldsymbol{%
0,V_{\beta }}(\boldsymbol{\phi }_{0})\right) ,
\end{equation*}%
with $\ \boldsymbol{V}_{\beta }(\boldsymbol{\phi }_{0})=\lim_{n\rightarrow
\infty }Var\left[ n^{(1-\alpha )/2}\boldsymbol{\hat{\beta}}_{TMG-TE}(%
\boldsymbol{\phi }_{0})\right] $. Suppose that $\func{plim}_{n\rightarrow
\infty }\left( \boldsymbol{\bar{Q}}_{n}^{\prime }\boldsymbol{M}_{T}%
\boldsymbol{\bar{X}}\right) =\boldsymbol{G}_{x}$, where $\boldsymbol{I}%
_{k^{\prime }}-\boldsymbol{G}_{x}$ is non-singular. For $\alpha >\frac{1}{%
1+2\alpha _{p}}$, we have $n^{(1-\alpha )/2}\left( \boldsymbol{\hat{\beta}}%
_{TMG-TE}-\boldsymbol{\beta }_{0}\right) \rightarrow _{d}N\left( \boldsymbol{%
0},\boldsymbol{V}_{\beta ,TMG-TE}\right) $, where 
\begin{equation}
\boldsymbol{V}_{\beta ,TMG-TE}=\left( \boldsymbol{I}_{k^{\prime }}-%
\boldsymbol{G}_{x}\right) ^{-1}\boldsymbol{V}_{\beta }(\boldsymbol{\phi }%
_{0})\left[ \left( \boldsymbol{I}_{k^{\prime }}-\boldsymbol{G}_{x}\right)
^{-1}\right] ^{\prime }.  \label{Vtmgte}
\end{equation}%
A consistent estimator is given by 
\begin{equation}
\widehat{Var(\boldsymbol{\hat{\beta}}_{TMG-TE})}=\left( \boldsymbol{I}%
_{k^{\prime }}-\boldsymbol{\bar{Q}}_{n}^{\prime }\boldsymbol{M}_{T}%
\boldsymbol{\bar{X}}\right) ^{-1}\boldsymbol{\widehat{\Sigma}}_{\beta }\left[
\left( \boldsymbol{I}_{k^{\prime }}-\boldsymbol{\bar{Q}}_{n}^{\prime }%
\boldsymbol{M}_{T}\boldsymbol{\bar{X}}\right) ^{-1}\right] ^{\prime },
\label{varbetaTE}
\end{equation}%
where 
\begin{equation*}
\boldsymbol{\widehat{\Sigma}}_{\beta }=\frac{1}{n(n-1)(1+\bar{\delta}%
_{n})^{2}}\sum_{i=1}^{n}(\boldsymbol{\tilde{\beta}}_{i}-\boldsymbol{Q}%
_{i}^{\prime }\boldsymbol{\hat{\phi}}-\boldsymbol{\hat{\beta}}_{TMG-TE})(%
\boldsymbol{\tilde{\beta}}_{i}-\boldsymbol{Q}_{i}^{\prime }\boldsymbol{\hat{%
\phi}}-\boldsymbol{\hat{\beta}}_{TMG-TE})^{\prime }.
\end{equation*}%
Consider now the asymptotic distribution of $\boldsymbol{\hat{\phi}}$. Given
(\ref{phigap}), two cases can arise depending on whether the probability
limit of $\boldsymbol{M}_{T}\boldsymbol{\bar{X}}$ tends to zero as $%
n\rightarrow \infty $, or not. Under (a) $\func{plim}_{n\rightarrow \infty }%
\boldsymbol{M}_{T}\boldsymbol{\bar{X}}=\boldsymbol{0}$, we have $n^{1/2}(%
\boldsymbol{\hat{\phi}}-\boldsymbol{\phi }_{0})\rightarrow _{d}N(\boldsymbol{%
0},\boldsymbol{\Omega }_{\nu })$, where $\boldsymbol{\Omega }_{\nu }$ is
given by (\ref{Omeganu}), namely $\boldsymbol{\hat{\phi}}\rightarrow _{p}%
\boldsymbol{\phi }_{0}$ at the regular rate of $n^{-1/2}$. Also since 
\begin{equation*}
\nu _{it}-\bar{\nu}_{i\circ }=(u_{it}-\bar{u}_{i\circ })+\left( \boldsymbol{x%
}_{it}-\boldsymbol{\bar{x}}_{i\circ }\right) ^{\prime }\boldsymbol{\eta }%
_{i}=y_{it}-\bar{y}_{i\circ }-(\boldsymbol{x}_{it}-\bar{\boldsymbol{x}}%
_{i\circ })^{\prime }\boldsymbol{\beta -}\phi _{t},
\end{equation*}%
$\boldsymbol{\Omega }_{\nu }$ can be consistently estimated by%
\begin{equation}
\boldsymbol{\widehat{\Omega }}_{\nu }=\frac{1}{n-1}\sum_{i=1}^{n}\boldsymbol{%
M}_{T}\left( \boldsymbol{y}_{i}-\boldsymbol{X}_{i}\boldsymbol{\hat{\beta}}%
_{TMG-TE}-\boldsymbol{\hat{\phi}}\right) \left( \boldsymbol{y}_{i}-%
\boldsymbol{X}_{i}\boldsymbol{\hat{\beta}}_{TMG-TE}-\boldsymbol{\hat{\phi}}%
\right) ^{\prime }\boldsymbol{M}_{T}.  \label{Omeganuhat}
\end{equation}%
Under case $(b)$, $\func{plim}_{n\rightarrow \infty }\boldsymbol{M}_{T}%
\boldsymbol{\bar{X}}\neq \boldsymbol{0}$, and convergence of $\boldsymbol{%
\hat{\phi}}$ to $\boldsymbol{\phi }_{0}$ cannot achieve the regular rate. To
see this, note that%
\begin{equation*}
n^{(1-\alpha )/2}\left( \boldsymbol{\hat{\phi}}-\boldsymbol{\phi }%
_{0}\right) =-\boldsymbol{M}_{T}\boldsymbol{\bar{X}}\left[ n^{(1-\alpha
)/2}\left( \boldsymbol{\hat{\beta}}_{TMG-TE}-\boldsymbol{\beta }_{0}\right) %
\right] +n^{-\alpha /2}\boldsymbol{M}_{T}\left( n^{1/2}\boldsymbol{\bar{\nu}}%
\right),
\end{equation*}%
where $\boldsymbol{M}_{T}\left( n^{1/2}\boldsymbol{\bar{\nu}}\right)
=O_{p}(1)$ and since $\alpha >0$ the second term tends to zero, but rather
slowly. In practice, where it is not known whether $\boldsymbol{M}_{T}%
\boldsymbol{\bar{X}}\rightarrow \boldsymbol{0}$ or not, one can consistently
estimate the variance of $\boldsymbol{\hat{\phi}}$ by%
\begin{equation}
\widehat{Var\left( \boldsymbol{\hat{\phi}}\right) }=\boldsymbol{\boldsymbol{M%
}}_{T}\left[ \boldsymbol{\bar{X}}\widehat{Var\left( \boldsymbol{\hat{\beta}}%
_{TMG-TE}\right) }\boldsymbol{\bar{X}}^{\prime }+n^{-1}\boldsymbol{\widehat{%
\Omega }}_{\nu }\right] \boldsymbol{M}_{T},  \label{VarPhicombined}
\end{equation}%
where $\widehat{Var\left( \boldsymbol{\hat{\beta}}_{TMG-TE}\right) }$ and $%
\boldsymbol{\widehat{\Omega }}_{\nu }$ are given by (\ref{varbetaTE}) and (%
\ref{Omeganuhat}), respectively. Note that $\widehat{Var\left( \boldsymbol{%
\hat{\phi}}\right) }$ is singular as $\widehat{Var\left( \boldsymbol{\hat{%
\phi}}\right) }\boldsymbol{\tau }_{T}=\boldsymbol{0}$, its diagonal elements
can be used to test the individual statistical significance of $\hat{\phi}_t$%
, for $t=1,2,...,T$, while their joint statistical significance can be
tested on the appropriate constrained subspace, subject to $\boldsymbol{\phi}%
^{\prime }\boldsymbol{\tau}_T=0$.
\end{proof}

\begin{example}
As an example of case (a) in Theorem \ref{thm_asytmgte}, suppose $%
\boldsymbol{x}_{it}=\boldsymbol{\alpha }_{ix}+\boldsymbol{u}_{x,it}$, where $%
\boldsymbol{u}_{x,it}$ are distributed independently over $i$ with zero
means. Then \thinspace $\boldsymbol{\bar{x}}_{\circ t}-\boldsymbol{\bar{x}}%
_{\circ \circ }=\boldsymbol{\bar{u}}_{x,\circ t}-\boldsymbol{\bar{u}}%
_{x,\circ \circ }\rightarrow _{p}\boldsymbol{0}$, and we have $\func{plim}%
_{n\rightarrow \infty }\boldsymbol{M}_{T}\boldsymbol{\bar{X}}=\boldsymbol{0}$%
. An example of case (b) arises when $\boldsymbol{x}_{it}$ contains latent
factors, $\boldsymbol{x}_{it}=\boldsymbol{\alpha }_{ix}+\boldsymbol{\Gamma }%
_{ix}\boldsymbol{f}_{t}+\boldsymbol{u}_{x,it}$. In this case $\,\boldsymbol{%
\bar{x}}_{\circ t}-\boldsymbol{\bar{x}}_{\circ \circ }=\boldsymbol{\bar{%
\Gamma}}\left( \boldsymbol{f}_{t}-\boldsymbol{\bar{f}}\right) +\boldsymbol{%
\bar{u}}_{x,\circ t}-\boldsymbol{\bar{u}}_{x,\circ \circ }$, where $%
\boldsymbol{\bar{\Gamma}}=\frac{1}{n}\sum_{i=1}^{n}\boldsymbol{\Gamma }%
_{ix}\rightarrow _{p}\boldsymbol{\Gamma}$. It follows that $\boldsymbol{\bar{%
x}}_{\circ t}-\boldsymbol{\bar{x}}_{\circ \circ }\rightarrow _{p}\boldsymbol{%
\Gamma}\left( \boldsymbol{f}_{t}-\boldsymbol{\bar{f}}\right) $ which is
non-zero if $\boldsymbol{f}_{t}$ varies over time and $\boldsymbol{\Gamma }%
\neq \boldsymbol{0}$, namely at least one of the factors has loadings with
non-zero means.
\end{example}

\newpage

\begin{center}
\thispagestyle{empty}

{\large Online Supplement} \smallbreak
{\large to Estimation of Average Effects in Short $T$ Heterogeneous Panels} 
{\Large \bigskip \bigskip }

{\normalsize {M. Hashem Pesaran }}

{\normalsize University of Southern California, and Trinity College,
Cambridge }

{\normalsize \vskip 0.5em }{\normalsize Liying Yang }

{\normalsize Shenzhen Audencia Financial Technology Institute, Shenzhen
University }

{\normalsize \vskip 1.5em \today
}

{\normalsize \vskip 30pt }
\end{center}

\makeatletter\setcounter{page}{1}\renewcommand{\thepage}{S\arabic{page}} %
\setcounter{table}{0} \renewcommand{\thetable}{S.\arabic{table}} %
\setcounter{section}{0}\renewcommand{\thesection}{S.\arabic{section}} %
\setcounter{section}{0}\renewcommand{\theHsection}{appendixsection.S.%
\arabic{section}} 
\setcounter{figure}{0} \renewcommand{\thefigure}{S.\arabic{figure}} %
\setcounter{footnote}{0} \renewcommand{\thefootnote}{S\arabic{footnote}} %
\renewcommand{\thetheorem}{S.\arabic{theorem}}\setcounter{theorem}{0} %
\renewcommand{\theproposition}{S.\arabic{proposition}}%
\setcounter{proposition}{0} \renewcommand{\theassumption}{S.%
\arabic{assumption}}\setcounter{assumption}{0} \renewcommand{\thelemma}{S.%
\arabic{lemma}}\setcounter{lemma}{0} \renewcommand{\theremark}{S.%
\arabic{remark}}\setcounter{remark}{0} \makeatother

This online supplement is structured as follows. Section \ref{mathsup}
provides additional mathematical derivations for the estimation and testing
procedures. Section \ref{mcsup} details the design of the Monte Carlo (MC)
experiments and reports additional MC results. Section \ref{appt3} presents
the estimation results for the cash transfer program using panels with $%
T=3>k $, where $k=k^{\prime}+1$ and $k^{\prime}$ is the number of regressors.

In Section \ref{mcsup} of MC simulations, sub-section \ref{secMCDGP}
provides details of the MC design. Sub-section \ref{MCap} summarizes MC
results for our estimation of $\alpha _{p}$, the tail index of the
distribution of $1/d_{i}$ ($d_{i}=\func{det}(\boldsymbol{X}_{i}^{\prime }%
\boldsymbol{M}_{T}\boldsymbol{X}_{i})$) using \cite{Hill1975}'s estimation
procedure, where $\boldsymbol{X}_{i}$ is the $T$ by $k^{\prime }$ matrix of
unit $i$'s observations including $k^{\prime }$ regressors and $\boldsymbol{M%
}_{T}=\boldsymbol{I}_{T}-\boldsymbol{\tau }_{T}\boldsymbol{\tau }%
_{T}^{\prime }/T$ with a $T\times T$ identity matrix $\boldsymbol{I}_{T}$
and a $T\times 1$ vector of ones, $\boldsymbol{\tau }_{T}$. Sub-section \ref%
{MCk2} presents MC results for TMG and TMG-TE estimators with Gaussian
distributed errors in the regressor ($x_{it}$) process, where sub-section %
\ref{MCthresh} provides MC evidence on TMG and GP estimators with different
choices of the trimming threshold parameters, $\alpha $ and $\alpha _{GP}$,
respectively. Sub-section \ref{MCk34} presents MC results for estimations
without time effects, focusing on cases with two and three regressors.
Sub-section \ref{MCte} provides MC results for estimations that include time
effects. Sub-section \ref{MCtest} provides and discusses MC results on the
Hausman test of uncorrelated heterogeneity in panel data models with time
effects.

\section{Mathematical supplement}

\label{mathsup}

\textbf{Notations:} Generic positive finite constants are denoted by $C$
when large, and $c$ when small. They can take different values at different
instances. $\lambda _{\max }\left( \boldsymbol{A}\right) $ and $\lambda
_{\min }\left( \boldsymbol{A}\right) $ denote the maximum and minimum
eigenvalues of matrix $\boldsymbol{A}$. $\boldsymbol{A}\succ \boldsymbol{0}$
and $\boldsymbol{A}\succeq \boldsymbol{0}$ denote that matrix $\boldsymbol{A}
$ is positive definite and is positive semi-definite, respectively. When
matrix $\boldsymbol{A}$ is square, its adjugate (adjoint) and determinant
are denoted by $\func{adj}(\boldsymbol{A)}$ and $\det (\boldsymbol{A})$,
respectively. If $\det (\boldsymbol{A})\neq 0$, then the inverse of $%
\boldsymbol{A}$ is given by $\boldsymbol{A}^{-1}=\func{adj}(\boldsymbol{A)}%
/\det (\boldsymbol{A})$. $\left\Vert \boldsymbol{A}\right\Vert =\lambda
_{\max }^{1/2}(\boldsymbol{A}^{\prime }\boldsymbol{A)}$ and $\left\Vert 
\boldsymbol{A}\right\Vert _{1}$ denote the spectral and column norms of
matrix $\boldsymbol{A}$, respectively. $\left\Vert \boldsymbol{x}\right\Vert
_{p}=\left[ E\left( \left\Vert \boldsymbol{x}\right\Vert ^{p}\right) \right]
^{1/p}$. If $\left\{ f_{n}\right\} _{n=1}^{\infty }$ is any real sequence
and $\left\{ g_{n}\right\} _{n=1}^{\infty }$ is a sequence of positive real
numbers, then $f_{n}=O(g_{n})$ if there exists $C$ such that $\left\vert
f_{n}\right\vert /g_{n}\leq C$ for all $n$, and $f_{n}=o(g_{n})$ if $%
f_{n}/g_{n}\rightarrow 0$ as $n\rightarrow \infty $. Similarly, $%
f_{n}=O_{p}(g_{n})$ if $f_{n}/g_{n}$ is stochastically bounded, and $%
f_{n}=o_{p}(g_{n})$, if $f_{n}/g_{n}\rightarrow _{p}0$. $f_{n}=\ominus
(g_{n})$ if there exist $n_{0}\geq 1$ and positive finite constants $C_{0}$
and $C_{1}$, such that $\inf_{n\geq n_{0}}\left( \left\vert f_{n}\right\vert
/g_{n}\right) \geq C_{0}$, and $\sup_{n\geq n_{0}}\left( \left\vert
f_{n}\right\vert /g_{n}\right) \leq C_{1}$. The operator $\rightarrow _{p}$
denotes convergence in probability, and $\rightarrow _{d}$ denotes
convergence in distribution. $IID$ stands for independently and identically
distributed. $\boldsymbol{u}\perp \boldsymbol{v}$ is used to show that
vectors of random variables $\boldsymbol{u}$ and $\boldsymbol{v}$ are
independently distributed.

\subsection{The Hausman test of correlated heterogeneity with time effects 
\label{TestTE}}

Given the panel data model with time effects in (\ref{panTE1}), a Hausman
test can be constructed based on the difference between the trimmed mean
group estimator with time effects (TMG-TE) and two-way fixed effects (TWFE)
estimator when $T\geq k$. The TWFE estimator is given by 
\begin{equation}
\boldsymbol{\hat{\beta}}_{TWFE}=\boldsymbol{\bar{\Psi}}_{n,TE}^{-1}\left[ 
\frac{1}{n}\sum_{i=1}^{n}\left( \boldsymbol{X}_{i}-\boldsymbol{\bar{X}}%
\right) ^{\prime }\boldsymbol{M}_{T}\left( \boldsymbol{y}_{i}-\boldsymbol{%
\bar{y}}\right) \right] ,  \label{TWFEhat}
\end{equation}%
and 
\begin{equation}
\boldsymbol{\hat{\phi}}_{TWFE}=\boldsymbol{M}_{T}(\boldsymbol{\bar{y}}-%
\boldsymbol{\bar{X}}\boldsymbol{\hat{\beta}}_{TWFE}),  \label{TWFEphihat}
\end{equation}%
where 
\begin{equation}
\boldsymbol{\bar{\Psi}}_{n,TE}=\frac{1}{n}\sum_{i=1}^{n}\left( \boldsymbol{X}%
_{i}-\boldsymbol{\bar{X}}\right) ^{\prime }\boldsymbol{M}_{T}\left( 
\boldsymbol{X}_{i}-\boldsymbol{\bar{X}}\right).  \label{Psi_nte}
\end{equation}
A consistent estimator of the variance of the TWFE estimator is given by 
\begin{equation}
\widehat{Var(\hat{\beta}_{TWFE})}=\frac{1}{n}\boldsymbol{\bar{\Psi}}%
_{n,TE}^{-1}\left[ \frac{1}{n}\sum_{i=1}^{n}\left( \boldsymbol{X}_{i}-%
\boldsymbol{\bar{X}}\right) ^{\prime }\boldsymbol{M}_{T}\boldsymbol{\hat{u}}%
_{i}^{\ast }\boldsymbol{\hat{u}}_{i}^{\ast \prime }\boldsymbol{M}_{T}\left( 
\boldsymbol{X}_{i}-\boldsymbol{\bar{X}}\right) \right] \boldsymbol{\bar{\Psi}%
}_{n,TE}^{-1}  \label{CvarTWFE}
\end{equation}%
where $\boldsymbol{\hat{u}}_{i}^{\ast }=\boldsymbol{\hat{u}}_{i,TWFE}-\frac{1%
}{n}\sum_{i=1}^{n}\boldsymbol{\hat{u}}_{i,TWFE}$, $\boldsymbol{\hat{u}}%
_{i,TWFE} = \boldsymbol{M}_{T}(\boldsymbol{y}_{i} - \boldsymbol{X}_{i} 
\boldsymbol{\hat{\beta}}_{TWFE})$, and $\boldsymbol{\bar{\Psi}}_{n,TE}$ is
given by (\ref{Psi_nte}).\footnote{%
The estimator in (\ref{CvarTWFE}) is a straightforward extension of the
variance estimator given for FE estimators in Section 26.7 of \cite%
{Pesaran2015}, and for a fixed $T$, it is robust to error serial correlation
and error variance heteroskedasticity.}

Then, 
\begin{equation}
\boldsymbol{\hat{\beta}}_{TWFE}-\boldsymbol{\beta }_{0}=\boldsymbol{\bar{\Psi%
}}_{n,TE}^{-1}\left[ \frac{1}{n}\sum_{i=1}^{n}\left( \boldsymbol{X}_{i}-%
\boldsymbol{\bar{X}}\right) ^{\prime }\boldsymbol{M}_{T}\boldsymbol{\tilde{%
\nu}}_{i}\right] ,  \label{TWFE}
\end{equation}%
where $\boldsymbol{\tilde{\nu}}_{i}=\boldsymbol{\nu }_{i}-\boldsymbol{\bar{%
\nu}}$, $\boldsymbol{\bar{\nu}}=n^{-1}\sum_{i=1}^{n}\boldsymbol{\nu }_{i}$,
and $\boldsymbol{\nu }_{i}=\boldsymbol{u}_{i}+\boldsymbol{X}_{i}\boldsymbol{%
\eta }_{i }$.

Consider $\boldsymbol{\hat{\Delta}}_{\beta ,TE}=\boldsymbol{\hat{\beta}}%
_{TWFE}-\boldsymbol{\hat{\beta}}_{TMG-TE}$. We derive the test statistics
under the null hypothesis in (\ref{null}) for two cases: the TMG-TE
estimator is given by (\ref{TMG_TE_2}) when $T= k$ and (\ref{BetaTE2}) when $%
T>k$. The implicit null is given by $\frac{1}{\sqrt{n}}\sum_{i=1}^{n}E\left[
\left( \boldsymbol{X}_{i}-\boldsymbol{\bar{X}}\right) ^{\prime }\boldsymbol{M%
}_{T}\left( \boldsymbol{X}_{i}-\boldsymbol{\bar{X}}\right) \boldsymbol{\eta }%
_{i }\right] \rightarrow \boldsymbol{0}$, which is implied by (\ref{null})
but not \textit{vice versa}. We make the following assumption that
corresponds to the pooling Assumption \ref{PoolA}:

\begin{assumption}
\label{PoolB} Let $\boldsymbol{\Psi }_{i,TE}=(\boldsymbol{X}_{i}-\boldsymbol{%
\bar{X}})^{\prime }\boldsymbol{M}_{T}(\boldsymbol{X}_{i}-\boldsymbol{\bar{X}}%
)$ and $\boldsymbol{\bar{\Psi}}_{n,TE}=\frac{1}{n}\sum_{i=1}^{n}\boldsymbol{%
\Psi }_{i,TE}\succ \boldsymbol{0}$, where $\boldsymbol{\bar{X}} = n^{-1}
\sum_{i=1}^n \boldsymbol{X}_{i}$. (a) $\lim\limits_{n\rightarrow \infty }E(%
\boldsymbol{\bar{\Psi}}_{n,TE})=\boldsymbol{\bar{\Psi}}_{TE}\succ 
\boldsymbol{0}$, $n^{-1}\sum_{i=1}^{n}E\left\Vert \boldsymbol{\Psi }%
_{i,TE}\right\Vert ^{2}<C$, and $\boldsymbol{\bar{\Psi}}_{n,TE}^{-1}=%
\boldsymbol{\bar{\Psi}}_{TE}^{-1}+o_{p}(1)$, with $\left\Vert \boldsymbol{%
\bar{\Psi}}_{TE}^{-1}\right\Vert ^{2+\varepsilon }<C$, for some $\varepsilon
>0$. (b) $\frac{1}{n} \sum_{i=1}^{n}(\boldsymbol{X}_{i}-\boldsymbol{\bar{X}}%
)^{\prime }\boldsymbol{M}_{T}\boldsymbol{H}_{i}\boldsymbol{M}_{T}(%
\boldsymbol{X}_{i}-\boldsymbol{\bar{X}})\rightarrow
_{p}\lim\limits_{n\rightarrow \infty }\frac{1}{n} \sum_{i=1}^{n}E\left[ (%
\boldsymbol{X}_{i}-\boldsymbol{\bar{X}})^{\prime }\boldsymbol{M}_{T}%
\boldsymbol{H}_{i}\boldsymbol{M}_{T}(\boldsymbol{X}_{i}-\boldsymbol{\bar{X}})%
\right] \succ \boldsymbol{0}$. (c) For a fixed $T\geq k$ and some $%
\varepsilon >0$, $\sup_i E\left\{ \lambda _{\max }^{1+\varepsilon }\left[ (%
\boldsymbol{X}_{i}-\boldsymbol{\bar{X}})^{\prime }\boldsymbol{M}_{T}%
\boldsymbol{H}_{i}\boldsymbol{M}_{T}(\boldsymbol{X}_{i}-\boldsymbol{\bar{X}})%
\right] \right\} <C$.
\end{assumption}

\subsubsection{Panels with $T = k$}

When $T=k$, we consider $\boldsymbol{\hat{\beta}}_{TMG-TE}$ given by (\ref%
{TMG_TE_2}). Given (\ref{thetagap}), (\ref{phigap}) and $\boldsymbol{\hat{%
\phi}}=\boldsymbol{M}_{T}(\boldsymbol{\bar{y}}-\boldsymbol{\bar{X}}%
\boldsymbol{\hat{\beta}}_{TMG-TE})$, we have 
\begin{equation}
\boldsymbol{\hat{\beta}}_{TMG-TE}-\boldsymbol{\beta }_{0}=(\boldsymbol{I}%
_{k^{\prime }}-\boldsymbol{\bar{Q}}_{n}^{\prime }\boldsymbol{M}_{T}%
\boldsymbol{\bar{X}})^{-1}\left[ \frac{1}{n(1+\bar{\delta}_{n})}%
\sum_{i=1}^{n}\boldsymbol{Q}_{i}^{\prime }\boldsymbol{M}_{T}\boldsymbol{%
\tilde{\nu}}_{i}\right] ,  \label{TMG-TE}
\end{equation}%
where $\delta _{i}$ is given by (\ref{deltai}), $\boldsymbol{Q}%
_{i}=(1+\delta _{i})\boldsymbol{M}_{T}\boldsymbol{X}_{i}(\boldsymbol{X}%
_{i}^{\prime }\boldsymbol{M}_{T}\boldsymbol{X}_{i})^{-1}$ given by (\ref{Qi}%
), and $\boldsymbol{\bar{Q}}_{n}=n^{-1}(1+\bar{\delta}_{n})^{-1}%
\sum_{i=1}^{n}\boldsymbol{Q}_{i}$ given by (\ref{Qn}). Using (\ref{TWFE})
and under Assumption \ref{PoolB}%
\begin{equation}
\boldsymbol{\hat{\beta}}_{TWFE}-\boldsymbol{\beta }_{0}=\boldsymbol{\bar{\Psi%
}}_{n,TE}^{-1}\left[ \frac{1}{n}\sum_{i=1}^{n}\left( \boldsymbol{X}_{i}-%
\boldsymbol{\bar{X}}\right) ^{\prime }\boldsymbol{M}_{T}\boldsymbol{\tilde{%
\nu}}_{i}\right] +o_{p}(1),  \label{TWFE2}
\end{equation}%
which in conjunction with (\ref{TMG-TE}) yields 
\begin{equation*}
\boldsymbol{\hat{\Delta}}_{\beta ,TE}=\boldsymbol{\hat{\beta}}_{TWFE}-%
\boldsymbol{\hat{\beta}}_{TMG-TE}=n^{-1}\sum_{i=1}^{n}\boldsymbol{G}%
_{i,TE}^{\prime }\boldsymbol{M}_{T}\boldsymbol{\tilde{\nu}}_{i},
\end{equation*}%
where $\boldsymbol{G}_{i,TE}$ is a $T\times k^{\prime }$ matrix given by 
\begin{equation}
\boldsymbol{G}_{i,TE}=\left( \boldsymbol{X}_{i}-\boldsymbol{\bar{X}}\right) 
\boldsymbol{\bar{\Psi}}_{n,TE}^{-1}-(1+\bar{\delta}_{n})^{-1}\boldsymbol{Q}%
_{i}\left[ \left( \boldsymbol{I}_{k^{\prime }}-\boldsymbol{\bar{Q}}%
_{n}^{\prime }\boldsymbol{M}_{T}\boldsymbol{\bar{X}}\right) ^{-1}\right]
^{\prime }.  \label{GiTE}
\end{equation}%
Under Assumption \ref{errorsx} and the null hypothesis given by (\ref{null}%
), $E(\tilde{\nu}_{it}|\boldsymbol{G}_{i,TE})=0$ for all $i$ and $t$, where $%
\tilde{\nu}_{it}=\nu _{it}-\bar{\nu}_{i\circ }$, and $\bar{\nu}_{i\circ } =
T^{-1} \sum_{t=1}^{T}\nu_{it}$. Also, by Assumptions \ref{errorsx} and \ref%
{CRE}, $u_{it}$ and $\boldsymbol{\eta }_{i}$ are cross-sectionally
independent so that $\tilde{\nu}_{it}$ conditional on $\boldsymbol{X}_{i}$
are also cross-sectionally independent. Then, following a similar line of
reasoning as in the proof Theorem \ref{AsyDTest}, for a fixed $T=k$, as $%
n\rightarrow \infty $, 
\begin{equation*}
\sqrt{n}\boldsymbol{\hat{\Delta}}_{\beta ,TE}\rightarrow _{d}N(\boldsymbol{0}%
,\boldsymbol{V}_{\Delta ,TE}),
\end{equation*}%
and 
\begin{equation}
\boldsymbol{V}_{\Delta ,TE}=\lim_{n\rightarrow \infty }\frac{1}{n}%
\sum_{i=1}^{n}E\left( \boldsymbol{G}_{i,TE}^{\prime }\boldsymbol{M}_{T}%
\boldsymbol{\tilde{\nu}}_{i}\boldsymbol{\tilde{\nu}}_{i}^{\prime }%
\boldsymbol{M}_{T}\boldsymbol{G}_{i,TE}\right) \succ \boldsymbol{0},
\label{Vdte}
\end{equation}%
where positive definiteness holds under conditions analogous to those in
Theorem \ref{AsyDTest}. The Hausman test statistic for panels with time
effects is given by 
\begin{equation}
H_{\beta ,TE}=n\boldsymbol{\hat{\Delta}}_{\beta ,TE}^{\prime }\boldsymbol{V}%
_{\Delta ,TE}^{-1}\boldsymbol{\hat{\Delta}}_{\beta ,TE}\text{, for }T=k.
\label{hte0}
\end{equation}%
As $n\rightarrow \infty $, $H_{\beta ,TE}\rightarrow _{d}\chi _{k^{\prime
}}^{2}$. For fixed $T$, $\boldsymbol{V}_{\Delta ,TE}$ can be consistently
estimated by 
\begin{equation}
\boldsymbol{\widehat{V}}_{\Delta ,TE}=\frac{1}{n}\sum_{i=1}^{n}\left( 
\boldsymbol{G}_{i,TE}^{\prime }\boldsymbol{M}_{T}\boldsymbol{\hat{\tilde{\nu}%
}}_{i,FE}\boldsymbol{\hat{\tilde{\nu}}}_{i,FE}^{\prime }\boldsymbol{M}_{T}%
\boldsymbol{G}_{i,TE}\right) ,  \label{varhte}
\end{equation}%
where 
\begin{equation}
\boldsymbol{\hat{\tilde{\nu}}}_{i,FE}=\widehat{\boldsymbol{\nu }_{i}-%
\boldsymbol{\bar{\nu}}}=(\boldsymbol{y}_{i}-\boldsymbol{\bar{y}})-(%
\boldsymbol{X}_{i}-\boldsymbol{\bar{X}})\boldsymbol{\hat{\beta}}_{TWFE},
\label{siTWFE}
\end{equation}
and $G_{i,TE}$ is defined by (\ref{GiTE}). Using the above estimator of $%
\boldsymbol{\widehat{V}}_{\Delta ,TE}$, a feasible statistic for testing the
null hypothesis given by (\ref{null}) in the case of panel regression models
with time effects and $T=k$ is given by 
\begin{equation}
\hat{H}_{\beta ,TE}=n\left( \boldsymbol{\hat{\beta}}_{TWFE}-\boldsymbol{\hat{%
\beta}}_{TMG-TE}\right) ^{\prime }\boldsymbol{\widehat{V}}_{\Delta
,TE}^{-1}\left( \boldsymbol{\hat{\beta}}_{TWFE}-\boldsymbol{\hat{\beta}}%
_{TMG-TE}\right) .  \label{htetest}
\end{equation}

\subsubsection{Panels with $T >k$}

In this case, we consider the TMG-TE estimator given by (\ref{BetaTE2}),
computed with $\boldsymbol{\hat{\phi}}_{C}$ given by (\ref{TE2}), and 
\begin{equation*}
\boldsymbol{\hat{\Delta}}_{\beta ,TE}=\boldsymbol{\hat{\beta}}_{TWFE}-%
\boldsymbol{\hat{\beta}}_{C,TMG-TE}\text{, for }T>k.
\end{equation*}%
Using (\ref{TE-Ca}) and noting that $\boldsymbol{M}_{i}\boldsymbol{M}_{T}%
\boldsymbol{X}_{i}=\boldsymbol{0}$, we have 
\begin{equation}
\boldsymbol{\hat{\phi}}_{C}-\boldsymbol{\phi }=\boldsymbol{\bar{M}}%
_{n}^{-1}\left( \frac{1}{n}\sum_{i=1}^{n}\boldsymbol{M}_{i}\boldsymbol{M}_{T}%
\boldsymbol{\nu }_{i}\right) ,  \label{phiC}
\end{equation}%
where $\boldsymbol{\nu }_{i}=\boldsymbol{X}_{i}\boldsymbol{\eta }_{i}+%
\boldsymbol{u}_{i}$, and $\boldsymbol{\bar{M}}_{n}=n^{-1}\sum_{i=1}^{n}%
\boldsymbol{M}_{i}$. Using (\ref{BetaTE2}), we have%
\begin{equation*}
\boldsymbol{\hat{\beta}}_{C,TMG-TE}=\frac{1}{1+\bar{\delta}_{n}}\left[
n^{-1}\sum_{i=1}^{n}\boldsymbol{Q}_{i}^{\prime }\boldsymbol{M}_{T}(%
\boldsymbol{y}_{i}-\boldsymbol{\hat{\phi}}_{C})\right] \text{, for }T>k,
\end{equation*}%
where $\boldsymbol{Q}_{i}$ is defined by (\ref{Qi}). Also, since $%
\boldsymbol{y}_{i}=\alpha _{i}\boldsymbol{\tau }_{T}+\boldsymbol{\phi }+%
\boldsymbol{X}_{i}\boldsymbol{\beta }_{0}+\boldsymbol{\nu }_{i}$, then
noting that $n^{-1}\sum_{i=1}^{n}\left( 1+\bar{\delta}_{n}\right) ^{-1}%
\boldsymbol{Q}_{i}^{\prime }\boldsymbol{M}_{T}\boldsymbol{X}_{i}=\boldsymbol{%
I}_{k^{\prime }}$, we have 
\begin{equation*}
\boldsymbol{\hat{\beta}}_{C,TMG-TE}-\boldsymbol{\beta }_{0}=\frac{1}{1+\bar{%
\delta}_{n}}\left[ n^{-1}\sum_{i=1}^{n}\boldsymbol{Q}_{i}^{\prime }%
\boldsymbol{M}_{T}(\boldsymbol{\nu }_{i}+\boldsymbol{\phi }-\boldsymbol{\hat{%
\phi}}_{C})\right] .
\end{equation*}%
Also using (\ref{phiC}), 
\begin{equation*}
\frac{1}{n(1+\bar{\delta}_{n})}\sum_{i=1}^{n}\boldsymbol{Q}_{i}^{\prime }%
\boldsymbol{M}_{T}\left( \boldsymbol{\hat{\phi}}_{C}-\boldsymbol{\phi }%
\right) =\boldsymbol{\bar{Q}}_{n}^{\prime }\boldsymbol{M}_{T}\boldsymbol{%
\bar{M}}_{n}^{-1}\left( \frac{1}{n}\sum_{i=1}^{n}\boldsymbol{M}_{i}%
\boldsymbol{M}_{T}\boldsymbol{\nu }_{i}\right) ,
\end{equation*}%
where $\boldsymbol{\bar{Q}}_{n}$ is given by (\ref{Qn}). Hence 
\begin{align*}
\boldsymbol{\hat{\beta}}_{C,TMG-TE}-\boldsymbol{\beta }_{0}& =\frac{1}{1+%
\bar{\delta}_{n}}n^{-1}\sum_{i=1}^{n}\boldsymbol{Q}_{i}^{\prime }\boldsymbol{%
M}_{T}\boldsymbol{\nu }_{i}-\boldsymbol{\bar{Q}}_{n}^{\prime }\boldsymbol{M}%
_{T}\boldsymbol{\bar{M}}_{n}^{-1}\left( \frac{1}{n}\sum_{i=1}^{n}\boldsymbol{%
M}_{i}\boldsymbol{M}_{T}\boldsymbol{\nu }_{i}\right) \\
& =n^{-1}\sum_{i=1}^{n}\left[ (1+\bar{\delta}_{n})^{-1}\boldsymbol{Q}%
_{i}^{\prime }-\boldsymbol{\bar{Q}}_{n}^{\prime }\boldsymbol{M}_{T}%
\boldsymbol{\bar{M}}_{n}^{-1}\boldsymbol{M}_{i}\right] \boldsymbol{M}_{T}%
\boldsymbol{\nu }_{i},
\end{align*}%
or equivalently in terms of $\boldsymbol{\tilde{\nu}}_{i}= (\tilde{\nu}%
_{i1}, \tilde{\nu}_{i2}, ..., \tilde{\nu}_{iT})^{\prime}$, where $\tilde{\nu}%
_{it}=\nu _{it}-\bar{\nu}_{i\circ }$, and $\bar{\nu}_{i\circ } =
T^{-1}\sum_{t=1}^{T}\nu_{it}$, 
\begin{equation*}
\boldsymbol{\hat{\beta}}_{C,TMG-TE}-\boldsymbol{\beta }_{0}=n^{-1}%
\sum_{i=1}^{n}\left[ (1+\bar{\delta}_{n})^{-1}\boldsymbol{Q}_{i}^{\prime }-%
\boldsymbol{\bar{Q}}_{n}^{\prime }\boldsymbol{M}_{T}\boldsymbol{\bar{M}}%
_{n}^{-1}\boldsymbol{M}_{i}\right] \boldsymbol{M}_{T}\boldsymbol{\tilde{\nu}}%
_{i},
\end{equation*}%
since $\frac{1}{n}\sum_{i=1}^{n}\left[ (1+\bar{\delta}_{n})^{-1}\boldsymbol{Q%
}_{i}^{\prime }-\boldsymbol{\bar{Q}}_{n}^{\prime }\boldsymbol{M}_{T}%
\boldsymbol{\bar{M}}_{n}^{-1}\boldsymbol{M}_{i}\right] \boldsymbol{M}_{T}%
\boldsymbol{\bar{\nu}}=\boldsymbol{0}$ given $\frac{1}{n}\sum_{i=1}^{n}%
\boldsymbol{\bar{M}}_{n}^{-1}\boldsymbol{M}_{i}=\boldsymbol{I}_{T}$. Using
this result together with (\ref{TWFE2}), we now have for $T>k$, 
\begin{equation*}
\boldsymbol{\hat{\Delta}}_{\beta ,TE}=\frac{1}{n}\sum_{i=1}^{n}\boldsymbol{G}%
_{i,C}^{\prime }\boldsymbol{M}_{T}\boldsymbol{\tilde{\nu}}_{i},
\end{equation*}%
where 
\begin{equation}
\boldsymbol{G}_{i,C}=\left( \boldsymbol{X}_{i}-\boldsymbol{\bar{X}}\right) 
\boldsymbol{\bar{\Psi}}_{n,TE}^{-1}-\left[ (1+\bar{\delta}_{n})^{-1}%
\boldsymbol{Q}_{i}-\boldsymbol{M}_{i}\boldsymbol{\bar{M}}_{n}^{-1}%
\boldsymbol{M}_{T}\boldsymbol{\bar{Q}}_{n}\right].  \label{GiC}
\end{equation}%
Under the null hypothesis given by (\ref{null}), we have $E(\tilde{\nu}_{it}|%
\boldsymbol{G}_{i,C})=0$, for all $i$ and $t$. Also, by Assumptions \ref%
{errorsx} and \ref{rcm}(a), $u_{it}$ and $\boldsymbol{\eta }_{i}$ are
cross-sectionally independent so that $\tilde{\nu}_{it}$ conditional on $%
\boldsymbol{X}_{i}$ are also cross-sectionally independent. Then, following
a similar line of reasoning as in the proof Theorem \ref{AsyDTest}, for a
fixed $T>k$, as $n\rightarrow \infty $, $\sqrt{n}\boldsymbol{\hat{\Delta}}%
_{\beta ,TE}\rightarrow _{d}N(\boldsymbol{0},\boldsymbol{V}_{\Delta ,C})$,
with 
\begin{equation}
\boldsymbol{V}_{\Delta ,C}=\lim_{n\rightarrow \infty }\frac{1}{n}%
\sum_{i=1}^{n}E\left( \boldsymbol{G}_{i,C}^{\prime }\boldsymbol{M}_{T}%
\boldsymbol{\tilde{\nu}}_{i}\boldsymbol{\tilde{\nu}}_{i}^{\prime }%
\boldsymbol{M}_{T}\boldsymbol{G}_{i,C}\right) \succ \boldsymbol{0},
\label{Vdc}
\end{equation}%
where positive definiteness holds under conditions analogous to those in
Theorem \ref{AsyDTest}.

Thus, when $T>k$, the Hausman test statistic for panels with time effects is
given by 
\begin{equation}
H_{\beta ,TE}=n\boldsymbol{\hat{\Delta}}_{\beta ,TE}^{\prime }\boldsymbol{V}%
_{\Delta ,C}^{-1}\boldsymbol{\hat{\Delta}}_{\beta ,TE},  \label{hte2}
\end{equation}%
and as $n\rightarrow \infty $, $H_{\beta ,TE}\rightarrow _{d}\chi
_{k^{\prime }}^{2}$. A consistent estimator of $\boldsymbol{V}_{\Delta ,C}$
for a fixed $T$ is given by 
\begin{equation}
\boldsymbol{\widehat{V}}_{\Delta ,C}=\frac{1}{n}\sum_{i=1}^{n}\left( 
\boldsymbol{G}_{i,C}^{\prime }\boldsymbol{M}_{T}\boldsymbol{\hat{\tilde{\nu}}%
}_{i,FE}\boldsymbol{\hat{\tilde{\nu}}}_{i,FE}^{\prime }\boldsymbol{M}_{T}%
\boldsymbol{G}_{i,C}\right) ,  \label{varhc}
\end{equation}%
where $\boldsymbol{G}_{i,C}$ and $\boldsymbol{\hat{\tilde{\nu}}}_{i,FE}$ are
given by (\ref{GiC}) and (\ref{siTWFE}), respectively. Then the test
statistics given by (\ref{hte2}) for panel regressions with time effects and 
$T>k$ can be consistently estimated by 
\begin{equation}
\hat{H}_{\beta ,TE}=n\left( \boldsymbol{\hat{\beta}}_{TWFE}-\boldsymbol{\hat{%
\beta}}_{C,TMG-TE}\right) ^{\prime }\boldsymbol{\widehat{V}}_{\Delta
,C}^{-1}\left( \boldsymbol{\hat{\beta}}_{TWFE}-\boldsymbol{\hat{\beta}}%
_{C,TMG-TE}\right) .  \label{htetest2}
\end{equation}

\subsection{The mean group estimator with time effects}

The mean group estimator of $\boldsymbol{\beta }_{0}=E\left( \boldsymbol{%
\beta }_{i}\right) $ in panel data models with time effects given by (\ref%
{panTE1}) can also be used when $T$ is sufficiently large such that the
underlying individual estimators have second-order moments. However, as
discussed in Pesaran and Yang (2026), the standard results established in
the literature for MG estimators do not apply to the mean group estimator of 
$\boldsymbol{\beta }_{0}$ in the case of panels with time effects (MG-TE).
In this more general setting, the MG-TE estimator is given by 
\begin{equation}
\boldsymbol{\hat{\beta}}_{MG-TE}=n^{-1}\sum_{i=1}^{n}\boldsymbol{\hat{\beta}}%
_{i,TE},  \label{MGTE}
\end{equation}%
where 
\begin{equation}
\boldsymbol{\hat{\beta}}_{i,TE}=(\boldsymbol{X}_{i}^{\prime }\boldsymbol{M}%
_{T}\boldsymbol{X}_{i})^{-1}\boldsymbol{X}_{i}^{\prime }\boldsymbol{M}_{T}(%
\boldsymbol{y}_{i}-\boldsymbol{\hat{\phi}}_{C}),  \label{biTE}
\end{equation}%
$\boldsymbol{\hat{\phi}}_{C}=\boldsymbol{\bar{M}}_{n}^{-1}\left( \frac{1}{n}%
\sum_{i=1}^{n}\boldsymbol{M}_{i}\boldsymbol{M}_{T}\boldsymbol{y}_{i}\right) $%
, $\boldsymbol{\bar{M}}_{n}=\frac{1}{n}\sum_{i=1}^{n}\boldsymbol{M}_{i}$, $%
\boldsymbol{M}_{i}=\boldsymbol{I}_{T}-\boldsymbol{M}_{T}\boldsymbol{X}_{i}(%
\boldsymbol{X}_{i}^{\prime }\boldsymbol{M}_{T}\boldsymbol{X}_{i})^{-1}%
\boldsymbol{X}_{i}^{\prime }\boldsymbol{M}_{T}$, $\boldsymbol{I}_{T}$ is a $%
T\times T$ identity matrix, and $\boldsymbol{M}_{T}=\boldsymbol{I}_{T}-%
\boldsymbol{\tau }_{T}\boldsymbol{\tau }_{T}^{\prime }/T$. In this case, the
derivation and estimation of the asymptotic variance of $\boldsymbol{\hat{%
\beta}}_{MG-TE}$ is complicated due to the sampling uncertainty associated
with the estimated time effects, $\boldsymbol{\hat{\phi}}_{C}$. Pesaran and
Yang (2026) propose the following estimator of $Var(\sqrt{n}\boldsymbol{\hat{%
\beta}}_{MG-TE})$ 
\begin{align}
\widehat{Var(\sqrt{n}\boldsymbol{\hat{\beta}}}_{MG-TE})=& \frac{1}{\left(
n-1\right) }\sum_{i=1}^{n}\left( \boldsymbol{\hat{\beta}}_{i,TE}-\boldsymbol{%
\hat{\beta}}_{MG-TE}\right) \left( \boldsymbol{\hat{\beta}}_{i,TE}-%
\boldsymbol{\hat{\beta}}_{MG-TE}\right) ^{\prime }  \notag \\
& +\boldsymbol{\bar{R}}_{n}^{\prime }\boldsymbol{\hat{A}}_{n}\boldsymbol{%
\bar{R}}_{n}-\left( \boldsymbol{\hat{B}}_{n}^{\prime }\boldsymbol{\bar{R}}%
_{n}+\boldsymbol{\bar{R}}_{n}^{\prime }\boldsymbol{\hat{B}}_{n}\right) ,
\label{Avarhat}
\end{align}%
where 
\begin{align*}
\boldsymbol{\bar{R}}_{n}^{\prime } =&n^{-1}\sum_{i=1}^{n}\boldsymbol{R}%
_{i}^{\prime },\text{ \ }\boldsymbol{R}_{i}^{\prime }=(\boldsymbol{X}%
_{i}^{\prime }\boldsymbol{M}_{T}\boldsymbol{X}_{i})^{-1}\boldsymbol{X}%
_{i}^{\prime }\boldsymbol{M}_{T}, \\
\boldsymbol{\hat{A}}_{n} =&\frac{1}{n}\sum_{i=1}^{n}\boldsymbol{S}%
_{i}^{\prime }\boldsymbol{\hat{v}_{i}}\boldsymbol{\hat{v}_{i}}^{\prime }%
\boldsymbol{S}_{i},\text{ }\boldsymbol{\hat{B}}_{n}=\frac{1}{n}\sum_{i=1}^{n}%
\boldsymbol{S}_{i}^{\prime }\boldsymbol{\hat{v}_{i}}(\boldsymbol{\hat{\beta}}%
_{i,TE}-\boldsymbol{\hat{\beta}}_{MG-TE})^{\prime },
\end{align*}%
$\boldsymbol{S}_{i}^{\prime }=\boldsymbol{\bar{M}}_{n}^{-1}\boldsymbol{M}_{i}%
\boldsymbol{M}_{T}$, and $\boldsymbol{\hat{v}}_{i}=\boldsymbol{y}_{i}-%
\boldsymbol{\hat{\phi}}_{C}$. \ It is further established that 
\begin{equation*}
E\left[ \widehat{Var(\sqrt{n}\boldsymbol{\hat{\beta}}}_{MG-TE})\right] -Var(%
\sqrt{n}\boldsymbol{\hat{\beta}}_{MG-TE})=O\left( n^{-1}\right) .
\end{equation*}%
Namely, the proposed estimator is asymptotically unbiased.

Following a similar line, the variance of the TMG-TE estimator, $\boldsymbol{%
\hat{\beta}}_{C,TMG-TE}$, given by (\ref{BetaTE2}), can be estimated by 
\begin{align}
\widehat{Var\left( \sqrt{n}\boldsymbol{\hat{\beta}}_{C,TMG-TE}\right) }=& 
\frac{1}{(n-1)(1+\bar{\delta}_{n})^{2}}\sum_{i=1}^{n}\left( \boldsymbol{%
\tilde{\beta}}_{i,C}-\boldsymbol{\hat{\beta}}_{C,TMG-TE}\right) \left( 
\boldsymbol{\tilde{\beta}}_{i,C}-\boldsymbol{\hat{\beta}}_{C,TMG-TE}\right)
^{\prime }  \notag \\
& +\boldsymbol{\bar{Q}}_{n}^{\prime }\boldsymbol{\hat{A}}_{n}\boldsymbol{%
\bar{Q}}_{n}-\left( \boldsymbol{\hat{B}}_{n}^{\prime }\boldsymbol{\bar{Q}}%
_{n}+\boldsymbol{\bar{Q}}_{n}^{\prime }\boldsymbol{\hat{B}}_{n}\right) ,
\label{VarbetaC2}
\end{align}%
where $\boldsymbol{\tilde{\beta}}_{i,C}=\boldsymbol{Q}_{i}^{\prime }(%
\boldsymbol{y}_{i}-\boldsymbol{\hat{\phi}}_{C})$, $\boldsymbol{\hat{B}}%
_{n}=n^{-1}(1+\bar{\delta}_{n})^{-1}\sum_{i=1}^{n}\boldsymbol{S}_{i}^{\prime
}\boldsymbol{\hat{v}_{i}}(\boldsymbol{\tilde{\beta}}_{i,C}-\boldsymbol{\hat{%
\beta}}_{C,TMG-TE})^{\prime }$, $\boldsymbol{Q}_{i}^{\prime }=(1+\delta _{i})%
\boldsymbol{R}_{i}^{\prime }$, $\boldsymbol{\bar{Q}}_{n}^{\prime }=n^{-1}(1+%
\bar{\delta}_{n})^{-1}\sum_{i=1}^{n}\boldsymbol{Q}_{i}^{\prime }$, $\bar{%
\delta}_{n}=n^{-1}\sum_{i=1}^{n}\delta _{i}$, $\delta _{i}=\left( \frac{%
d_{i}-a_{n}}{a_{n}}\right) \boldsymbol{1}\{d_{i}\leq a_{n}\}\leq 0$, and $%
a_{n}=Cn^{-\alpha }$.

\newpage\clearpage

\section{Monte Carlo supplement}

\label{mcsup}

\subsection{Design of Monte Carlo experiments}

\label{secMCDGP}

The DGP for $y_{it}$ and $x_{it}$ and the baseline experiments are described
in Section \ref{DGP} of the main paper. Section \ref{dgprobust} describes
the experimental design used for robustness analysis. Section \ref{Para}
describes different DGPs considered in the MC experiments, with their key
parameters summarized in Table \ref{tab:paramcs}. Section \ref{Simuk}
describes how the value of $\kappa _{T}$ in the $y_{it}$ process has been
calibrated by stochastic simulations to achieve a given level of overall
fit, $PR^{2}$ for a given value of $T$, when $k^{\prime}=1$.

\subsubsection{Data generating processes for robustness checks \label%
{dgprobust}}

To check the robustness of the TMG estimator, the following variations in
the DGP of the errors and regressors are considered. When the errors in $%
y_{it}$ are serially correlated, we generate $\rho _{ie}\sim IIDU(0,0.95)$
and $e_{i0}\sim IIDN(0,1)$ for all $i$. When the regressors $\left\{
x_{j,it}\right\} $, for $j=1,2,...,k^{\prime }$ are autocorrelated, the
associated coefficients are generated as $\rho _{j,ix}\sim IIDU(0,0.95)$.
When there is an interactive effect in $\left\{ x_{j,it}\right\} $, we
generate the factor loadings as $\gamma _{j,ix}\sim IIDU(0,2)$, and the
factors as $f_{j,t}=0.9f_{j,t-1}+(1-0.9^{2})^{1/2}v_{j,t}$, for $%
t=-49,-48,...,-1,0,1,...,T$, where $v_{j,t}\sim IIDN(0,1)$, with $%
f_{j,-50}=0 $.

To examine the relative efficiency of TMG and FE estimators, we set $\psi
_{\beta _{1}}=0$ (uncorrelated heterogeneity) in equation (\ref{eta_i}) of
the main paper, but allow error heteroskedasticity to be correlated with the
processes generating $x_{it}$. We consider the following two scenarios: (a)
cross-sectional heteroskedasticity: 
\begin{equation*}
\sigma _{it}^{2}=\lambda _{i}^{2}=\left[ \frac{\boldsymbol{e}_{1,ix}^{\prime
}\boldsymbol{M}_{T}\boldsymbol{e}_{1,ix}-E\left( \boldsymbol{e}%
_{1,ix}^{\prime }\boldsymbol{M}_{T}\boldsymbol{e}_{1,ix}\right) }{\sqrt{%
Var\left( \boldsymbol{e}_{1,ix}^{\prime }\boldsymbol{M}_{T}\boldsymbol{e}%
_{1,ix}\right) }}\right] ^{2},\text{ for all }i\text{ and }t,
\end{equation*}%
where $\boldsymbol{e}_{1,ix}=(e_{x1,i1},e_{x1,i2},...,e_{x1,iT})^{\prime
}\sim IID(\boldsymbol{0},\boldsymbol{I}_{T})$, and it follows that $\lambda
_{i}$ is $IID(0,1)$. And (b) cross-sectional and time series
heteroskedasticity, $\sigma _{it}^{2}=e_{x1,it}^{2}$, for all $i$ and $t$,
where $e_{x1,it}$ is the innovation to the $x_{1,it}$ process. In both
cases, we have $E(\sigma _{it}^{2})=1$, which matches the case of randomly
generated heteroskedasticity.

\subsubsection{Calibration of the parameters with one regressor\label{Para}}

\begin{enumerate}
\item Generation of $y_{it}$ and $x_{j,it}$ , for $i=1,2,...,n$, $%
t=1,2,...,T $, and $j=1,2,...,k^{\prime}$

\begin{enumerate}
\item $x_{j,it}$ are generated with $\rho _{j,ix}=0$ for all $i$ (in the
static case) and as heterogeneous AR(1) processes with $\rho _{j,ix}\sim
IIDU(0,0.95)$ for the dynamic case. See (\ref{xdgp}) in the main paper. The
errors $e_{xj,it}$ of the $x_{j,it}$ equation are generated with $%
E(e_{xj,it})=0$ and $E(e_{xj,it}^{2})=1$ according to the following two
distributions:

\begin{enumerate}
\item Gaussian with $e_{xj,it}\sim IIDN(0,1)$,

\item Uniform distribution with $e_{xj,it}=\sqrt{12}(\mathfrak{z}%
_{j,it}-1/2) $ and $\mathfrak{z}_{j,it}$ $\sim IIDU(0,1)$.
\end{enumerate}

\item $\alpha _{j,ix}\sim IIDN(1,1)$, and $\sigma _{j,ix}^{2}\sim IID\frac{1%
}{2}(z_{j,ix}^{2}+1)$, with $z_{j,ix}\sim IIDN(0,1)$.

\item The errors in the $y_{it}$ equation are composed of three components, $%
\kappa \sigma _{it}e_{it}$. See (\ref{ydgp}) in the main paper. $e_{it}$ are
generated as heterogeneous AR(1) processes given by (\ref{yerror}) in the
main paper, with $\rho _{ie}=0$ for all $i$ (serially uncorrelated case) and 
$\rho _{ie}\sim IIDU(0,0.95)$ (serially correlated case). The innovations to
the $y_{it}$, $\varsigma _{it}$, are generated as $\varsigma _{it}\sim
IIDN(0,1)$, or $IID\frac{1}{2}(\chi _{2}^{2}-2)$. $\sigma _{it}^{2}$ are
generated based on different cases described in Section \ref{DGP} in the
main paper and Section \ref{dgprobust} with $E(\sigma _{it}^{2})=1$. The
scalar, $\kappa $, is calibrated for each $T$ to achieve a given level of
fit, $PR^{2}\in \{0.2,0.4\}$, See sub-section \ref{Simuk} below.
\end{enumerate}

\item Generation of heterogeneous coefficients, $(\alpha _{i}, \beta _{i1},
\beta_{i2}, ...,\beta_{ik^{\prime}})^{\prime }$ for $i=1,2,...,n$.

\begin{enumerate}
\item $(\alpha _{i},\beta _{i1})^{\prime }$ are generated using (\ref%
{coefdgp}) in sub-section \ref{DGP} of the main paper, with $\alpha
_{0}=E(\alpha _{i})=1$ and $\beta _{01}=E(\beta _{i1})=1$.

\item $\sigma _{\alpha }^{2}=0.5$ and $\sigma _{\beta_{1}}^{2} = 0.75$.

\item $\psi_{\alpha} = 0.5$ and $\psi_{\beta_{1}} \in \{0, 0.5, 0.8\}$.

\item $\epsilon _{i\alpha }\sim IID(0,\sigma _{\epsilon \alpha }^{2})$ and $%
\epsilon _{i\beta_{1}}\sim IID(0,\sigma _{\epsilon \beta_{1} }^{2})$, where $%
\sigma _{\epsilon \alpha }^{2}=\sigma _{\alpha }^{2} - \psi_{\alpha}^{2} =
0.25$ and $\sigma _{\epsilon \beta }^{2}=\sigma _{\beta_{1}}^{2} -
\psi_{\beta_{1}}^{2} \in \{0.75, 0.5, 0.11\}$.

\item $\beta _{ij}=\beta _{0j}+\epsilon _{i\beta _{j}}$ with $\epsilon
_{i\beta _{j}}\sim IIDN (0,\sigma _{\epsilon \beta _{j}}^{2})$, for $%
j=2,3,...,k^{\prime}$.
\end{enumerate}
\end{enumerate}

\begin{table}[h!]
\caption{Summary of key parameters in the Monte Carlo experiments with one
regressor}
\label{tab:paramcs}\vspace{-6mm}
\par
\begin{center}
\scalebox{0.9}{
\renewcommand{\arraystretch}{1.05}
\begin{tabular}{lcr@{}lr@{}lr@{}lr@{}lr@{}l}
\hline\hline
Case &  & \multicolumn{2}{l}{(1)} & \multicolumn{2}{l}{(2)} & \multicolumn{2}{l}{(3)} & \multicolumn{2}{l}{(4)} & \multicolumn{2}{l}{(5)} \\ \hline
$E(\alpha_{i})$ &  & 1 &  & 1 &  & 1 &  & 1 & & 1 &  \\
$E(\beta_{i1})$ &  & 1 &  & 1 &  & 1 &  & 1 & & 1 &  \\
$E(\sigma_{it}^2)$ &  & 1 &  & 1 &  & 1 &  & 1 & & 1 &  \\
$E(\sigma_{1,ix}^2)$ &  & 1 &  & 1 &  & 1 &  & 1 &  & 1 &  \\
$PR^2$ &  & 0 & .2 & 0 & .2 & 0 & .2 & 0 & .4 & 0 & .2 \\
$\sigma_{\alpha}^2 $ &  & 0 & .5 & 0 & .5 & 0 & .5 & 0 & .5 & 0 & .5 \\
$\sigma_{\beta_{1}}^2 $ &  & 0 & .75 & 0 & .75 & 0 & .75 & 0 & .75 & 0 &  \\
$\psi_{ \alpha}$ &  & 0 & .5 & 0 & .5 & 0 & .5 & 0 & .5 & 0 & .5 \\
$\psi_{ \beta_{1}}$ &  & 0 &  & 0 & .5 & 0 & .8 & 0 & .5 & 0 & \\
$\sigma_{\epsilon\alpha}^2$ &  & 0 & .25 & 0 & .25 & 0 & .25 & 0 & .25 & 0 & .25 \\
$\sigma_{\epsilon\beta_{1}}^2 $ &  & 0 & .75 & 0 & .5 & 0 & .11 & 0 & .5 & 0 &  \\
$Cov(\alpha_{i}, \beta_{i1})$ &  & 0 &  & 0 & .25 & 0 & .4 & 0 & .25 & 0 & 
 \\\hline\hline
\end{tabular}}
\end{center}
\par
\vspace{-2mm} 
\begin{spacing}{1}
{\footnotesize 
Notes: The values of key parameters under columns (1), (2), and (3) are
set according to the description in Section \ref{DGP} with
zero, medium and high degrees of correlated heterogeneity ($\psi_{\beta_{1} }$ defined in (\ref{eta_i}) in the main paper) with $PR^{2} = 0.2$ and one regressor. Column (4) corresponds to the case of medium correlated heterogeneity with $PR^{2} = 0.4$ and one regressor. Column (5) corresponds to the case of homogeneous slope coefficients. 
For further details see Section \ref{DGP} of the main paper.}
\end{spacing}
\end{table}

\subsubsection{Calibration of $\protect\kappa^{2}$ by stochastic simulation
with one regressor \label{Simuk}}

The scaling parameter $\kappa $ in (\ref{ydgp}) in the main paper is set to
achieve a given level of fit as measured by the pooled $R^{2}$ ($PR^{2}$)
given by 
\begin{eqnarray}
PR^{2} &=&\lim_{n\rightarrow \infty }PR_{n}^{2}=1-\frac{\lim_{n\rightarrow
\infty }n^{-1}T^{-1}\sum_{i=1}^{n}\sum_{t=1}^{T}Var(u_{it})}{%
\lim_{n\rightarrow \infty
}n^{-1}T^{-1}\sum_{i=1}^{n}\sum_{t=1}^{T}Var(y_{it}-\alpha _{i}-\phi _{t})} 
\notag \\
&=&1-\frac{\kappa ^{2}}{\lim_{n\rightarrow \infty
}n^{-1}T^{-1}\sum_{i=1}^{n}\sum_{t=1}^{T}Var(\beta _{i1}x_{1,it})+\kappa ^{2}%
}.  \label{PR2}
\end{eqnarray}%
Since $T$ is fixed, the value of $\kappa $ in general depends on $T$, and we
have (noting that $Var(\beta _{i1}x_{1,it})=E(\beta _{1i}^{2}x_{1,it}^{2})-%
\left[ E(\beta _{i1}x_{1,it})\right] ^{2}$) 
\begin{equation}
\kappa _{T}^{2}=\left( \frac{1-PR^{2}}{PR^{2}}\right) \lim_{n\rightarrow
\infty }\frac{1}{nT}\sum_{i=1}^{n}\sum_{t=1}^{T}\left\{ E(\beta
_{i1}^{2}x_{1,it}^{2})-\left[ E(\beta _{i1}x_{1,it})\right] ^{2}\right\} .
\label{s-kappa2T}
\end{equation}%
Due to the non-linear dependence of $\beta _{i1}$ on $x_{1,it}$ (through $%
\sigma _{1,ix}^{2}$ ) we use stochastic simulations to compute $E(\beta
_{i1}^{2}x_{1,it}^{2})$ and $E(\beta _{i1}x_{1,it})$, which can be carried
out in a straightforward manner since the values of $x_{1,it}$ and $\beta
_{i1}$ do not depend on $\kappa $ and can be jointly simulated using the
equations (\ref{xdgp}) and (\ref{coefdgp}) in the main paper.

The total number of simulations is $R_{\kappa }=1,000$ with $n=5,000$ and $%
T=2,3,...,15$. For each replication $r=1,2,...,R_{\kappa}$, we generate a
new sample of $\{\beta _{i1}^{(r)}\}$ and $\{x_{1,it}^{(r)}\}$ given the DGP
set up in our paper. The random variables that are drawn independently
across replications are denoted with a superscript $(r)$. The random
variables that are drawn once and used for all replications are denoted
without a superscript $(r)$.

\begin{enumerate}
\item Generate $x_{1,it}^{(r)}$:

\begin{enumerate}
\item First generate $e_{x1,it}^{(r)}$ as $IID(0,1)$ according to the two
distributions specified in Section \ref{Para}, namely Gaussian or uniform
distributions, and generate $(\sigma _{1,ix}^{2})^{(r)}$ as $IID \frac{1}{2}%
\left[\left(z_{1,ix}^{(r)}\right)^{2}+1\right]$, where $z_{1,ix}^{(r)} $ are
generated as $IIDN(0,1)$.

\item Generate $\rho_{1,ix}^{(r)} = 0\, \forall i$ for static $x_{1,it}$, or 
$\rho_{1,ix}^{(r)} \sim IIDU(0,0.95)$ for dynamic $x_{1,it}$. Then generate $%
\alpha_{1,ix}^{(r)}$ as $IIDN(1,1)$, and $x_{1,it}^{(r)}$ iteratively for
the dynamic case for $t=-49,-48,... ,-1,0,1,...,T$ 
\begin{equation*}
x_{1,it}^{(r)}=\alpha_{1,ix}^{(r)}\left(1-\rho_{1,ix}^{(r)}\right) +
\gamma_{1,ix}^{(r)} f_{1,t} +\rho_{1,ix}^{(r)}x_{1,i,t-1}^{(r)}+\left[%
1-\left(\rho_{1,ix}^{(r)}\right)^{2}\right]^{1/2}\sigma
_{1,ix}^{(r)}e_{x1,it}^{(r)},
\end{equation*}%
without or with interactive effects, $\gamma_{1,ix}^{(r)} \sim IIDU(0,2)$, $%
f_{1,t}=0.9 f_{1,t-1} + (1-0.9^{2})^{1/2} v_{1,t}$, and $v_{1,t} \sim
IIDN(0,1)$, where $x_{1,i,-50} =0$ and $f_{1,-50}=0$. The first $50$
observations are dropped, and $\{x_{1,i1},x_{1,i2},...,x_{1,iT}\}$ are used
in the simulations.
\end{enumerate}

\item Generate $\beta_{i1}^{(r)}$

\begin{enumerate}
\item Generate $\epsilon _{i\beta_{1} }^{(r)}$ as $IIDN(0, \sigma _{\epsilon
\beta_{1} }^{2})$ where $\sigma _{\epsilon \beta_{1} }^{2} =
\sigma_{\beta_{1}}^{2} - \psi_{\beta_{1}}^{2}$.

\item Given $\psi _{\beta_{1}}$, $\sigma _{1,ix}^{(r)}$ and $\epsilon
_{i\beta_{1} }^{(r)}$, generate $\eta _{i \beta_{1}}^{(r)}$ given by (\ref%
{eta_i}). Then $\beta_{i1}^{(r)}=\beta_{01}+\eta _{i \beta_{1}}^{(r)}$.
\end{enumerate}

\item Given $\beta _{i1}^{(r)}$ and $x_{1,it}^{(r)}$, we then simulate 
\begin{equation*}
A_{RT}=R_{\kappa }^{-1}T^{-1}n^{-1}\sum_{r=1}^{R_{\kappa
}}\sum_{i=1}^{n}\sum_{t=1}^{T}\left( \beta _{i1}^{(r)}\right) ^{2}\left(
x_{1,it}^{(r)}\right) ^{2},
\end{equation*}%
\begin{equation*}
B_{RT}=R_{\kappa }^{-1}T^{-1}n^{-1}\sum_{r=1}^{R_{\kappa
}}\sum_{i=1}^{n}\sum_{t=1}^{T}\beta _{i1}^{(r)}x_{1,it}^{(r)},
\end{equation*}%
and 
\begin{equation*}
\widehat{Var(\beta _{i1}x_{1,it})}_{RT}=A_{RT}-B_{RT}^{2}.
\end{equation*}%
Then, for given values of $PR^{2}=0.2\text{ or }0.4$, compute $\kappa
_{T}^{2} $ as 
\begin{equation*}
\kappa _{T}^{2}=\left( \frac{1-PR^{2}}{PR^{2}}\right) \widehat{Var(\beta
_{i1}x_{1,it})}_{RT}.
\end{equation*}
\end{enumerate}

The simulated values of $\kappa _{T}^{2}$ for different DGPs are reported in
Table \ref{tab:simuk2}. 
\begin{table}[!t]
\caption{Simulated values of $\protect\kappa_{T}^{2}$ for $%
T=2,3,4,5,6,8,10,15$ with one regressor}
\label{tab:simuk2}
\begin{center}
\vspace{-6mm} 
\scalebox{0.9}{
\begin{tabular}{rrrrrrrrrrrrr}
\hline\hline
 &  & \multicolumn{5}{c}{Gaussian} &  & \multicolumn{5}{c}{Uniform} \\ \cline{3-7} \cline{9-13}
 $T$ &  & (1) & (2) & (3) & (4) & (5) &  & (1) & (2) & (3) & (4) & (5) \\\hline
 &  & \multicolumn{11}{l}{There are no autoregressions or interactive effects in the $x_{1,it}$ equation.} \\\hline
2 &  & 17.03 & 21.87 & 26.67 & 8.20 & 8.01 &  & 17.02 & 21.87 & 26.72 & 8.20 & 8.01 \\
3 &  & 17.00 & 21.89 & 26.78 & 8.21 & 8.00 &  & 17.00 & 21.84 & 26.67 & 8.19 & 8.00 \\
4 &  & 17.00 & 21.84 & 26.68 & 8.19 & 8.00 &  & 17.01 & 21.87 & 26.70 & 8.20 & 8.00 \\
5 &  & 17.00 & 21.83 & 26.67 & 8.18 & 8.01 &  & 17.00 & 21.82 & 26.65 & 8.18 & 8.01 \\
6 &  & 17.00 & 21.85 & 26.69 & 8.19 & 8.01 &  & 17.00 & 21.87 & 26.73 & 8.20 & 8.01 \\
8 &  & 17.01 & 21.83 & 26.64 & 8.19 & 8.00 &  & 16.99 & 21.84 & 26.68 & 8.19 & 8.01 \\
10 &  & 17.02 & 21.88 & 26.72 & 8.20 & 8.00 &  & 16.99 & 21.83 & 26.69 & 8.19 & 8.00 \\
15 &  & 17.03 & 21.85 & 26.67 & 8.19 & 8.01 &  & 17.02 & 21.86 & 26.68 & 8.20 & 8.00 \\
[1mm]
 &  & \multicolumn{11}{l}{$x_{1,it}$ are generated with interactive effects.} \\ \hline
2 &  & 23.06 & 27.92 & 32.77 & 10.47 & 8.85 &  & 23.04 & 27.88 & 32.70 & 10.45 & 8.84 \\
3 &  & 23.88 & 28.74 & 33.61 & 10.78 & 8.90 &  & 23.90 & 28.78 & 33.61 & 10.79 & 8.90 \\
4 &  & 22.70 & 27.56 & 32.42 & 10.33 & 8.85 &  & 22.71 & 27.55 & 32.38 & 10.33 & 8.85 \\
5 &  & 22.78 & 27.61 & 32.42 & 10.35 & 8.78 &  & 22.77 & 27.64 & 32.51 & 10.36 & 8.79 \\
6 &  & 21.98 & 26.83 & 31.66 & 10.06 & 8.84 &  & 21.97 & 26.82 & 31.66 & 10.06 & 8.84 \\
8 &  & 21.26 & 26.07 & 30.86 & 9.78 & 9.62 &  & 21.26 & 26.10 & 30.93 & 9.79 & 9.62 \\
10 &  & 20.55 & 25.39 & 30.23 & 9.52 & 9.37 &  & 20.54 & 25.33 & 30.10 & 9.50 & 9.36 \\
15 &  & 20.59 & 25.40 & 30.22 & 9.52 & 9.08 &  & 20.60 & 25.44 & 30.25 & 9.54 & 9.08 \\
[1mm]
 &  & \multicolumn{11}{l}{$x_{1,it}$ are generated as heterogeneous AR(1) processes.} \\ \hline
2 &  & 16.97 & 21.79 & 26.60 & 8.17 & 7.99 &  & 17.01 & 21.82 & 26.61 & 8.18 & 7.98 \\
3 &  & 16.98 & 21.86 & 26.74 & 8.20 & 7.99 &  & 16.98 & 21.82 & 26.66 & 8.18 & 7.99 \\
4 &  & 16.96 & 21.81 & 26.67 & 8.18 & 8.00 &  & 16.99 & 21.81 & 26.60 & 8.18 & 7.99 \\
5 &  & 16.98 & 21.81 & 26.63 & 8.18 & 7.99 &  & 16.97 & 21.79 & 26.60 & 8.17 & 7.99 \\
6 &  & 17.01 & 21.83 & 26.64 & 8.19 & 7.99 &  & 16.99 & 21.81 & 26.64 & 8.18 & 7.99 \\
8 &  & 16.97 & 21.87 & 26.77 & 8.20 & 8.00 &  & 16.95 & 21.86 & 26.78 & 8.20 & 8.00 \\
10 &  & 16.98 & 21.80 & 26.58 & 8.17 & 7.99 &  & 16.99 & 21.85 & 26.70 & 8.19 & 7.99 \\
15 &  & 16.99 & 21.79 & 26.60 & 8.17 & 8.00 &  & 16.98 & 21.86 & 26.74 & 8.20 & 7.99 \\
[1mm]
 &  & \multicolumn{11}{l}{$x_{1,it}$ are generated as heterogeneous AR(1) processes with interactive effects.} \\ \hline
2 &  & 30.49 & 35.33 & 40.15 & 13.25 & 10.86 &  & 30.47 & 35.35 & 40.24 & 13.26 & 10.86 \\
3 &  & 32.03 & 36.89 & 41.74 & 13.83 & 11.08 &  & 32.07 & 36.88 & 41.63 & 13.83 & 11.08 \\
4 &  & 31.32 & 36.17 & 41.01 & 13.56 & 10.97 &  & 31.34 & 36.15 & 40.94 & 13.56 & 10.97 \\
5 &  & 31.82 & 36.70 & 41.57 & 13.76 & 11.03 &  & 31.84 & 36.65 & 41.48 & 13.74 & 11.03 \\
6 &  & 31.04 & 35.81 & 40.58 & 13.43 & 11.01 &  & 31.02 & 35.87 & 40.68 & 13.45 & 11.00 \\
8 &  & 28.87 & 33.75 & 38.60 & 12.66 & 11.84 &  & 28.88 & 33.76 & 38.60 & 12.66 & 11.84 \\
10 &  & 26.91 & 31.79 & 36.62 & 11.92 & 11.53 &  & 26.85 & 31.72 & 36.58 & 11.89 & 11.52 \\
15 &  & 26.21 & 31.06 & 35.93 & 11.65 & 10.80 &  & 26.24 & 31.10 & 35.97 & 11.66 & 10.80
\\ \hline\hline
\end{tabular}}
\end{center}
\par
\vspace{-2mm} 
\begin{spacing}{1}
{\footnotesize 
Notes: The values of $\kappa_T^2$ are computed based on the stochastic simulations described in Section \ref{Simuk}, using 1,000 replications for $k^{\prime}= 1$. The key parameter values for the different cases are summarized in Table \ref{tab:paramcs}; columns (1)--(5) correspond to columns (1)--(5) of Table \ref{tab:paramcs}. See also footnotes to Table \ref{tab:paramcs}.}
\end{spacing}
\end{table}

\clearpage

\subsection{Monte Carlo results for the tail index of the distribution of $%
1/d_{i}$}

\label{MCap}

Table \ref{tab:ap_mc_2} summarizes the estimates of $\alpha _{P}$, the tail
index of the distribution of $1/d_{i}$, for different cut-off values using
Hill's estimation procedure with one, two, and three regressors,
respectively. Table \ref{tab:ap_mc_3} reports the estimates of $\alpha _{P}$
with one regressor and interactive effects in the regressor process. See
also sub-section \ref{MCalphap} of the main paper.

\begin{table}[h]
\caption{The estimates of $\protect\alpha _{p}$, the tail index of the
distribution of $1/d_{i}$, where $d_{i}=\func{det}(\boldsymbol{X}%
_{i}^{\prime }\boldsymbol{M}_{T}\boldsymbol{X}_{i})$, by Hill's method for
two choices of cut-off values}
\label{tab:ap_mc_2}\vspace{-6mm}
\par
\begin{center}
\scalebox{0.7}{
\begin{tabular}{ccccccccccccccccc}
\hline \hline
\multicolumn{5}{c}{One regressor $(k^{\prime}=1)$} &  & \multicolumn{5}{c}{Two regressors $(k^{\prime}=2)$} &  & \multicolumn{5}{c}{Three regressors $(k^{\prime}=3)$} \\ \cline{1-5} \cline{7-11} \cline{13-17}
Cut-off & \multicolumn{2}{c}{$n^{1/2}$} & \multicolumn{2}{c}{$n^{1/3}$} &  &  & \multicolumn{2}{c}{$n^{1/2}$} & \multicolumn{2}{c}{$n^{1/3}$} &  &  & \multicolumn{2}{c}{$n^{1/2}$} & \multicolumn{2}{c}{$n^{1/3}$} \\ \cline{1-5} \cline{8-11} \cline{14-17}
$T$ & $\hat{\alpha}_{p}$ &  & $\hat{\alpha}_{p}$ &  &  & $T$ & $\hat{\alpha}_{p}$ &  & $\hat{\alpha}_{p}$ &  &  & $T$ & $\hat{\alpha}_{p}$ &  & $\hat{\alpha}_{p}$ &  \\ \hline 
 & \multicolumn{16}{c}{$n=1,000$} \\ \hline
2 & 0.54 & (0.09) & 0.62 & (0.19) &  & 3 & 0.52 & (0.09) & 0.60 & (0.18) &  & 4 & 0.52 & (0.09) & 0.61 & (0.18) \\
3 & 1.05 & (0.18) & 1.21 & (0.36) &  & 4 & 0.99 & (0.17) & 1.16 & (0.35) &  & 5 & 0.95 & (0.17) & 1.14 & (0.35) \\
4 & 1.54 & (0.27) & 1.82 & (0.55) &  & 5 & 1.37 & (0.24) & 1.65 & (0.50) &  & 6 & 1.28 & (0.22) & 1.58 & (0.48) \\
5 & 1.98 & (0.34) & 2.35 & (0.71) &  & 6 & 1.68 & (0.29) & 2.08 & (0.63) &  & 7 & 1.53 & (0.27) & 1.92 & (0.58) \\
6 & 2.38 & (0.41) & 2.83 & (0.85) &  & 7 & 1.97 & (0.34) & 2.45 & (0.74) &  & 8 & 1.75 & (0.30) & 2.22 & (0.67) \\
8 & 3.05 & (0.53) & 3.68 & (1.11) &  & 8 & 2.22 & (0.39) & 2.79 & (0.84) &  & 9 & 1.93 & (0.34) & 2.42 & (0.73) \\
10 & 3.65 & (0.64) & 4.53 & (1.37) &  & 10 & 2.65 & (0.46) & 3.35 & (1.01) &  & 10 & 2.10 & (0.37) & 2.70 & (0.81) \\
15 & 4.87 & (0.85) & 6.14 & (1.85) &  & 15 & 3.46 & (0.60) & 4.49 & (1.36) &  & 15 & 2.73 & (0.48) & 3.57 & (1.08) \\
[1mm]  
& \multicolumn{16}{c}{$n=2,000$} \\ \hline
2 & 0.52 & (0.08) & 0.58 & (0.16) &  & 3 & 0.52 & (0.08) & 0.58 & (0.16) &  & 4 & 0.52 & (0.08) & 0.58 & (0.16) \\
3 & 1.04 & (0.15) & 1.16 & (0.31) &  & 4 & 0.98 & (0.15) & 1.12 & (0.30) &  & 5 & 0.96 & (0.14) & 1.11 & (0.30) \\
4 & 1.53 & (0.23) & 1.74 & (0.47) &  & 5 & 1.38 & (0.20) & 1.62 & (0.43) &  & 6 & 1.28 & (0.19) & 1.51 & (0.40) \\
5 & 1.96 & (0.29) & 2.25 & (0.60) &  & 6 & 1.70 & (0.25) & 2.00 & (0.53) &  & 7 & 1.55 & (0.23) & 1.87 & (0.50) \\
6 & 2.38 & (0.35) & 2.77 & (0.74) &  & 7 & 1.99 & (0.29) & 2.37 & (0.63) &  & 8 & 1.78 & (0.26) & 2.18 & (0.58) \\
8 & 3.09 & (0.46) & 3.62 & (0.97) &  & 8 & 2.24 & (0.33) & 2.71 & (0.72) &  & 9 & 1.98 & (0.29) & 2.42 & (0.65) \\
10 & 3.71 & (0.55) & 4.42 & (1.18) &  & 10 & 2.68 & (0.39) & 3.32 & (0.89) &  & 10 & 2.14 & (0.32) & 2.67 & (0.71) \\
15 & 4.97 & (0.73) & 6.06 & (1.62) &  & 15 & 3.53 & (0.52) & 4.39 & (1.17) &  & 15 & 2.82 & (0.42) & 3.54 & (0.95) \\
[1mm]  
& \multicolumn{16}{c}{$n=5,000$} \\ \hline
2 & 0.51 & (0.06) & 0.56 & (0.13) &  & 3 & 0.51 & (0.06) & 0.57 & (0.13) &  & 4 & 0.51 & (0.06) & 0.57 & (0.13) \\
3 & 1.03 & (0.12) & 1.12 & (0.26) &  & 4 & 0.98 & (0.12) & 1.09 & (0.26) &  & 5 & 0.95 & (0.11) & 1.08 & (0.25) \\
4 & 1.50 & (0.18) & 1.66 & (0.39) &  & 5 & 1.38 & (0.16) & 1.58 & (0.37) &  & 6 & 1.30 & (0.15) & 1.52 & (0.36) \\
5 & 1.97 & (0.23) & 2.21 & (0.52) &  & 6 & 1.72 & (0.20) & 2.00 & (0.47) &  & 7 & 1.59 & (0.19) & 1.88 & (0.44) \\
6 & 2.38 & (0.28) & 2.70 & (0.64) &  & 7 & 2.03 & (0.24) & 2.35 & (0.55) &  & 8 & 1.83 & (0.22) & 2.20 & (0.52) \\
8 & 3.12 & (0.37) & 3.56 & (0.84) &  & 8 & 2.30 & (0.27) & 2.72 & (0.64) &  & 9 & 2.03 & (0.24) & 2.47 & (0.58) \\
10 & 3.76 & (0.44) & 4.37 & (1.03) &  & 10 & 2.75 & (0.32) & 3.27 & (0.77) &  & 10 & 2.22 & (0.26) & 2.68 & (0.63) \\
15 & 5.10 & (0.60) & 6.12 & (1.44) &  & 15 & 3.70 & (0.44) & 4.57 & (1.08) &  & 15 & 2.93 & (0.35) & 3.66 & (0.86) \\
[1mm]  
& \multicolumn{16}{c}{$n=10,000$} \\ \hline
2 & 0.51 & (0.05) & 0.54 & (0.11) &  & 3 & 0.51 & (0.05) & 0.55 & (0.11) &  & 4 & 0.51 & (0.05) & 0.55 & (0.11) \\
3 & 1.02 & (0.10) & 1.10 & (0.23) &  & 4 & 0.98 & (0.10) & 1.06 & (0.22) &  & 5 & 0.96 & (0.10) & 1.06 & (0.22) \\
4 & 1.51 & (0.15) & 1.64 & (0.34) &  & 5 & 1.39 & (0.14) & 1.55 & (0.32) &  & 6 & 1.32 & (0.13) & 1.50 & (0.31) \\
5 & 1.96 & (0.20) & 2.14 & (0.45) &  & 6 & 1.75 & (0.17) & 1.99 & (0.41) &  & 7 & 1.61 & (0.16) & 1.86 & (0.39) \\
6 & 2.39 & (0.24) & 2.63 & (0.55) &  & 7 & 2.06 & (0.20) & 2.34 & (0.49) &  & 8 & 1.86 & (0.18) & 2.17 & (0.45) \\
8 & 3.14 & (0.31) & 3.53 & (0.74) &  & 8 & 2.33 & (0.23) & 2.71 & (0.56) &  & 9 & 2.07 & (0.21) & 2.44 & (0.51) \\
10 & 3.79 & (0.38) & 4.32 & (0.90) &  & 10 & 2.82 & (0.28) & 3.33 & (0.69) &  & 10 & 2.28 & (0.23) & 2.69 & (0.56) \\
15 & 5.19 & (0.52) & 6.03 & (1.26) &  & 15 & 3.77 & (0.38) & 4.56 & (0.95) &  & 15 & 3.03 & (0.30) & 3.64 & (0.76)
\\\hline\hline
\end{tabular}}
\end{center}
\par
\vspace{-2mm} 
\begin{spacing}{0.95}
{\footnotesize 
Notes: The estimates of $\alpha_{p}$ and their standard errors are computed using Hill's estimation procedure given by (\ref{aphill}) in the main paper, with the cut-off values $n^{1/2}$ and $n^{1/3}$. $\boldsymbol{X}_{i} = (\boldsymbol{x}_{i1}, \boldsymbol{x}_{i2}, ..., \boldsymbol{x}_{iT})^{\prime}$ where the $k^{\prime} \times 1$ regressors, $\boldsymbol{x}_{it}$, are generated as specified in the baseline DGP. $\boldsymbol{M}_{T} = \boldsymbol{I}_{T} - \boldsymbol{\tau}_{T}\boldsymbol{\tau}_{T}^{\prime}/T$, where $ \boldsymbol{I}_{T}$ is a $T \times T$ identity matrix and $\boldsymbol{\tau}_{T}$ is a $T \times 1$ vector of ones. For details of the DGP, see sub-section \ref{DGP} in the main paper and Section \ref{secMCDGP}. The numbers in brackets are standard errors. }
\end{spacing}
\end{table}

\begin{table}[t]
\caption{The estimates of $\protect\alpha _{p}$, the tail index of the
distribution of $1/d_{i}$, where $d_{i}=\func{det}(\boldsymbol{X}%
_{i}^{\prime }\boldsymbol{M}_{T}\boldsymbol{X}_{i})$, by Hill's method for
two choices of cut-off values with one regressor and interactive effects in
the regressor process}
\label{tab:ap_mc_3}\vspace{-6mm}
\par
\begin{center}
\scalebox{0.68}{
\begin{tabular}{cccccccccccccccccccc}
\hline\hline
 & \multicolumn{4}{c}{$n=1,000$} &  & \multicolumn{4}{c}{$n=2,000$} &  & \multicolumn{4}{c}{$n=5,000$} &  & \multicolumn{4}{c}{$n=10,000$} \\ \cline{2-5} \cline{7-10} \cline{12-15} \cline{17-20}
Cut-off & \multicolumn{2}{c}{$n^{1/2}$} & \multicolumn{2}{c}{$n^{1/3}$} &  & \multicolumn{2}{c}{$n^{1/2}$} & \multicolumn{2}{c}{$n^{1/3}$} &  & \multicolumn{2}{c}{$n^{1/2}$} & \multicolumn{2}{c}{$n^{1/3}$} &  & \multicolumn{2}{c}{$n^{1/2}$} & \multicolumn{2}{c}{$n^{1/3}$} \\ \cline{1-5} \cline{7-10} \cline{12-15} \cline{17-20}
$T$ & $\hat{\alpha}_{p}$ & & $\hat{\alpha}_{p}$ &  &  & $\hat{\alpha}_{p}$ &  & $\hat{\alpha}_{p}$ &  &  & $\hat{\alpha}_{p}$ &  & $\hat{\alpha}_{p}$ &  &  & $\hat{\alpha}_{p}$ &  & $\hat{\alpha}_{p}$ &  \\ \hline
2 & 0.53 & (0.09) & 0.61 & (0.18) &  & 0.52 & (0.08) & 0.58 & (0.15) &  & 0.51 & (0.06) & 0.56 & (0.13) &  & 0.51 & (0.05) & 0.54 & (0.11) \\
3 & 1.05 & (0.18) & 1.21 & (0.36) &  & 1.03 & (0.15) & 1.16 & (0.31) &  & 1.03 & (0.12) & 1.13 & (0.27) &  & 1.02 & (0.10) & 1.09 & (0.23) \\
4 & 1.53 & (0.27) & 1.80 & (0.54) &  & 1.52 & (0.22) & 1.73 & (0.46) &  & 1.52 & (0.18) & 1.69 & (0.40) &  & 1.50 & (0.15) & 1.62 & (0.34) \\
5 & 1.97 & (0.34) & 2.31 & (0.70) &  & 1.97 & (0.29) & 2.27 & (0.61) &  & 1.96 & (0.23) & 2.19 & (0.52) &  & 1.96 & (0.20) & 2.16 & (0.45) \\
6 & 2.39 & (0.42) & 2.85 & (0.86) &  & 2.37 & (0.35) & 2.75 & (0.73) &  & 2.37 & (0.28) & 2.68 & (0.63) &  & 2.38 & (0.24) & 2.65 & (0.55) \\
8 & 2.97 & (0.52) & 3.63 & (1.09) &  & 3.01 & (0.44) & 3.56 & (0.95) &  & 3.04 & (0.36) & 3.53 & (0.83) &  & 3.07 & (0.31) & 3.49 & (0.73) \\
10 & 3.50 & (0.61) & 4.35 & (1.31) &  & 3.54 & (0.52) & 4.25 & (1.13) &  & 3.63 & (0.43) & 4.26 & (1.00) &  & 3.69 & (0.37) & 4.23 & (0.88) \\
15 & 4.62 & (0.80) & 5.88 & (1.77) &  & 4.71 & (0.69) & 5.81 & (1.55) &  & 4.87 & (0.57) & 5.84 & (1.38) &  & 5.00 & (0.50) & 5.88 & (1.23)
\\\hline\hline
\end{tabular}}
\end{center}
\par
\vspace{-2mm} 
\begin{spacing}{1}
{\footnotesize
Notes: The estimates of $\alpha_{p}$ and their standard errors are computed using Hill's estimation procedure given by (\ref{aphill}) in the main paper, with the cut-off values $n^{1/2}$ and $n^{1/3}$. $\boldsymbol{x}_{1,i} = (x_{1,i1}, x_{1,i2}, ..., x_{1,iT})^{\prime}$, where $x_{1,it}$ is generated with interactive effects without heterogeneous autoregressions. $\boldsymbol{M}_{T} = \boldsymbol{I}_{T} - \boldsymbol{\tau}_{T}\boldsymbol{\tau}_{T}^{\prime}/T$, where $ \boldsymbol{I}_{T}$ is a $T \times T$ identity matrix and $\boldsymbol{\tau}_{T}$ is a $T \times 1$ vector of ones. For details of the DGP, see sub-section \ref{DGP} in the main paper and Section \ref{secMCDGP}. The numbers in brackets are standard errors.}
\end{spacing}
\end{table}

\subsection{Monte Carlo evidence of estimation in the baseline DGP with one
regressor, without time effects \label{MCk2}}

\subsubsection{Comparing TMG, FE and MG estimators \label{sectmgvsfe}}

Figure \ref{fig:fe_tmg_k2_base_2} displays the empirical power functions for
FE, MG and TMG estimators in the case of the baseline DGP (with one
regressor but without time effects) with the sample sizes $n=10,000$, and $%
T=5,6,8$. The corresponding MC results are summarized in Table \ref%
{tab:T_d1_c12_chi2_tex0} in the main paper, and Figure \ref%
{fig:fe_tmg_k2_base_1} in the main paper shows the empirical power for $%
T=2,3,4$. See sub-section \ref{MCfe} in the main paper for a discussion of
the results.

Table \ref{tab:T_d1_c2_hk23_chi2_tex0} summarizes the MC results for FE, MG
and TMG estimators in panel data models under uncorrelated heterogeneity ($%
\psi _{\beta _{1}}=0$), but with correlated heteroskedastic errors (in the $%
y_{it}$ equation). It gives bias, RMSE, and size for case (a) $\sigma
_{it}^{2}=\lambda _{i}^{2}$, on the left panel and for case (b) $\sigma
_{it}^{2}=e_{x1,it}^{2}$, on the right panel of the table.\footnote{%
For details on the DGP and the rationale behind these specifications, see
Section \ref{DGP} in the main paper.} As expected, the MG estimator performs
very poorly when $T$ is ultra-short and suffers from substantial bias. In
contrast, the bias of the TMG estimator remains small even when $T=2$.
Turning to the comparison of FE and TMG estimators, we note that under both
specifications of error heteroskedasticity, the biases of FE and TMG
estimators are close to zero, and both estimators have the correct size for
all $T$ and $n$ combinations. The main difference between FE and TMG
estimators lies in their relative efficiency (in the RMSE sense), when $T$
is ultra-short. For example, when $T=2$ and $n=1,000$, the FE estimator is
more efficient than the TMG estimator under case (b), whilst the reverse is
true under case (a). This ranking of the two estimators is also reflected in
their empirical power functions shown on the left and right panels of Figure %
\ref{fig:fe_tmg_k2_hetrosk}, for $T=2,3,4$ and $n=10,000$. The empirical
power functions for both estimators are correctly centered around $\beta
_{01}=1$. But under heteroskedasticity of type (a), the empirical power
function of the TMG estimator is steeper and for $T=2$ lies within that of
the FE estimator, with the reverse being true when error heteroskedasticity
is generated under case (b).\footnote{%
These results are in line with Proposition \ref{prop_mgvsfe} and Example \ref%
{ExampleMG-FE} in the main paper.} However, differences between FE, MG and
TMG estimators vanish very rapidly as $n$ and $T$ are increased. The
respective empirical power functions for $T=2,3,4,5,6$, and 8 are shown in
Figure \ref{fig:fe_tmg_k2_hetrosk}.

As a general rule, the FE estimator performs well when heterogeneity is
uncorrelated. But in line with our theoretical results, the FE estimator
suffers from substantial bias and size distortions under correlated
heterogeneity, irrespective of whether the errors are heteroskedastic. The
degree of bias and size distortion of the FE estimator rises with the degree
of heterogeneity, $\psi_{\beta_{1}}$. Table \ref{tab:T_d1_c34_chi2_tex0}
provides additional MC results for $\psi_{\beta_{1}}=0.8$ and $PR^{2}=0.2$
on the left panel and for $\psi_{\beta_{1}}=0.5$ and $PR^{2}=0.4$ on the
right panel for sample size combinations $T=2,3,4,5,6,8,10,15$, and $n=1,000$%
, $2,000$, $5,000$, $10,000$. Comparing these results with those already
reported in Table \ref{tab:T_d1_c12_chi2_tex0} we also note that the FE
estimator shows a higher degree of size distortion when $PR^{2}$ is
increased from $0.2$ to $0.4$, with $\psi_{\beta_{1}}$ fixed at $0.5$.

Table \ref{tab:T_d1_c23_chi2_tex2} reports bias, RMSE, and size of FE, MG,
and TMG estimators of $\beta_{01}$ for $T=2,3,4,5,6,8,10,15$, and $n=1,000$, 
$2,000 $, $5,000$ and $10,000$. When $\{x_{1,it}\}$ contains the interactive
effects, the bias magnitude in the FE estimator varies across the finite
values of $T$. The comparative performances of TMG, FE and MG estimators are
similar to those in the baseline DGP.

\begin{figure}[h!]
\caption{Empirical power functions for FE, MG and TMG estimators of $\protect%
\beta_{01}$ $(E(\protect\beta_{i1}) = \protect\beta_{01}=1)$ in the baseline
DGP with one regressor, without time effects for $n=10,000$ and $T=5,6,8$}
\label{fig:fe_tmg_k2_base_2}\vspace{-6mm}
\par
\begin{center}
\includegraphics[scale=0.28]{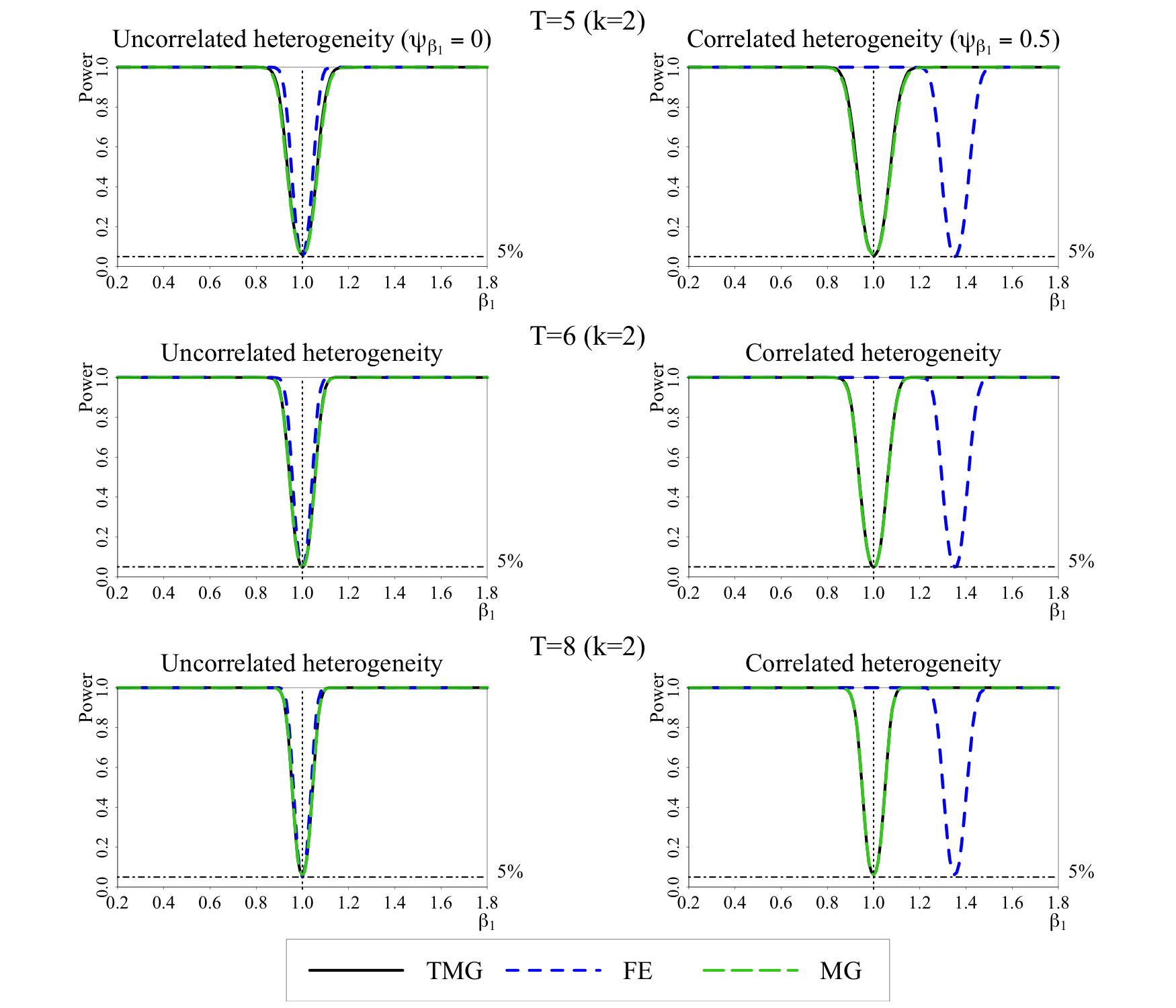}
\end{center}
\par
\vspace{-2mm} 
\begin{spacing}{1}
{\footnotesize
Notes: (i) In the baseline DGP, the outcome variable is generated as $y_{it}=\alpha_{i} + \beta_{i1} x_{1,it} + u_{it}$.  
For further details, see Section \ref{DGP} in the main paper and Section \ref{secMCDGP}. 
(ii) FE and MG estimators are given by (\ref{fee}) and (\ref{mge}) in the main paper, respectively. 
The TMG estimator and its variance are given by (\ref{TMGb}) and (\ref{varC}) in the main paper, respectively. 
(iii) The trimming threshold value for the TMG estimator is given by $a_{n}=\bar{d}_{n} n^{-\alpha}$, where $\bar{d}_{n} =\frac{1}{n} \sum_{i}^{n} d_{i}$, $d_{i} =\func{det}(\boldsymbol{X}_{i}^{\prime} \boldsymbol{M}_{T} \boldsymbol{X}_{i})$, $\boldsymbol{X}_{i}=(\boldsymbol{x}_{i1},\boldsymbol{x}_{i2},...,\boldsymbol{x}_{iT})^{\prime}$ and $\boldsymbol{M}_{T} = \boldsymbol{I}_{T} - \boldsymbol{\tau}_{T}\boldsymbol{\tau}_{T}^{\prime}/T$ where $ \boldsymbol{I}_{T}$ is a $T \times T$ identity matrix and $\boldsymbol{\tau}_{T}$ is a $T \times 1$ vector of ones. 
$\alpha$ is set to $1/3$. }
\end{spacing}
\end{figure}

\begin{sidewaystable}
\caption{Bias, RMSE and size of FE, MG and TMG estimators of $\beta_{01}$ $(E(\beta_{i1}) = \beta_{01}=1)$ in panel data models with one regressor, no time effects, uncorrelated heterogeneity, $\psi_{\beta_{1}}=0$, but correlated heteroskedasticity (cases (a) and (b))} 
\label{tab:T_d1_c2_hk23_chi2_tex0}
\vspace{-5mm}
\begin{center}
\scalebox{0.65}{
\begin{tabular}{rrcrrrrrrrrrrrrrrrrrrrrrrrrr}
 \hline\hline &  \multicolumn{13}{c}{Case (a) of correlated heteroskedasticity} &  & \multicolumn{13}{c}{Case (b) of correlated heteroskedasticity}  \\ \cline{2-14} \cline{16-28}   
& $\hat{\pi}$ $(\times 100)$  & & \multicolumn{3}{c}{Bias}  & & \multicolumn{3}{c}{RMSE} && \multicolumn{3}{c}{Size $(\times 100)$} && $\hat{\pi}$ $(\times 100)$ & & \multicolumn{3}{c}{Bias}  & & \multicolumn{3}{c}{RMSE} && \multicolumn{3}{c}{Size $(\times 100)$} \\ \cline{2-2} \cline{4-6} \cline{8-10} \cline{12-14} \cline{16-16} \cline{18-20} \cline{22-24} \cline{26-28}  
 $T$ & TMG & & FE & MG & TMG && FE & MG & TMG && FE & MG & TMG &&  TMG & & FE & MG & TMG && FE & MG & TMG && FE & MG & TMG    \\ \hline 
&   \multicolumn{27}{c}{$n=1,000$} \\ \hline
2 & 27.50 &  & 0.015 & 1.056 & 0.005 &  & 0.306 & 48.187 & 0.162 &  & 5.2 & 1.6 & 5.2 &  & 27.50 &  & 0.001 & 0.986 & -0.001 &  & 0.141 & 61.738 & 0.265 &  & 4.9 & 1.8 & 5.2 \\
3 & 12.20 &  & -0.006 & 0.002 & -0.003 &  & 0.164 & 0.361 & 0.126 &  & 5.0 & 3.9 & 5.1 &  & 12.20 &  & -0.001 & 0.005 & -0.001 &  & 0.103 & 0.338 & 0.162 &  & 5.4 & 4.7 & 4.9 \\
4 & 6.10 &  & -0.002 & 0.000 & -0.001 &  & 0.123 & 0.151 & 0.106 &  & 5.0 & 5.7 & 5.0 &  & 6.10 &  & 0.001 & -0.001 & -0.002 &  & 0.085 & 0.149 & 0.123 &  & 4.8 & 4.6 & 5.8 \\
5 & 3.20 &  & 0.001 & 0.004 & 0.002 &  & 0.102 & 0.107 & 0.091 &  & 5.4 & 4.3 & 5.0 &  & 3.20 &  & 0.000 & 0.002 & 0.001 &  & 0.077 & 0.108 & 0.100 &  & 5.1 & 4.7 & 5.0 \\
6 & 1.80 &  & 0.001 & 0.001 & 0.001 &  & 0.086 & 0.092 & 0.083 &  & 4.2 & 4.5 & 4.9 &  & 1.80 &  & 0.001 & 0.001 & 0.001 &  & 0.069 & 0.089 & 0.085 &  & 4.7 & 4.1 & 4.1 \\
8 & 0.60 &  & -0.001 & 0.001 & 0.001 &  & 0.073 & 0.073 & 0.071 &  & 5.4 & 5.2 & 5.3 &  & 0.60 &  & 0.000 & 0.001 & 0.001 &  & 0.062 & 0.070 & 0.070 &  & 5.1 & 5.1 & 5.0 \\
10 & 0.20 &  & 0.001 & 0.001 & 0.001 &  & 0.063 & 0.063 & 0.062 &  & 4.7 & 5.0 & 5.2 &  & 0.20 &  & -0.001 & 0.001 & 0.001 &  & 0.056 & 0.061 & 0.061 &  & 4.6 & 4.5 & 4.5 \\
15 & 0.00 &  & 0.001 & 0.000 & 0.000 &  & 0.055 & 0.052 & 0.052 &  & 5.3 & 4.8 & 4.9 &  & 0.00 &  & 0.001 & 0.000 & 0.000 &  & 0.051 & 0.052 & 0.052 &  & 5.2 & 4.7 & 4.7 \\
&   \multicolumn{27}{c}{$n=2,000$} \\ \hline
2 & 24.60 &  & -0.005 & -1.267 & -0.003 &  & 0.217 & 73.545 & 0.122 &  & 5.9 & 1.4 & 5.7 &  & 24.60 &  & -0.005 & 0.754 & -0.008 &  & 0.101 & 114.301 & 0.194 &  & 5.1 & 1.7 & 4.7 \\
3 & 9.80 &  & 0.000 & -0.021 & -0.002 &  & 0.115 & 0.807 & 0.095 &  & 4.9 & 4.3 & 5.1 &  & 9.80 &  & 0.000 & -0.010 & -0.003 &  & 0.071 & 0.343 & 0.119 &  & 4.0 & 4.4 & 5.1 \\
4 & 4.40 &  & -0.001 & -0.001 & -0.002 &  & 0.086 & 0.109 & 0.077 &  & 5.0 & 5.4 & 5.7 &  & 4.40 &  & -0.001 & -0.001 & -0.001 &  & 0.062 & 0.107 & 0.087 &  & 6.1 & 4.9 & 4.4 \\
5 & 2.10 &  & 0.000 & -0.002 & -0.002 &  & 0.073 & 0.081 & 0.071 &  & 6.1 & 6.0 & 5.9 &  & 2.10 &  & 0.000 & -0.002 & -0.002 &  & 0.054 & 0.077 & 0.072 &  & 4.6 & 4.7 & 4.8 \\
6 & 1.10 &  & -0.001 & 0.002 & 0.002 &  & 0.060 & 0.063 & 0.059 &  & 4.3 & 4.8 & 4.7 &  & 1.10 &  & -0.001 & 0.001 & 0.001 &  & 0.049 & 0.063 & 0.061 &  & 5.1 & 4.9 & 4.6 \\
8 & 0.30 &  & 0.000 & 0.000 & 0.000 &  & 0.051 & 0.051 & 0.050 &  & 4.4 & 4.4 & 4.4 &  & 0.30 &  & -0.001 & -0.001 & -0.001 &  & 0.043 & 0.049 & 0.049 &  & 4.8 & 4.0 & 4.2 \\
10 & 0.10 &  & 0.000 & 0.000 & 0.000 &  & 0.045 & 0.044 & 0.043 &  & 4.2 & 4.4 & 4.6 &  & 0.10 &  & 0.000 & 0.000 & 0.000 &  & 0.040 & 0.043 & 0.043 &  & 4.5 & 4.7 & 4.8 \\
15 & 0.00 &  & -0.001 & 0.000 & 0.000 &  & 0.038 & 0.037 & 0.037 &  & 5.3 & 5.4 & 5.6 &  & 0.00 &  & -0.001 & -0.001 & -0.001 &  & 0.035 & 0.036 & 0.036 &  & 4.7 & 5.3 & 5.3 \\
&   \multicolumn{27}{c}{$n=5,000$} \\ \hline
2 & 21.20 &  & -0.002 & 1.571 & 0.000 &  & 0.135 & 81.325 & 0.083 &  & 5.7 & 2.3 & 5.2 &  & 21.20 &  & -0.001 & 3.893 & 0.002 &  & 0.066 & 128.395 & 0.132 &  & 5.6 & 2.3 & 4.4 \\
3 & 7.30 &  & -0.002 & 0.002 & 0.000 &  & 0.075 & 0.273 & 0.062 &  & 4.8 & 4.4 & 5.1 &  & 7.30 &  & 0.000 & -0.002 & 0.001 &  & 0.047 & 0.226 & 0.078 &  & 5.4 & 4.3 & 4.9 \\
4 & 2.90 &  & -0.001 & 0.000 & 0.000 &  & 0.055 & 0.068 & 0.051 &  & 4.8 & 5.0 & 5.7 &  & 2.90 &  & 0.000 & -0.001 & -0.001 &  & 0.039 & 0.069 & 0.058 &  & 5.5 & 4.7 & 5.1 \\
5 & 1.20 &  & 0.000 & -0.001 & 0.000 &  & 0.045 & 0.048 & 0.043 &  & 5.0 & 4.6 & 4.7 &  & 1.20 &  & -0.002 & -0.002 & -0.001 &  & 0.033 & 0.049 & 0.046 &  & 3.9 & 4.8 & 4.8 \\
6 & 0.50 &  & 0.000 & -0.001 & -0.001 &  & 0.039 & 0.040 & 0.039 &  & 4.9 & 5.4 & 4.9 &  & 0.50 &  & 0.000 & -0.001 & -0.001 &  & 0.031 & 0.040 & 0.040 &  & 4.2 & 5.0 & 5.1 \\
8 & 0.10 &  & -0.001 & -0.001 & -0.001 &  & 0.032 & 0.032 & 0.032 &  & 5.1 & 4.7 & 4.6 &  & 0.10 &  & -0.001 & -0.001 & -0.001 &  & 0.028 & 0.032 & 0.032 &  & 4.8 & 5.3 & 5.3 \\
10 & 0.00 &  & 0.001 & 0.001 & 0.001 &  & 0.029 & 0.028 & 0.028 &  & 5.7 & 4.7 & 4.7 &  & 0.00 &  & 0.000 & 0.000 & 0.000 &  & 0.026 & 0.029 & 0.029 &  & 5.6 & 5.1 & 5.1 \\
15 & 0.00 &  & 0.000 & 0.000 & 0.000 &  & 0.023 & 0.023 & 0.023 &  & 4.7 & 5.1 & 5.1 &  & 0.00 &  & 0.000 & 0.000 & 0.000 &  & 0.022 & 0.023 & 0.023 &  & 5.0 & 4.9 & 4.9 \\
&   \multicolumn{27}{c}{$n=10,000$} \\ \hline
2 & 19.00 &  & 0.003 & 0.545 & 0.000 &  & 0.094 & 65.869 & 0.064 &  & 5.2 & 2.3 & 5.7 &  & 19.00 &  & 0.002 & -1.264 & 0.002 &  & 0.046 & 114.412 & 0.102 &  & 5.7 & 2.6 & 5.5 \\
3 & 5.90 &  & 0.000 & 0.000 & 0.000 &  & 0.052 & 0.130 & 0.047 &  & 4.8 & 4.4 & 5.0 &  & 5.90 &  & 0.000 & 0.000 & 0.001 &  & 0.033 & 0.135 & 0.058 &  & 5.4 & 5.3 & 5.1 \\
4 & 2.10 &  & -0.001 & 0.000 & -0.001 &  & 0.039 & 0.047 & 0.037 &  & 4.7 & 4.7 & 5.3 &  & 2.10 &  & -0.001 & 0.000 & -0.001 &  & 0.028 & 0.047 & 0.042 &  & 4.6 & 5.0 & 5.8 \\
5 & 0.80 &  & -0.001 & 0.001 & 0.001 &  & 0.032 & 0.035 & 0.033 &  & 5.3 & 5.7 & 6.0 &  & 0.80 &  & 0.000 & 0.001 & 0.001 &  & 0.025 & 0.036 & 0.034 &  & 5.5 & 6.6 & 6.5 \\
6 & 0.30 &  & -0.001 & -0.001 & -0.001 &  & 0.028 & 0.028 & 0.027 &  & 4.8 & 4.7 & 4.8 &  & 0.30 &  & 0.000 & 0.000 & 0.000 &  & 0.022 & 0.029 & 0.028 &  & 4.7 & 5.1 & 5.3 \\
8 & 0.10 &  & -0.001 & 0.000 & 0.000 &  & 0.023 & 0.023 & 0.022 &  & 5.7 & 4.9 & 4.8 &  & 0.10 &  & -0.001 & 0.000 & 0.000 &  & 0.020 & 0.023 & 0.023 &  & 5.8 & 5.2 & 5.3 \\
10 & 0.00 &  & 0.000 & -0.001 & -0.001 &  & 0.020 & 0.020 & 0.020 &  & 4.9 & 5.1 & 5.1 &  & 0.00 &  & 0.000 & 0.000 & 0.000 &  & 0.018 & 0.020 & 0.020 &  & 5.4 & 5.3 & 5.3 \\
15 & 0.00 &  & 0.000 & 0.001 & 0.001 &  & 0.017 & 0.016 & 0.016 &  & 5.1 & 4.8 & 4.8 &  & 0.00 &  & 0.000 & 0.000 & 0.000 &  & 0.016 & 0.016 & 0.016 &  & 6.2 & 5.3 & 5.3 \\
\hline
\hline
\end{tabular}
}
\end{center}
\vspace{-3mm}
\begin{spacing}{1}
{\footnotesize 
Notes: 
(i) The data generating process is given by $y_{it}=\alpha_{i}
+ \beta_{i1} x_{1,it} + \sigma_{it} e_{it}$,  where $\sigma_{it}^{2}$ are 
generated as case (a): $ \sigma_{it}^{2} =\lambda_{i}^{2}$, and case (b): $\sigma_{it}^{2}=e_{x1,it}^{2}$, for all $i$ and $t$. 
For further details, see Section \ref{DGP} in the main paper and Section \ref{secMCDGP}. 
(ii) For FE, MG and TMG estimators, see footnotes to Figure \ref{fig:fe_tmg_k2_base_2}.
$\hat{\pi}$ is the simulated fraction of individual estimates being trimmed, defined by (\ref{pin}) in the main paper. 
} 
\end{spacing}
\end{sidewaystable}

\clearpage \newpage 
\begin{figure}[h]
\caption{Empirical power functions for FE, MG and TMG estimators of $\protect%
\beta_{01}$ $(E(\protect\beta_{i1})=\protect\beta_{01}=1)$ in panel data
models with one regressor, no time effects, uncorrelated heterogeneity, $%
\protect\psi _{\protect\beta_{1}}=0$, but correlated heteroskedasticity
(cases (a) and (b)), for $n=10,000$ and $T=2,3,4,5,6,8$}
\label{fig:fe_tmg_k2_hetrosk}\vspace{-5mm}
\par
\begin{center}
\includegraphics[scale=0.21]{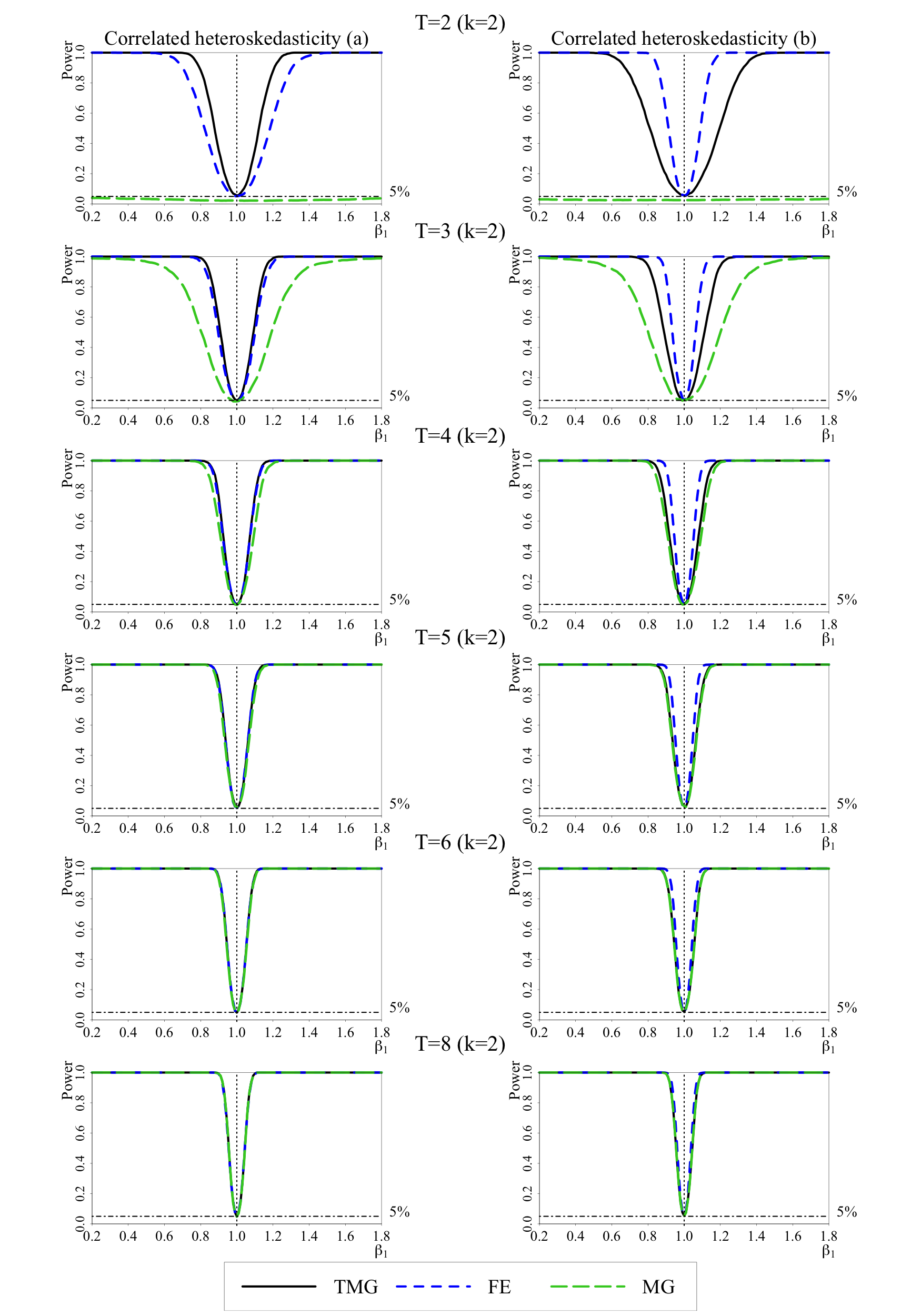}
\end{center}
\par
\vspace{-2mm} 
\begin{spacing}{1}
{\footnotesize
Notes: See footnotes to Table \ref{tab:T_d1_c2_hk23_chi2_tex0}.
}
\end{spacing}
\end{figure}

\clearpage\newpage 
\begin{sidewaystable}
\caption{Bias, RMSE and size of FE, MG and TMG estimators of $\beta_{01}$ $(E(\beta_{i1})=\beta_{01}=1)$ in panel data models with one regressor, without time effects, but with different degrees of correlated heterogeneity, $\psi_{\beta_{1}}$, and overall fit, $PR^{2}$} 
\label{tab:T_d1_c34_chi2_tex0}
\vspace{-5mm}
\begin{center}
\scalebox{0.65}{
\begin{tabular}{rrrrrrrrrrrrrrrrrrrrrrrrrrrr}
 \hline\hline &  \multicolumn{13}{c}{ High correlated heterogeneity $\psi_{\beta_{1}}=0.8$ with $PR^{2}=0.2$} &  & \multicolumn{13}{c}{ Medium correlated heterogeneity $\psi_{\beta_{1}}=0.5$ with $PR^{2}=0.4$}  \\ \cline{2-14} \cline{16-28}   
&\multicolumn{1}{c}{$\hat{\pi}$ $(\times 100)$}   & & \multicolumn{3}{c}{Bias}  & & \multicolumn{3}{c}{RMSE} && \multicolumn{3}{c}{Size $(\times 100)$} &&\multicolumn{1}{c}{$\hat{\pi}$ $(\times 100)$} & & \multicolumn{3}{c}{Bias}  & & \multicolumn{3}{c}{RMSE} && \multicolumn{3}{c}{Size $(\times 100)$} \\ \cline{2-2} \cline{4-6} \cline{8-10} \cline{12-14} \cline{16-16} \cline{18-20} \cline{22-24} \cline{26-28}  
 $T$ & TMG & & FE & MG & TMG && FE & MG & TMG && FE & MG & TMG &&  TMG & & FE & MG & TMG && FE & MG & TMG && FE & MG & TMG    \\ \hline 
& \multicolumn{27}{c}{$n=1,000$} \\ \hline
2 & 27.50 &  & 0.569 & 2.022 & 0.040 &  & 0.606 & 91.089 & 0.331 &  & 80.0 & 1.7 & 5.3 &  & 27.50 &  & 0.356 & 1.122 & 0.025 &  & 0.379 & 50.516 & 0.186 &  & 81.3 & 1.8 & 5.1 \\
3 & 12.20 &  & 0.564 & -0.001 & 0.013 &  & 0.585 & 0.469 & 0.202 &  & 96.6 & 4.2 & 4.5 &  & 12.20 &  & 0.353 & 0.000 & 0.009 &  & 0.367 & 0.261 & 0.115 &  & 96.8 & 4.5 & 4.8 \\
4 & 6.10 &  & 0.564 & 0.001 & 0.010 &  & 0.580 & 0.188 & 0.151 &  & 99.2 & 5.4 & 4.9 &  & 6.10 &  & 0.353 & 0.000 & 0.006 &  & 0.364 & 0.107 & 0.087 &  & 99.3 & 5.2 & 4.8 \\
5 & 3.20 &  & 0.563 & 0.002 & 0.006 &  & 0.577 & 0.135 & 0.124 &  & 99.9 & 4.3 & 4.7 &  & 3.20 &  & 0.352 & 0.001 & 0.003 &  & 0.361 & 0.078 & 0.072 &  & 99.8 & 4.4 & 4.8 \\
6 & 1.80 &  & 0.564 & 0.001 & 0.003 &  & 0.575 & 0.114 & 0.109 &  & 99.9 & 5.3 & 5.1 &  & 1.80 &  & 0.353 & 0.001 & 0.002 &  & 0.361 & 0.066 & 0.064 &  & 99.9 & 4.9 & 4.8 \\
8 & 0.60 &  & 0.561 & 0.002 & 0.002 &  & 0.571 & 0.088 & 0.088 &  & 100.0 & 5.2 & 5.2 &  & 0.60 &  & 0.350 & 0.000 & 0.001 &  & 0.357 & 0.054 & 0.053 &  & 100.0 & 5.6 & 5.7 \\
10 & 0.20 &  & 0.563 & 0.001 & 0.001 &  & 0.572 & 0.075 & 0.075 &  & 100.0 & 4.9 & 4.9 &  & 0.20 &  & 0.351 & 0.000 & 0.000 &  & 0.358 & 0.048 & 0.048 &  & 100.0 & 4.8 & 4.9 \\
15 & 0.00 &  & 0.562 & -0.001 & -0.001 &  & 0.569 & 0.061 & 0.061 &  & 100.0 & 5.4 & 5.4 &  & 0.00 &  & 0.351 & 0.000 & 0.000 &  & 0.357 & 0.042 & 0.042 &  & 100.0 & 6.0 & 6.1 \\
& \multicolumn{27}{c}{$n=2,000$} \\ \hline
2 & 24.60 &  & 0.561 & -1.138 & 0.026 &  & 0.580 & 127.613 & 0.248 &  & 96.9 & 1.5 & 5.2 &  & 24.60 &  & 0.350 & -0.631 & 0.016 &  & 0.362 & 70.770 & 0.139 &  & 97.4 & 1.5 & 5.7 \\
3 & 9.80 &  & 0.566 & -0.030 & 0.010 &  & 0.577 & 1.167 & 0.151 &  & 100.0 & 4.2 & 5.1 &  & 9.80 &  & 0.354 & -0.017 & 0.007 &  & 0.361 & 0.646 & 0.085 &  & 100.0 & 4.0 & 4.5 \\
4 & 4.40 &  & 0.565 & -0.002 & 0.004 &  & 0.573 & 0.135 & 0.109 &  & 100.0 & 4.9 & 5.3 &  & 4.40 &  & 0.353 & -0.001 & 0.003 &  & 0.358 & 0.076 & 0.062 &  & 100.0 & 4.8 & 4.9 \\
5 & 2.10 &  & 0.563 & -0.004 & -0.001 &  & 0.570 & 0.099 & 0.093 &  & 100.0 & 5.8 & 6.0 &  & 2.10 &  & 0.352 & -0.002 & 0.000 &  & 0.357 & 0.057 & 0.054 &  & 100.0 & 6.0 & 5.8 \\
6 & 1.10 &  & 0.567 & 0.002 & 0.003 &  & 0.573 & 0.076 & 0.074 &  & 100.0 & 4.7 & 4.2 &  & 1.10 &  & 0.354 & 0.001 & 0.001 &  & 0.358 & 0.045 & 0.044 &  & 100.0 & 4.6 & 4.6 \\
8 & 0.30 &  & 0.564 & -0.001 & -0.001 &  & 0.569 & 0.061 & 0.060 &  & 100.0 & 4.6 & 4.5 &  & 0.30 &  & 0.352 & -0.001 & -0.001 &  & 0.356 & 0.037 & 0.037 &  & 100.0 & 4.8 & 4.6 \\
10 & 0.10 &  & 0.567 & 0.000 & 0.000 &  & 0.571 & 0.052 & 0.052 &  & 100.0 & 4.6 & 4.6 &  & 0.10 &  & 0.354 & 0.000 & 0.000 &  & 0.357 & 0.033 & 0.033 &  & 100.0 & 4.6 & 4.6 \\
15 & 0.00 &  & 0.566 & 0.000 & 0.000 &  & 0.570 & 0.043 & 0.043 &  & 100.0 & 4.7 & 4.7 &  & 0.00 &  & 0.353 & 0.000 & 0.000 &  & 0.356 & 0.029 & 0.029 &  & 100.0 & 4.9 & 4.9 \\
& \multicolumn{27}{c}{$n=5,000$} \\ \hline
2 & 21.20 &  & 0.564 & 1.771 & 0.026 &  & 0.572 & 123.341 & 0.166 &  & 100.0 & 2.0 & 4.8 &  & 21.20 &  & 0.353 & 0.982 & 0.016 &  & 0.358 & 68.402 & 0.093 &  & 100.0 & 2.0 & 4.8 \\
3 & 7.30 &  & 0.565 & 0.007 & 0.010 &  & 0.570 & 0.482 & 0.097 &  & 100.0 & 4.6 & 5.6 &  & 7.30 &  & 0.353 & 0.004 & 0.007 &  & 0.356 & 0.267 & 0.055 &  & 100.0 & 4.4 & 6.0 \\
4 & 2.90 &  & 0.564 & -0.001 & 0.003 &  & 0.567 & 0.085 & 0.071 &  & 100.0 & 5.1 & 4.9 &  & 2.90 &  & 0.352 & 0.000 & 0.002 &  & 0.354 & 0.048 & 0.041 &  & 100.0 & 4.9 & 5.1 \\
5 & 1.20 &  & 0.567 & 0.001 & 0.003 &  & 0.569 & 0.059 & 0.056 &  & 100.0 & 4.4 & 4.4 &  & 1.20 &  & 0.354 & 0.000 & 0.002 &  & 0.355 & 0.034 & 0.033 &  & 100.0 & 4.3 & 4.2 \\
6 & 0.50 &  & 0.566 & 0.000 & 0.001 &  & 0.569 & 0.050 & 0.049 &  & 100.0 & 4.7 & 5.1 &  & 0.50 &  & 0.354 & 0.000 & 0.000 &  & 0.356 & 0.029 & 0.029 &  & 100.0 & 4.8 & 4.7 \\
8 & 0.10 &  & 0.566 & 0.000 & 0.000 &  & 0.568 & 0.039 & 0.039 &  & 100.0 & 5.3 & 5.3 &  & 0.10 &  & 0.353 & 0.000 & 0.000 &  & 0.355 & 0.024 & 0.024 &  & 100.0 & 5.3 & 5.2 \\
10 & 0.00 &  & 0.564 & -0.001 & -0.001 &  & 0.566 & 0.034 & 0.034 &  & 100.0 & 4.8 & 4.8 &  & 0.00 &  & 0.353 & 0.000 & 0.000 &  & 0.354 & 0.022 & 0.022 &  & 100.0 & 5.1 & 5.1 \\
15 & 0.00 &  & 0.566 & 0.000 & 0.000 &  & 0.568 & 0.026 & 0.026 &  & 100.0 & 4.1 & 4.1 &  & 0.00 &  & 0.354 & 0.000 & 0.000 &  & 0.355 & 0.018 & 0.018 &  & 100.0 & 4.2 & 4.2 \\
& \multicolumn{27}{c}{$n=10,000$} \\ \hline
2 & 19.00 &  & 0.567 & 0.962 & 0.025 &  & 0.571 & 142.909 & 0.129 &  & 100.0 & 2.4 & 6.3 &  & 19.00 &  & 0.355 & 0.534 & 0.015 &  & 0.357 & 79.253 & 0.073 &  & 100.0 & 2.4 & 6.8 \\
3 & 5.90 &  & 0.565 & 0.000 & 0.009 &  & 0.568 & 0.163 & 0.072 &  & 100.0 & 4.6 & 4.4 &  & 5.90 &  & 0.353 & 0.000 & 0.006 &  & 0.355 & 0.091 & 0.041 &  & 100.0 & 4.8 & 4.9 \\
4 & 2.10 &  & 0.563 & 0.000 & 0.002 &  & 0.565 & 0.059 & 0.052 &  & 100.0 & 5.5 & 5.8 &  & 2.10 &  & 0.352 & 0.000 & 0.001 &  & 0.353 & 0.034 & 0.030 &  & 100.0 & 5.1 & 5.1 \\
5 & 0.80 &  & 0.566 & 0.000 & 0.001 &  & 0.567 & 0.044 & 0.042 &  & 100.0 & 5.9 & 5.7 &  & 0.80 &  & 0.354 & 0.000 & 0.001 &  & 0.355 & 0.025 & 0.024 &  & 100.0 & 5.8 & 5.4 \\
6 & 0.30 &  & 0.565 & -0.001 & 0.000 &  & 0.566 & 0.035 & 0.034 &  & 100.0 & 4.9 & 4.6 &  & 0.30 &  & 0.353 & -0.001 & 0.000 &  & 0.354 & 0.021 & 0.020 &  & 100.0 & 4.8 & 4.9 \\
8 & 0.10 &  & 0.564 & -0.001 & -0.001 &  & 0.565 & 0.028 & 0.028 &  & 100.0 & 6.2 & 6.2 &  & 0.10 &  & 0.353 & 0.000 & 0.000 &  & 0.354 & 0.017 & 0.017 &  & 100.0 & 5.6 & 5.5 \\
10 & 0.00 &  & 0.565 & -0.001 & -0.001 &  & 0.566 & 0.024 & 0.024 &  & 100.0 & 5.6 & 5.6 &  & 0.00 &  & 0.353 & 0.000 & 0.000 &  & 0.354 & 0.015 & 0.015 &  & 100.0 & 4.2 & 4.2 \\
15 & 0.00 &  & 0.565 & 0.000 & 0.000 &  & 0.566 & 0.019 & 0.019 &  & 100.0 & 5.6 & 5.6 &  & 0.00 &  & 0.353 & 0.000 & 0.000 &  & 0.354 & 0.013 & 0.013 &  & 100.0 & 5.9 & 5.9\\
 \hline
\hline
\end{tabular}}
\end{center}
\vspace{-3mm}
\begin{spacing}{1}
{\footnotesize
Notes: 
(i) The data generating process is given by $y_{it}=\alpha_{i} + \beta_{i1} x_{1,it} + \sigma_{it}e_{it}$ with random heteroskedasticity, different degrees of correlated heterogeneity, $\psi_{\beta_{1}}$ (defined by (\ref{eta_i}) in the main paper), and levels of overall fit, $PR^{2}$ (defined by (\ref{PR2})). For further details, see Section \ref{DGP} in the main paper and Section \ref{secMCDGP}. 
(ii) For FE, MG and TMG estimators, see footnotes to Figure \ref{fig:fe_tmg_k2_base_2}. 
$\hat{\pi}$ is the simulated fraction of individual estimates being trimmed, defined by (\ref{pin}) in the main paper. 
}
\end{spacing}
\end{sidewaystable}

\clearpage \newpage 
\begin{sidewaystable}
\caption{Bias, RMSE and size of FE, MG and TMG estimators of $\protect%
\beta_{01}$ $(E(\beta_{i1})=\beta_{01}=1)$ in panel data models with one regressor, without time effects, but with correlated heterogeneity and interactive effects in the $x_{1,it}$ process} 
\label{tab:T_d1_c23_chi2_tex2}
\vspace{-6mm}
\begin{center}
\scalebox{0.65}{
\begin{tabular}{rrcrrrrrrrrrrrrrrrrrrrrrrrrr}
 \hline\hline &  \multicolumn{13}{c}{ Medium correlated heterogeneity: $\psi_{\beta_{1}} = 0.5 $ } &  & \multicolumn{13}{c}{ High correlated heterogeneity: $\psi_{\beta_{1}}=0.8$}  \\ \cline{2-14} \cline{16-28}   
& $\hat{\pi}$  $(\times 100)$  & & \multicolumn{3}{c}{Bias}  & & \multicolumn{3}{c}{RMSE} && \multicolumn{3}{c}{Size $(\times 100)$} && $\hat{\pi}$  $(\times 100)$ & & \multicolumn{3}{c}{Bias}  & & \multicolumn{3}{c}{RMSE} && \multicolumn{3}{c}{Size $(\times 100)$} \\ \cline{2-2} \cline{4-6} \cline{8-10} \cline{12-14} \cline{16-16} \cline{18-20} \cline{22-24} \cline{26-28}  
 $T$ & TMG & & FE & MG & TMG && FE & MG & TMG && FE & MG & TMG &&  TMG & & FE & MG & TMG && FE & MG & TMG && FE & MG & TMG    \\ \hline 
   \multicolumn{27}{c}{$n=1,000$} \\ \hline
2 & 26.60 &  & 0.296 & -27.929 & 0.020 &  & 0.341 & 862.796 & 0.298 &  & 39.4 & 1.8 & 4.3 &  & 26.60 &  & 0.474 & -30.258 & 0.031 &  & 0.513 & 934.746 & 0.323 &  & 65.8 & 1.8 & 4.3 \\
3 & 11.60 &  & 0.315 & -0.010 & 0.012 &  & 0.343 & 0.522 & 0.192 &  & 65.6 & 4.3 & 4.2 &  & 11.60 &  & 0.503 & -0.011 & 0.018 &  & 0.529 & 0.565 & 0.208 &  & 90.2 & 4.2 & 4.2 \\
4 & 5.60 &  & 0.312 & 0.002 & 0.008 &  & 0.332 & 0.177 & 0.145 &  & 81.2 & 5.3 & 5.3 &  & 5.60 &  & 0.498 & 0.002 & 0.011 &  & 0.515 & 0.192 & 0.157 &  & 97.6 & 5.0 & 5.6 \\
5 & 2.90 &  & 0.320 & -0.001 & 0.001 &  & 0.336 & 0.130 & 0.121 &  & 89.2 & 4.9 & 5.6 &  & 2.90 &  & 0.512 & 0.000 & 0.003 &  & 0.527 & 0.140 & 0.130 &  & 99.1 & 5.2 & 5.9 \\
6 & 1.60 &  & 0.304 & 0.001 & 0.002 &  & 0.317 & 0.101 & 0.097 &  & 92.8 & 4.3 & 4.4 &  & 1.60 &  & 0.486 & 0.001 & 0.003 &  & 0.499 & 0.108 & 0.104 &  & 99.6 & 4.2 & 4.1 \\
8 & 0.60 &  & 0.243 & 0.001 & 0.002 &  & 0.254 & 0.071 & 0.071 &  & 93.3 & 4.9 & 4.7 &  & 0.60 &  & 0.390 & 0.002 & 0.002 &  & 0.400 & 0.077 & 0.076 &  & 99.7 & 5.2 & 5.1 \\
10 & 0.20 &  & 0.255 & -0.004 & -0.003 &  & 0.265 & 0.063 & 0.063 &  & 97.3 & 5.2 & 5.1 &  & 0.20 &  & 0.410 & -0.004 & -0.004 &  & 0.419 & 0.068 & 0.068 &  & 100.0 & 5.3 & 5.1 \\
15 & 0.00 &  & 0.281 & 0.000 & 0.000 &  & 0.288 & 0.052 & 0.052 &  & 99.7 & 4.7 & 4.7 &  & 0.00 &  & 0.449 & 0.001 & 0.001 &  & 0.456 & 0.056 & 0.056 &  & 100.0 & 4.7 & 4.7 \\
   \multicolumn{27}{c}{$n=2,000$} \\ \hline
2 & 23.80 &  & 0.298 & -18.992 & 0.020 &  & 0.323 & 422.775 & 0.224 &  & 68.0 & 2.5 & 5.3 &  & 23.80 &  & 0.476 & -20.576 & 0.031 &  & 0.497 & 458.031 & 0.244 &  & 92.0 & 2.5 & 5.3 \\
3 & 9.40 &  & 0.315 & -0.007 & 0.003 &  & 0.329 & 0.345 & 0.142 &  & 90.7 & 3.8 & 5.5 &  & 9.40 &  & 0.504 & -0.007 & 0.007 &  & 0.517 & 0.372 & 0.154 &  & 99.7 & 3.9 & 5.6 \\
4 & 4.00 &  & 0.308 & 0.000 & 0.004 &  & 0.318 & 0.124 & 0.104 &  & 98.1 & 5.4 & 5.6 &  & 4.00 &  & 0.492 & -0.001 & 0.005 &  & 0.502 & 0.134 & 0.112 &  & 100.0 & 5.3 & 5.4 \\
5 & 1.90 &  & 0.317 & 0.000 & 0.001 &  & 0.325 & 0.091 & 0.086 &  & 99.6 & 5.3 & 5.2 &  & 1.90 &  & 0.508 & 0.000 & 0.002 &  & 0.515 & 0.099 & 0.093 &  & 100.0 & 5.1 & 5.5 \\
6 & 0.90 &  & 0.304 & 0.000 & 0.001 &  & 0.310 & 0.073 & 0.071 &  & 99.7 & 5.2 & 5.7 &  & 0.90 &  & 0.485 & 0.000 & 0.001 &  & 0.492 & 0.079 & 0.077 &  & 100.0 & 5.2 & 5.4 \\
8 & 0.30 &  & 0.244 & 0.000 & 0.000 &  & 0.249 & 0.053 & 0.052 &  & 99.9 & 5.3 & 5.1 &  & 0.30 &  & 0.391 & 0.000 & 0.001 &  & 0.396 & 0.057 & 0.057 &  & 100.0 & 5.2 & 5.1 \\
10 & 0.10 &  & 0.258 & -0.001 & -0.001 &  & 0.263 & 0.044 & 0.044 &  & 100.0 & 4.9 & 4.8 &  & 0.10 &  & 0.413 & -0.001 & -0.001 &  & 0.417 & 0.048 & 0.048 &  & 100.0 & 5.1 & 5.1 \\
15 & 0.00 &  & 0.280 & 0.001 & 0.001 &  & 0.283 & 0.037 & 0.037 &  & 100.0 & 5.6 & 5.6 &  & 0.00 &  & 0.447 & 0.001 & 0.001 &  & 0.451 & 0.040 & 0.040 &  & 100.0 & 4.5 & 4.6 \\
   \multicolumn{27}{c}{$n=5,000$} \\ \hline
2 & 20.50 &  & 0.300 & 3.977 & 0.014 &  & 0.310 & 186.355 & 0.154 &  & 96.9 & 1.9 & 5.0 &  & 20.50 &  & 0.479 & 4.309 & 0.022 &  & 0.488 & 201.895 & 0.167 &  & 100.0 & 1.9 & 5.1 \\
3 & 7.00 &  & 0.318 & 0.000 & 0.008 &  & 0.323 & 0.243 & 0.094 &  & 100.0 & 4.3 & 5.0 &  & 7.00 &  & 0.507 & -0.001 & 0.011 &  & 0.512 & 0.263 & 0.102 &  & 100.0 & 4.4 & 5.2 \\
4 & 2.60 &  & 0.310 & 0.001 & 0.003 &  & 0.314 & 0.078 & 0.067 &  & 100.0 & 4.9 & 4.4 &  & 2.60 &  & 0.495 & 0.001 & 0.004 &  & 0.499 & 0.084 & 0.072 &  & 100.0 & 4.8 & 4.4 \\
5 & 1.10 &  & 0.318 & 0.000 & 0.001 &  & 0.322 & 0.058 & 0.055 &  & 100.0 & 5.2 & 5.1 &  & 1.10 &  & 0.509 & 0.000 & 0.001 &  & 0.512 & 0.062 & 0.059 &  & 100.0 & 5.2 & 5.3 \\
6 & 0.50 &  & 0.303 & -0.001 & 0.000 &  & 0.306 & 0.045 & 0.044 &  & 100.0 & 4.9 & 5.2 &  & 0.50 &  & 0.487 & 0.000 & 0.000 &  & 0.489 & 0.048 & 0.047 &  & 100.0 & 5.1 & 5.1 \\
8 & 0.10 &  & 0.244 & -0.001 & -0.001 &  & 0.247 & 0.032 & 0.032 &  & 100.0 & 5.0 & 4.7 &  & 0.10 &  & 0.392 & 0.000 & 0.000 &  & 0.394 & 0.035 & 0.035 &  & 100.0 & 5.1 & 5.0 \\
10 & 0.00 &  & 0.258 & -0.001 & -0.001 &  & 0.260 & 0.028 & 0.028 &  & 100.0 & 5.3 & 5.4 &  & 0.00 &  & 0.412 & -0.001 & -0.001 &  & 0.414 & 0.031 & 0.031 &  & 100.0 & 4.6 & 4.6 \\
15 & 0.00 &  & 0.279 & -0.001 & -0.001 &  & 0.281 & 0.024 & 0.024 &  & 100.0 & 4.8 & 4.8 &  & 0.00 &  & 0.447 & -0.001 & -0.001 &  & 0.448 & 0.026 & 0.026 &  & 100.0 & 5.4 & 5.4 \\
   \multicolumn{27}{c}{$n=10,000$} \\ \hline
2 & 18.30 &  & 0.297 & 5.578 & 0.011 &  & 0.303 & 177.762 & 0.114 &  & 100.0 & 2.0 & 4.4 &  & 18.30 &  & 0.476 & 6.043 & 0.018 &  & 0.481 & 192.586 & 0.124 &  & 100.0 & 2.0 & 4.6 \\
3 & 5.60 &  & 0.317 & 0.003 & 0.004 &  & 0.319 & 0.159 & 0.069 &  & 100.0 & 5.3 & 5.1 &  & 5.60 &  & 0.507 & 0.003 & 0.007 &  & 0.509 & 0.171 & 0.075 &  & 100.0 & 5.2 & 5.2 \\
4 & 1.90 &  & 0.309 & 0.001 & 0.002 &  & 0.311 & 0.056 & 0.048 &  & 100.0 & 4.9 & 5.2 &  & 1.90 &  & 0.494 & 0.001 & 0.003 &  & 0.495 & 0.060 & 0.052 &  & 100.0 & 5.1 & 5.2 \\
5 & 0.70 &  & 0.318 & 0.000 & 0.001 &  & 0.320 & 0.041 & 0.039 &  & 100.0 & 5.6 & 5.4 &  & 0.70 &  & 0.509 & 0.000 & 0.001 &  & 0.511 & 0.044 & 0.042 &  & 100.0 & 5.5 & 5.4 \\
6 & 0.30 &  & 0.304 & 0.000 & 0.000 &  & 0.305 & 0.032 & 0.032 &  & 100.0 & 4.6 & 4.8 &  & 0.30 &  & 0.487 & 0.000 & 0.000 &  & 0.488 & 0.035 & 0.034 &  & 100.0 & 4.8 & 4.8 \\
8 & 0.10 &  & 0.246 & 0.001 & 0.001 &  & 0.247 & 0.023 & 0.023 &  & 100.0 & 5.1 & 5.1 &  & 0.10 &  & 0.392 & 0.001 & 0.001 &  & 0.393 & 0.025 & 0.025 &  & 100.0 & 5.0 & 4.9 \\
10 & 0.00 &  & 0.259 & 0.000 & 0.000 &  & 0.260 & 0.020 & 0.020 &  & 100.0 & 5.7 & 5.7 &  & 0.00 &  & 0.415 & 0.000 & 0.000 &  & 0.415 & 0.022 & 0.022 &  & 100.0 & 4.9 & 4.9 \\
15 & 0.00 &  & 0.279 & 0.000 & 0.000 &  & 0.280 & 0.017 & 0.017 &  & 100.0 & 4.8 & 4.9 &  & 0.00 &  & 0.447 & 0.001 & 0.001 &  & 0.448 & 0.018 & 0.018 &  & 100.0 & 4.8 & 4.8\\
 \hline
\hline
\end{tabular}}
\end{center}
\vspace{-3mm}
\begin{spacing}{1}
{\footnotesize
Notes:
(i) The data generating process is given by $y_{it}=\alpha_{i} + \beta_{i1} x_{1,it} + \sigma_{it}e_{it}$ with random heteroskedasticity, $x_{1,it}$ are generated with interactive effects without heterogeneous autoregressions, and $\psi_{\beta_{1}}$ is defined by (\ref{eta_i}) in the main paper. 
The error processes for $y_{it}$ and $x_{1,it}$ equations are chi-squared and Gaussian, respectively. 
For further details, see Section \ref{DGP} in the main paper and Section \ref{secMCDGP}. 
(ii) For FE, MG and TMG estimators, see footnotes to Figure \ref{fig:fe_tmg_k2_base_2}. 
$\hat{\pi}$ is the simulated fraction of individual estimates being trimmed, defined by (\ref{pin}) in the main paper. 
}
\end{spacing}
\end{sidewaystable}

\clearpage \newpage

\subsubsection{Comparing the TMG estimator with different threshold exponents%
}

\label{MChatalpha}

We also carried out MC experiments to evaluate the finite-sample performance
of the TMG estimator under alternative choices of the threshold exponent, $%
\alpha $. Specifically, we consider $\alpha =1/3$ as well as $\hat{\alpha}%
=1/(1+2\hat{\alpha}_{p})+0.01$, where $\hat{\alpha}_{p}$ is the tail index
of the distribution of $1/d_{i}$, where $d_{i}=\func{det}(\boldsymbol{X}%
_{i}^{\prime }\boldsymbol{M}_{T}\boldsymbol{X}_{i})$, using \cite{Hill1975}%
's estimator given by (\ref{aphill}) in the main paper, with the cut-off
value $n^{1/3}$.

Table \ref{tab:tmg_thresh_d1_c2_chi2_tex0_b} reports the trimmed fraction,
bias, RMSE, and empirical size for the two choices of the threshold
exponents. The empirical power functions for $T=2,3,4$ are shown in Figure %
\ref{fig:tmg_alpha_hat}. When $T=k=2$, setting $\alpha =1/3$ yields
substantially larger trimmed fractions and outperforms the one based on $%
\hat{\alpha}$, yielding a larger bias but smaller RMSE and better empirical
power. When $T>k$, the differences in RMSE and empirical power between the
two choices are generally modest. The trimmed fraction with $\alpha =1/3$
declines toward zero as $T$ increases. But using the estimated value of $%
\alpha $ given by $\hat{\alpha}=1/(1+2\hat{\alpha}_{p})+0.01$, the fraction
of the estimates being trimmed increases with $T$, and as a result, there is
some over-rejection when $T=15$. This finding shows that when $T$ is
sufficiently large and the second-order moment of the MG estimator exists,
trimming is not required and its use could be counter productive.

Taken together, these results suggest that the gains from setting $\alpha $
based on the estimated $\alpha _{p}$ are limited in finite samples, while
estimation of $\alpha $ by $\hat{\alpha}=1/(1+2\hat{\alpha}_{p})+0.01$
introduces additional estimation errors. Given the uncertainty associated
with the estimates of $\alpha _{p}$, in practice it is reasonable to set $%
\alpha _{p}=1$ and adopt the simple threshold function $\boldsymbol{1}%
\{d_{i}/\bar{d}_{n}>n^{-1/3}\}$, which performs comparably while being
straightforward to implement.

\begin{table}[ht]
\caption{Bias, RMSE and size of the TMG estimator of $\protect\beta_{01}$ $%
(E(\protect\beta_{i1})=\protect\beta_{01}=1)$ for the threshold exponents $%
\protect\alpha =1/3$ and $\protect\alpha=\hat{\protect\alpha}$ in the
baseline DGP with one regressor, without time effects, but with correlated
heterogeneity, $\protect\psi _{\protect\beta_{1}}=0.5$}
\label{tab:tmg_thresh_d1_c2_chi2_tex0_b}\vspace{-6mm}
\par
\begin{center}
\scalebox{0.75}{
\begin{tabular}{rrrlrrlcclrc}
\hline\hline
 & \multicolumn{2}{c}{$\hat{\pi} (\times 100)$} &  & \multicolumn{2}{c}{Bias} &  & \multicolumn{2}{c}{RMSE} &  & \multicolumn{2}{c}{Size $(\times 100)$} \\ \cline{2-3} \cline{5-6} \cline{8-9} \cline{11-12}
$T$/$\alpha$ & $\hat{\alpha}$ & 1/3 &  & $\hat{\alpha}$ & 1/3 &  & $\hat{\alpha}$ & 1/3 &  & $\hat{\alpha}$ & 1/3 \\ \hline
 & \multicolumn{11}{c}{$n=1,000$}  \\ \hline
2 & 17.90 & 27.50 &  & 0.014 & 0.026 &  & 0.387 & 0.299 &  & 5.7 & 5.1 \\
3 & 14.50 & 12.20 &  & 0.009 & 0.007 &  & 0.180 & 0.183 &  & 4.7 & 4.7 \\
4 & 15.40 & 6.10 &  & 0.016 & 0.006 &  & 0.128 & 0.138 &  & 5.1 & 5.0 \\
5 & 16.60 & 3.20 &  & 0.015 & 0.004 &  & 0.103 & 0.113 &  & 4.6 & 4.6 \\
6 & 17.80 & 1.80 &  & 0.018 & 0.002 &  & 0.093 & 0.099 &  & 5.6 & 5.0 \\
8 & 20.20 & 0.60 &  & 0.019 & 0.002 &  & 0.077 & 0.079 &  & 6.2 & 5.7 \\
10 & 23.30 & 0.20 &  & 0.021 & 0.000 &  & 0.069 & 0.069 &  & 6.3 & 5.3 \\
15 & 28.40 & 0.00 &  & 0.024 & 0.000 &  & 0.062 & 0.057 &  & 8.7 & 6.0 \\
[2mm] & \multicolumn{11}{c}{$n=2,000$}  \\ \hline
2 & 14.50 & 24.60 &  & 0.002 & 0.015 &  & 0.298 & 0.224 &  & 5.4 & 5.2 \\
3 & 11.40 & 9.80 &  & 0.008 & 0.005 &  & 0.135 & 0.136 &  & 5.1 & 4.8 \\
4 & 12.10 & 4.40 &  & 0.009 & 0.002 &  & 0.092 & 0.098 &  & 5.1 & 5.1 \\
5 & 13.00 & 2.10 &  & 0.008 & -0.001 &  & 0.079 & 0.084 &  & 6.3 & 6.0 \\
6 & 14.60 & 1.10 &  & 0.014 & 0.002 &  & 0.065 & 0.068 &  & 5.3 & 4.4 \\
8 & 17.00 & 0.30 &  & 0.014 & -0.001 &  & 0.054 & 0.055 &  & 5.6 & 4.3 \\
10 & 19.70 & 0.10 &  & 0.017 & 0.000 &  & 0.049 & 0.048 &  & 6.0 & 4.4 \\
15 & 25.00 & 0.00 &  & 0.020 & 0.000 &  & 0.044 & 0.040 &  & 8.6 & 5.4 \\
[2mm] & \multicolumn{11}{c}{$n=5,000$}  \\ \hline
2 & 11.40 & 21.20 &  & 0.010 & 0.016 &  & 0.213 & 0.150 &  & 5.8 & 4.6 \\
3 & 8.40 & 7.30 &  & 0.007 & 0.006 &  & 0.087 & 0.088 &  & 5.7 & 5.8 \\
4 & 8.50 & 2.90 &  & 0.007 & 0.002 &  & 0.060 & 0.065 &  & 4.8 & 5.1 \\
5 & 9.80 & 1.20 &  & 0.010 & 0.002 &  & 0.049 & 0.051 &  & 4.9 & 4.2 \\
6 & 10.90 & 0.50 &  & 0.010 & 0.000 &  & 0.043 & 0.044 &  & 5.2 & 4.6 \\
8 & 13.10 & 0.10 &  & 0.011 & 0.000 &  & 0.037 & 0.036 &  & 7.2 & 5.6 \\
10 & 15.60 & 0.00 &  & 0.012 & 0.000 &  & 0.033 & 0.031 &  & 8.0 & 4.6 \\
15 & 21.30 & 0.00 &  & 0.016 & 0.000 &  & 0.030 & 0.024 &  & 10.8 & 3.7 \\
[2mm] & \multicolumn{11}{c}{$n=10,000$}  \\ \hline
2 & 9.40 & 19.00 &  & 0.009 & 0.016 &  & 0.166 & 0.116 &  & 6.0 & 6.2 \\
3 & 6.50 & 5.90 &  & 0.006 & 0.006 &  & 0.065 & 0.065 &  & 4.8 & 4.3 \\
4 & 6.70 & 2.10 &  & 0.005 & 0.001 &  & 0.045 & 0.047 &  & 5.5 & 5.3 \\
5 & 7.40 & 0.80 &  & 0.006 & 0.001 &  & 0.036 & 0.038 &  & 6.4 & 5.4 \\
6 & 8.40 & 0.30 &  & 0.007 & 0.000 &  & 0.031 & 0.031 &  & 5.9 & 4.8 \\
8 & 10.70 & 0.10 &  & 0.008 & 0.000 &  & 0.027 & 0.026 &  & 7.7 & 6.2 \\
10 & 12.80 & 0.00 &  & 0.010 & -0.001 &  & 0.024 & 0.022 &  & 8.3 & 5.4 \\
15 & 18.00 & 0.00 &  & 0.014 & 0.000 &  & 0.023 & 0.018 &  & 13.8 & 5.2\\
\hline\hline
\end{tabular}}
\end{center}
\par
\vspace{-2mm} 
\begin{spacing}{1}
{\footnotesize
Notes: We set $\hat{\alpha} = \frac{1}{1+2\hat{\alpha}_{p}}+0.01$, where $\hat{\alpha}_{p}$ is computed by Hill's estimation procedure with the cut-off value $n^{1/3}$. See Table \ref{tab:ap_mc_2} for the estimates of $\hat{\alpha}_{p}$ with $k^{\prime}$=1. For details of the baseline DGP and the TMG estimator, see footnotes to Figure \ref{fig:fe_tmg_k2_base_2}.
$\hat{\pi}$ is the simulated fraction of individual estimates being trimmed, defined by (\ref{pin}) in the main paper. 
}
\end{spacing}
\end{table}

\newpage \clearpage
\begin{figure}[h]
\caption{Empirical power functions for the TMG estimator of $\protect\beta%
_{01}$ $(E(\protect\beta_{i1})=\protect\beta_{01}=1)$ for the threshold
exponents $\protect\alpha =1/3$ and $\protect\alpha=\hat{\protect\alpha}$ in
the baseline DGP with one regressor, without time effects, but with
correlated heterogeneity, $\protect\psi _{\protect\beta_{1}}=0.5$}
\label{fig:tmg_alpha_hat}\vspace{-5mm}
\par
\begin{center}
\includegraphics[scale=0.16]{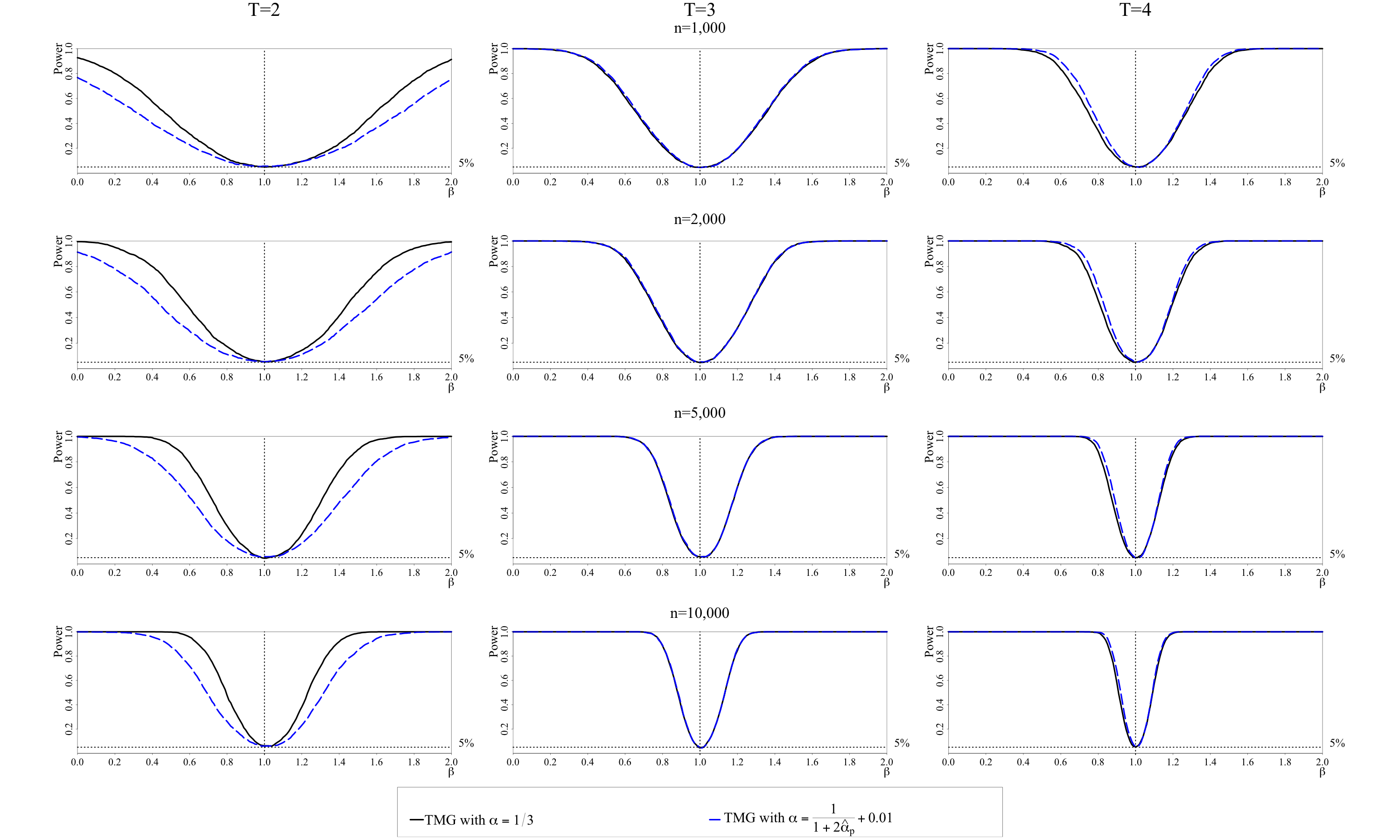}
\end{center}
\par
\vspace{-2mm} 
\begin{spacing}{1}
{\footnotesize 
Notes: See footnotes to Table \ref{tab:tmg_thresh_d1_c2_chi2_tex0_b}. 
}
\end{spacing}
\end{figure}

\clearpage

\subsubsection{Sensitivity of TMG and GP estimators to the choice of the
threshold parameters \label{MCthresh}}

Here we investigate the sensitivity of TMG and GP estimators to the choice
of threshold exponents with $k^{\prime }=1$. Recall that our threshold
function is $\boldsymbol{1}\{d_{i}>a_{n}\}$ with $a_{n}=C_{n}n^{-\alpha }$,
whereas the threshold function used by GP is $\boldsymbol{1}\{|\det (%
\boldsymbol{W}_{i})|>h_{n}\}$ with $h_{n}=C_{GP}n^{-\alpha _{GP}}$. When $%
T=k=2$, GP's threshold function can be written equivalently as $\boldsymbol{1%
}\{d_{i}>h_{n}^{2}/2\}$. Consequently, our choice of the threshold value
differs from that of GP in both the scaling constant and the exponent of the
threshold value.

For the TMG estimator, in addition to the baseline value of $\alpha =1/3$,
we consider $\alpha =0.35$ and $1/2$. For the GP estimator the baseline
value of $\alpha _{GP}$ is set to $1/3$, as recommended by GP.\footnote{%
See equation (\ref{gpe}) and the discussion for the implementation of the GP
estimator in sub-section \ref{Threshold} in the main paper.} But for the
purpose of comparisons with the TMG estimator (computed for $\alpha
=1/3,0.35 $ and $1/2)$, we also consider $\alpha _{GP}=0.35/2$ and $1/4$.
Recall that the bandwidth, $h_{n}^{2}/2$, used by GP corresponds to $a_{n}$
used in the specification of TMG. Hence, $2\alpha _{GP}$ corresponds to $%
\alpha $. For comparability, we decided to consider values of $\alpha _{GP}$
below $1/3$ required by GP's theory. This allows us to compare GP and TMG
estimators focusing on the utility of including both trimmed and untrimmed
estimates of $\boldsymbol{\beta }_{i}$ in estimation of $\boldsymbol{\beta }%
_{0}$.

As can be seen from Table \ref{tab:thresh_d1_c2_chi2_tex0_b} there is a
clear trade-off between bias and variance as $\alpha $ and $\alpha _{GP}$
are increased. For $T=2$, the TMG estimator is biased when $\alpha =1/3$ (as
predicted by the theory) but has a lower variance, with its RMSE rising as $%
\alpha $ is increased from $1/3$ to $1/2$. This trade-off is less
consequential when $T$ is increased to $T=3$. The same is also true for the
GP estimator. But for all choices of $\alpha $ and $\alpha _{GP}$, the TMG
performs better in terms of RMSE when $T=2$. For $T=3$, TMG and GP
estimators share the same threshold exponent by setting $\alpha _{GP}=\alpha
/2$, resulting in identical trimmed fractions for $\alpha =2\alpha _{GP}\in
\{0.35,1/2\}$. While RMSEs are similar, the GP estimator exhibits higher
bias than the TMG estimator, particularly for $\alpha=2\alpha_{GP}=0.35$, as
data of the trimmed units are not exploited by the GP estimator.

Figure \ref{fig:tmg_gp_alpha_1} compares power functions of TMG and GP
estimators with $\alpha =2\alpha _{GP}=0.35$. For $T=2$, a higher trimmed
fraction results in a steeper power function for the TMG estimator as
compared to that of the GP estimator. When $T=3$, with the same trimmed
fraction, the power function of the GP estimator shifts to the right, away
from the true value. The substantial differences in the power performance of
the TMG estimator with $\alpha =1/3$ and the GP estimator with $\alpha
_{GP}=1/3$ are also illustrated in Figure \ref{fig:tmg_gp_alpha_2}.

The power comparisons of TMG and GP estimators for different values of $%
\alpha $ and $\alpha _{GP}$ are given in Figures \ref{fig:tmg_alpha} and \ref%
{fig:gp_alpha}, respectively, and convey the same message, suggesting that
for the TMG estimator, the value of $\alpha =1/3$ tends to produce the best
bias-variance trade-off. Increasing $\alpha $ reduces the bias but increases
the variance, and the value $\alpha =1/3$ seems to strike a sensible balance
and is recommended.

\begin{table}[ht]
\caption{Bias, RMSE and size of TMG and GP estimators of $\protect\beta_{01}$
$(E(\protect\beta_{i1})=\protect\beta_{01}=1)$ for different threshold
exponents, $\protect\alpha $ and $\protect\alpha _{GP}$, in the baseline DGP
with one regressor, without time effects, but with correlated heterogeneity, 
$\protect\psi _{\protect\beta_{1}}=0.5$}
\label{tab:thresh_d1_c2_chi2_tex0_b}\vspace{-6mm}
\par
\begin{center}
\scalebox{0.75}{
\begin{tabular}{lcrrccrrrcc}
\hline\hline & & \multicolumn{4}{c}{$T=2$} &  & \multicolumn{4}{c}{$T=3$}   \\  \cline{3-6} \cline{8-11}
& & \multicolumn{1}{c}{$\hat{\pi}$ } & & & Size & & \multicolumn{1}{c}{$\hat{\pi}$ } & & & Size\\ 
Estimator &  $\alpha$/$\alpha_{GP}$  & $(\times 100)$ & Bias & RMSE & $(\times 100)$ & & $(\times 100)$ & Bias& RMSE & $(\times 100)$ \\ \hline   
&& \multicolumn{9}{c}{$n=1,000$}  \\ \hline  
TMG & 1/3 & 27.5 & 0.026 & 0.299 & 5.1 &  & 12.2 & 0.007 & 0.183 & 4.7 \\
TMG & 0.35 & 26.0 & 0.025 & 0.307 & 5.1 &  & 10.9 & 0.006 & 0.186 & 4.8 \\
TMG & 0.50 & 15.7 & 0.013 & 0.395 & 5.5 &  & 4.0 & 0.001 & 0.216 & 4.7 \\
GP & 0.35/2 & 12.0 & 0.018 & 0.392 & 5.1 &  & 10.9 & 0.016 & 0.175 & 4.7 \\
GP & 0.25 & 7.2 & 0.000 & 0.513 & 6.2 &  & 4.0 & 0.004 & 0.201 & 4.8 \\
GP & 1/3 & 4.1 & -0.003 & 0.653 & 4.6 &  & 1.3 & -0.001 & 0.234 & 4.7 \\
&& \multicolumn{9}{c}{$n=2,000$}  \\ \hline  
TMG & 1/3 & 24.6 & 0.015 & 0.224 & 5.2 &  & 9.8 & 0.005 & 0.136 & 4.8 \\
TMG & 0.35 & 23.2 & 0.013 & 0.230 & 5.6 &  & 8.7 & 0.004 & 0.139 & 5.1 \\
TMG & 0.50 & 13.3 & 0.001 & 0.300 & 5.4 &  & 2.9 & -0.002 & 0.161 & 4.9 \\
GP & 0.35/2 & 10.7 & 0.002 & 0.294 & 5.7 &  & 8.7 & 0.014 & 0.132 & 5.0 \\
GP & 0.25 & 6.1 & -0.005 & 0.383 & 5.4 &  & 2.9 & 0.001 & 0.152 & 5.4 \\
GP & 1/3 & 3.2 & 0.004 & 0.523 & 4.7 &  & 0.8 & -0.003 & 0.176 & 4.9 \\
&& \multicolumn{9}{c}{$n=5,000$}  \\ \hline  
TMG & 1/3 & 21.2 & 0.016 & 0.150 & 4.6 &  & 7.3 & 0.006 & 0.088 & 5.8 \\
TMG & 0.35 & 19.8 & 0.015 & 0.155 & 4.7 &  & 6.4 & 0.006 & 0.090 & 5.7 \\
TMG & 0.50 & 10.5 & 0.008 & 0.215 & 5.9 &  & 1.8 & 0.001 & 0.105 & 5.6 \\
GP & 0.35/2 & 9.1 & 0.011 & 0.201 & 5.0 &  & 6.4 & 0.011 & 0.085 & 5.0 \\
GP & 0.25 & 4.8 & 0.007 & 0.279 & 5.8 &  & 1.8 & 0.003 & 0.099 & 4.9 \\
GP & 1/3 & 2.4 & 0.008 & 0.398 & 6.1 &  & 0.4 & 0.001 & 0.117 & 5.3 \\
&& \multicolumn{9}{c}{$n=10,000$}  \\ \hline  
TMG & 1/3 & 19.0 & 0.016 & 0.116 & 6.2 &  & 5.9 & 0.006 & 0.065 & 4.3 \\
TMG & 0.35 & 17.6 & 0.014 & 0.120 & 6.3 &  & 5.1 & 0.005 & 0.066 & 4.6 \\
TMG & 0.50 & 8.9 & 0.008 & 0.165 & 5.6 &  & 1.3 & 0.001 & 0.077 & 4.4 \\
GP & 0.35/2 & 8.1 & 0.010 & 0.153 & 5.6 &  & 5.1 & 0.010 & 0.063 & 5.2 \\
GP & 0.25 & 4.1 & 0.008 & 0.212 & 5.5 &  & 1.3 & 0.003 & 0.074 & 4.3 \\
GP & 1/3 & 1.9 & 0.008 & 0.302 & 4.4 &  & 0.3 & 0.000 & 0.086 & 4.7 \\
   \hline
\hline
\end{tabular}}
\end{center}
\par
\vspace{-2mm} 
\begin{spacing}{1}
{\footnotesize
Notes: 
(i) The GP estimator proposed by \cite{GrahamPowell2012} is given by (\ref{gpe}) in the main paper. For $T=k$, GP compare $d_{i,GP}^{1/2}$ with the bandwidth $h_{n} = C_{GP}n^{-\alpha_{GP}}$, where $d_{i,GP} = \func{det}(\boldsymbol{W}_{i}^{\prime}\boldsymbol{W}_{i})$ with $\boldsymbol{W}_{i} = (\boldsymbol{\tau}_{T}, \boldsymbol{X}_{i})$ and $\boldsymbol{\tau}_{T}$ being a $T \times 1$ vector of ones. $C_{GP}=\frac{1}{2}\min \left( \hat{\sigma}_{D},\hat{r}_{D}/1.34\right)$, where $\hat{\sigma}_{D}$ and $\hat{r}_{D}$ are the respective sample standard deviation and interquartile range of $d_{i,GP}^{1/2}$. For $T>k$, $C_{GP}=\left(n^{-1}\sum_{i=1}^{n}d_{i,GP}\right)^{1/2}$. See sub-section \ref{MCcorr} in the main paper for details. 
(ii) For details of the baseline DGP and the TMG estimator, see footnotes to Figure \ref{fig:fe_tmg_k2_base_2}.
$\hat{\pi}$ is the simulated fraction of individual estimates being trimmed, defined by (\ref{pin}) in the main paper. 
}
\end{spacing}
\end{table}

\newpage \clearpage
\begin{figure}[h]
\caption{Empirical power functions for TMG and GP estimators of $\protect%
\beta_{01}$ $(E(\protect\beta_{i1})=\protect\beta_{01}=1)$ for the threshold
exponents $\protect\alpha =0.35$ and $\protect\alpha _{GP}=0.35/2$ in the
baseline DGP with one regressor, without time effects, but with correlated
heterogeneity, $\protect\psi _{\protect\beta_{1}}=0.5$}
\label{fig:tmg_gp_alpha_1}\vspace{-5mm}
\par
\begin{center}
\includegraphics[scale=0.2]{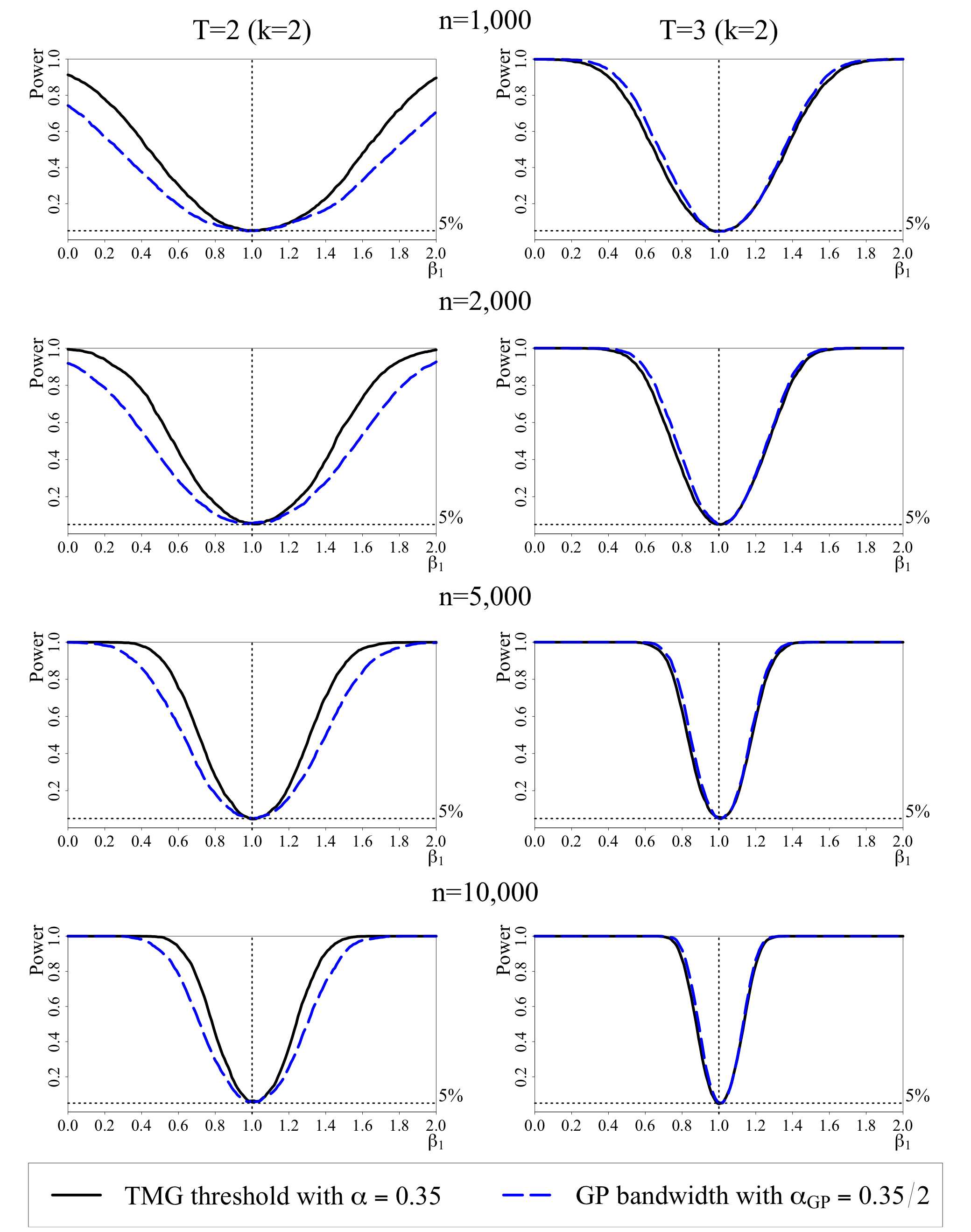}
\end{center}
\par
\vspace{-2mm} 
\begin{spacing}{1}
{\footnotesize 
Notes: See footnotes to Table \ref{tab:thresh_d1_c2_chi2_tex0_b}. 
}
\end{spacing}
\end{figure}

\begin{figure}[h!]
\caption{Empirical power functions for TMG and GP estimators of $\protect%
\beta_{01} $ $(E(\protect\beta_{i1})=\protect\beta_{01}=1)$ for the
threshold exponents $\protect\alpha=1/3$ and $\protect\alpha _{GP}=1/3$ in
the baseline DGP with one regressor, without time effects, but with
correlated heterogeneity, $\protect\psi _{\protect\beta_{1}}=0.5$}
\label{fig:tmg_gp_alpha_2}\vspace{-5mm}
\par
\begin{center}
\includegraphics[scale=0.2]{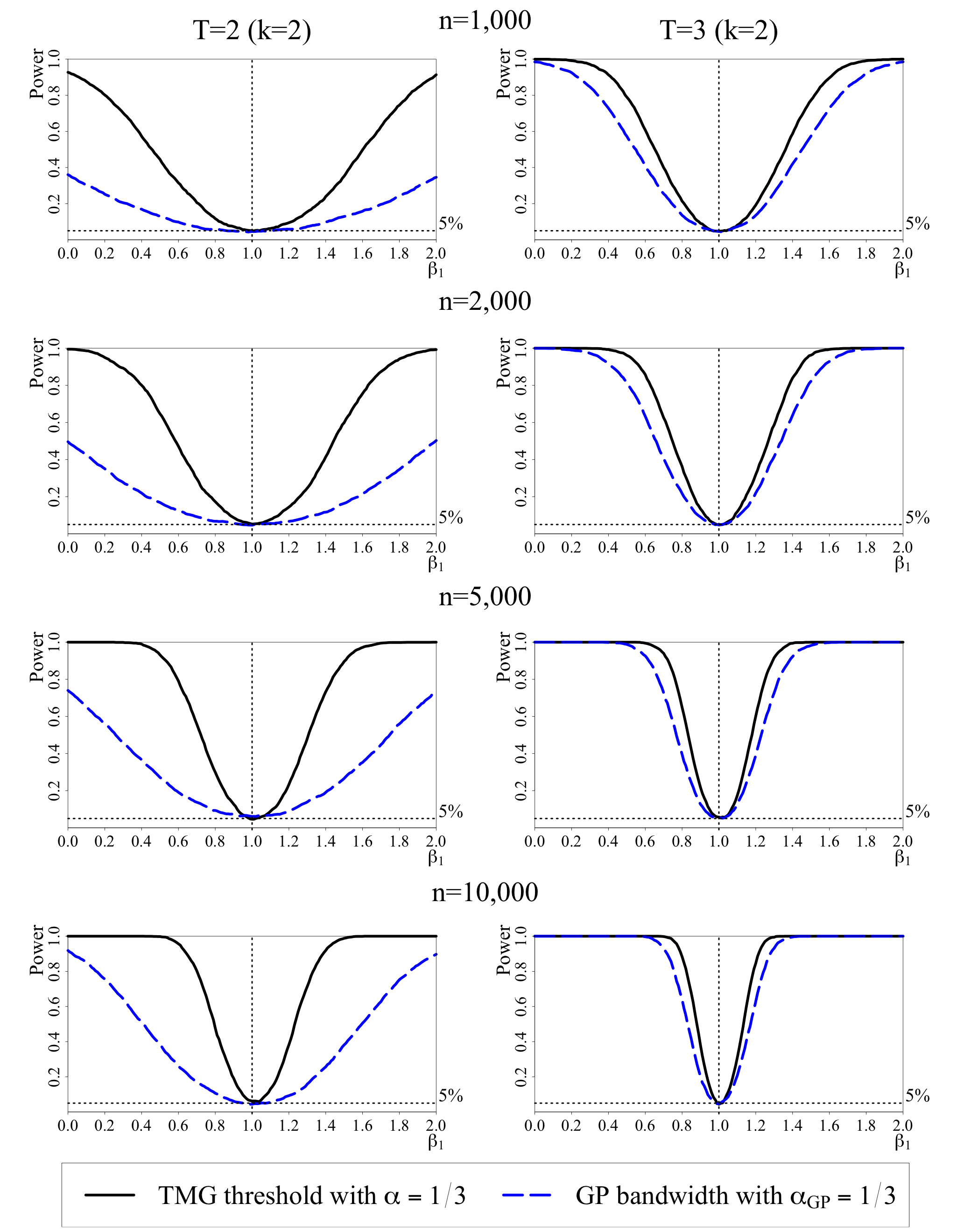}
\end{center}
\par
\vspace{-2mm} 
\begin{spacing}{1}
{\footnotesize 
Notes: See footnotes to Table \ref{tab:thresh_d1_c2_chi2_tex0_b}. 
}
\end{spacing}
\end{figure}

\begin{figure}[h!]
\caption{Empirical power functions for the TMG estimator of $\protect\beta %
_{01}$ $(E(\protect\beta_{i1})=\protect\beta_{01}=1)$ for different
threshold exponents, $\protect\alpha \in \{1/3,0.35,1/2\}$, in the baseline
DGP with one regressor, without time effects, but with correlated
heterogeneity, $\protect\psi _{\protect\beta_{1}}=0.5$}
\label{fig:tmg_alpha}\vspace{-5mm}
\par
\begin{center}
\includegraphics[scale=0.26]{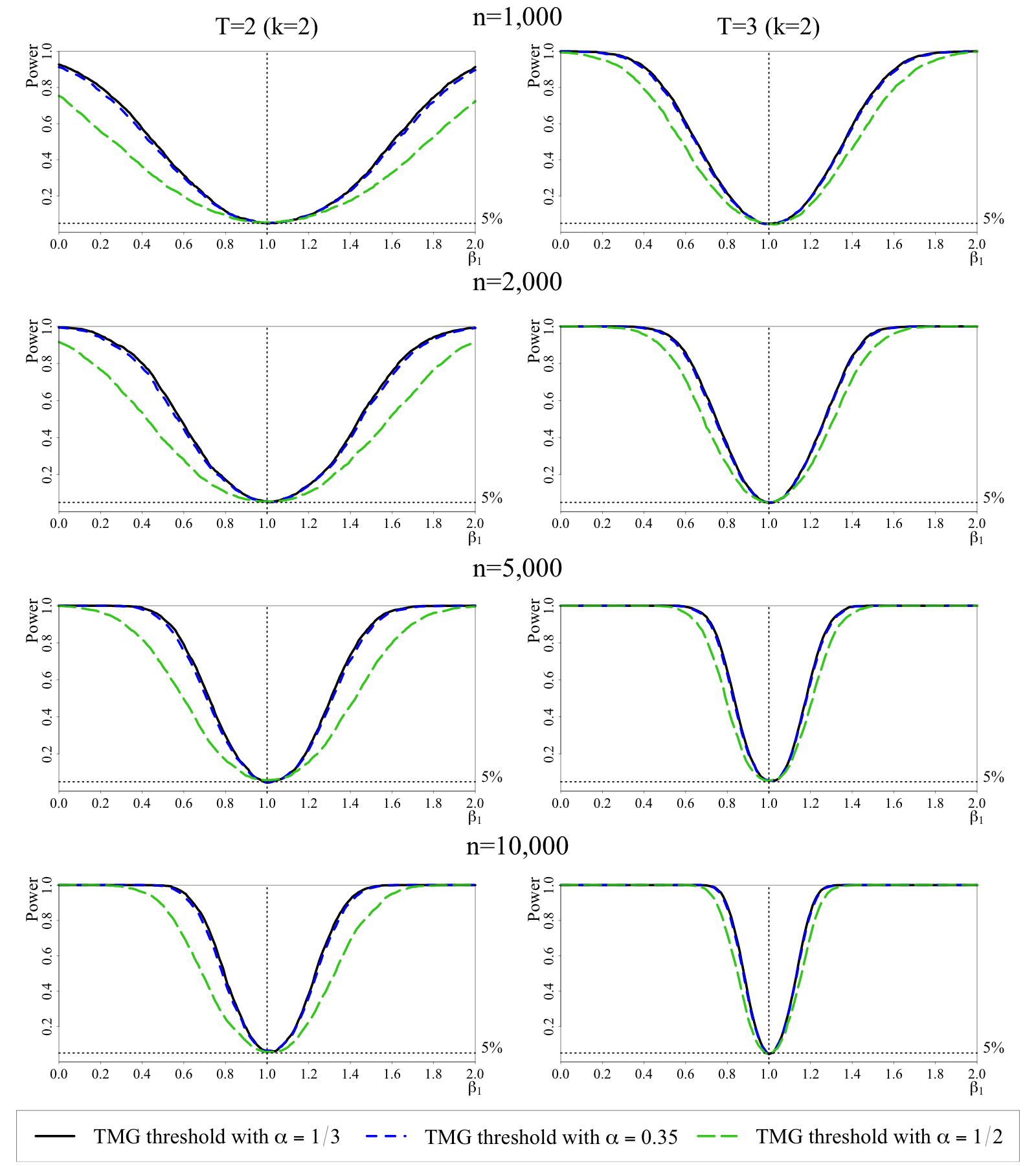}
\end{center}
\par
\vspace{-2mm} 
\begin{spacing}{1}
{\footnotesize 
Notes: See footnotes to Table \ref{tab:thresh_d1_c2_chi2_tex0_b}. 
}
\end{spacing}
\end{figure}

\begin{figure}[h]
\caption{Empirical power functions for the GP estimator of $\protect\beta %
_{01}$ $(E(\protect\beta_{i1})=\protect\beta_{01}=1)$ for different
threshold exponents, $\protect\alpha _{GP}\in \{0.35/2,1/4,1/3\}$ in the
baseline DGP with one regressor, without time effects, but with correlated
heterogeneity, $\protect\psi _{\protect\beta_{1}}=0.5$}
\label{fig:gp_alpha}\vspace{-5mm}
\par
\begin{center}
\includegraphics[scale=0.26]{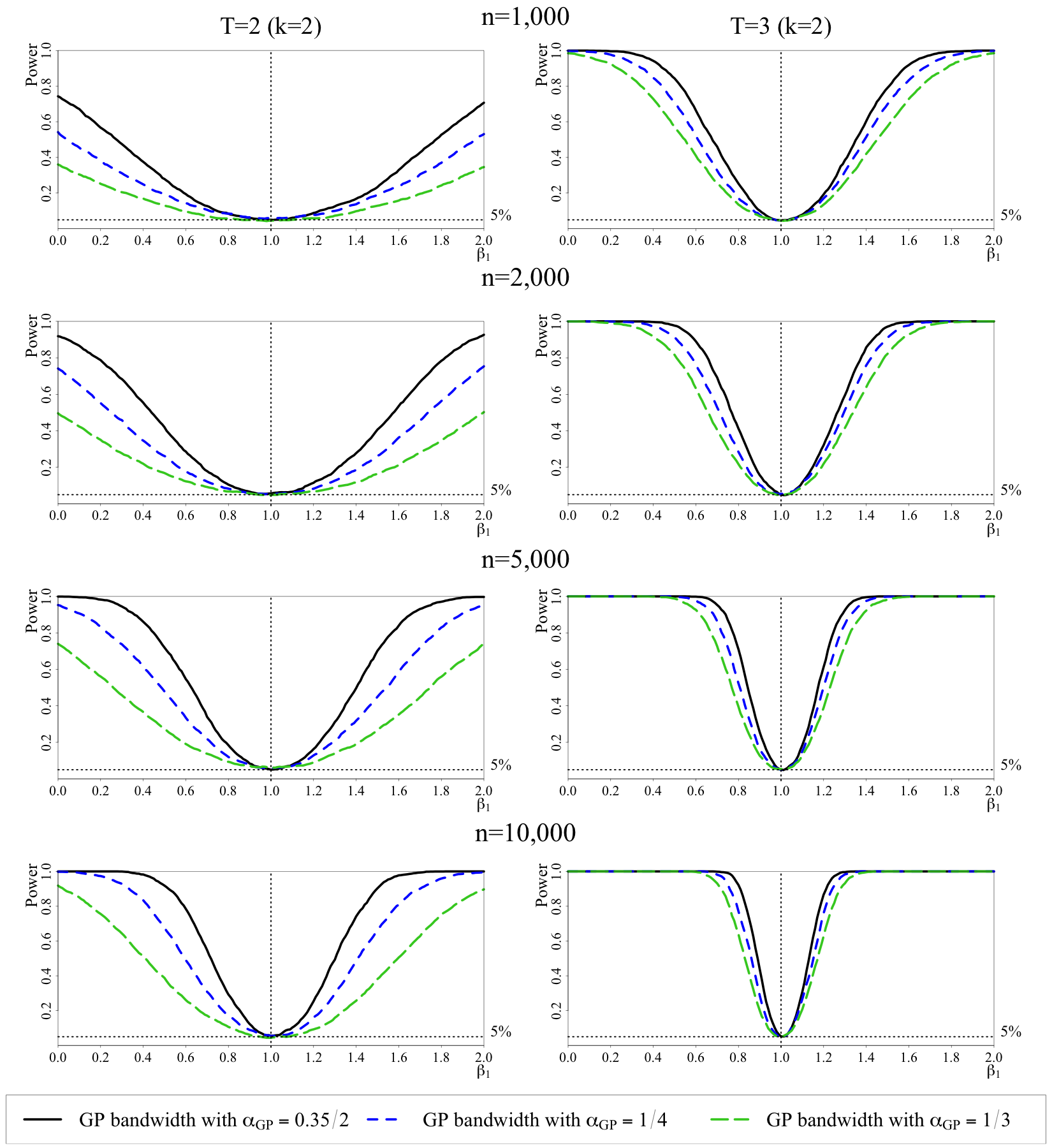}
\end{center}
\par
\vspace{-2mm} 
\begin{spacing}{1}
{\footnotesize 
Notes: See footnotes to Table \ref{tab:thresh_d1_c2_chi2_tex0_b}. 
}
\end{spacing}
\end{figure}

\newpage\clearpage

\subsubsection{Comparing TMG and GP estimators with correlated
heteroskedasticity}

Table \ref{tab:beta_d1_c2_h23_chi2_tex0_b} provides additional MC results on
small sample properties of TMG and GP estimators of $\beta_{01}$ in panel
data models with correlated heterogeneity, $\psi_{\beta_{1} }=0.5$, as well
as correlated error heteroskedasticity, which are generated as in case (a) $%
\sigma _{it}^{2}=\lambda _{i}^{2}$, and case (b) $\sigma
_{it}^{2}=e_{x1,it}^{2}$, for all $i$ and $t$. These results are to be
compared to the ones in Table \ref{tab:beta_d1_c2_chi2_tex0_b} in
sub-section \ref{MCcorr} of the main paper, which are for random
heteroskedasticity. The TMG estimator continues to perform better than the
GP estimator when $T=2$ and $3$, and allowing for correlated
heteroskedasticity does not alter this conclusion.

\begin{table}[ht]
\caption{Bias, RMSE and size of TMG and GP estimators of $\protect\beta_{01}$
$(E(\protect\beta_{i1})=\protect\beta_{01}=1)$ in panel data models with one
regressor, without time effects, but with correlated heterogeneity, $\protect%
\psi_{\protect\beta_{1}}=0.5$, and correlated heteroskedasticity (cases (a)
and (b))}
\label{tab:beta_d1_c2_h23_chi2_tex0_b}\vspace{-6mm}
\par
\begin{center}
\scalebox{0.7}{
\renewcommand{\arraystretch}{1.05}
\begin{tabular}{lcrrcccrrcc}
\hline\hline & & \multicolumn{4}{c}{$T=2$}                 &  & \multicolumn{4}{c}{$T=3$}             \\ \cline{3-6} \cline{8-11}
&  & \multicolumn{1}{c}{$\hat{\pi}$ }  & &   & Size   &  & \multicolumn{1}{c}{$\hat{\pi}$ } & & & Size   \\ 
Estimator &  & $(\times 100)$  &  Bias   & RMSE   &  $(\times 100)$  &  &$(\times 100)$ &   Bias   & RMSE  & $(\times 100)$  \\ \hline
  \multicolumn{10}{l}{\textbf{Case (a)}: $\sigma_{it}^{2}=\lambda_{i}^{2}$} \\ \hline 
  & \multicolumn{10}{c}{$n=1,000$}  \\ \hline 
TMG &  & 27.5 & 0.028 & 0.184 & 5.8 &  & 12.2 & 0.008 & 0.142 & 5.1 \\
GP &  & 4.1 & 0.003 & 0.444 & 4.7 &  & 1.3 & 0.002 & 0.203 & 4.8 \\
  & \multicolumn{10}{c}{$n=2,000$}  \\ \hline 
TMG &  & 24.6 & 0.016 & 0.139 & 5.8 &  & 9.8 & 0.006 & 0.107 & 4.8 \\
GP &  & 3.2 & 0.002 & 0.359 & 4.6 &  & 0.8 & -0.003 & 0.153 & 4.8 \\
  & \multicolumn{10}{c}{$n=5,000$}  \\ \hline 
TMG &  & 21.2 & 0.016 & 0.095 & 5.4 &  & 7.3 & 0.007 & 0.070 & 5.2 \\
GP &  & 2.4 & 0.006 & 0.276 & 5.8 &  & 0.4 & 0.001 & 0.103 & 5.1 \\
  & \multicolumn{10}{c}{$n=10,000$}  \\ \hline 
TMG &  & 19.0 & 0.015 & 0.074 & 6.6 &  & 5.9 & 0.006 & 0.053 & 5.2 \\
GP &  & 1.9 & 0.006 & 0.207 & 4.3 &  & 0.3 & 0.000 & 0.077 & 4.6 \\
[2mm]
    \multicolumn{10}{l}{\textbf{Case (b)}: $\sigma_{it}^{2}=e_{x1,it}^{2}$} \\ \hline 
  & \multicolumn{10}{c}{$n=1,000$}  \\ \hline 
 &  & \multicolumn{9}{l}{n=1,000} \\
TMG &  & 27.5 & 0.022 & 0.300 & 5.0 &  & 12.2 & 0.011 & 0.183 & 4.9 \\
GP &  & 4.1 & -0.017 & 0.679 & 5.2 &  & 1.3 & 0.006 & 0.234 & 4.6 \\
  & \multicolumn{10}{c}{$n=2,000$}  \\ \hline 
TMG &  & 24.6 & 0.010 & 0.220 & 4.6 &  & 9.8 & 0.005 & 0.135 & 5.1 \\
GP &  & 3.2 & -0.006 & 0.534 & 5.7 &  & 0.8 & -0.002 & 0.176 & 4.7 \\
  & \multicolumn{10}{c}{$n=5,000$}  \\ \hline 
TMG &  & 21.2 & 0.019 & 0.151 & 4.6 &  & 7.3 & 0.007 & 0.088 & 4.8 \\
GP &  & 2.4 & 0.011 & 0.387 & 5.1 &  & 0.4 & 0.002 & 0.118 & 5.2 \\
  & \multicolumn{10}{c}{$n=10,000$}  \\ \hline 
TMG &  & 19.0 & 0.016 & 0.117 & 6.1 &  & 5.9 & 0.006 & 0.066 & 5.0 \\
GP &  & 1.9 & -0.001 & 0.300 & 4.4 &  & 0.3 & 0.001 & 0.087 & 4.7 \\
\hline\hline
\end{tabular}
}
\end{center}
\par
\vspace{-3mm} 
\begin{spacing}{0.92}
{\footnotesize 
Notes: 
(i) The data generating process is given by $y_{it}=\alpha_{i} + \beta_{i1}
x_{1,it} + \sigma_{it} e_{it}$, where $\sigma_{it}^{2}$ are generated as case (a): $ \sigma_{it}^{2} =\lambda_{i}^{2}$, and case (b): $\sigma_{it}^{2}=e_{x1,it}^{2}$, for all $i$ and $t$. $\psi_{\beta_{1}}$ is defined in (\ref{eta_i}) in the main paper. 
For further details, see Section \ref{DGP} in the main paper and Section \ref{secMCDGP}. 
(ii) For the TMG estimator, see footnotes to Figure \ref{fig:fe_tmg_k2_base_2}. 
For the GP estimator, see footnote (i) to Table \ref{tab:thresh_d1_c2_chi2_tex0_b}. 
$\hat{\pi}$ is the simulated fraction of individual estimates being trimmed, given by (\ref{pin}) in the main paper.}
\end{spacing}
\end{table}

\newpage \clearpage

\subsection{Monte Carlo evidence for panels with two and three regressors in
the baseline DGP without time effects\label{MCk34}}

For $k^{\prime}=2$ and 3, the results of FE, MG, and TMG estimators in the
baseline DGP without time effects are reported in Tables \ref%
{tab:T_d1_c12_chi2_tex0_k3} and \ref{tab:T_d1_c12_chi2_tex0_k4} and the
empirical power functions are plotted in Figures \ref{fig:fe_tmg_k3_base_1}
through \ref{fig:fe_tmg_k4_base_2}.

Tables \ref{tab:beta_d1_c2_chi2_tex0_b_k3} and \ref%
{tab:beta_d1_c2_chi2_tex0_b_k4} summarize the results of the TMG and GP
estimators in the baseline DGP without time effects and with correlated
heterogeneity and multiple regressors, $k^{\prime}=2$ and 3. The respective
empirical power functions are shown in Figures \ref{fig:tmg_gp_k3_base} and %
\ref{fig:tmg_gp_k4_base} for $n=10,000$.

\newpage \clearpage
\begin{sidewaystable}
\caption{Bias, RMSE and size of FE, MG and TMG estimators of $\beta_{01}$ $(E(\beta_{i1}) = \beta_{01}=1)$ in the baseline DGP with two regressors, without time effects} 
\label{tab:T_d1_c12_chi2_tex0_k3}
\vspace{-6mm}
\begin{center}
\scalebox{0.65}{
\begin{tabular}{rrcrrrrrrrrrrrrrrrrrrrrrrrrr}
 \hline\hline &  \multicolumn{13}{c}{ Uncorrelated heterogeneity: $\psi_{\beta_{1}} = 0 $} &  & \multicolumn{13}{c}{ Correlated heterogeneity: $\psi_{\beta_{1}}=0.5$}  \\ \cline{2-14} \cline{16-28}   
& $\hat{\pi}$ $(\times 100)$  & & \multicolumn{3}{c}{Bias}  & & \multicolumn{3}{c}{RMSE} && \multicolumn{3}{c}{Size $(\times 100)$} && $\hat{\pi}$ $(\times 100)$ & & \multicolumn{3}{c}{Bias}  & & \multicolumn{3}{c}{RMSE} && \multicolumn{3}{c}{Size $(\times 100)$} \\ \cline{2-2} \cline{4-6} \cline{8-10} \cline{12-14} \cline{16-16} \cline{18-20} \cline{22-24} \cline{26-28}  
 $T$ & TMG & & FE & MG & TMG && FE & MG & TMG && FE & MG & TMG &&  TMG & & FE & MG & TMG && FE & MG & TMG && FE & MG & TMG    \\ \hline 
&   \multicolumn{27}{c}{$n=1,000$} \\ \hline
3 & 41.80 &  & 0.000 & -1.007 & 0.002 &  & 0.106 & 97.026 & 0.277 &  & 5.2 & 1.4 & 4.6 &  & 41.80 &  & 0.354 & -1.142 & 0.040 &  & 0.377 & 110.092 & 0.317 &  & 79.4 & 1.5 & 5.0 \\
4 & 25.30 &  & 0.002 & -0.015 & 0.002 &  & 0.088 & 0.446 & 0.159 &  & 5.3 & 4.0 & 5.1 &  & 25.30 &  & 0.356 & -0.016 & 0.027 &  & 0.373 & 0.504 & 0.183 &  & 90.0 & 3.8 & 5.7 \\
5 & 16.90 &  & 0.002 & -0.002 & 0.000 &  & 0.078 & 0.157 & 0.119 &  & 5.4 & 5.3 & 5.3 &  & 16.90 &  & 0.355 & -0.003 & 0.016 &  & 0.370 & 0.176 & 0.135 &  & 94.8 & 5.1 & 4.9 \\
6 & 12.00 &  & -0.001 & 0.001 & 0.001 &  & 0.072 & 0.108 & 0.097 &  & 5.4 & 4.7 & 4.9 &  & 12.00 &  & 0.352 & 0.002 & 0.013 &  & 0.364 & 0.121 & 0.110 &  & 97.2 & 4.6 & 5.2 \\
7 & 8.90 &  & -0.001 & -0.001 & -0.001 &  & 0.068 & 0.089 & 0.084 &  & 5.6 & 4.1 & 4.2 &  & 8.90 &  & 0.352 & -0.001 & 0.007 &  & 0.364 & 0.100 & 0.095 &  & 98.7 & 4.4 & 4.8 \\
8 & 6.70 &  & 0.000 & -0.003 & -0.002 &  & 0.064 & 0.080 & 0.078 &  & 4.6 & 5.2 & 4.8 &  & 6.70 &  & 0.352 & -0.003 & 0.004 &  & 0.363 & 0.090 & 0.087 &  & 99.3 & 5.2 & 4.9 \\
10 & 4.10 &  & 0.001 & 0.002 & 0.002 &  & 0.058 & 0.067 & 0.066 &  & 5.2 & 5.0 & 5.0 &  & 4.10 &  & 0.353 & 0.002 & 0.005 &  & 0.361 & 0.075 & 0.074 &  & 99.6 & 5.4 & 5.2 \\
15 & 1.50 &  & 0.000 & 0.000 & 0.000 &  & 0.050 & 0.053 & 0.052 &  & 5.7 & 4.5 & 4.5 &  & 1.50 &  & 0.353 & 0.000 & 0.001 &  & 0.360 & 0.058 & 0.057 &  & 99.9 & 5.1 & 4.9 \\
&   \multicolumn{27}{c}{$n=2,000$} \\ \hline
3 & 38.30 &  & -0.001 & -4.744 & -0.003 &  & 0.073 & 101.175 & 0.207 &  & 4.7 & 1.6 & 5.1 &  & 38.30 &  & 0.352 & -5.383 & 0.030 &  & 0.364 & 114.798 & 0.236 &  & 97.6 & 1.6 & 5.7 \\
4 & 21.60 &  & -0.001 & 0.009 & 0.000 &  & 0.063 & 0.319 & 0.118 &  & 5.2 & 4.6 & 5.3 &  & 21.60 &  & 0.352 & 0.010 & 0.021 &  & 0.362 & 0.362 & 0.134 &  & 99.4 & 4.6 & 5.7 \\
5 & 13.50 &  & -0.001 & 0.000 & -0.001 &  & 0.055 & 0.108 & 0.085 &  & 4.9 & 5.4 & 5.2 &  & 13.50 &  & 0.352 & 0.000 & 0.012 &  & 0.359 & 0.122 & 0.097 &  & 99.9 & 5.4 & 5.2 \\
6 & 9.00 &  & 0.000 & 0.000 & -0.001 &  & 0.051 & 0.075 & 0.068 &  & 5.5 & 4.3 & 4.2 &  & 9.00 &  & 0.353 & -0.001 & 0.007 &  & 0.359 & 0.085 & 0.077 &  & 100.0 & 4.9 & 4.2 \\
7 & 6.30 &  & 0.000 & 0.000 & 0.001 &  & 0.047 & 0.065 & 0.062 &  & 4.7 & 5.7 & 5.5 &  & 6.30 &  & 0.355 & 0.001 & 0.007 &  & 0.361 & 0.072 & 0.070 &  & 100.0 & 6.0 & 5.2 \\
8 & 4.50 &  & 0.000 & 0.001 & 0.001 &  & 0.043 & 0.056 & 0.054 &  & 4.0 & 4.3 & 4.1 &  & 4.50 &  & 0.353 & 0.001 & 0.005 &  & 0.358 & 0.062 & 0.061 &  & 100.0 & 4.5 & 4.5 \\
10 & 2.40 &  & 0.001 & 0.002 & 0.002 &  & 0.041 & 0.047 & 0.046 &  & 5.2 & 4.4 & 4.4 &  & 2.40 &  & 0.355 & 0.002 & 0.004 &  & 0.359 & 0.052 & 0.051 &  & 100.0 & 4.3 & 4.4 \\
15 & 0.70 &  & 0.000 & 0.000 & 0.000 &  & 0.035 & 0.037 & 0.037 &  & 4.5 & 5.3 & 5.2 &  & 0.70 &  & 0.354 & 0.001 & 0.001 &  & 0.358 & 0.041 & 0.041 &  & 100.0 & 5.1 & 5.1 \\
&   \multicolumn{27}{c}{$n=5,000$} \\ \hline
3 & 34.10 &  & 0.000 & -4.523 & -0.002 &  & 0.047 & 160.307 & 0.133 &  & 4.9 & 2.2 & 4.3 &  & 34.10 &  & 0.352 & -5.132 & 0.026 &  & 0.357 & 181.894 & 0.153 &  & 100.0 & 2.2 & 4.5 \\
4 & 17.30 &  & 0.000 & -0.004 & 0.000 &  & 0.040 & 0.178 & 0.075 &  & 5.4 & 4.3 & 4.9 &  & 17.30 &  & 0.354 & -0.004 & 0.015 &  & 0.358 & 0.201 & 0.086 &  & 100.0 & 4.4 & 5.4 \\
5 & 9.90 &  & -0.002 & -0.001 & -0.001 &  & 0.035 & 0.066 & 0.055 &  & 4.9 & 4.8 & 4.6 &  & 9.90 &  & 0.351 & -0.001 & 0.008 &  & 0.354 & 0.075 & 0.063 &  & 100.0 & 4.7 & 5.0 \\
6 & 6.00 &  & -0.001 & -0.002 & -0.002 &  & 0.032 & 0.049 & 0.045 &  & 5.3 & 5.6 & 4.9 &  & 6.00 &  & 0.353 & -0.003 & 0.003 &  & 0.355 & 0.055 & 0.051 &  & 100.0 & 5.2 & 4.6 \\
7 & 3.80 &  & -0.001 & -0.001 & -0.002 &  & 0.031 & 0.041 & 0.040 &  & 5.9 & 5.2 & 5.2 &  & 3.80 &  & 0.353 & -0.001 & 0.002 &  & 0.355 & 0.046 & 0.045 &  & 100.0 & 5.3 & 5.3 \\
8 & 2.50 &  & 0.001 & 0.001 & 0.001 &  & 0.027 & 0.035 & 0.035 &  & 4.3 & 4.5 & 4.4 &  & 2.50 &  & 0.354 & 0.001 & 0.003 &  & 0.356 & 0.040 & 0.039 &  & 100.0 & 5.2 & 5.1 \\
10 & 1.20 &  & 0.000 & 0.000 & 0.000 &  & 0.026 & 0.030 & 0.030 &  & 5.6 & 5.2 & 5.6 &  & 1.20 &  & 0.353 & 0.000 & 0.001 &  & 0.354 & 0.033 & 0.033 &  & 100.0 & 5.2 & 5.3 \\
15 & 0.20 &  & 0.000 & 0.000 & 0.000 &  & 0.022 & 0.023 & 0.023 &  & 4.8 & 3.8 & 3.9 &  & 0.20 &  & 0.353 & 0.000 & 0.000 &  & 0.354 & 0.025 & 0.025 &  & 100.0 & 4.0 & 3.9 \\
&   \multicolumn{27}{c}{$n=10,000$} \\ \hline
3 & 31.10 &  & -0.002 & 0.104 & -0.003 &  & 0.033 & 53.326 & 0.101 &  & 5.9 & 2.1 & 5.6 &  & 31.10 &  & 0.351 & 0.118 & 0.021 &  & 0.354 & 60.507 & 0.117 &  & 100.0 & 2.1 & 5.4 \\
4 & 14.60 &  & 0.001 & -0.004 & -0.001 &  & 0.028 & 0.128 & 0.056 &  & 5.2 & 5.0 & 5.3 &  & 14.60 &  & 0.355 & -0.004 & 0.012 &  & 0.357 & 0.145 & 0.064 &  & 100.0 & 5.2 & 5.9 \\
5 & 7.70 &  & -0.001 & 0.000 & 0.000 &  & 0.024 & 0.048 & 0.040 &  & 4.2 & 4.6 & 4.4 &  & 7.70 &  & 0.353 & 0.000 & 0.007 &  & 0.354 & 0.054 & 0.045 &  & 100.0 & 4.7 & 4.3 \\
6 & 4.30 &  & 0.000 & 0.000 & 0.000 &  & 0.023 & 0.035 & 0.033 &  & 5.2 & 4.7 & 5.7 &  & 4.30 &  & 0.354 & 0.000 & 0.004 &  & 0.355 & 0.039 & 0.037 &  & 100.0 & 5.2 & 5.3 \\
7 & 2.60 &  & 0.000 & 0.000 & 0.000 &  & 0.020 & 0.028 & 0.027 &  & 4.6 & 4.8 & 5.0 &  & 2.60 &  & 0.353 & 0.000 & 0.002 &  & 0.354 & 0.031 & 0.031 &  & 100.0 & 5.6 & 5.0 \\
8 & 1.60 &  & -0.001 & -0.001 & -0.001 &  & 0.020 & 0.025 & 0.025 &  & 5.4 & 5.1 & 5.6 &  & 1.60 &  & 0.353 & -0.001 & 0.000 &  & 0.354 & 0.028 & 0.028 &  & 100.0 & 5.4 & 5.1 \\
10 & 0.60 &  & 0.000 & -0.001 & -0.001 &  & 0.018 & 0.020 & 0.020 &  & 4.8 & 3.6 & 3.6 &  & 0.60 &  & 0.353 & -0.001 & 0.000 &  & 0.354 & 0.023 & 0.022 &  & 100.0 & 3.9 & 3.6 \\
15 & 0.10 &  & -0.001 & 0.000 & 0.000 &  & 0.016 & 0.017 & 0.017 &  & 4.9 & 5.8 & 5.8 &  & 0.10 &  & 0.353 & 0.000 & 0.000 &  & 0.354 & 0.019 & 0.019 &  & 100.0 & 5.9 & 5.9\\
   \hline
\hline
\end{tabular}
}
\end{center}
\vspace{-3mm}
{\footnotesize
Notes: 
(i) For the DGP see Section \ref{DGP} in the main paper and Section \ref{secMCDGP}. 
(ii) FE and MG estimators are given by (\ref{fee}) and (\ref{mge}), respectively, in the main paper. 
The TMG estimator and its variance are given by (\ref{TMGb}) and (\ref{varC}) in the main paper. 
(iii) The trimming threshold value for the TMG estimator is given by $a_{n}=\bar{d}_{n} n^{-\alpha}$, where $\bar{d}_{n} =\frac{1}{n} \sum_{i}^{n} d_{i}$, $d_{i} =\func{det}(\boldsymbol{X}_{i}^{\prime} \boldsymbol{M}_{T} \boldsymbol{X}_{i})$, $\boldsymbol{X}_{i}=(\boldsymbol{x}_{i1},\boldsymbol{x}_{i2},...,\boldsymbol{x}_{iT})^{\prime}$ and $\boldsymbol{M}_{T} = \boldsymbol{I}_{T} - \boldsymbol{\tau}_{T}\boldsymbol{\tau}_{T}^{\prime}/T$. 
$\alpha$ is set to $1/3$. 
$\hat{\pi}$ is the simulated fraction of individual estimates being trimmed, defined by (\ref{pin}) in the main paper. 
}
\end{sidewaystable}

\newpage \clearpage
\begin{sidewaystable}
\caption{Bias, RMSE and size of FE, MG and TMG estimators of $\beta_{01}$ $(E(\beta_{i1}) = \beta_{01}=1)$ in the baseline DGP with three regressors, without time effects} 
\label{tab:T_d1_c12_chi2_tex0_k4}
\vspace{-6mm}
\begin{center}
\scalebox{0.65}{
\begin{tabular}{rrcrrrrrrrrrrrrrrrrrrrrrrrrr}
 \hline\hline &  \multicolumn{13}{c}{ Uncorrelated heterogeneity: $\psi_{\beta_{1}} = 0 $} &  & \multicolumn{13}{c}{ Correlated heterogeneity: $\psi_{\beta_{1}}=0.5$}  \\ \cline{2-14} \cline{16-28}   
& $\hat{\pi}$ $(\times 100)$  & & \multicolumn{3}{c}{Bias}  & & \multicolumn{3}{c}{RMSE} && \multicolumn{3}{c}{Size $(\times 100)$} && $\hat{\pi}$ $(\times 100)$ & & \multicolumn{3}{c}{Bias}  & & \multicolumn{3}{c}{RMSE} && \multicolumn{3}{c}{Size $(\times 100)$} \\ \cline{2-2} \cline{4-6} \cline{8-10} \cline{12-14} \cline{16-16} \cline{18-20} \cline{22-24} \cline{26-28}  
 $T$ & TMG & & FE & MG & TMG && FE & MG & TMG && FE & MG & TMG &&  TMG & & FE & MG & TMG && FE & MG & TMG && FE & MG & TMG    \\ \hline 
&   \multicolumn{27}{c}{$n=1,000$} \\ \hline
4 & 50.80 &  & 0.001 & -0.967 & 0.005 &  & 0.090 & 116.328 & 0.284 &  & 5.5 & 2.7 & 4.4 &  & 50.80 &  & 0.353 & -1.096 & 0.054 &  & 0.372 & 131.868 & 0.326 &  & 89.3 & 2.7 & 5.3 \\
5 & 35.50 &  & 0.002 & 0.003 & 0.008 &  & 0.076 & 0.387 & 0.160 &  & 4.0 & 4.2 & 4.6 &  & 35.50 &  & 0.357 & 0.004 & 0.045 &  & 0.371 & 0.439 & 0.186 &  & 95.8 & 4.2 & 5.8 \\
6 & 27.10 &  & 0.003 & 0.004 & 0.003 &  & 0.073 & 0.150 & 0.119 &  & 5.1 & 4.4 & 4.9 &  & 27.10 &  & 0.354 & 0.004 & 0.031 &  & 0.366 & 0.169 & 0.138 &  & 97.6 & 4.6 & 5.5 \\
7 & 21.60 &  & -0.002 & -0.002 & -0.003 &  & 0.068 & 0.108 & 0.098 &  & 5.7 & 4.5 & 5.0 &  & 21.60 &  & 0.351 & -0.003 & 0.018 &  & 0.362 & 0.122 & 0.112 &  & 98.8 & 4.9 & 5.7 \\
8 & 17.90 &  & 0.001 & 0.001 & 0.002 &  & 0.065 & 0.091 & 0.086 &  & 5.6 & 5.0 & 5.8 &  & 17.90 &  & 0.354 & 0.002 & 0.020 &  & 0.364 & 0.102 & 0.099 &  & 99.1 & 5.2 & 6.0 \\
9 & 15.20 &  & -0.001 & -0.001 & -0.001 &  & 0.061 & 0.078 & 0.075 &  & 5.6 & 4.9 & 4.9 &  & 15.20 &  & 0.353 & -0.001 & 0.014 &  & 0.362 & 0.087 & 0.085 &  & 99.2 & 4.8 & 4.6 \\
10 & 13.00 &  & 0.001 & 0.001 & 0.001 &  & 0.060 & 0.072 & 0.070 &  & 5.7 & 4.8 & 4.6 &  & 13.00 &  & 0.353 & 0.000 & 0.012 &  & 0.362 & 0.081 & 0.079 &  & 99.8 & 4.4 & 4.7 \\
15 & 7.40 &  & -0.002 & -0.003 & -0.003 &  & 0.050 & 0.055 & 0.054 &  & 4.9 & 4.7 & 4.8 &  & 7.40 &  & 0.350 & -0.003 & 0.003 &  & 0.357 & 0.060 & 0.059 &  & 99.9 & 5.1 & 4.6 \\
&   \multicolumn{27}{c}{$n=2,000$} \\ \hline
4 & 47.30 &  & 0.000 & -0.833 & -0.004 &  & 0.064 & 50.815 & 0.206 &  & 5.8 & 2.1 & 4.2 &  & 47.30 &  & 0.352 & -0.945 & 0.038 &  & 0.361 & 57.603 & 0.236 &  & 99.5 & 2.1 & 4.3 \\
5 & 31.50 &  & 0.000 & -0.002 & 0.001 &  & 0.057 & 0.383 & 0.120 &  & 5.6 & 5.1 & 5.7 &  & 31.50 &  & 0.353 & -0.002 & 0.032 &  & 0.360 & 0.433 & 0.139 &  & 100.0 & 4.8 & 6.7 \\
6 & 23.00 &  & 0.001 & 0.001 & 0.000 &  & 0.052 & 0.106 & 0.086 &  & 5.7 & 4.6 & 4.7 &  & 23.00 &  & 0.354 & 0.002 & 0.022 &  & 0.360 & 0.120 & 0.100 &  & 100.0 & 4.7 & 5.3 \\
7 & 17.60 &  & -0.001 & -0.001 & -0.002 &  & 0.048 & 0.077 & 0.069 &  & 5.5 & 5.2 & 5.1 &  & 17.60 &  & 0.352 & -0.001 & 0.016 &  & 0.358 & 0.087 & 0.080 &  & 100.0 & 5.3 & 5.4 \\
8 & 14.00 &  & -0.001 & -0.002 & -0.002 &  & 0.044 & 0.064 & 0.060 &  & 4.6 & 5.0 & 5.0 &  & 14.00 &  & 0.353 & -0.002 & 0.011 &  & 0.359 & 0.072 & 0.069 &  & 100.0 & 5.6 & 5.9 \\
9 & 11.50 &  & -0.001 & 0.000 & 0.000 &  & 0.043 & 0.058 & 0.056 &  & 4.7 & 6.2 & 5.9 &  & 11.50 &  & 0.353 & 0.000 & 0.011 &  & 0.358 & 0.065 & 0.063 &  & 100.0 & 6.0 & 5.9 \\
10 & 9.60 &  & 0.001 & 0.001 & 0.001 &  & 0.040 & 0.051 & 0.050 &  & 4.2 & 5.4 & 5.5 &  & 9.60 &  & 0.353 & 0.001 & 0.010 &  & 0.358 & 0.057 & 0.056 &  & 100.0 & 5.0 & 5.5 \\
15 & 4.70 &  & 0.001 & 0.000 & 0.000 &  & 0.036 & 0.039 & 0.038 &  & 4.8 & 4.3 & 5.0 &  & 4.70 &  & 0.354 & 0.000 & 0.003 &  & 0.358 & 0.043 & 0.043 &  & 100.0 & 5.3 & 5.2 \\
&   \multicolumn{27}{c}{$n=5,000$} \\ \hline
4 & 42.90 &  & -0.001 & -0.736 & -0.005 &  & 0.039 & 49.455 & 0.141 &  & 4.6 & 2.2 & 5.1 &  & 42.90 &  & 0.353 & -0.834 & 0.032 &  & 0.356 & 56.062 & 0.163 &  & 100.0 & 2.2 & 5.9 \\
5 & 26.50 &  & -0.002 & -0.008 & -0.002 &  & 0.035 & 0.177 & 0.076 &  & 4.9 & 4.2 & 4.6 &  & 26.50 &  & 0.352 & -0.009 & 0.022 &  & 0.354 & 0.201 & 0.089 &  & 100.0 & 4.2 & 5.6 \\
6 & 18.00 &  & 0.001 & 0.001 & 0.000 &  & 0.032 & 0.068 & 0.054 &  & 5.1 & 5.1 & 5.1 &  & 18.00 &  & 0.354 & 0.001 & 0.017 &  & 0.356 & 0.077 & 0.063 &  & 100.0 & 4.9 & 5.9 \\
7 & 13.10 &  & 0.000 & 0.000 & 0.000 &  & 0.030 & 0.049 & 0.044 &  & 4.9 & 4.8 & 5.0 &  & 13.10 &  & 0.353 & 0.000 & 0.012 &  & 0.355 & 0.055 & 0.051 &  & 100.0 & 4.7 & 5.1 \\
8 & 9.80 &  & -0.001 & -0.001 & -0.001 &  & 0.029 & 0.040 & 0.038 &  & 5.1 & 4.6 & 4.8 &  & 9.80 &  & 0.351 & -0.001 & 0.008 &  & 0.353 & 0.045 & 0.043 &  & 100.0 & 4.6 & 4.7 \\
9 & 7.60 &  & 0.001 & 0.001 & 0.001 &  & 0.027 & 0.036 & 0.035 &  & 5.1 & 6.1 & 6.0 &  & 7.60 &  & 0.355 & 0.001 & 0.008 &  & 0.357 & 0.041 & 0.041 &  & 100.0 & 6.2 & 6.5 \\
10 & 6.00 &  & -0.001 & 0.000 & 0.000 &  & 0.026 & 0.031 & 0.031 &  & 5.5 & 4.9 & 4.6 &  & 6.00 &  & 0.352 & 0.000 & 0.005 &  & 0.354 & 0.035 & 0.035 &  & 100.0 & 4.6 & 4.9 \\
15 & 2.40 &  & -0.001 & -0.001 & -0.001 &  & 0.023 & 0.025 & 0.025 &  & 5.9 & 5.8 & 5.9 &  & 2.40 &  & 0.353 & -0.001 & 0.001 &  & 0.354 & 0.028 & 0.028 &  & 100.0 & 6.0 & 6.0 \\
&   \multicolumn{27}{c}{$n=10,000$} \\ \hline
4 & 39.80 &  & -0.001 & -0.828 & -0.003 &  & 0.028 & 135.675 & 0.105 &  & 5.7 & 2.2 & 5.4 &  & 39.80 &  & 0.353 & -0.939 & 0.030 &  & 0.355 & 153.799 & 0.123 &  & 100.0 & 2.2 & 6.2 \\
5 & 23.10 &  & -0.001 & 0.002 & 0.000 &  & 0.025 & 0.132 & 0.055 &  & 5.1 & 5.2 & 4.9 &  & 23.10 &  & 0.352 & 0.002 & 0.022 &  & 0.354 & 0.150 & 0.065 &  & 100.0 & 5.4 & 6.1 \\
6 & 14.90 &  & 0.000 & -0.001 & 0.000 &  & 0.023 & 0.048 & 0.039 &  & 5.6 & 5.9 & 4.0 &  & 14.90 &  & 0.353 & -0.001 & 0.014 &  & 0.355 & 0.054 & 0.046 &  & 100.0 & 5.8 & 5.3 \\
7 & 10.20 &  & 0.000 & 0.000 & 0.000 &  & 0.022 & 0.035 & 0.032 &  & 5.3 & 5.3 & 4.9 &  & 10.20 &  & 0.353 & 0.000 & 0.009 &  & 0.354 & 0.039 & 0.037 &  & 100.0 & 5.6 & 6.0 \\
8 & 7.30 &  & 0.000 & 0.000 & 0.000 &  & 0.020 & 0.029 & 0.028 &  & 5.7 & 5.8 & 5.6 &  & 7.30 &  & 0.354 & 0.000 & 0.007 &  & 0.355 & 0.033 & 0.032 &  & 100.0 & 6.1 & 6.1 \\
9 & 5.40 &  & 0.000 & 0.000 & 0.000 &  & 0.019 & 0.025 & 0.024 &  & 4.4 & 4.4 & 4.7 &  & 5.40 &  & 0.353 & 0.000 & 0.004 &  & 0.354 & 0.027 & 0.027 &  & 100.0 & 4.4 & 5.1 \\
10 & 4.10 &  & -0.001 & 0.000 & -0.001 &  & 0.019 & 0.023 & 0.023 &  & 6.4 & 5.6 & 5.4 &  & 4.10 &  & 0.352 & 0.000 & 0.003 &  & 0.353 & 0.026 & 0.026 &  & 100.0 & 5.6 & 5.9 \\
15 & 1.30 &  & 0.000 & 0.000 & 0.000 &  & 0.016 & 0.018 & 0.018 &  & 4.7 & 5.3 & 5.4 &  & 1.30 &  & 0.354 & 0.000 & 0.001 &  & 0.354 & 0.019 & 0.019 &  & 100.0 & 5.1 & 5.2\\
   \hline
\hline
\end{tabular}
}
\end{center}
\vspace{-3mm}
{\footnotesize
Notes: 
(i) For the DGP see Section \ref{DGP} in the main paper and Section \ref{secMCDGP}. 
(ii) FE and MG estimators are given by (\ref{fee}) and (\ref{mge}), respectively, in the main paper. 
The TMG estimator and its variance are given by (\ref{TMGb}) and (\ref{varC}) in the main paper. 
(iii) The trimming threshold value for the TMG estimator is given by $a_{n}=\bar{d}_{n} n^{-\alpha}$, where $\bar{d}_{n} =\frac{1}{n} \sum_{i}^{n} d_{i}$, $d_{i} =\func{det}(\boldsymbol{X}_{i}^{\prime} \boldsymbol{M}_{T} \boldsymbol{X}_{i})$, $\boldsymbol{X}_{i}=(\boldsymbol{x}_{i1},\boldsymbol{x}_{i2},...,\boldsymbol{x}_{iT})^{\prime}$ and $\boldsymbol{M}_{T} = \boldsymbol{I}_{T} - \boldsymbol{\tau}_{T}\boldsymbol{\tau}_{T}^{\prime}/T$. 
$\alpha$ is set to $1/3$. 
$\hat{\pi}$ is the simulated fraction of individual estimates being trimmed, defined by (\ref{pin}) in the main paper. 
}
\end{sidewaystable}

\newpage \clearpage
\begin{figure}[h!]
\caption{Empirical power functions for FE, MG and TMG estimators of $\protect%
\beta_{01}$ $(E(\protect\beta_{i1}) = \protect\beta_{01}=1)$ in the baseline
DGP with two regressors, without time effects, for $n=10,000$ and $T=3,4,5$}
\label{fig:fe_tmg_k3_base_1}\vspace{-6mm}
\par
\begin{center}
\includegraphics[scale=0.28]{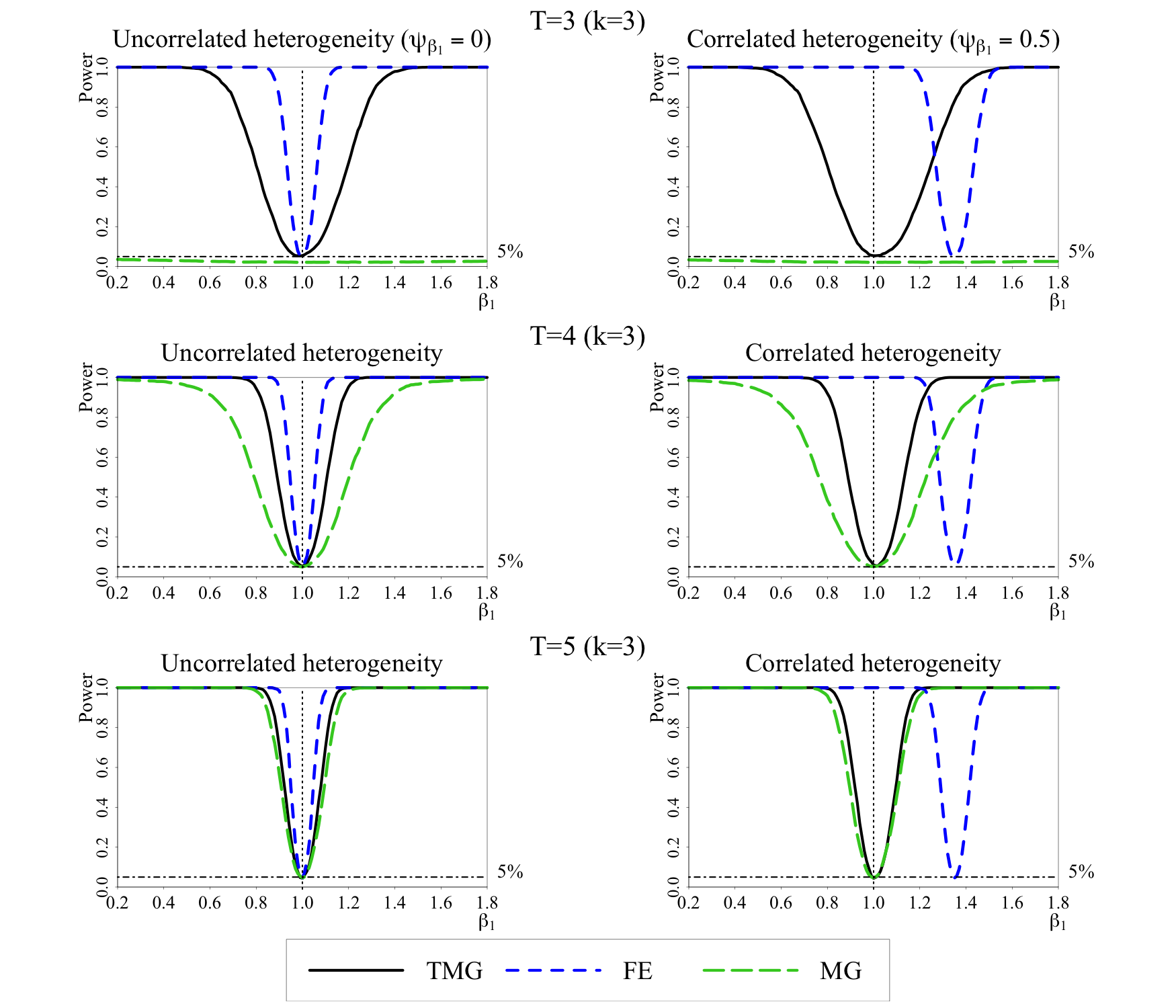}
\end{center}
\par
\vspace{-2mm} 
\begin{spacing}{1}
{\footnotesize
Notes: See footnotes to Table \ref{tab:T_d1_c12_chi2_tex0_k3}. }
\end{spacing}
\end{figure}

\begin{figure}[h!]
\caption{Empirical power functions for FE, MG and TMG estimators of $\protect%
\beta_{01}$ $(E(\protect\beta_{i1}) = \protect\beta_{01}=1)$ in the baseline
DGP with two regressors, without time effects, for $n=10,000$ and $T=6,7,8$}
\label{fig:fe_tmg_k3_base_2}\vspace{-6mm}
\par
\begin{center}
\includegraphics[scale=0.28]{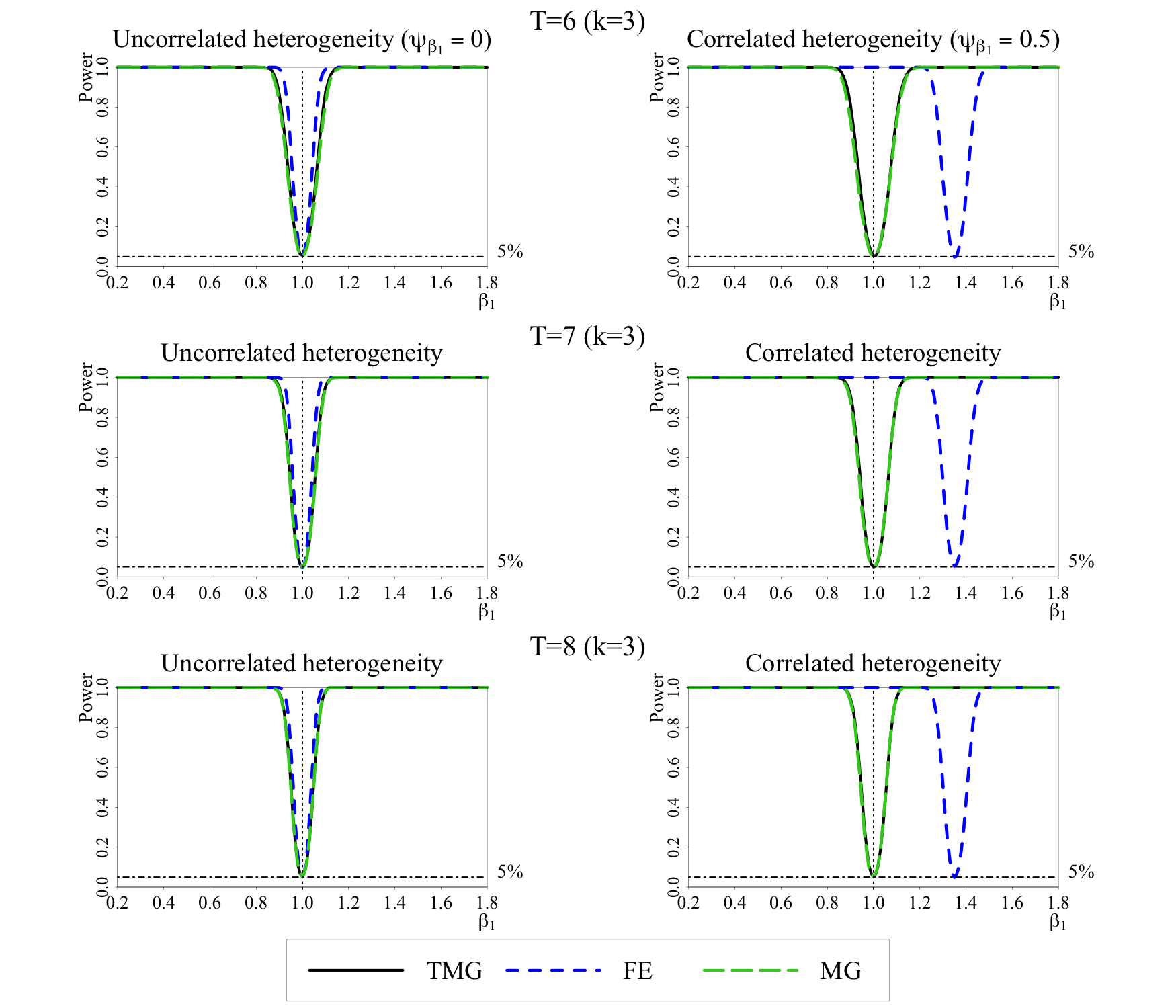}
\end{center}
\par
\vspace{-2mm} 
\begin{spacing}{1}
{\footnotesize
Notes: See footnotes to Table \ref{tab:T_d1_c12_chi2_tex0_k3}. }
\end{spacing}
\end{figure}

\newpage\clearpage

\begin{figure}[h!]
\caption{Empirical power functions for FE, MG and TMG estimators of $\protect%
\beta_{01}$ $(E(\protect\beta_{i1}) = \protect\beta_{01}=1)$ in the baseline
DGP with three regressors, without time effects, for $n=10,000$ and $T=4,5,6$%
}
\label{fig:fe_tmg_k4_base_1}\vspace{-6mm}
\par
\begin{center}
\includegraphics[scale=0.28]{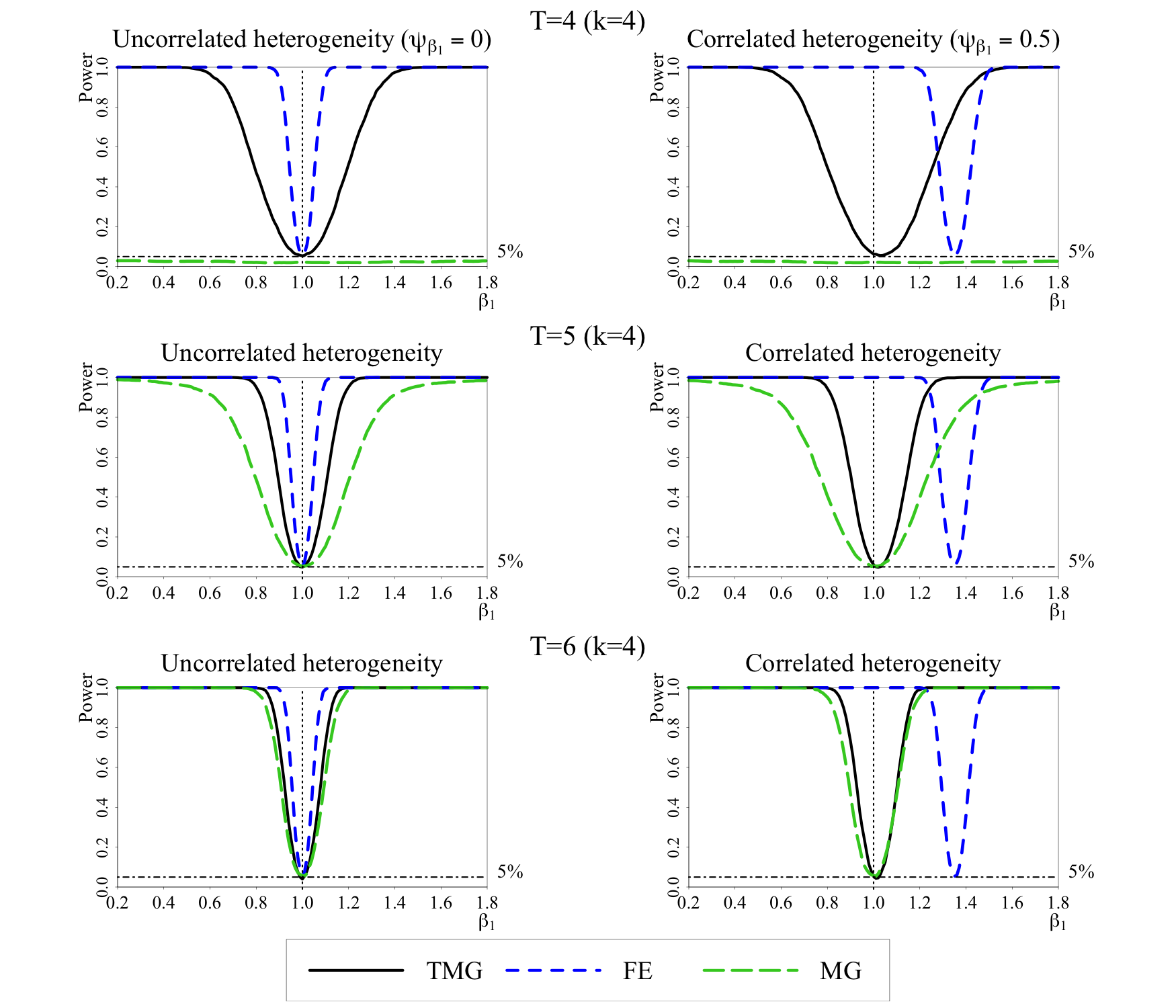}
\end{center}
\par
\vspace{-2mm} 
\begin{spacing}{1}
{\footnotesize
Notes: See footnotes to Table \ref{tab:T_d1_c12_chi2_tex0_k4}. }
\end{spacing}
\end{figure}

\begin{figure}[h!]
\caption{Empirical power functions for FE, MG and TMG estimators of $\protect%
\beta_{01}$ $(E(\protect\beta_{i1}) = \protect\beta_{01}=1)$ in the baseline
DGP with three regressors, without time effects, for $n=10,000$ and $T=7,8,9$%
}
\label{fig:fe_tmg_k4_base_2}\vspace{-6mm}
\par
\begin{center}
\includegraphics[scale=0.28]{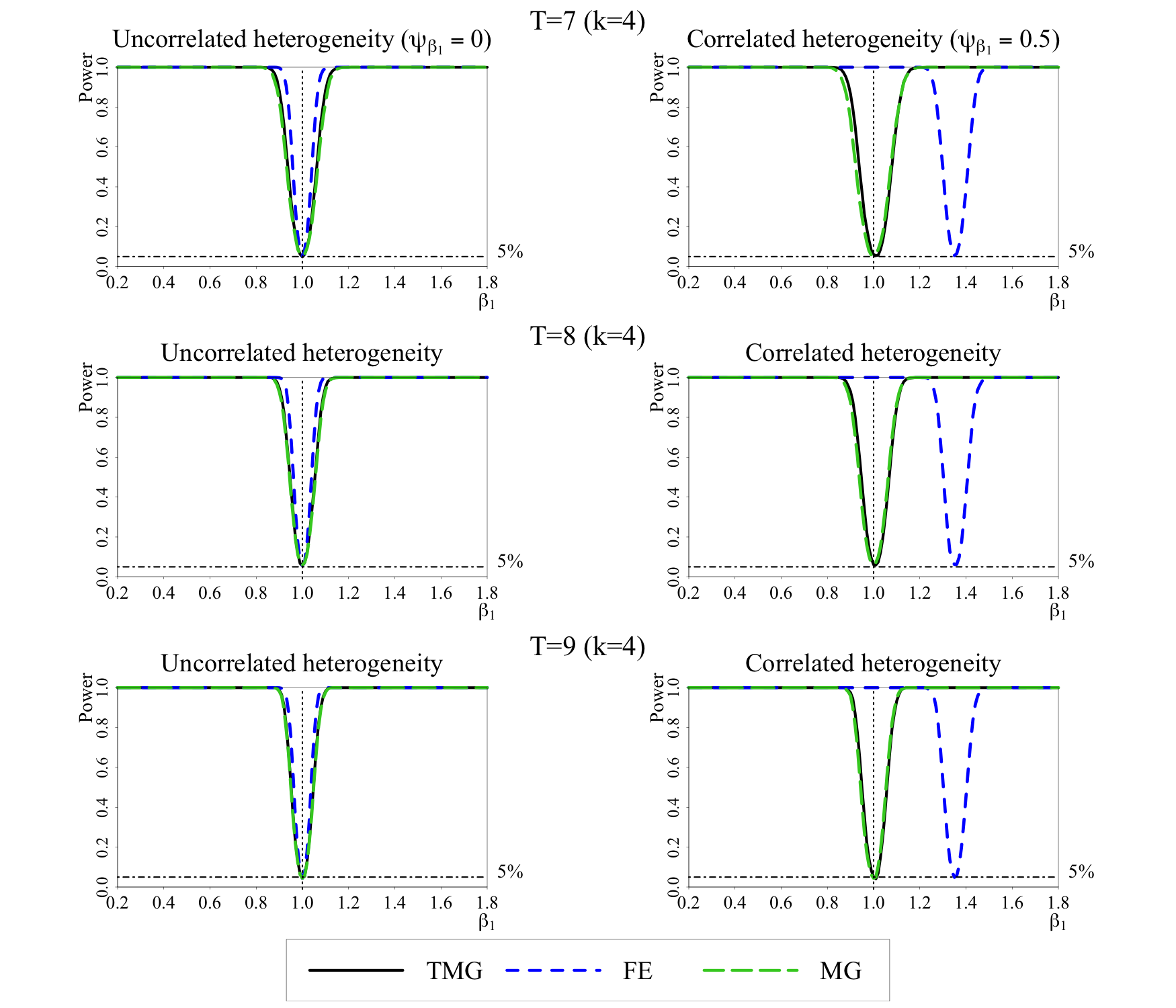}
\end{center}
\par
\vspace{-2mm} 
\begin{spacing}{1}
{\footnotesize
Notes: See footnotes to Table \ref{tab:T_d1_c12_chi2_tex0_k4}. }
\end{spacing}
\end{figure}

\newpage \clearpage

\begin{table}[h]
\caption{Bias, RMSE and size of TMG and GP estimators of $\protect\beta_{01}$
$(E(\protect\beta_{i1})=\protect\beta_{01}=1)$ in the baseline DGP with two
regressors, without time effects, but with correlated heterogeneity, $%
\protect\psi _{\protect\beta_{1}}=0.5$}
\label{tab:beta_d1_c2_chi2_tex0_b_k3}\vspace{-5mm}
\par
\begin{center}
\scalebox{0.8}{
\begin{tabular}{rrrrrrrrrrrr}
\hline\hline
 & \multicolumn{2}{c}{$\hat{\pi}$ $(\times 100)$} &  & \multicolumn{2}{c}{Bias} &  & \multicolumn{2}{c}{RMSE} &  & \multicolumn{2}{c}{Size $(\times 100)$} \\ \cline{2-3} \cline{5-6} \cline{8-9} \cline{11-12}
$T$ & TMG   & GP   &  & TMG    & GP     &  & TMG   & GP    &  & TMG & GP  \\ \hline
  & \multicolumn{11}{c}{$n=1,000$}                                        \\ \hline
3 & 41.8 & 5.2 &  & 0.040 & 0.021 &  & 0.317 & 0.779 &  & 5.0 & 5.1 \\
4 & 25.3 & 4.0 &  & 0.027 & 0.007 &  & 0.183 & 0.227 &  & 5.7 & 5.6 \\
5 & 16.9 & 1.2 &  & 0.016 & 0.001 &  & 0.135 & 0.150 &  & 4.9 & 4.6 \\
6 & 12.0 & 0.4 &  & 0.013 & 0.003 &  & 0.110 & 0.117 &  & 5.2 & 4.9 \\
7 & 8.9 & 0.1 &  & 0.007 & 0.000 &  & 0.095 & 0.099 &  & 4.8 & 4.4 \\
8 & 6.7 & 0.0 &  & 0.004 & -0.003 &  & 0.087 & 0.089 &  & 4.9 & 5.2 \\
[2mm]  & \multicolumn{11}{c}{$n=2,000$}                                        \\ \hline
3 & 38.3 & 4.2 &  & 0.030 & 0.002 &  & 0.236 & 0.606 &  & 5.7 & 4.7 \\
4 & 21.6 & 2.7 &  & 0.021 & 0.005 &  & 0.134 & 0.169 &  & 5.7 & 5.8 \\
5 & 13.5 & 0.6 &  & 0.012 & 0.001 &  & 0.097 & 0.112 &  & 5.2 & 6.2 \\
6 & 9.0 & 0.2 &  & 0.007 & 0.000 &  & 0.077 & 0.083 &  & 4.2 & 4.9 \\
7 & 6.3 & 0.0 &  & 0.007 & 0.001 &  & 0.070 & 0.072 &  & 5.2 & 5.8 \\
8 & 4.5 & 0.0 &  & 0.005 & 0.001 &  & 0.061 & 0.062 &  & 4.5 & 4.6 \\
[2mm]  & \multicolumn{11}{c}{$n=5,000$}                                        \\ \hline
3 & 34.1 & 3.1 &  & 0.026 & 0.008 &  & 0.153 & 0.440 &  & 4.5 & 5.1 \\
4 & 17.3 & 1.5 &  & 0.015 & 0.001 &  & 0.086 & 0.112 &  & 5.4 & 5.0 \\
5 & 9.9 & 0.3 &  & 0.008 & -0.001 &  & 0.063 & 0.071 &  & 5.0 & 4.8 \\
6 & 6.0 & 0.1 &  & 0.003 & -0.002 &  & 0.051 & 0.054 &  & 4.6 & 4.9 \\
7 & 3.8 & 0.0 &  & 0.002 & -0.001 &  & 0.045 & 0.047 &  & 5.3 & 5.4 \\
8 & 2.5 & 0.0 &  & 0.003 & 0.001 &  & 0.039 & 0.040 &  & 5.1 & 5.1 \\
[2mm]  & \multicolumn{11}{c}{$n=10,000$}                                        \\ \hline
3 & 31.1 & 2.5 &  & 0.021 & -0.007 &  & 0.117 & 0.343 &  & 5.4 & 5.2 \\
4 & 14.6 & 1.0 &  & 0.012 & -0.001 &  & 0.064 & 0.085 &  & 5.9 & 5.8 \\
5 & 7.7 & 0.1 &  & 0.007 & 0.000 &  & 0.045 & 0.051 &  & 4.3 & 4.8 \\
6 & 4.3 & 0.0 &  & 0.004 & 0.000 &  & 0.037 & 0.039 &  & 5.3 & 5.2 \\
7 & 2.6 & 0.0 &  & 0.002 & 0.000 &  & 0.031 & 0.031 &  & 5.0 & 5.6 \\
8 & 1.6 & 0.0 &  & 0.000 & -0.001 &  & 0.028 & 0.028 &  & 5.1 & 5.3 \\
\hline\hline
\end{tabular}
}
\end{center}
\par
\vspace{-1mm} 
\begin{spacing}{1}
{\footnotesize 
Notes: 
(i) For details of the baseline DGP without time effects and the TMG estimator, see footnotes to Figure \ref{fig:fe_tmg_k2_base_2}.
(ii) For the GP estimator, see footnote (i) to Table \ref{tab:thresh_d1_c2_chi2_tex0_b}. 
(iii) $\hat{\pi}$ is the simulated fraction of individual estimates being trimmed, defined by (\ref{pin}) in the main paper. } 
\end{spacing}
\end{table}

\begin{table}[h]
\caption{Bias, RMSE and size of TMG and GP estimators of $\protect\beta_{01}$
$(E(\protect\beta_{i1})=\protect\beta_{01}=1)$ in the baseline DGP with
three regressors, without time effects, but with correlated heterogeneity, $%
\protect\psi _{\protect\beta_{1}}=0.5$}
\label{tab:beta_d1_c2_chi2_tex0_b_k4}\vspace{-5mm}
\par
\begin{center}
\scalebox{0.8}{
\begin{tabular}{rrrrrrrrrrrr}
\hline\hline
 & \multicolumn{2}{c}{$\hat{\pi}$ $(\times 100)$} &  & \multicolumn{2}{c}{Bias} &  & \multicolumn{2}{c}{RMSE} &  & \multicolumn{2}{c}{Size $(\times 100)$} \\ \cline{2-3} \cline{5-6} \cline{8-9} \cline{11-12}
$T$ & TMG   & GP   &  & TMG    & GP     &  & TMG   & GP    &  & TMG & GP  \\ \hline
  & \multicolumn{11}{c}{$n=1,000$}                                        \\ \hline
4 & 50.8 & 5.9 &  & 0.054 & 0.043 &  & 0.326 & 0.840 &  & 5.3 & 4.6 \\
5 & 35.5 & 7.9 &  & 0.045 & 0.024 &  & 0.186 & 0.219 &  & 5.8 & 5.0 \\
6 & 27.1 & 3.2 &  & 0.031 & 0.011 &  & 0.138 & 0.151 &  & 5.5 & 5.5 \\
7 & 21.6 & 1.4 &  & 0.018 & 0.000 &  & 0.112 & 0.119 &  & 5.7 & 5.1 \\
8 & 17.9 & 0.7 &  & 0.020 & 0.004 &  & 0.099 & 0.101 &  & 6.0 & 5.1 \\
9 & 15.2 & 0.3 &  & 0.014 & 0.000 &  & 0.085 & 0.087 &  & 4.6 & 4.7 \\
10 & 13.0 & 0.2 &  & 0.012 & 0.000 &  & 0.079 & 0.080 &  & 4.7 & 4.2 \\
[2mm]  & \multicolumn{11}{c}{$n=2,000$}                                        \\ \hline
4 & 47.3 & 4.8 &  & 0.038 & 0.000 &  & 0.236 & 0.660 &  & 4.3 & 5.2 \\
5 & 31.5 & 5.5 &  & 0.032 & 0.008 &  & 0.139 & 0.167 &  & 6.7 & 5.2 \\
6 & 23.0 & 1.9 &  & 0.022 & 0.005 &  & 0.100 & 0.108 &  & 5.3 & 5.0 \\
7 & 17.6 & 0.7 &  & 0.016 & 0.001 &  & 0.080 & 0.085 &  & 5.4 & 5.3 \\
8 & 14.0 & 0.3 &  & 0.011 & -0.002 &  & 0.069 & 0.071 &  & 5.9 & 5.4 \\
9 & 11.5 & 0.1 &  & 0.011 & 0.000 &  & 0.063 & 0.064 &  & 5.9 & 5.8 \\
10 & 9.6 & 0.1 &  & 0.010 & 0.001 &  & 0.056 & 0.057 &  & 5.5 & 5.4 \\
[2mm]  & \multicolumn{11}{c}{$n=5,000$}                                        \\ \hline
4 & 42.9 & 3.5 &  & 0.032 & 0.002 &  & 0.163 & 0.475 &  & 5.9 & 4.2 \\
5 & 26.5 & 3.3 &  & 0.022 & 0.000 &  & 0.089 & 0.109 &  & 5.6 & 4.6 \\
6 & 18.0 & 0.9 &  & 0.017 & 0.002 &  & 0.063 & 0.070 &  & 5.9 & 5.1 \\
7 & 13.1 & 0.3 &  & 0.012 & 0.001 &  & 0.051 & 0.054 &  & 5.1 & 5.2 \\
8 & 9.8 & 0.1 &  & 0.008 & -0.001 &  & 0.043 & 0.044 &  & 4.7 & 4.5 \\
9 & 7.6 & 0.0 &  & 0.008 & 0.001 &  & 0.041 & 0.041 &  & 6.5 & 6.2 \\
10 & 6.0 & 0.0 &  & 0.005 & 0.000 &  & 0.035 & 0.035 &  & 4.9 & 4.7 \\
[2mm]  & \multicolumn{11}{c}{$n=10,000$}                                        \\ \hline
4 & 39.8 & 2.8 &  & 0.030 & -0.013 &  & 0.123 & 0.376 &  & 6.2 & 4.4 \\
5 & 23.1 & 2.2 &  & 0.022 & 0.006 &  & 0.065 & 0.082 &  & 6.1 & 5.7 \\
6 & 14.9 & 0.5 &  & 0.014 & 0.001 &  & 0.046 & 0.050 &  & 5.3 & 4.9 \\
7 & 10.2 & 0.1 &  & 0.009 & 0.000 &  & 0.037 & 0.039 &  & 6.0 & 5.4 \\
8 & 7.3 & 0.0 &  & 0.007 & 0.000 &  & 0.032 & 0.033 &  & 6.1 & 5.6 \\
9 & 5.4 & 0.0 &  & 0.004 & 0.000 &  & 0.027 & 0.027 &  & 5.1 & 4.4 \\
10 & 4.1 & 0.0 &  & 0.003 & 0.000 &  & 0.026 & 0.026 &  & 5.9 & 5.6 \\
\hline\hline
\end{tabular}
}
\end{center}
\par
\vspace{-1mm} 
\begin{spacing}{1}
{\footnotesize 
Notes: 
(i) For details of the baseline DGP without time effects and the TMG estimator, see footnotes to Figure \ref{fig:fe_tmg_k2_base_2}.
(ii) For the GP estimator, see footnote (i) to Table \ref{tab:thresh_d1_c2_chi2_tex0_b}. 
(iii) $\hat{\pi}$ is the simulated fraction of individual estimates being trimmed, defined by (\ref{pin}) in the main paper. } 
\end{spacing}
\end{table}

\newpage \clearpage
\begin{figure}[h!]
\caption{Empirical power functions for TMG, GP, and MG estimators of $%
\protect\beta_{01}$ $(E(\protect\beta_{i1}) = \protect\beta_{01}=1)$ in the
baseline DGP with two regressors, without time effects, but with correlated
heterogeneity, $\protect\psi_{\protect\beta_{1}}=0.5$, for $n=10,000$ and $%
T=3,4,5,6$}
\label{fig:tmg_gp_k3_base}\vspace{-7mm}
\par
\begin{center}
\includegraphics[scale=0.13]{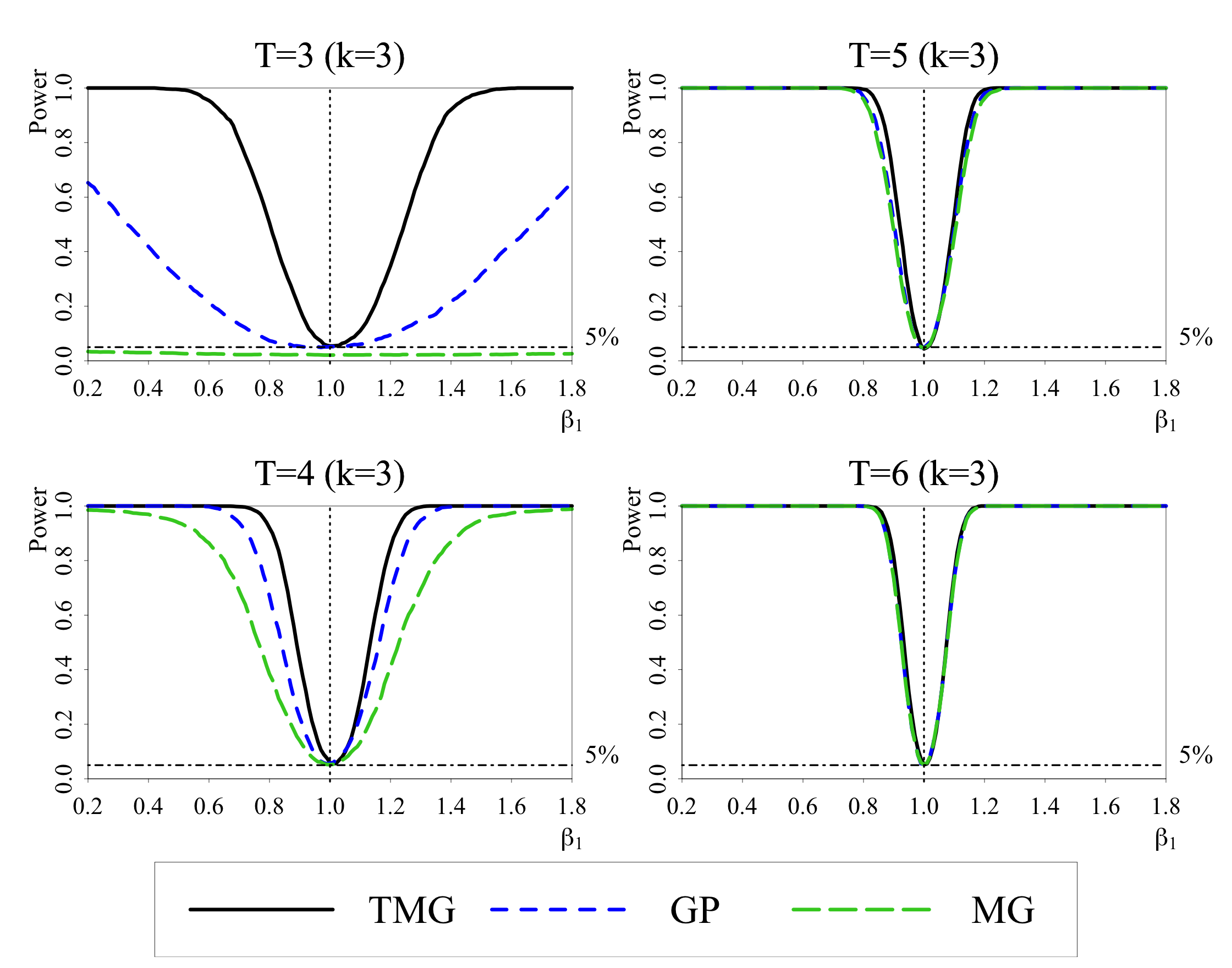}
\end{center}
\par
\vspace{-2mm} 
\begin{spacing}{1}
{\footnotesize
Notes: See footnotes to Table \ref{tab:beta_d1_c2_chi2_tex0_b_k3}. 
}
\end{spacing}
\end{figure}

\begin{figure}[h!]
\caption{Empirical power functions for TMG, GP, and MG estimators of $%
\protect\beta_{01}$ $(E(\protect\beta_{i1}) = \protect\beta_{01}=1)$ in the
baseline DGP with three regressors, without time effects, but with
correlated heterogeneity, $\protect\psi_{\protect\beta_{1}}=0.5$, for $%
n=10,000$ and $T=4,5,6,7$}
\label{fig:tmg_gp_k4_base}\vspace{-7mm}
\par
\begin{center}
\includegraphics[scale=0.13]{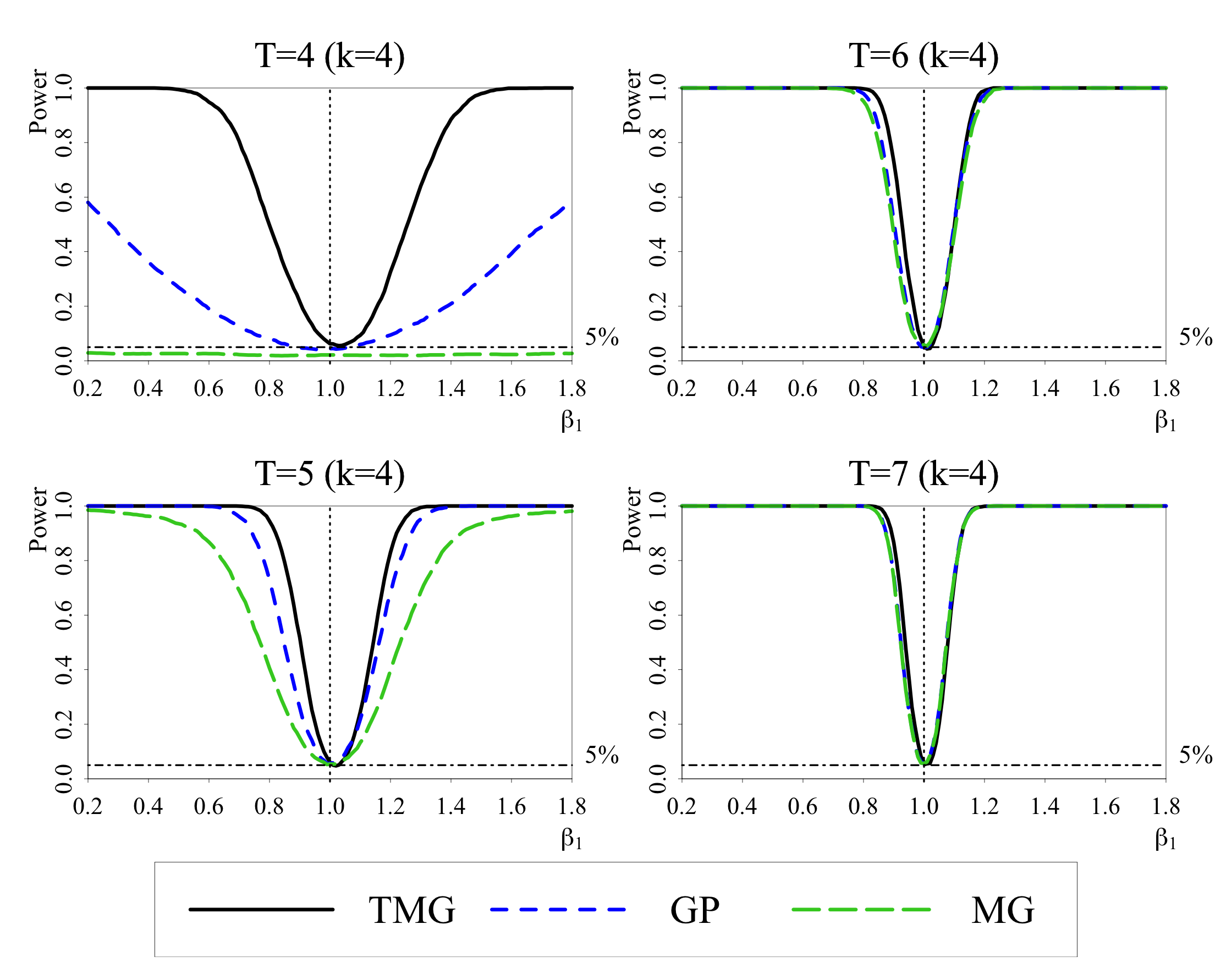}
\end{center}
\par
\vspace{-2mm} 
\begin{spacing}{1}
{\footnotesize
Notes: See footnotes to Table \ref{tab:beta_d1_c2_chi2_tex0_b_k4}. 
}
\end{spacing}
\end{figure}

\newpage \clearpage

\subsection{Monte Carlo evidence in the baseline DGP with time effects\label%
{MCte}}

Tables \ref{tab:te_d1_c12_chi2_tex0}--\ref{tab:te_d1_c2_chi2_tex2_b_k4}
present the results of TWFE, MG-TE and TMG-TE estimators of $\beta_{01}$ in
panel data models with time effects and $k^{\prime}=1$, 2 and 3.

Tables \ref{tab:thetate_d1_c2_chi2_tex0_b_k2}--\ref%
{tab:thetate_d1_c2_chi2_tex0_b_k4} report the results of TMG-TE and GP-TE
estimators of $E(\beta _{i1})=\beta_{01}$ in the baseline DGP with time
effects and $k^{\prime}=1$, 2 and 3. Tables \ref{tab:phi1_d1_c2_chi2_tex0}
and \ref{tab:phi2_d1_c2_chi2_tex0} summarize results of the TMG-TE and GP-TE
estimators for the time effects $\phi _{1}=1$ and $\phi _{2}=2$ with $T=k$,
where the empirical power functions are shown in Figures \ref%
{fig:tmg_gp_te1_base} and \ref{fig:tmg_gp_te2_base}.

\newpage \clearpage
\begin{sidewaystable}
\caption{Bias, RMSE, and size of TWFE, MG-TE, and TMG-TE estimators of $\beta_{01}$ $(E(\beta_{i1}) = \beta_{01}=1)$ in the baseline DGP with one regressor and time effects} 
\label{tab:te_d1_c12_chi2_tex0}
\vspace{-6mm}
\begin{center}
\scalebox{0.52}{
\begin{tabular}{rrcrrrrrrrrrrrrrrrrrrrrrrrrr}
 \hline\hline &  \multicolumn{13}{c}{ Uncorrelated heterogeneity: $\psi_{\beta_{1}} = 0 $, $PR^{2}=0.2$ } &  & \multicolumn{13}{c}{ Correlated heterogeneity: $\psi_{\beta_{1}}=0.5$, $PR^{2}=0.2$ }  \\ \cline{2-14} \cline{16-28}   
& $\hat{\pi}$ $(\times 100)$  & & \multicolumn{3}{c}{Bias}  & & \multicolumn{3}{c}{RMSE} && \multicolumn{3}{c}{Size $(\times 100)$} && $\hat{\pi}$ $(\times 100)$ & & \multicolumn{3}{c}{Bias}  & & \multicolumn{3}{c}{RMSE} && \multicolumn{3}{c}{Size $(\times 100)$} \\ \cline{2-2} \cline{4-6} \cline{8-10} \cline{12-14} \cline{16-16} \cline{18-20} \cline{22-24} \cline{26-28}  
 $T$ & TMG-TE & & TWFE & MG-TE & TMG-TE && TWFE & MG-TE & TMG-TE && TWFE & MG-TE & TMG-TE &&  TMG-TE & & TWFE & MG-TE & TMG-TE && TWFE & MG-TE & TMG-TE && TWFE & MG-TE & TMG-TE    \\ \hline 
&   \multicolumn{27}{c}{$n=1,000$} \\ \hline
2 & 27.50 &  & -0.001 & 3.245 & 0.002 &  & 0.139 & 63.498 & 0.265 &  & 4.6 & 4.0 & 5.3 &  & 27.50 &  & 0.349 & 3.691 & 0.025 &  & 0.389 & 72.308 & 0.300 &  & 54.0 & 3.9 & 5.3 \\
3 & 12.20 &  & -0.001 & 0.013 & 0.004 &  & 0.104 & 0.399 & 0.167 &  & 5.1 & 5.0 & 5.6 &  & 12.20 &  & 0.352 & 0.016 & 0.016 &  & 0.376 & 0.453 & 0.190 &  & 77.8 & 5.0 & 5.9 \\
4 & 6.10 &  & 0.004 & 0.009 & 0.007 &  & 0.087 & 0.149 & 0.120 &  & 5.3 & 3.5 & 4.7 &  & 6.10 &  & 0.357 & 0.010 & 0.013 &  & 0.375 & 0.168 & 0.136 &  & 90.2 & 3.4 & 4.5 \\
5 & 3.20 &  & -0.001 & -0.003 & -0.003 &  & 0.077 & 0.111 & 0.101 &  & 5.3 & 4.3 & 5.3 &  & 3.20 &  & 0.351 & -0.003 & 0.000 &  & 0.365 & 0.125 & 0.113 &  & 95.1 & 3.9 & 5.0 \\
6 & 1.80 &  & 0.000 & -0.001 & -0.001 &  & 0.072 & 0.091 & 0.088 &  & 5.9 & 3.8 & 5.0 &  & 1.80 &  & 0.353 & -0.002 & 0.000 &  & 0.366 & 0.102 & 0.099 &  & 97.9 & 4.0 & 4.9 \\
8 & 0.60 &  & 0.000 & 0.001 & 0.000 &  & 0.063 & 0.073 & 0.072 &  & 5.1 & 3.8 & 5.1 &  & 0.60 &  & 0.351 & 0.000 & 0.000 &  & 0.360 & 0.081 & 0.080 &  & 99.4 & 3.8 & 4.8 \\
10 & 0.20 &  & -0.002 & -0.002 & -0.002 &  & 0.059 & 0.064 & 0.064 &  & 5.8 & 4.3 & 5.6 &  & 0.20 &  & 0.352 & -0.001 & -0.001 &  & 0.361 & 0.071 & 0.070 &  & 99.6 & 3.8 & 4.8 \\
15 & 0.00 &  & 0.001 & 0.002 & 0.002 &  & 0.049 & 0.051 & 0.051 &  & 4.3 & 3.3 & 4.0 &  & 0.00 &  & 0.353 & 0.001 & 0.001 &  & 0.360 & 0.055 & 0.055 &  & 100.0 & 3.8 & 4.4 \\
&   \multicolumn{27}{c}{$n=2,000$} \\ \hline
2 & 24.60 &  & 0.000 & 5.570 & 0.002 &  & 0.102 & 120.569 & 0.198 &  & 5.8 & 2.5 & 5.4 &  & 24.60 &  & 0.354 & 6.320 & 0.023 &  & 0.374 & 136.897 & 0.225 &  & 82.4 & 2.4 & 5.8 \\
3 & 9.80 &  & -0.001 & -0.004 & 0.001 &  & 0.074 & 0.324 & 0.116 &  & 5.6 & 3.6 & 4.7 &  & 9.80 &  & 0.351 & -0.004 & 0.010 &  & 0.363 & 0.367 & 0.132 &  & 96.9 & 3.5 & 4.6 \\
4 & 4.40 &  & 0.000 & -0.003 & 0.000 &  & 0.061 & 0.112 & 0.092 &  & 5.1 & 4.9 & 6.4 &  & 4.40 &  & 0.353 & -0.003 & 0.004 &  & 0.362 & 0.127 & 0.105 &  & 99.7 & 4.9 & 6.2 \\
5 & 2.10 &  & 0.001 & 0.001 & 0.001 &  & 0.055 & 0.076 & 0.071 &  & 5.3 & 4.3 & 5.0 &  & 2.10 &  & 0.354 & 0.001 & 0.003 &  & 0.361 & 0.085 & 0.080 &  & 100.0 & 4.2 & 4.8 \\
6 & 1.10 &  & 0.000 & -0.001 & -0.001 &  & 0.049 & 0.064 & 0.062 &  & 4.4 & 4.3 & 5.4 &  & 1.10 &  & 0.351 & -0.001 & 0.000 &  & 0.357 & 0.071 & 0.069 &  & 100.0 & 4.0 & 5.3 \\
8 & 0.30 &  & 0.000 & -0.001 & -0.001 &  & 0.043 & 0.050 & 0.050 &  & 5.0 & 2.8 & 4.6 &  & 0.30 &  & 0.355 & -0.001 & -0.001 &  & 0.359 & 0.055 & 0.055 &  & 100.0 & 2.6 & 4.0 \\
10 & 0.10 &  & -0.001 & 0.000 & 0.000 &  & 0.040 & 0.044 & 0.044 &  & 5.3 & 4.2 & 5.5 &  & 0.10 &  & 0.353 & 0.000 & 0.000 &  & 0.357 & 0.049 & 0.049 &  & 100.0 & 3.7 & 5.4 \\
15 & 0.00 &  & 0.000 & -0.001 & -0.001 &  & 0.035 & 0.037 & 0.037 &  & 4.7 & 4.8 & 5.9 &  & 0.00 &  & 0.352 & -0.002 & -0.002 &  & 0.355 & 0.040 & 0.040 &  & 100.0 & 4.6 & 5.6 \\
&   \multicolumn{27}{c}{$n=5,000$} \\ \hline
2 & 21.20 &  & 0.002 & 7.434 & 0.001 &  & 0.065 & 222.380 & 0.130 &  & 5.0 & 2.1 & 4.5 &  & 21.20 &  & 0.353 & 8.421 & 0.017 &  & 0.362 & 251.908 & 0.148 &  & 99.6 & 2.1 & 4.4 \\
3 & 7.30 &  & 0.001 & -0.001 & 0.004 &  & 0.047 & 0.192 & 0.078 &  & 5.2 & 4.1 & 4.7 &  & 7.30 &  & 0.355 & -0.001 & 0.011 &  & 0.360 & 0.217 & 0.089 &  & 100.0 & 4.2 & 4.9 \\
4 & 2.90 &  & 0.001 & -0.001 & 0.000 &  & 0.038 & 0.070 & 0.058 &  & 4.6 & 4.9 & 4.9 &  & 2.90 &  & 0.354 & -0.001 & 0.002 &  & 0.357 & 0.079 & 0.066 &  & 100.0 & 4.7 & 4.9 \\
5 & 1.20 &  & 0.000 & 0.000 & 0.000 &  & 0.035 & 0.049 & 0.046 &  & 5.6 & 3.8 & 4.8 &  & 1.20 &  & 0.353 & 0.000 & 0.001 &  & 0.356 & 0.055 & 0.052 &  & 100.0 & 3.6 & 4.5 \\
6 & 0.50 &  & 0.000 & 0.000 & 0.000 &  & 0.032 & 0.042 & 0.041 &  & 4.9 & 4.4 & 5.6 &  & 0.50 &  & 0.353 & 0.000 & 0.001 &  & 0.355 & 0.047 & 0.046 &  & 100.0 & 4.0 & 5.6 \\
8 & 0.10 &  & 0.000 & -0.001 & -0.001 &  & 0.028 & 0.033 & 0.032 &  & 5.4 & 4.2 & 5.4 &  & 0.10 &  & 0.353 & -0.001 & -0.001 &  & 0.355 & 0.036 & 0.036 &  & 100.0 & 4.0 & 5.1 \\
10 & 0.00 &  & 0.000 & 0.000 & 0.000 &  & 0.025 & 0.029 & 0.029 &  & 5.2 & 4.3 & 5.6 &  & 0.00 &  & 0.353 & 0.000 & 0.000 &  & 0.355 & 0.032 & 0.032 &  & 100.0 & 4.0 & 5.6 \\
15 & 0.00 &  & 0.000 & -0.001 & -0.001 &  & 0.023 & 0.023 & 0.023 &  & 5.4 & 4.2 & 5.3 &  & 0.00 &  & 0.353 & -0.001 & -0.001 &  & 0.354 & 0.026 & 0.026 &  & 100.0 & 4.5 & 5.8 \\
&   \multicolumn{27}{c}{$n=10,000$} \\ \hline
2 & 19.00 &  & 0.000 & 5.290 & 0.002 &  & 0.045 & 192.190 & 0.099 &  & 5.1 & 3.0 & 4.9 &  & 19.00 &  & 0.351 & 5.997 & 0.017 &  & 0.356 & 218.008 & 0.113 &  & 100.0 & 3.0 & 4.7 \\
3 & 5.90 &  & 0.000 & 0.000 & 0.000 &  & 0.033 & 0.124 & 0.058 &  & 5.2 & 3.4 & 5.6 &  & 5.90 &  & 0.353 & 0.000 & 0.005 &  & 0.355 & 0.141 & 0.065 &  & 100.0 & 3.4 & 5.3 \\
4 & 2.10 &  & 0.000 & 0.001 & 0.000 &  & 0.028 & 0.047 & 0.041 &  & 4.7 & 4.2 & 4.9 &  & 2.10 &  & 0.352 & 0.001 & 0.002 &  & 0.354 & 0.054 & 0.046 &  & 100.0 & 4.3 & 5.0 \\
5 & 0.80 &  & 0.000 & 0.000 & 0.000 &  & 0.025 & 0.034 & 0.033 &  & 5.3 & 3.5 & 4.2 &  & 0.80 &  & 0.353 & 0.000 & 0.001 &  & 0.354 & 0.038 & 0.037 &  & 100.0 & 3.6 & 4.4 \\
6 & 0.30 &  & 0.001 & 0.002 & 0.002 &  & 0.023 & 0.030 & 0.029 &  & 4.8 & 5.0 & 6.3 &  & 0.30 &  & 0.355 & 0.002 & 0.002 &  & 0.356 & 0.033 & 0.033 &  & 100.0 & 4.9 & 5.9 \\
8 & 0.10 &  & 0.000 & 0.000 & 0.000 &  & 0.020 & 0.023 & 0.023 &  & 4.6 & 3.4 & 4.7 &  & 0.10 &  & 0.354 & 0.001 & 0.001 &  & 0.355 & 0.025 & 0.025 &  & 100.0 & 3.5 & 4.7 \\
10 & 0.00 &  & 0.001 & 0.001 & 0.001 &  & 0.018 & 0.020 & 0.020 &  & 5.2 & 4.3 & 5.4 &  & 0.00 &  & 0.354 & 0.001 & 0.001 &  & 0.355 & 0.022 & 0.022 &  & 100.0 & 3.6 & 5.1 \\
15 & 0.00 &  & -0.001 & -0.001 & -0.001 &  & 0.016 & 0.016 & 0.016 &  & 4.4 & 4.6 & 5.6 &  & 0.00 &  & 0.353 & -0.001 & -0.001 &  & 0.354 & 0.018 & 0.018 &  & 100.0 & 4.0 & 5.9\\
   \hline
\hline
\end{tabular}
}
\end{center}
\vspace{-3mm}
{\footnotesize
Notes: 
(i) In the baseline DGP, the outcome variable is generated as $y_{it}=\alpha_{i} +\phi_{t} + \beta_{i1} x_{1,it} + u_{it}$, and $\psi_{\beta_{1}}$ (the degree of correlated heterogeneity) is defined by (\ref{eta_i}) in the main paper. For further details, see Section \ref{DGP} in the main paper and Section \ref{secMCDGP}. 
(ii) The TWFE estimator is given by (\ref{TWFEhat}), and its variance is estimated by (\ref{CvarTWFE}). The MG-TE estimator uses the same estimator of time fixed effects as the TMG-TE estimator. When $T=k$, the TMG-TE estimator of $\boldsymbol{\beta}_{0}$ and $\boldsymbol{\phi}_{0}$ are given by (\ref{TMG-TE1}) and (\ref{phihat_b}), and their variances are estimated by (\ref{varbetaTE}) and (\ref{VarPhicombined}), respectively, in the mathematical appendix. When $T>k$, the TMG-TE estimator of $\boldsymbol{\beta}_{0}$ is given by (\ref{BetaTE2}), and its variance is estimated by (\ref{VarbetaC}) in the main paper for $T<2k+1$ and (\ref{VarbetaC2}) for $T\geq2k+1$; and the TMG-TE estimator of $\boldsymbol{\phi}_{0}$ is given by (\ref{TE2}) and its variance is estimated by (\ref{Varphi2}) in the main paper.  (iii) The trimming threshold value for the TMG-TE estimator is given by $a_{n}=\bar{d}_{n} n^{-\alpha}$, where $\bar{d}_{n} =\frac{1}{n} \sum_{i}^{n} d_{i}$, $d_{i} =\func{det}(\boldsymbol{X}_{i}^{\prime} \boldsymbol{M}_{T} \boldsymbol{X}_{i})$, $\boldsymbol{X}_{i}=(\boldsymbol{x}_{i1},\boldsymbol{x}_{i2},...,\boldsymbol{x}_{iT})^{\prime}$, and $\boldsymbol{M}_{T} = \boldsymbol{I}_{T} - \boldsymbol{\tau}_{T}\boldsymbol{\tau}_{T}^{\prime}/T$. 
$\alpha$ is set to $1/3$. 
$\hat{\pi}$ is the simulated fraction of individual estimates being trimmed, defined by (\ref{pin}) in the main paper. }
\end{sidewaystable}

\begin{table}[h]
\caption{Bias, RMSE and size of TWFE, MG-TE, and TMG-TE estimators of $%
\protect\beta_{01}$ $(E(\protect\beta_{i1})=\protect\beta_{01}=1)$ in the
panel data model with one regressor, time effects, correlated heterogeneity, 
$\protect\psi _{\protect\beta_{1}}=0.5$, and interactive effects in the
regressor process}
\label{tab:te_d1_c2_chi2_tex2_b}\vspace{-5mm}
\par
\begin{center}
\scalebox{0.75}{
\begin{tabular}{rrrrrrrrrrrrrr}
\hline\hline
   & $\hat{\pi}$ $(\times 100)$ &  & \multicolumn{3}{c}{Bias} &  & \multicolumn{3}{c}{RMSE} &  & \multicolumn{3}{c}{Size $(\times 100)$} \\ \cline{2-2} \cline{4-6} \cline{8-10} \cline{12-14}
$T$  & TMG-TE  &  & TWFE  & MG-TE   & TMG-TE    &  & TWFE     & MG-TE      & TMG-TE   &  & TWFE & MG-TE & TMG-TE     \\ \hline
   & \multicolumn{13}{c}{$n=1,000$}   \\ \hline
2 & 26.60 &  & 0.336 & -1.858 & 0.019 &  & 0.386 & 161.185 & 0.348 &  & 45.5 & 14.1 & 4.2 \\
3 & 11.60 &  & 0.344 & 0.004 & 0.011 &  & 0.370 & 0.454 & 0.207 &  & 69.7 & 3.6 & 5.4 \\
4 & 5.60 &  & 0.338 & -0.005 & 0.002 &  & 0.358 & 0.177 & 0.145 &  & 84.0 & 2.2 & 4.2 \\
5 & 2.90 &  & 0.342 & -0.002 & 0.000 &  & 0.358 & 0.134 & 0.125 &  & 91.1 & 5.0 & 5.2 \\
6 & 1.60 &  & 0.338 & -0.001 & 0.000 &  & 0.351 & 0.106 & 0.103 &  & 95.3 & 3.2 & 5.1 \\
8 & 0.60 &  & 0.320 & 0.001 & 0.001 &  & 0.330 & 0.080 & 0.079 &  & 97.8 & 4.0 & 5.1 \\
10 & 0.20 &  & 0.325 & 0.003 & 0.003 &  & 0.334 & 0.069 & 0.069 &  & 99.0 & 4.5 & 5.1 \\
15 & 0.00 &  & 0.332 & 0.002 & 0.002 &  & 0.340 & 0.057 & 0.057 &  & 99.8 & 4.2 & 5.0 \\
[2mm]   & \multicolumn{13}{c}{$n=2,000$}   \\ \hline
2 & 23.80 &  & 0.334 & -12.835 & 0.013 &  & 0.361 & 523.171 & 0.258 &  & 71.6 & 13.5 & 4.6 \\
3 & 9.40 &  & 0.341 & -0.004 & -0.001 &  & 0.355 & 0.632 & 0.151 &  & 93.2 & 4.5 & 5.3 \\
4 & 4.00 &  & 0.338 & -0.002 & 0.001 &  & 0.348 & 0.126 & 0.105 &  & 98.7 & 3.5 & 5.1 \\
5 & 1.90 &  & 0.347 & 0.002 & 0.003 &  & 0.355 & 0.093 & 0.088 &  & 99.8 & 4.2 & 5.0 \\
6 & 0.90 &  & 0.342 & 0.001 & 0.002 &  & 0.349 & 0.075 & 0.074 &  & 100.0 & 3.2 & 5.6 \\
8 & 0.30 &  & 0.317 & 0.002 & 0.002 &  & 0.323 & 0.057 & 0.056 &  & 99.9 & 4.0 & 4.5 \\
10 & 0.10 &  & 0.322 & 0.000 & 0.000 &  & 0.327 & 0.049 & 0.049 &  & 100.0 & 4.4 & 5.1 \\
15 & 0.00 &  & 0.331 & -0.002 & -0.002 &  & 0.335 & 0.039 & 0.039 &  & 100.0 & 3.5 & 4.0 \\
[2mm]   & \multicolumn{13}{c}{$n=5,000$}   \\ \hline
2 & 20.50 &  & 0.336 & -0.050 & 0.017 &  & 0.347 & 96.380 & 0.179 &  & 97.8 & 13.8 & 4.8 \\
3 & 7.00 &  & 0.345 & -0.002 & 0.008 &  & 0.351 & 0.223 & 0.099 &  & 100.0 & 4.7 & 5.7 \\
4 & 2.60 &  & 0.339 & -0.002 & 0.000 &  & 0.343 & 0.080 & 0.069 &  & 100.0 & 3.2 & 4.8 \\
5 & 1.10 &  & 0.343 & 0.000 & 0.001 &  & 0.346 & 0.060 & 0.057 &  & 100.0 & 4.5 & 5.2 \\
6 & 0.50 &  & 0.340 & -0.001 & -0.001 &  & 0.342 & 0.048 & 0.047 &  & 100.0 & 3.3 & 5.5 \\
8 & 0.10 &  & 0.318 & 0.000 & 0.000 &  & 0.320 & 0.036 & 0.035 &  & 100.0 & 3.5 & 4.7 \\
10 & 0.00 &  & 0.323 & -0.001 & -0.001 &  & 0.325 & 0.030 & 0.030 &  & 100.0 & 4.2 & 4.8 \\
15 & 0.00 &  & 0.331 & 0.000 & 0.000 &  & 0.332 & 0.025 & 0.025 &  & 100.0 & 4.4 & 4.8 \\
[2mm]   & \multicolumn{13}{c}{$n=10,000$}   \\ \hline
2 & 18.30 &  & 0.339 & 11.918 & 0.021 &  & 0.344 & 477.221 & 0.130 &  & 100.0 & 12.6 & 4.0 \\
3 & 5.60 &  & 0.342 & -0.001 & 0.002 &  & 0.345 & 0.158 & 0.073 &  & 100.0 & 4.5 & 5.3 \\
4 & 1.90 &  & 0.341 & -0.001 & 0.001 &  & 0.343 & 0.058 & 0.050 &  & 100.0 & 3.2 & 4.5 \\
5 & 0.70 &  & 0.344 & 0.000 & 0.000 &  & 0.346 & 0.042 & 0.040 &  & 100.0 & 4.6 & 5.4 \\
6 & 0.30 &  & 0.340 & 0.000 & 0.000 &  & 0.341 & 0.034 & 0.033 &  & 100.0 & 3.3 & 5.2 \\
8 & 0.10 &  & 0.318 & 0.000 & 0.000 &  & 0.319 & 0.025 & 0.025 &  & 100.0 & 3.5 & 4.3 \\
10 & 0.00 &  & 0.323 & -0.001 & 0.000 &  & 0.324 & 0.022 & 0.022 &  & 100.0 & 4.9 & 5.4 \\
15 & 0.00 &  & 0.331 & 0.000 & 0.000 &  & 0.332 & 0.018 & 0.018 &  & 100.0 & 4.2 & 5.1 \\
\hline\hline
\end{tabular}
}
\end{center}
\par
\vspace{-1mm} 
\begin{spacing}{1}
{\footnotesize 
Notes: 
(i) The outcome variable is generated as $y_{it}=\alpha_{i} +\phi_{t} + \beta_{i1} x_{1,it} + u_{it}$ with random heteroskedasticity, $x_{1,it}$ are generated with interactive effects without heterogeneous autoregressions, and $\psi_{\beta_{1}}$ is defined by (\ref{eta_i}) in the main paper. For further details, see Section \ref{DGP} in the main paper and Section \ref{secMCDGP}. 
(ii) For the TWFE, MG-TE, and TMG-TE estimators, see footnotes to Table \ref{tab:te_d1_c12_chi2_tex0}. 
$\hat{\pi}$ is the simulated fraction of individual estimates being trimmed, defined by (\ref{pin}) in the main paper. } 
\end{spacing}
\end{table}

\begin{table}[h]
\caption{Bias, RMSE and size of TWFE, MG-TE, and TMG-TE estimators of $%
\protect\beta_{01}$ $(E(\protect\beta_{i1})=\protect\beta_{01}=1)$ in the
baseline DGP with two regressors, time effects, and correlated
heterogeneity, $\protect\psi _{\protect\beta_{1}}=0.5$}
\label{tab:te_d1_c2_chi2_tex2_b_k3}\vspace{-5mm}
\par
\begin{center}
\scalebox{0.75}{
\begin{tabular}{rrrrrrrrrrrrrr}
\hline\hline
   & $\hat{\pi}$ $(\times 100)$ &  & \multicolumn{3}{c}{Bias} &  & \multicolumn{3}{c}{RMSE} &  & \multicolumn{3}{c}{Size $(\times 100)$} \\ \cline{2-2} \cline{4-6} \cline{8-10} \cline{12-14}
$T$  & TMG-TE  &  & TWFE  & MG-TE   & TMG-TE    &  & TWFE     & MG-TE      & TMG-TE   &  & TWFE & MG-TE & TMG-TE     \\ \hline
   & \multicolumn{13}{c}{$n=1,000$}   \\ \hline
3 & 41.70 &  & 0.354 & -93.854 & 0.033 &  & 0.376 & 3786.261 & 0.300 &  & 79.6 & 4.6 & 4.0 \\
4 & 25.30 &  & 0.352 & 0.000 & 0.023 &  & 0.368 & 0.461 & 0.178 &  & 90.6 & 3.5 & 5.0 \\
5 & 16.90 &  & 0.353 & 0.002 & 0.016 &  & 0.368 & 0.190 & 0.137 &  & 94.8 & 4.7 & 5.2 \\
6 & 12.00 &  & 0.352 & -0.004 & 0.010 &  & 0.364 & 0.123 & 0.112 &  & 97.5 & 3.8 & 4.9 \\
7 & 8.80 &  & 0.351 & -0.001 & 0.008 &  & 0.363 & 0.101 & 0.096 &  & 98.2 & 4.2 & 5.9 \\
8 & 6.70 &  & 0.352 & 0.000 & 0.006 &  & 0.362 & 0.087 & 0.084 &  & 99.2 & 3.4 & 4.3 \\
10 & 4.10 &  & 0.351 & -0.001 & 0.002 &  & 0.360 & 0.074 & 0.073 &  & 99.7 & 3.9 & 4.4 \\
15 & 1.50 &  & 0.352 & -0.001 & 0.000 &  & 0.359 & 0.057 & 0.057 &  & 100.0 & 3.2 & 4.6 \\
[2mm]   & \multicolumn{13}{c}{$n=2,000$}   \\ \hline
3 & 38.30 &  & 0.351 & -28.419 & 0.029 &  & 0.363 & 1248.537 & 0.228 &  & 97.3 & 3.6 & 5.0 \\
4 & 21.60 &  & 0.356 & 0.005 & 0.026 &  & 0.365 & 0.301 & 0.131 &  & 99.4 & 4.1 & 5.1 \\
5 & 13.50 &  & 0.355 & -0.001 & 0.013 &  & 0.362 & 0.122 & 0.099 &  & 99.8 & 4.4 & 5.9 \\
6 & 9.00 &  & 0.353 & -0.001 & 0.008 &  & 0.358 & 0.087 & 0.079 &  & 100.0 & 4.1 & 5.1 \\
7 & 6.30 &  & 0.354 & -0.001 & 0.005 &  & 0.359 & 0.073 & 0.070 &  & 100.0 & 4.2 & 5.7 \\
8 & 4.50 &  & 0.355 & 0.001 & 0.005 &  & 0.359 & 0.062 & 0.060 &  & 100.0 & 4.2 & 4.7 \\
10 & 2.40 &  & 0.355 & 0.000 & 0.002 &  & 0.360 & 0.053 & 0.052 &  & 100.0 & 3.9 & 5.2 \\
15 & 0.70 &  & 0.350 & -0.001 & 0.000 &  & 0.354 & 0.040 & 0.040 &  & 100.0 & 3.2 & 4.3 \\
[2mm]   & \multicolumn{13}{c}{$n=5,000$}   \\ \hline
3 & 34.00 &  & 0.353 & -2.588 & 0.030 &  & 0.357 & 148.444 & 0.155 &  & 100.0 & 2.9 & 5.1 \\
4 & 17.30 &  & 0.353 & -0.003 & 0.017 &  & 0.357 & 0.187 & 0.089 &  & 100.0 & 3.5 & 6.2 \\
5 & 9.90 &  & 0.351 & -0.002 & 0.007 &  & 0.354 & 0.077 & 0.063 &  & 100.0 & 4.1 & 5.6 \\
6 & 6.00 &  & 0.353 & 0.000 & 0.006 &  & 0.356 & 0.054 & 0.050 &  & 100.0 & 3.1 & 4.8 \\
7 & 3.80 &  & 0.355 & 0.000 & 0.004 &  & 0.357 & 0.044 & 0.043 &  & 100.0 & 3.4 & 4.6 \\
8 & 2.50 &  & 0.352 & -0.002 & 0.001 &  & 0.354 & 0.040 & 0.039 &  & 100.0 & 4.0 & 5.2 \\
10 & 1.20 &  & 0.354 & 0.000 & 0.001 &  & 0.356 & 0.034 & 0.033 &  & 100.0 & 4.3 & 5.4 \\
15 & 0.20 &  & 0.354 & 0.000 & 0.000 &  & 0.355 & 0.026 & 0.026 &  & 100.0 & 3.1 & 4.3 \\
[2mm]   & \multicolumn{13}{c}{$n=10,000$}   \\ \hline
3 & 31.10 &  & 0.354 & 0.964 & 0.029 &  & 0.357 & 133.213 & 0.117 &  & 100.0 & 2.9 & 6.0 \\
4 & 14.60 &  & 0.353 & -0.005 & 0.012 &  & 0.355 & 0.146 & 0.064 &  & 100.0 & 4.0 & 5.6 \\
5 & 7.70 &  & 0.353 & 0.001 & 0.007 &  & 0.355 & 0.055 & 0.047 &  & 100.0 & 4.3 & 6.0 \\
6 & 4.30 &  & 0.354 & 0.000 & 0.003 &  & 0.355 & 0.038 & 0.036 &  & 100.0 & 3.6 & 4.9 \\
7 & 2.60 &  & 0.354 & 0.000 & 0.002 &  & 0.355 & 0.033 & 0.033 &  & 100.0 & 4.2 & 5.6 \\
8 & 1.60 &  & 0.353 & -0.001 & 0.000 &  & 0.354 & 0.028 & 0.028 &  & 100.0 & 4.0 & 4.9 \\
10 & 0.60 &  & 0.353 & 0.001 & 0.001 &  & 0.354 & 0.023 & 0.023 &  & 100.0 & 4.2 & 5.0 \\
15 & 0.10 &  & 0.354 & 0.000 & 0.000 &  & 0.354 & 0.018 & 0.018 &  & 100.0 & 3.6 & 4.7 \\
\hline\hline
\end{tabular}
}
\end{center}
\par
\vspace{-1mm} 
\begin{spacing}{1}
{\footnotesize 
Notes: 
For details of the baseline DGP with time effects and the TWFE, MG-TE, and TMG-TE estimators, see footnotes to Table \ref{tab:te_d1_c12_chi2_tex0}. $\hat{\pi}$ is the simulated fraction of individual estimates being trimmed, defined by (\ref{pin}) in the main paper. } 
\end{spacing}
\end{table}

\begin{table}[h]
\caption{Bias, RMSE and size of TWFE, MG-TE, and TMG-TE estimators of $%
\protect\beta_{01}$ $(E(\protect\beta_{i1})=\protect\beta_{01}=1)$ in the
baseline DGP with three regressors, time effects, and correlated
heterogeneity, $\protect\psi _{\protect\beta_{1}}=0.5$}
\label{tab:te_d1_c2_chi2_tex2_b_k4}\vspace{-5mm}
\par
\begin{center}
\scalebox{0.75}{
\begin{tabular}{rrrrrrrrrrrrrr}
\hline\hline
   & $\hat{\pi}$ $(\times 100)$ &  & \multicolumn{3}{c}{Bias} &  & \multicolumn{3}{c}{RMSE} &  & \multicolumn{3}{c}{Size $(\times 100)$} \\ \cline{2-2} \cline{4-6} \cline{8-10} \cline{12-14}
$T$  & TMG-TE  &  & TWFE  & MG-TE   & TMG-TE    &  & TWFE     & MG-TE      & TMG-TE   &  & TWFE & MG-TE & TMG-TE     \\ \hline
   & \multicolumn{13}{c}{$n=1,000$}   \\ \hline
4 & 50.70 &  & 0.356 & -1.685 & 0.049 &  & 0.373 & 173.153 & 0.332 &  & 90.3 & 3.9 & 5.0 \\
5 & 35.50 &  & 0.350 & -0.009 & 0.034 &  & 0.364 & 0.448 & 0.189 &  & 94.6 & 3.9 & 6.0 \\
6 & 27.20 &  & 0.354 & -0.002 & 0.024 &  & 0.366 & 0.171 & 0.134 &  & 97.4 & 4.3 & 4.9 \\
7 & 21.60 &  & 0.355 & 0.003 & 0.025 &  & 0.366 & 0.121 & 0.111 &  & 98.6 & 3.6 & 5.0 \\
8 & 17.90 &  & 0.350 & -0.001 & 0.016 &  & 0.360 & 0.103 & 0.098 &  & 99.2 & 4.1 & 5.1 \\
9 & 15.20 &  & 0.352 & -0.001 & 0.014 &  & 0.361 & 0.087 & 0.084 &  & 99.7 & 3.6 & 5.0 \\
10 & 13.00 &  & 0.352 & -0.002 & 0.011 &  & 0.360 & 0.080 & 0.079 &  & 99.7 & 3.9 & 4.9 \\
15 & 7.40 &  & 0.353 & 0.001 & 0.007 &  & 0.360 & 0.060 & 0.060 &  & 100.0 & 4.4 & 5.1 \\
[2mm]   & \multicolumn{13}{c}{$n=2,000$}   \\ \hline
4 & 47.40 &  & 0.352 & -3.605 & 0.040 &  & 0.360 & 121.242 & 0.236 &  & 99.7 & 3.1 & 4.3 \\
5 & 31.50 &  & 0.351 & 0.003 & 0.038 &  & 0.359 & 0.324 & 0.136 &  & 99.9 & 4.0 & 5.8 \\
6 & 22.90 &  & 0.351 & -0.005 & 0.020 &  & 0.357 & 0.121 & 0.096 &  & 100.0 & 4.3 & 5.2 \\
7 & 17.60 &  & 0.354 & 0.003 & 0.019 &  & 0.359 & 0.086 & 0.080 &  & 100.0 & 3.8 & 4.9 \\
8 & 14.10 &  & 0.353 & 0.000 & 0.014 &  & 0.358 & 0.071 & 0.068 &  & 100.0 & 3.9 & 5.6 \\
9 & 11.50 &  & 0.355 & 0.002 & 0.013 &  & 0.360 & 0.063 & 0.063 &  & 100.0 & 3.8 & 6.2 \\
10 & 9.60 &  & 0.354 & 0.001 & 0.009 &  & 0.358 & 0.057 & 0.057 &  & 100.0 & 4.3 & 5.1 \\
15 & 4.70 &  & 0.352 & -0.001 & 0.003 &  & 0.356 & 0.043 & 0.043 &  & 100.0 & 4.6 & 5.3 \\
[2mm]   & \multicolumn{13}{c}{$n=5,000$}   \\ \hline
4 & 43.00 &  & 0.353 & -17.739 & 0.043 &  & 0.356 & 427.544 & 0.166 &  & 100.0 & 3.1 & 5.9 \\
5 & 26.50 &  & 0.353 & 0.008 & 0.027 &  & 0.356 & 0.225 & 0.089 &  & 100.0 & 4.2 & 5.7 \\
6 & 18.00 &  & 0.354 & 0.003 & 0.019 &  & 0.356 & 0.076 & 0.065 &  & 100.0 & 4.0 & 5.6 \\
7 & 13.00 &  & 0.352 & 0.000 & 0.011 &  & 0.355 & 0.057 & 0.053 &  & 100.0 & 4.2 & 5.9 \\
8 & 9.80 &  & 0.355 & 0.002 & 0.011 &  & 0.357 & 0.046 & 0.044 &  & 100.0 & 4.8 & 5.6 \\
9 & 7.60 &  & 0.353 & 0.000 & 0.007 &  & 0.355 & 0.040 & 0.040 &  & 100.0 & 4.8 & 6.2 \\
10 & 6.00 &  & 0.353 & 0.000 & 0.006 &  & 0.355 & 0.036 & 0.036 &  & 100.0 & 4.4 & 5.9 \\
15 & 2.40 &  & 0.354 & 0.000 & 0.001 &  & 0.355 & 0.027 & 0.027 &  & 100.0 & 4.5 & 5.4 \\
[2mm]   & \multicolumn{13}{c}{$n=10,000$}   \\ \hline
4 & 39.80 &  & 0.352 & -0.235 & 0.032 &  & 0.354 & 122.363 & 0.120 &  & 100.0 & 2.6 & 6.3 \\
5 & 23.10 &  & 0.354 & 0.002 & 0.023 &  & 0.356 & 0.148 & 0.066 &  & 100.0 & 4.6 & 6.7 \\
6 & 14.90 &  & 0.354 & -0.001 & 0.014 &  & 0.356 & 0.053 & 0.045 &  & 100.0 & 3.2 & 5.4 \\
7 & 10.20 &  & 0.352 & -0.001 & 0.009 &  & 0.354 & 0.039 & 0.037 &  & 100.0 & 4.4 & 5.9 \\
8 & 7.30 &  & 0.355 & 0.001 & 0.008 &  & 0.355 & 0.033 & 0.032 &  & 100.0 & 4.4 & 6.2 \\
9 & 5.40 &  & 0.353 & 0.000 & 0.005 &  & 0.354 & 0.028 & 0.027 &  & 100.0 & 3.5 & 5.0 \\
10 & 4.10 &  & 0.354 & -0.001 & 0.003 &  & 0.354 & 0.025 & 0.025 &  & 100.0 & 4.2 & 5.8 \\
15 & 1.30 &  & 0.353 & 0.000 & 0.001 &  & 0.354 & 0.019 & 0.019 &  & 100.0 & 4.7 & 5.4 \\
\hline\hline
\end{tabular}
}
\end{center}
\par
\vspace{-1mm} 
\begin{spacing}{1}
{\footnotesize 
Notes: 
For details of the baseline DGP with time effects and the TWFE, MG-TE, and TMG-TE estimators, see footnotes to Table \ref{tab:te_d1_c12_chi2_tex0}. $\hat{\pi}$ is the simulated fraction of individual estimates being trimmed, defined by (\ref{pin}) in the main paper. } 
\end{spacing}
\end{table}

\newpage \clearpage

\begin{table}[h]
\caption{Bias, RMSE and size of TMG-TE and GP-TE estimators of $\protect%
\beta _{01}$ $(E(\protect\beta_{i1})=\protect\beta_{01}=1)$ in the baseline
DGP with one regressor, time effects, and correlated heterogeneity, $\protect%
\psi _{\protect\beta _{1}}=0.5$}
\label{tab:thetate_d1_c2_chi2_tex0_b_k2}\vspace{-5mm}
\par
\begin{center}
\scalebox{0.8}{
\begin{tabular}{rrrrrrrrrrrr}
\hline\hline
 & \multicolumn{2}{c}{$\hat{\pi}$ $(\times 100)$} &  & \multicolumn{2}{c}{Bias} &  & \multicolumn{2}{c}{RMSE} &  & \multicolumn{2}{c}{Size $(\times 100)$} \\ \cline{2-3} \cline{5-6} \cline{8-9} \cline{11-12}
$T$ & TMG-TE   & GP-TE   &  & TMG-TE    & GP-TE     &  & TMG-TE   & GP-TE    &  & TMG-TE & GP-TE  \\ \hline
  & \multicolumn{11}{c}{$n=1,000$}                                        \\ \hline
2 & 27.50 & 4.10 &  & 0.025 & 0.004 &  & 0.300 & 0.689 &  & 5.3 & 5.3 \\
3 & 12.20 & 1.30 &  & 0.016 & 0.010 &  & 0.190 & 0.244 &  & 5.9 & 5.6 \\
4 & 6.10 & 0.20 &  & 0.013 & 0.008 &  & 0.136 & 0.155 &  & 4.5 & 4.3 \\
5 & 3.20 & 0.00 &  & 0.000 & -0.003 &  & 0.113 & 0.122 &  & 5.0 & 4.8 \\
6 & 1.80 & 0.00 &  & 0.000 & -0.002 &  & 0.099 & 0.102 &  & 4.9 & 5.1 \\
[2mm]  & \multicolumn{11}{c}{$n=2,000$}                                        \\ \hline
2 & 24.60 & 3.20 &  & 0.023 & 0.013 &  & 0.225 & 0.534 &  & 5.8 & 4.7 \\
3 & 9.80 & 0.80 &  & 0.010 & 0.002 &  & 0.132 & 0.174 &  & 4.6 & 4.4 \\
4 & 4.40 & 0.10 &  & 0.004 & -0.001 &  & 0.105 & 0.120 &  & 6.2 & 5.9 \\
5 & 2.10 & 0.00 &  & 0.003 & 0.001 &  & 0.080 & 0.085 &  & 4.8 & 5.1 \\
6 & 1.10 & 0.00 &  & 0.000 & -0.001 &  & 0.069 & 0.071 &  & 5.3 & 5.4 \\
[2mm]  & \multicolumn{11}{c}{$n=5,000$}                                        \\ \hline
2 & 21.20 & 2.40 &  & 0.017 & -0.009 &  & 0.148 & 0.390 &  & 4.4 & 5.1 \\
3 & 7.30 & 0.40 &  & 0.011 & 0.007 &  & 0.089 & 0.118 &  & 4.9 & 5.7 \\
4 & 2.90 & 0.00 &  & 0.002 & -0.001 &  & 0.066 & 0.074 &  & 4.9 & 5.1 \\
5 & 1.20 & 0.00 &  & 0.001 & 0.000 &  & 0.052 & 0.054 &  & 4.5 & 4.5 \\
6 & 0.50 & 0.00 &  & 0.001 & 0.000 &  & 0.046 & 0.047 &  & 5.6 & 5.7 \\
[2mm]  & \multicolumn{11}{c}{$n=10,000$}                                        \\ \hline
2 & 19.00 & 1.90 &  & 0.017 & 0.014 &  & 0.113 & 0.309 &  & 4.7 & 5.1 \\
3 & 5.90 & 0.30 &  & 0.005 & 0.000 &  & 0.065 & 0.086 &  & 5.3 & 4.7 \\
4 & 2.10 & 0.00 &  & 0.002 & 0.001 &  & 0.046 & 0.052 &  & 5.0 & 4.8 \\
5 & 0.80 & 0.00 &  & 0.001 & 0.000 &  & 0.037 & 0.038 &  & 4.4 & 4.6 \\
6 & 0.30 & 0.00 &  & 0.002 & 0.002 &  & 0.033 & 0.033 &  & 5.9 & 5.8 \\
\hline\hline
\end{tabular}
}
\end{center}
\par
\vspace{-1mm} 
\begin{spacing}{1}
{\footnotesize 
Notes: 
(i) For details of the baseline DGP with time effects and the TMG-TE estimator, see footnotes to Table \ref{tab:te_d1_c12_chi2_tex0}. 
(ii) The GP-TE estimator of $\boldsymbol{\beta }_{0}$ and $\boldsymbol{\phi}$ is given by equations (25) and (24) on p. 2120, and the variance estimator is given by equation (30) on p. 2126 in \cite{GrahamPowell2012}. For $T=k$, GP compare $d_{i,GP}^{1/2}$ with the bandwidth $h_{n} = C_{GP}n^{-\alpha_{GP}}$, where $d_{i,GP} = \func{det}(\boldsymbol{W}_{i}^{\prime}\boldsymbol{W}_{i})$ with $\boldsymbol{W}_{i} = (\boldsymbol{\tau}_{T}, \boldsymbol{X}_{i})$ and $\boldsymbol{\tau}_{T}$ being a $T \times 1$ vector of ones. $\alpha_{GP}$ is set to 1/3. $C_{GP}=\frac{1}{2}\min \left( \hat{\sigma}_{D},\hat{r}_{D}/1.34\right) $, where $\hat{\sigma}_{D}$ and $\hat{r}_{D}$ are the respective sample standard deviation and interquartile range of $d_{i,GP}^{1/2}$. For $T>k$, $C_{GP}=\left(n^{-1}\sum_{i=1}^{n}d_{i,GP}\right)^{1/2}$. See sub-section \ref{MCcorr} in the main paper for details. 
(iii) $\hat{\pi}$ is the simulated fraction of individual estimates being trimmed, defined by (\ref{pin}) in the main paper. } 
\end{spacing}
\end{table}

\begin{table}[h]
\caption{Bias, RMSE and size of TMG-TE and GP-TE estimators of $\protect%
\beta _{01}$ $(E(\protect\beta_{i1})=\protect\beta_{01}=1)$ in the baseline
DGP with two regressors, time effects, and correlated heterogeneity, $%
\protect\psi _{\protect\beta_{1}}=0.5$}
\label{tab:thetate_d1_c2_chi2_tex0_b_k3}\vspace{-5mm}
\par
\begin{center}
\scalebox{0.8}{
\begin{tabular}{rrrrrrrrrrrr}
\hline\hline
 & \multicolumn{2}{c}{$\hat{\pi}$ $(\times 100)$} &  & \multicolumn{2}{c}{Bias} &  & \multicolumn{2}{c}{RMSE} &  & \multicolumn{2}{c}{Size $(\times 100)$} \\ \cline{2-3} \cline{5-6} \cline{8-9} \cline{11-12}
$T$ & TMG-TE   & GP-TE   &  & TMG-TE    & GP-TE     &  & TMG-TE   & GP-TE    &  & TMG-TE & GP-TE  \\ \hline
  & \multicolumn{11}{c}{$n=1,000$}                                        \\ \hline
3 & 41.70 & 5.20 &  & 0.033 & 0.001 &  & 0.300 & 0.786 &  & 4.0 & 3.0 \\
4 & 25.30 & 4.00 &  & 0.023 & 0.006 &  & 0.178 & 0.224 &  & 5.0 & 5.3 \\
5 & 16.90 & 1.20 &  & 0.016 & 0.003 &  & 0.137 & 0.154 &  & 5.2 & 5.8 \\
6 & 12.00 & 0.40 &  & 0.010 & -0.002 &  & 0.112 & 0.120 &  & 4.9 & 4.8 \\
7 & 8.80 & 0.10 &  & 0.008 & 0.000 &  & 0.096 & 0.101 &  & 5.9 & 5.5 \\
8 & 6.70 & 0.00 &  & 0.006 & 0.000 &  & 0.084 & 0.087 &  & 4.3 & 4.2 \\
[2mm]  & \multicolumn{11}{c}{$n=2,000$}                                        \\ \hline
3 & 38.30 & 4.20 &  & 0.029 & 0.002 &  & 0.228 & 0.628 &  & 5.0 & 4.2 \\
4 & 21.60 & 2.70 &  & 0.026 & 0.011 &  & 0.131 & 0.166 &  & 5.1 & 4.9 \\
5 & 13.50 & 0.60 &  & 0.013 & 0.001 &  & 0.099 & 0.112 &  & 5.9 & 5.1 \\
6 & 9.00 & 0.20 &  & 0.008 & 0.000 &  & 0.079 & 0.085 &  & 5.1 & 5.1 \\
7 & 6.30 & 0.00 &  & 0.005 & -0.001 &  & 0.070 & 0.072 &  & 5.7 & 5.2 \\
8 & 4.50 & 0.00 &  & 0.005 & 0.001 &  & 0.060 & 0.062 &  & 4.7 & 5.3 \\
[2mm]  & \multicolumn{11}{c}{$n=5,000$}                                        \\ \hline
3 & 34.00 & 3.10 &  & 0.030 & 0.014 &  & 0.155 & 0.447 &  & 5.1 & 4.2 \\
4 & 17.30 & 1.50 &  & 0.017 & 0.004 &  & 0.089 & 0.115 &  & 6.2 & 4.7 \\
5 & 9.90 & 0.30 &  & 0.007 & -0.001 &  & 0.063 & 0.072 &  & 5.6 & 5.1 \\
6 & 6.00 & 0.10 &  & 0.006 & 0.001 &  & 0.050 & 0.053 &  & 4.8 & 4.2 \\
7 & 3.80 & 0.00 &  & 0.004 & 0.000 &  & 0.043 & 0.044 &  & 4.6 & 4.0 \\
8 & 2.50 & 0.00 &  & 0.001 & -0.002 &  & 0.039 & 0.040 &  & 5.2 & 4.9 \\
[2mm]  & \multicolumn{11}{c}{$n=10,000$}                                        \\ \hline
3 & 31.10 & 2.50 &  & 0.029 & 0.010 &  & 0.117 & 0.351 &  & 6.0 & 4.0 \\
4 & 14.60 & 1.00 &  & 0.012 & 0.001 &  & 0.064 & 0.083 &  & 5.6 & 5.1 \\
5 & 7.70 & 0.10 &  & 0.007 & 0.000 &  & 0.047 & 0.052 &  & 6.0 & 5.8 \\
6 & 4.30 & 0.00 &  & 0.003 & 0.000 &  & 0.036 & 0.038 &  & 4.9 & 4.5 \\
7 & 2.60 & 0.00 &  & 0.002 & 0.000 &  & 0.033 & 0.033 &  & 5.6 & 5.3 \\
8 & 1.60 & 0.00 &  & 0.000 & -0.001 &  & 0.028 & 0.028 &  & 4.9 & 4.9 \\
\hline\hline
\end{tabular}
}
\end{center}
\par
\vspace{-1mm} 
\begin{spacing}{1}
{\footnotesize 
Notes: 
For details of the baseline DGP with time effects and the TMG-TE estimator, see footnotes to Table \ref{tab:te_d1_c12_chi2_tex0}. For the GP-TE estimator, see footnotes to Table \ref{tab:thetate_d1_c2_chi2_tex0_b_k2}.
$\hat{\pi}$ is the simulated fraction of individual estimates being trimmed, defined by (\ref{pin}) in the main paper. } 
\end{spacing}
\end{table}

\begin{table}[h]
\caption{Bias, RMSE and size of TMG-TE and GP-TE estimators of $\protect%
\beta _{01}$ $(E(\protect\beta_{i1})=\protect\beta_{01}=1)$ in the baseline
DGP with three regressors, time effects, and correlated heterogeneity, $%
\protect\psi _{\protect\beta_{1}}=0.5$}
\label{tab:thetate_d1_c2_chi2_tex0_b_k4}\vspace{-5mm}
\par
\begin{center}
\scalebox{0.8}{
\begin{tabular}{rrrrrrrrrrrr}
\hline\hline
 & \multicolumn{2}{c}{$\hat{\pi}$ $(\times 100)$} &  & \multicolumn{2}{c}{Bias} &  & \multicolumn{2}{c}{RMSE} &  & \multicolumn{2}{c}{Size $(\times 100)$} \\ \cline{2-3} \cline{5-6} \cline{8-9} \cline{11-12}
$T$ & TMG-TE   & GP-TE   &  & TMG-TE    & GP-TE     &  & TMG-TE   & GP-TE    &  & TMG-TE & GP-TE  \\ \hline
  & \multicolumn{11}{c}{$n=1,000$}                                        \\ \hline
4 & 50.70 & 5.90 &  & 0.049 & -0.008 &  & 0.332 & 0.920 &  & 5.0 & 3.5 \\
5 & 35.50 & 7.80 &  & 0.034 & 0.009 &  & 0.189 & 0.224 &  & 6.0 & 5.4 \\
6 & 27.20 & 3.20 &  & 0.024 & 0.004 &  & 0.134 & 0.148 &  & 4.9 & 4.5 \\
7 & 21.60 & 1.40 &  & 0.025 & 0.006 &  & 0.111 & 0.117 &  & 5.0 & 4.5 \\
8 & 17.90 & 0.70 &  & 0.016 & 0.000 &  & 0.098 & 0.102 &  & 5.1 & 5.0 \\
9 & 15.20 & 0.30 &  & 0.014 & -0.001 &  & 0.084 & 0.087 &  & 5.0 & 4.2 \\
10 & 13.00 & 0.20 &  & 0.011 & -0.002 &  & 0.079 & 0.080 &  & 4.9 & 5.4 \\
[2mm]  & \multicolumn{11}{c}{$n=2,000$}                                        \\ \hline
4 & 47.40 & 4.70 &  & 0.040 & 0.001 &  & 0.236 & 0.707 &  & 4.3 & 3.5 \\
5 & 31.50 & 5.50 &  & 0.038 & 0.016 &  & 0.136 & 0.161 &  & 5.8 & 5.1 \\
6 & 22.90 & 1.90 &  & 0.020 & 0.001 &  & 0.096 & 0.108 &  & 5.2 & 5.3 \\
7 & 17.60 & 0.70 &  & 0.019 & 0.004 &  & 0.080 & 0.083 &  & 4.9 & 4.3 \\
8 & 14.10 & 0.30 &  & 0.014 & 0.001 &  & 0.068 & 0.070 &  & 5.6 & 4.9 \\
9 & 11.50 & 0.10 &  & 0.013 & 0.003 &  & 0.063 & 0.063 &  & 6.2 & 5.0 \\
10 & 9.60 & 0.10 &  & 0.009 & 0.001 &  & 0.057 & 0.057 &  & 5.1 & 5.5 \\
[2mm]  & \multicolumn{11}{c}{$n=5,000$}                                        \\ \hline
4 & 43.00 & 3.50 &  & 0.043 & 0.016 &  & 0.166 & 0.474 &  & 5.9 & 3.4 \\
5 & 26.50 & 3.30 &  & 0.027 & 0.007 &  & 0.089 & 0.109 &  & 5.7 & 4.5 \\
6 & 18.00 & 0.90 &  & 0.019 & 0.003 &  & 0.065 & 0.071 &  & 5.6 & 5.3 \\
7 & 13.00 & 0.30 &  & 0.011 & 0.000 &  & 0.053 & 0.055 &  & 5.9 & 5.0 \\
8 & 9.80 & 0.10 &  & 0.011 & 0.002 &  & 0.044 & 0.045 &  & 5.6 & 5.5 \\
9 & 7.60 & 0.00 &  & 0.007 & 0.000 &  & 0.040 & 0.040 &  & 6.2 & 5.6 \\
10 & 6.00 & 0.00 &  & 0.006 & 0.000 &  & 0.036 & 0.036 &  & 5.9 & 5.4 \\
[2mm]  & \multicolumn{11}{c}{$n=10,000$}                                        \\ \hline
4 & 39.80 & 2.80 &  & 0.032 & 0.004 &  & 0.120 & 0.392 &  & 6.3 & 4.2 \\
5 & 23.10 & 2.20 &  & 0.023 & 0.004 &  & 0.066 & 0.080 &  & 6.7 & 4.7 \\
6 & 14.90 & 0.50 &  & 0.014 & 0.001 &  & 0.045 & 0.049 &  & 5.4 & 4.2 \\
7 & 10.20 & 0.10 &  & 0.009 & -0.001 &  & 0.037 & 0.039 &  & 5.9 & 5.8 \\
8 & 7.30 & 0.00 &  & 0.008 & 0.001 &  & 0.032 & 0.033 &  & 6.2 & 5.3 \\
9 & 5.40 & 0.00 &  & 0.005 & 0.000 &  & 0.027 & 0.028 &  & 5.0 & 4.4 \\
10 & 4.10 & 0.00 &  & 0.003 & -0.001 &  & 0.025 & 0.025 &  & 5.8 & 5.0 \\
\hline\hline
\end{tabular}
}
\end{center}
\par
\vspace{-1mm} 
\begin{spacing}{1}
{\footnotesize 
Notes: 
For details of the baseline DGP with time effects and the TMG-TE estimator, see footnotes to Table \ref{tab:te_d1_c12_chi2_tex0}. For the GP-TE estimator, see footnotes to Table \ref{tab:thetate_d1_c2_chi2_tex0_b_k2}.
$\hat{\pi}$ is the simulated fraction of individual estimates being trimmed, defined by (\ref{pin}) in the main paper. } 
\end{spacing}
\end{table}

\begin{table}[h]
\caption{Bias, RMSE and size of TMG-TE and GP-TE estimators of the time
effect $\protect\phi_{1} = 1$ in the baseline DGP with one, two, or three
regressors and correlated heterogeneity, $\protect\psi _{\protect\beta%
_{1}}=0.5$, for $T=k$}
\label{tab:phi1_d1_c2_chi2_tex0}\vspace{-5mm}
\par
\begin{center}
\scalebox{0.8}{
\begin{tabular}{rrrrrrrrrrrrr}
\hline\hline
 &  & \multicolumn{2}{c}{$\hat{\pi}$ $(\times 100)$} &  & \multicolumn{2}{c}{Bias} &  & \multicolumn{2}{c}{RMSE} &  & \multicolumn{2}{c}{Size $(\times 100)$} \\ \cline{3-4} \cline{6-7} \cline{9-10} \cline{12-13}
$T$ & \multicolumn{1}{c}{$n$}   & TMG-TE   & GP-TE   &  & TMG-TE & GP-TE  &  & TMG-TE   & GP-TE    &  & TMG-TE & GP-TE   \\ \hline
  & \multicolumn{12}{c}{One regressor $(k^{\prime}=1)$}                                           \\ \hline
2 & 1,000 & 27.50 & 4.10 &  & 0.000 & 0.007 &  & 0.110 & 0.578 &  & 5.0 & 6.0 \\
2 & 2,000 & 24.60 & 3.20 &  & 0.004 & -0.009 &  & 0.076 & 0.464 &  & 5.3 & 6.9 \\
2 & 5,000 & 21.20 & 2.40 &  & -0.002 & -0.004 &  & 0.049 & 0.337 &  & 5.5 & 5.1 \\
2 & 10,000 & 19.00 & 1.90 &  & 0.000 & -0.001 &  & 0.033 & 0.268 &  & 4.8 & 4.9 \\
[2mm]
  & \multicolumn{12}{c}{Two regressors $(k^{\prime}=2)$}                                           \\ \hline
3 & 1,000 & 41.70 & 5.20 &  & -0.001 & 0.004 &  & 0.126 & 1.390 &  & 5.1 & 12.5 \\
3 & 2,000 & 38.30 & 4.20 &  & 0.001 & 0.013 &  & 0.090 & 1.142 &  & 5.6 & 10.9 \\
3 & 5,000 & 34.00 & 3.10 &  & 0.000 & -0.026 &  & 0.056 & 0.903 &  & 4.5 & 8.3 \\
3 & 10,000 & 31.10 & 2.50 &  & -0.001 & 0.007 &  & 0.042 & 0.694 &  & 6.4 & 7.0 \\
[2mm]
  & \multicolumn{12}{c}{Three regressors $(k^{\prime}=3)$}                                           \\ \hline
4 & 1,000 & 50.70 & 5.90 &  & 0.000 & -0.059 &  & 0.133 & 1.832 &  & 4.2 & 16.5 \\
4 & 2,000 & 47.40 & 4.70 &  & 0.005 & 0.017 &  & 0.096 & 1.625 &  & 4.7 & 13.3 \\
4 & 5,000 & 43.00 & 3.50 &  & -0.001 & -0.026 &  & 0.062 & 1.235 &  & 6.4 & 12.3 \\
4 & 10,000 & 39.80 & 2.80 &  & 0.000 & 0.003 &  & 0.043 & 1.050 &  & 5.8 & 8.6 \\
\hline\hline
\end{tabular}
}
\end{center}
\par
\vspace{-1mm} 
\begin{spacing}{1}
{\footnotesize 
Notes: 
For details of the baseline DGP with time effects and the TMG-TE estimator, see footnotes to Table \ref{tab:te_d1_c12_chi2_tex0}. For the GP-TE estimator, see footnotes to Table \ref{tab:thetate_d1_c2_chi2_tex0_b_k2}.
$\hat{\pi}$ is the simulated fraction of individual estimates being trimmed, defined by (\ref{pin}) in the main paper. } 
\end{spacing}
\end{table}

\begin{table}[h]
\caption{Bias, RMSE and size of TMG-TE and GP-TE estimators of the time
effect $\protect\phi_{2} = 2$ in the baseline DGP with two or three
regressors and correlated heterogeneity, $\protect\psi _{\protect\beta%
_{1}}=0.5$, for $T=k$}
\label{tab:phi2_d1_c2_chi2_tex0}\vspace{-5mm}
\par
\begin{center}
\scalebox{0.8}{
\begin{tabular}{rrrrrrrrrrrrr}
\hline\hline
 &  & \multicolumn{2}{c}{$\hat{\pi}$ $(\times 100)$} &  & \multicolumn{2}{c}{Bias} &  & \multicolumn{2}{c}{RMSE} &  & \multicolumn{2}{c}{Size $(\times 100)$} \\ \cline{3-4} \cline{6-7} \cline{9-10} \cline{12-13}
$T$ & \multicolumn{1}{c}{$n$} & TMG-TE   & GP-TE   &  & TMG-TE    & GP-TE  &  & TMG-TE   & GP-TE    &  & TMG-TE & GP-TE   \\ \hline
  & \multicolumn{12}{c}{Two regressors $(k^{\prime}=2)$}                                           \\ \hline
3 & 1,000 & 41.70 & 5.20 &  & -0.002 & 0.003 &  & 0.126 & 1.317 &  & 5.0 & 10.8 \\
3 & 2,000 & 38.30 & 4.20 &  & 0.002 & 0.004 &  & 0.089 & 1.116 &  & 4.1 & 9.4 \\
3 & 5,000 & 34.00 & 3.10 &  & 0.000 & -0.011 &  & 0.055 & 0.862 &  & 4.2 & 7.5 \\
3 & 10,000 & 31.10 & 2.50 &  & 0.003 & -0.015 &  & 0.040 & 0.714 &  & 5.0 & 6.5 \\
[2mm]
  & \multicolumn{12}{c}{Three regressors $(k^{\prime}=3)$}                                           \\ \hline
4 & 1,000 & 50.70 & 5.90 &  & 0.003 & 0.048 &  & 0.137 & 1.849 &  & 5.1 & 15.8 \\
4 & 2,000 & 47.40 & 4.70 &  & 0.001 & -0.004 &  & 0.093 & 1.539 &  & 4.0 & 12.4 \\
4 & 5,000 & 43.00 & 3.50 &  & 0.001 & 0.043 &  & 0.060 & 1.253 &  & 5.0 & 10.9 \\
4 & 10,000 & 39.80 & 2.80 &  & 0.000 & 0.027 &  & 0.043 & 1.031 &  & 4.7 & 8.5 \\
\hline\hline
\end{tabular}
}
\end{center}
\par
\vspace{-1mm} 
\begin{spacing}{1}
{\footnotesize 
Notes: 
For details of the baseline DGP with time effects and the TMG-TE estimator, see footnotes to Table \ref{tab:te_d1_c12_chi2_tex0}. For the GP-TE estimator, see footnotes to Table \ref{tab:thetate_d1_c2_chi2_tex0_b_k2}.
$\hat{\pi}$ is the simulated fraction of individual estimates being trimmed, defined by (\ref{pin}) in the main paper. } 
\end{spacing}
\end{table}

\clearpage \newpage 
\begin{figure}[h!]
\caption{Empirical power functions for TMG-TE and GP-TE estimators of the
time effect $\protect\phi_{1}=1$ in the baseline DGP with one, two or three
regressors and correlated heterogeneity, $\protect\psi_{\protect\beta%
_{1}}=0.5$, for $T=k = \{2,3,4\}$}
\label{fig:tmg_gp_te1_base}\vspace{-6mm}
\par
\begin{center}
\includegraphics[scale=0.17]{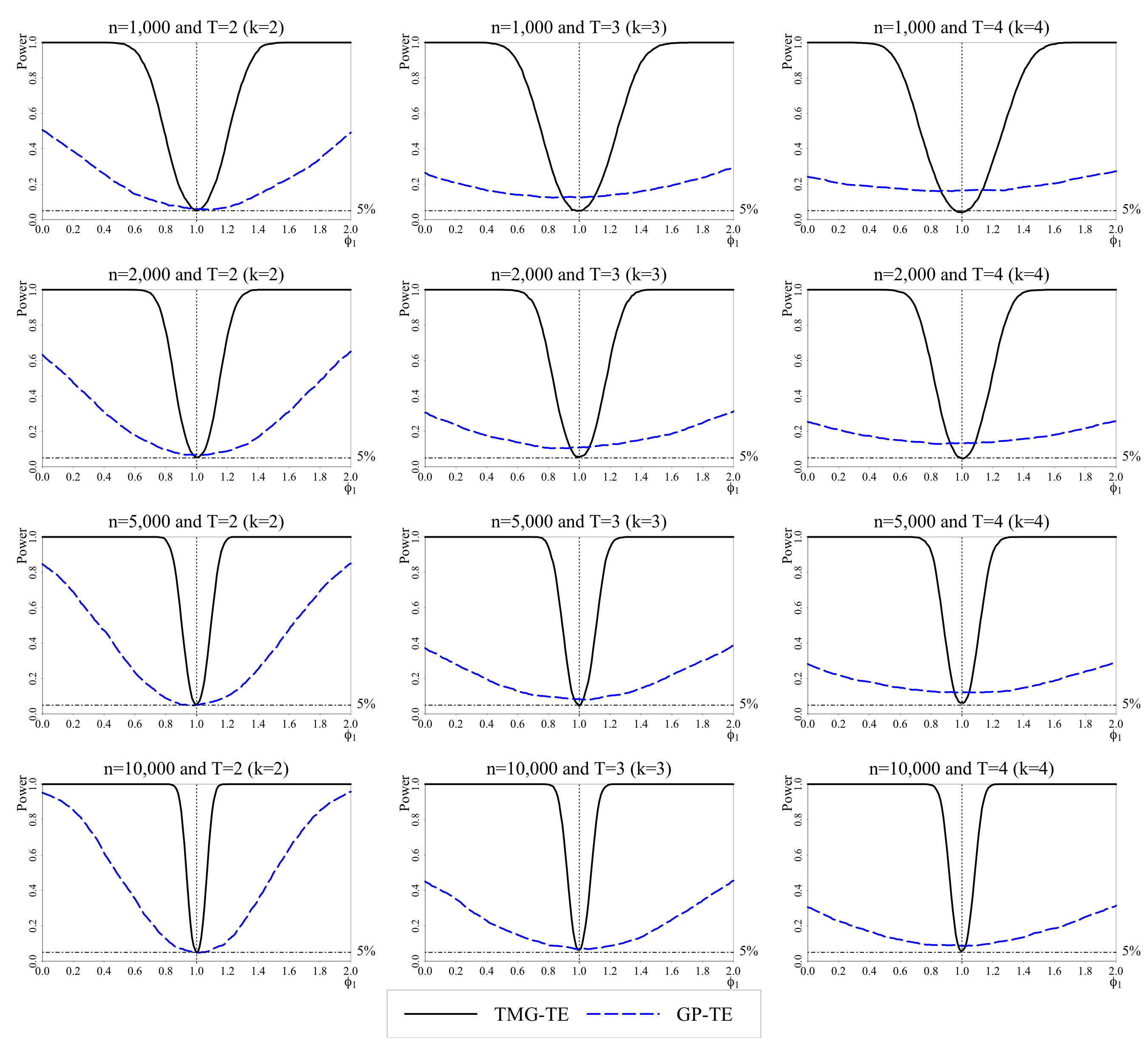}
\end{center}
\par
\vspace{-2mm} 
\begin{spacing}{1}
{\footnotesize
Notes: See footnotes to Table \ref{tab:phi1_d1_c2_chi2_tex0}. }
\end{spacing}
\end{figure}

\begin{figure}[h!]
\caption{Empirical power functions for TMG-TE and GP-TE estimators of the
time effect $\protect\phi_{2}=2$ in the baseline DGP with two or three
regressors and correlated heterogeneity, $\protect\psi_{\protect\beta%
_{1}}=0.5$, for $T=k = \{3,4\}$}
\label{fig:tmg_gp_te2_base}\vspace{-6mm}
\par
\begin{center}
\includegraphics[scale=0.2]{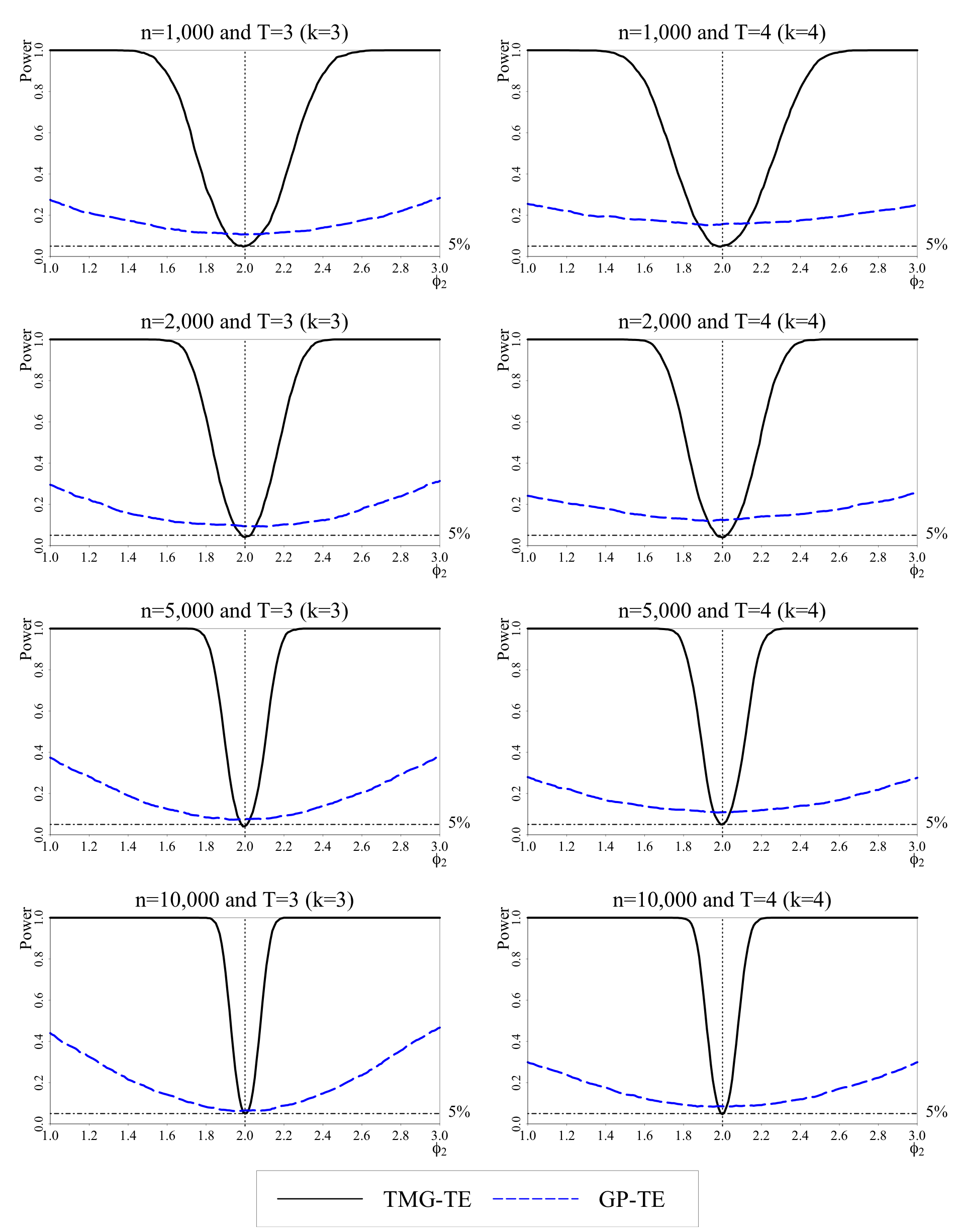}
\end{center}
\par
\vspace{-2mm} 
\begin{spacing}{1}
{\footnotesize
Notes: See footnotes to Table \ref{tab:phi2_d1_c2_chi2_tex0}. }
\end{spacing}
\end{figure}

\newpage \clearpage

\subsection{Monte Carlo evidence for the Hausman test of correlated
heterogeneity with one or more regressors}

\label{MCtest}

The MC evidence on the performance of our proposed test of correlated
heterogeneity in the case of panels with time effects is given in Table \ref%
{tab:Test_te_d1_chi2_tex0}. Depending on how time effects are filtered out
in the case of $T=k$ and $T>k$, the test statistics are given by equations (%
\ref{htetest}) and (\ref{htetest2}), respectively, in Section \ref{TestTE}.
Comparing the results including the time effects in Table \ref%
{tab:Test_te_d1_chi2_tex0} with those without time effects in Table \ref%
{tab:Test_d1_chi2_tex0}, it can be seen that allowing for time effects does
not make much of a difference to the small sample performance of the test.
Increases in $n$ and/or $T$ result in a rapid rise in power, illustrating
the consistent property of the proposed test.

\begin{table}[h]
\caption{Empirical size and power of the Hausman test of correlated
heterogeneity in the baseline DGP with one regressor and time effects}
\label{tab:Test_te_d1_chi2_tex0}\vspace{-6mm}
\par
\begin{center}
\scalebox{0.8}{
\begin{tabular}{lccccccccccrrrr}
\hline\hline 
     &   \multicolumn{9}{c}{Under $H_{0}$: Uncorrelated heterogeneity} &  & \multicolumn{4}{c}{Under $H_{1}$: Correlated heterogeneity}   \\ \cline{2-10} \cline{12-15}
     &   \multicolumn{4}{c}{$\sigma_{\beta_{1}}^{2}=0$}   &  & \multicolumn{4}{c}{$\sigma_{\beta_{1}}^{2}=0.75$, $\psi_{\beta_{1}} =0$} &  & \multicolumn{4}{c}{$\sigma_{\beta_{1}}^{2}=0.75$, $\psi_{\beta_{1}} =0.5$ }  \\ \cline{2-5} \cline{7-10} \cline{12-15} 
$T/n$   & 1,000   & 2,000   & 5,000   & 10,000  &  & 1,000 & 2,000 & 5,000  & 10,000  &  & $\,\,\,\,\,$1,000 & $\,\,\,\,\,$2,000 & $\,\,\,\,\,$5,000 & $\,\,\,\,$10,000 \\ \hline  
2   & 3.7 & 5.6 & 5.2 & 5.5 &  & 6.1 & 5.6 & 4.7 & 5.0 &  & 23.6 & 37.8 & 67.3 & 89.4 \\
3   & 4.9 & 5.6 & 4.7 & 4.3 &  & 5.5 & 5.2 & 4.4 & 5.2 &  & 55.1 & 82.6 & 99.2 & 100.0 \\
4   & 5.3 & 4.6 & 4.6 & 4.9 &  & 4.1 & 5.7 & 5.0 & 4.9 &  & 83.7 & 97.6 & 100.0 & 100.0 \\
5   & 4.9 & 4.7 & 4.7 & 3.7 &  & 5.1 & 4.0 & 5.2 & 5.5 &  & 96.7 & 100.0 & 100.0 & 100.0 \\
6   & 5.4 & 5.2 & 4.6 & 4.8 &  & 4.1 & 5.0 & 5.4 & 6.3 &  & 99.2 & 100.0 & 100.0 & 100.0 \\
8   & 5.1 & 4.8 & 5.2 & 4.7 &  & 4.6 & 4.8 & 4.7 & 3.4 &  & 100.0 & 100.0 & 100.0 & 100.0 \\
\hline
\hline
\end{tabular}}
\end{center}
\par
\vspace{-1mm} 
\begin{spacing}{1}
\footnotesize{
Notes: 
(i) In the baseline DGP for the test, the outcome variable is generated as $y_{it}=\alpha_{i} + \phi_{t}+ \beta_{i1} x_{1,it} + u_{it}$, with $\alpha_{i}$ correlated with $x_{1,it}$ under both the null and alternative hypotheses. 
Time effects are set as $\phi_{t}=t$ for $t=1,2,...,T-1$, and $\phi_{T}=-T(T-1)/2$. 
The error processes for $y_{it}$ and $x_{1,it}$ equations are chi-squared and Gaussian, respectively, and $x_{1, it}$ is generated without autoregressions or interactive effects. For further details, see Section \ref{DGP} in the main paper and Section \ref{secMCDGP}. 
(ii) The null hypothesis is given by (\ref{null}) in the main paper, including the case of homogeneity with $\sigma_{\beta_{1}}^{2}=0$ and the case of uncorrelated heterogeneity with $\psi_{\beta_{1}}=0$ (the degree of correlated heterogeneity defined by (\ref{eta_i}) in the main paper) and $\sigma_{\beta_{1}}^{2}=0.75$. The alternative of correlated heterogeneity is generated with $\psi_{\beta_{1}}=0.5$ and $\sigma_{\beta_{1}}^{2}=0.75$. 
(iii) The test statistics for panels with time effects are based on the difference between the TWFE and TMG-TE estimators given by (\ref{htetest}) when $T = k$ and (\ref{htetest2}) when $T>k$. 
For further details, see Section \ref{TestTE}. 
Size and power are in per cent. }
\end{spacing}
\end{table}

\begin{table}[h]
\caption{Empirical size and power of the Hausman test of correlated
heterogeneity in the baseline DGP with two regressors, without time effects}
\label{tab:Test_d1_chi2_tex0_k3}\vspace{-6mm}
\par
\begin{center}
\scalebox{0.8}{
\begin{tabular}{rccccccccccrrrr}
\hline\hline 
     &   \multicolumn{9}{c}{Under $H_{0}$: Uncorrelated heterogeneity} &  & \multicolumn{4}{c}{Under $H_{1}$: Correlated heterogeneity}   \\ \cline{2-10} \cline{12-15}
     &   \multicolumn{4}{c}{$\sigma_{\beta_{1}}^{2}=0$}   &  & \multicolumn{4}{c}{$\sigma_{\beta_{1}}^{2}=0.75$, $\psi_{\beta_{1}} =0$} &  & \multicolumn{4}{c}{$\sigma_{\beta_{1}}^{2}=0.75$, $\psi_{\beta_{1}} =0.5$ }  \\ \cline{2-5} \cline{7-10} \cline{12-15} 
$T/n$   & 1,000   & 2,000   & 5,000   & 10,000  &  & 1,000 & 2,000 & 5,000  & 10,000  &  & $\,\,\,\,\,$1,000 & $\,\,\,\,\,$2,000 & $\,\,\,\,\,$5,000 & $\,\,\,\,$10,000 \\ \hline  
3   & 5.0 & 4.7 & 4.9 & 6.0 &  & 5.0 & 4.7 & 4.3 & 4.6 &  & 14.6 & 24.2 & 50.4 & 77.2 \\
4   & 4.4 & 5.2 & 4.6 & 5.2 &  & 5.4 & 5.0 & 4.8 & 5.3 &  & 40.8 & 68.8 & 97.5 & 100.0 \\
5   & 4.4 & 4.7 & 5.5 & 5.1 &  & 5.4 & 5.6 & 5.3 & 5.0 &  & 72.5 & 94.8 & 100.0 & 100.0 \\
6   & 5.4 & 5.6 & 4.4 & 4.9 &  & 4.2 & 4.2 & 5.4 & 4.7 &  & 88.7 & 99.4 & 100.0 & 100.0 \\
7   & 5.3 & 5.0 & 4.7 & 5.5 &  & 4.2 & 5.1 & 4.8 & 4.9 &  & 97.4 & 100.0 & 100.0 & 100.0 \\
8   & 5.5 & 4.7 & 4.6 & 5.0 &  & 4.3 & 4.7 & 4.9 & 4.2 &  & 99.3 & 100.0 & 100.0 & 100.0 \\
9   & 4.9 & 4.8 & 5.4 & 4.8 &  & 3.9 & 4.1 & 5.3 & 5.5 &  & 100.0 & 100.0 & 100.0 & 100.0 \\
10   & 5.2 & 4.6 & 5.0 & 4.4 &  & 5.4 & 5.7 & 5.0 & 4.9 &  & 100.0 & 100.0 & 100.0 & 100.0 \\
11   & 5.2 & 4.7 & 5.5 & 5.1 &  & 4.7 & 4.9 & 5.8 & 6.0 &  & 100.0 & 100.0 & 100.0 & 100.0 \\
12   & 6.0 & 4.7 & 4.7 & 4.8 &  & 5.4 & 5.2 & 6.2 & 4.9 &  & 100.0 & 100.0 & 100.0 & 100.0 \\
13   & 5.8 & 5.2 & 5.8 & 4.7 &  & 4.2 & 4.9 & 4.9 & 5.4 &  & 100.0 & 100.0 & 100.0 & 100.0 \\
14   & 5.0 & 4.7 & 5.1 & 4.1 &  & 5.0 & 5.2 & 5.4 & 4.7 &  & 100.0 & 100.0 & 100.0 & 100.0 \\
15   & 4.4 & 4.7 & 4.7 & 6.5 &  & 5.3 & 4.9 & 5.0 & 4.8 &  & 100.0 & 100.0 & 100.0 & 100.0 \\
\hline\hline
\end{tabular}}
\end{center}
\par
\vspace{-3mm} 
\begin{spacing}{1}
\footnotesize{
Notes: 
(i) In the baseline DGP for the test, the outcome variable is generated as $y_{it}=\alpha_{i} + \boldsymbol{\beta}_{i}^{\prime} \boldsymbol{x}_{it} + u_{it}$, with $\alpha_{i}$ correlated with $x_{it,1}$ under both the null and alternative hypotheses. For details of the baseline DGP without time effects, see sub-section \ref{DGP} in the main paper. (ii) The null hypothesis is given by (\ref{null}) in the main paper, including the case of homogeneity with $\sigma_{\beta_{1}}^{2}=0$ and the case of uncorrelated heterogeneity with $\psi_{\beta_{1}}=0$ (the degree of correlated heterogeneity defined by (\ref{eta_i})) and $\sigma_{\beta_{1}}^{2}=0.75$. The alternative of correlated heterogeneity is generated with $\psi_{\beta_{1}}=0.5$ and $\sigma_{\beta_{1}}^{2}=0.75$. 
(iii) The test statistic is calculated based on the difference between FE and TMG estimators, given by (\ref{htest}) in the main paper. 
Size and power are in per cent.}
\end{spacing}
\end{table}

\begin{table}[h]
\caption{Empirical size and power of the Hausman test of correlated
heterogeneity in the baseline DGP with two regressors and time effects}
\label{tab:Test_te_d1_chi2_tex0_k3}\vspace{-6mm}
\par
\begin{center}
\scalebox{0.8}{
\begin{tabular}{rccccccccccrrrr}
\hline\hline 
     &   \multicolumn{9}{c}{Under $H_{0}$: Uncorrelated heterogeneity} &  & \multicolumn{4}{c}{Under $H_{1}$: Correlated heterogeneity}   \\ \cline{2-10} \cline{12-15}
     &   \multicolumn{4}{c}{$\sigma_{\beta_{1}}^{2}=0$}   &  & \multicolumn{4}{c}{$\sigma_{\beta_{1}}^{2}=0.75$, $\psi_{\beta_{1}} =0$} &  & \multicolumn{4}{c}{$\sigma_{\beta_{1}}^{2}=0.75$, $\psi_{\beta_{1}} =0.5$ }  \\ \cline{2-5} \cline{7-10} \cline{12-15} 
$T/n$   & 1,000   & 2,000   & 5,000   & 10,000  &  & 1,000 & 2,000 & 5,000  & 10,000  &  & $\,\,\,\,\,$1,000 & $\,\,\,\,\,$2,000 & $\,\,\,\,\,$5,000 & $\,\,\,\,$10,000 \\ \hline   
3   & 5.5 & 4.3 & 4.6 & 4.9 &  & 4.2 & 5.3 & 5.8 & 5.1 &  & 14.0 & 23.0 & 49.0 & 75.0 \\
4   & 4.1 & 4.5 & 4.4 & 4.7 &  & 5.9 & 4.9 & 4.9 & 5.4 &  & 43.2 & 68.6 & 97.0 & 100.0 \\
5   & 5.3 & 5.6 & 4.6 & 5.2 &  & 5.7 & 4.9 & 4.6 & 5.1 &  & 72.4 & 95.3 & 100.0 & 100.0 \\
6   & 4.8 & 4.5 & 4.7 & 4.6 &  & 4.4 & 4.4 & 4.4 & 4.9 &  & 89.7 & 99.5 & 100.0 & 100.0 \\
7   & 4.3 & 5.2 & 4.9 & 5.1 &  & 5.4 & 5.4 & 4.9 & 4.8 &  & 97.0 & 100.0 & 100.0 & 100.0 \\
8   & 5.0 & 4.9 & 4.4 & 5.2 &  & 4.7 & 4.9 & 4.7 & 4.9 &  & 99.6 & 100.0 & 100.0 & 100.0 \\
9   & 4.9 & 5.2 & 5.0 & 4.9 &  & 4.7 & 4.5 & 5.6 & 5.1 &  & 100.0 & 100.0 & 100.0 & 100.0 \\
10   & 4.9 & 4.9 & 4.2 & 4.9 &  & 4.2 & 4.8 & 4.8 & 5.5 &  & 100.0 & 100.0 & 100.0 & 100.0 \\
11   & 5.2 & 5.2 & 5.0 & 5.4 &  & 4.2 & 4.8 & 5.2 & 5.4 &  & 100.0 & 100.0 & 100.0 & 100.0 \\
12   & 5.6 & 5.4 & 4.7 & 5.4 &  & 4.9 & 5.1 & 4.9 & 4.9 &  & 100.0 & 100.0 & 100.0 & 100.0 \\
13   & 5.2 & 5.5 & 4.3 & 4.2 &  & 5.2 & 5.2 & 4.8 & 5.5 &  & 100.0 & 100.0 & 100.0 & 100.0 \\
14   & 4.8 & 4.9 & 5.0 & 4.6 &  & 4.9 & 4.9 & 5.6 & 4.4 &  & 100.0 & 100.0 & 100.0 & 100.0 \\
15   & 5.8 & 4.4 & 4.9 & 5.4 &  & 4.0 & 5.2 & 4.9 & 4.5 &  & 100.0 & 100.0 & 100.0 & 100.0 \\
\hline
\hline
\end{tabular}}
\end{center}
\par
\vspace{-3mm} 
\begin{spacing}{1}
\footnotesize{
Notes: 
In the baseline DGP for the test, the outcome variable is generated as $y_{it}=\alpha_{i} + \phi_{t}+ \boldsymbol{\beta}_{i}^{\prime} \boldsymbol{x}_{it} + u_{it}$, with $\alpha_{i}$ correlated with $x_{it,1}$ under both the null and alternative hypotheses. For further details of the baseline DGP with time effects, see Section \ref{DGP} in the main paper and Section \ref{secMCDGP}. 
See also footnotes to Table \ref{tab:Test_te_d1_chi2_tex0}. 
Size and power are in per cent. }
\end{spacing}
\end{table}

\begin{table}[h]
\caption{Empirical size and power of the Hausman test of correlated
heterogeneity in the baseline DGP with three regressors, without time
effects }
\label{tab:Test_d1_chi2_tex0_k4}\vspace{-6mm}
\par
\begin{center}
\scalebox{0.8}{
\begin{tabular}{rccccccccccrrrr}
\hline\hline 
     &   \multicolumn{9}{c}{Under $H_{0}$: Uncorrelated heterogeneity} &  & \multicolumn{4}{c}{Under $H_{1}$: Correlated heterogeneity}   \\ \cline{2-10} \cline{12-15}
     &   \multicolumn{4}{c}{$\sigma_{\beta_{1}}^{2}=0$}   &  & \multicolumn{4}{c}{$\sigma_{\beta_{1}}^{2}=0.75$, $\psi_{\beta_{1}} =0$} &  & \multicolumn{4}{c}{$\sigma_{\beta_{1}}^{2}=0.75$, $\psi_{\beta_{1}} =0.5$ }  \\ \cline{2-5} \cline{7-10} \cline{12-15} 
$T/n$   & 1,000   & 2,000   & 5,000   & 10,000  &  & 1,000 & 2,000 & 5,000  & 10,000  &  & $\,\,\,\,\,$1,000 & $\,\,\,\,\,$2,000 & $\,\,\,\,\,$5,000 & $\,\,\,\,$10,000 \\ \hline  
4   & 4.8 & 5.1 & 5.7 & 4.7 &  & 4.4 & 4.5 & 5.1 & 4.7 &  & 10.8 & 19.2 & 40.5 & 66.0 \\
5   & 5.2 & 5.0 & 4.6 & 4.0 &  & 4.6 & 4.8 & 4.6 & 4.7 &  & 30.3 & 56.0 & 94.7 & 99.7 \\
6   & 4.5 & 5.4 & 4.7 & 4.6 &  & 5.2 & 6.0 & 5.6 & 4.8 &  & 58.1 & 89.8 & 100.0 & 100.0 \\
7   & 4.7 & 6.0 & 5.1 & 5.1 &  & 3.6 & 4.7 & 5.2 & 5.0 &  & 81.8 & 98.9 & 100.0 & 100.0 \\
8   & 5.6 & 5.4 & 4.7 & 4.1 &  & 4.8 & 5.1 & 4.2 & 4.7 &  & 92.6 & 100.0 & 100.0 & 100.0 \\
9   & 4.4 & 5.1 & 4.6 & 4.6 &  & 4.8 & 4.7 & 4.8 & 4.4 &  & 97.8 & 100.0 & 100.0 & 100.0 \\
10   & 4.9 & 4.7 & 4.7 & 4.7 &  & 3.6 & 4.6 & 5.5 & 4.6 &  & 99.6 & 100.0 & 100.0 & 100.0 \\
11   & 5.4 & 4.7 & 4.7 & 4.8 &  & 5.2 & 4.5 & 5.5 & 5.3 &  & 99.6 & 100.0 & 100.0 & 100.0 \\
12   & 5.4 & 4.8 & 5.6 & 4.4 &  & 4.5 & 4.6 & 4.6 & 5.4 &  & 99.9 & 100.0 & 100.0 & 100.0 \\
13   & 4.4 & 5.0 & 5.0 & 4.6 &  & 5.5 & 4.2 & 4.8 & 5.0 &  & 100.0 & 100.0 & 100.0 & 100.0 \\
14   & 4.6 & 5.2 & 4.7 & 5.6 &  & 5.1 & 4.8 & 4.4 & 5.1 &  & 100.0 & 100.0 & 100.0 & 100.0 \\
15   & 4.8 & 4.9 & 5.5 & 5.3 &  & 5.1 & 4.0 & 5.2 & 4.4 &  & 100.0 & 100.0 & 100.0 & 100.0 \\
\hline\hline
\end{tabular}}
\end{center}
\par
\vspace{-2mm} 
\begin{spacing}{1}
\footnotesize{
Notes: In the baseline DGP for the test, the outcome variable is generated as $y_{it}=\alpha_{i} + \boldsymbol{\beta}_{i}^{\prime} \boldsymbol{x}_{it} + u_{it}$, with $\alpha_{i}$ correlated with $x_{it,1}$ under both the null and alternative hypotheses. For details of the baseline DGP without time effects, see sub-section \ref{DGP} in the main paper. See also footnotes to Table \ref{tab:Test_d1_chi2_tex0_k3}. 
Size and power are in per cent.}
\end{spacing}
\end{table}

\begin{table}[h]
\caption{Empirical size and power of the Hausman test of correlated
heterogeneity in the baseline DGP with three regressors and time effects}
\label{tab:Test_te_d1_chi2_tex0_k4}\vspace{-6mm}
\par
\begin{center}
\scalebox{0.8}{
\begin{tabular}{rccccccccccrrrr}
\hline\hline 
     &   \multicolumn{9}{c}{Under $H_{0}$: Uncorrelated heterogeneity} &  & \multicolumn{4}{c}{Under $H_{1}$: Correlated heterogeneity}   \\ \cline{2-10} \cline{12-15}
     &   \multicolumn{4}{c}{$\sigma_{\beta_{1}}^{2}=0$}   &  & \multicolumn{4}{c}{$\sigma_{\beta_{1}}^{2}=0.75$, $\psi_{\beta_{1}} =0$} &  & \multicolumn{4}{c}{$\sigma_{\beta_{1}}^{2}=0.75$, $\psi_{\beta_{1}} =0.5$ }  \\ \cline{2-5} \cline{7-10} \cline{12-15} 
$T/n$   & 1,000   & 2,000   & 5,000   & 10,000  &  & 1,000 & 2,000 & 5,000  & 10,000  &  & $\,\,\,\,\,$1,000 & $\,\,\,\,\,$2,000 & $\,\,\,\,\,$5,000 & $\,\,\,\,$10,000 \\ \hline  
4   & 5.6 & 5.5 & 5.2 & 4.3 &  & 5.3 & 4.5 & 4.7 & 4.8 &  & 11.5 & 18.2 & 35.9 & 65.0 \\
5   & 4.2 & 4.7 & 5.6 & 4.7 &  & 5.0 & 4.9 & 5.5 & 4.8 &  & 30.1 & 55.3 & 93.8 & 99.9 \\
6   & 4.7 & 5.2 & 5.3 & 5.0 &  & 5.2 & 4.9 & 5.0 & 5.2 &  & 59.2 & 89.4 & 100.0 & 100.0 \\
7   & 5.1 & 5.0 & 4.6 & 5.4 &  & 5.2 & 5.0 & 5.4 & 4.8 &  & 80.8 & 98.7 & 100.0 & 100.0 \\
8   & 5.5 & 5.2 & 4.9 & 4.7 &  & 5.4 & 5.3 & 5.4 & 5.2 &  & 92.4 & 100.0 & 100.0 & 100.0 \\
9   & 4.4 & 4.4 & 5.0 & 4.7 &  & 5.4 & 5.0 & 5.1 & 4.7 &  & 98.3 & 100.0 & 100.0 & 100.0 \\
10   & 4.1 & 6.4 & 4.7 & 4.9 &  & 5.6 & 4.8 & 5.0 & 5.4 &  & 99.4 & 100.0 & 100.0 & 100.0 \\
11   & 4.9 & 4.4 & 4.4 & 4.9 &  & 5.4 & 4.7 & 5.6 & 4.6 &  & 100.0 & 100.0 & 100.0 & 100.0 \\
12   & 4.7 & 5.1 & 5.0 & 5.2 &  & 4.9 & 5.5 & 5.3 & 4.5 &  & 99.9 & 100.0 & 100.0 & 100.0 \\
13   & 5.2 & 5.5 & 5.1 & 4.7 &  & 5.4 & 4.5 & 5.4 & 4.8 &  & 100.0 & 100.0 & 100.0 & 100.0 \\
14   & 4.8 & 5.4 & 5.0 & 5.8 &  & 4.4 & 4.7 & 4.4 & 5.0 &  & 100.0 & 100.0 & 100.0 & 100.0 \\
15   & 5.0 & 4.8 & 5.1 & 4.2 &  & 4.4 & 4.8 & 4.8 & 5.0 &  & 100.0 & 100.0 & 100.0 & 100.0 \\
\hline
\hline
\end{tabular}}
\end{center}
\par
\vspace{-2mm} 
\begin{spacing}{1}
\footnotesize{
Notes: In the baseline DGP for the test, the outcome variable for the test is generated as $y_{it}=\alpha_{i} + \phi_{t}+ \boldsymbol{\beta}_{i} ^{\prime}\boldsymbol{x}_{it} + u_{it}$, with $\alpha_{i}$ correlated with $x_{it,1}$ under both the null and alternative hypotheses. For further details of the baseline DGP with time effects, see Section \ref{DGP} in the main paper and Section \ref{secMCDGP}. 
See also footnotes to Table \ref{tab:Test_te_d1_chi2_tex0}. 
Size and power are in per cent. }
\end{spacing}
\end{table}

\newpage \clearpage

\section{Empirical supplement}

\label{appt3}

Table \ref{tab:RPS_cal_T3} gives the estimates of $\beta_{0}$ for the
extended panel, 2000--2002 (with $T=3$), both with and without time effects.
As in the case of the 2001--2002 panel shown in Table \ref{tab:RPS_cal_T2}
in the main paper, the FE estimates are larger than the GP and TMG
estimates, but these differences are reduced somewhat, particularly when
time effects are included in the panel regressions. Further, as expected,
increasing $T$ reduces the rate of trimming and brings the GP and TMG
estimators closer to one another. The trimming rate for the GP estimator is
very low indeed (only 1.2 per cent), as compared to around 11 per cent for
the TMG estimator in the case of the 2000--2002 panel. The TMG-TE estimates
of the time effects ($\phi _{2001}$ and $\phi _{2002}$) are highly
statistically significant and positive, capturing the upward trend in the
calorie intake.

\begin{table}[h]
\caption{Alternative estimates of the average effect of household
expenditures on calorie demand in Nicaragua over the period 2000--2002 ($T=3$%
) }
\label{tab:RPS_cal_T3}
\begin{center}
\vspace{-6mm} 
\scalebox{0.9}{
\begin{tabular}{lccccccc}
\hline\hline 
& \multicolumn{3}{c}{Without time effects} &  & \multicolumn{3}{c}{With time effects} \\ \cline{2-4} \cline{6-8}
 & (1) & (2) & (3) &  & (4) & (5) & (6)  \\ 
 & FE & GP & TMG &  & TWFE & GP-TE & TMG-TE  \\ \hline
$\hat{\beta}_{0}$ & 0.6588 & 0.6034 & 0.5900 &  & 0.6968 & 0.6448  & 0.6338 \\
 & (0.0233) & (0.0350) & (0.0284) &  & (0.0211) & (0.0330)  & (0.0259) \\
$\hat{\phi}_{2001}$ & ... & ... & ... &  & 0.0727 & 0.0682  & 0.0682 \\
 & ... & ... & ... &  & (0.0080) & (0.0108)  & (0.0108) \\
$\hat{\phi}_{2002}$ & ... & ... & ... &  & 0.1066 & 0.0954  & 0.0954 \\
 & ... & ... & ... &  & (0.0074) & (0.0101)  & (0.0101) \\
$\hat{\pi}$  $(\times 100)$ & ... & 1.2 & 10.9 &  & ... & 1.2  & 10.9\\
\hline\hline
\end{tabular}
}
\end{center}
\par
\vspace{-2mm} {\footnotesize Notes: (i) The estimates of $\beta _{0}=E(\beta
_{i}) $, $\phi _{2001}$ and $\phi _{2002}$ in the model (\ref{calm1}) in the
main paper are based on a panel of $1,358$ households. (ii) FE and GP
estimators are given by (\ref{fee}) and (\ref{gpe}) in the main paper,
respectively. The TMG estimator is given by (\ref{TMGb}), and its variance
estimator is given by (\ref{varC}) in the main paper. (iii) The TWFE
estimator is given by (\ref{TWFEhat}), and its variance estimator is given
by (\ref{CvarTWFE}). The GP-TE estimator of $\boldsymbol{\beta }_{0}$ and $%
\boldsymbol{\phi}_{0}$ is given by equations (25) and (24) in \cite%
{GrahamPowell2012}. (iv) When $T=3 > k$, the TMG-TE estimator of $%
\boldsymbol{\beta}_{0}$ and $\boldsymbol{\phi}_{0}$ are given by (\ref%
{BetaTE2}) and (\ref{TE2}), and their variance estimators are given by (\ref%
{VarbetaC}) and (\ref{Varphi2}), respectively, in the main paper. The
trimming threshold value for TMG and TMG-TE estimators is given by $a_{n}=%
\bar{d}_{n} n^{-\alpha}$, where $\bar{d}_{n} =\frac{1}{n} \sum_{i}^{n} d_{i}$%
, $d_{i} =\func{det}(\boldsymbol{X}_{i}^{\prime} \boldsymbol{M}_{T} 
\boldsymbol{X}_{i})$, $\boldsymbol{X}_{i}=(\boldsymbol{x}_{i1},\boldsymbol{x}%
_{i2},...,\boldsymbol{x}_{iT})^{\prime}$ and $\boldsymbol{M}_{T} = 
\boldsymbol{I}_{T} - \boldsymbol{\tau}_{T}\boldsymbol{\tau}_{T}^{\prime}/T$. 
$\alpha$ is set to $1/3$. (v) $\hat{\pi}$ is the estimated fraction of
individual estimates being trimmed, given by (\ref{pin}) in the main paper.
\textquotedblleft ..." denotes that the estimation algorithms are not
applicable. The numbers in brackets are standard errors. }
\end{table}

\newpage \clearpage
\noindent {\large \textbf{References} }

\bigskip

\noindent Graham, B. S. and J. L. Powell (2012). \href{https://doi.org/10.3982/ECTA8220}%
{Identification and estimation of average partial effects in
\textquotedblleft irregular\textquotedblright correlated random coefficient
panel data models}. \textit{Econometrica} 80, 2105--2152.

\noindent Hill, B. M. (1975). \href{https://www.jstor.org/stable/2958370}{A
simple general approach to inference about the tail of a distribution}. 
\textit{The Annals of Statistics} 3, 1163--1174.

\noindent Pesaran, M. H. (2015). \href{https://doi.org/10.1093/acprof:oso/9780198736912.001.0001}%
{\textit{Time Series and Panel Data Econometrics}}. Oxford University Press,
Oxford.

\noindent Pesaran, M. H. and L. Yang (2026). Mean group estimation with time
effects. Unpublished manuscript.

\smallskip

\end{document}